\documentclass[a4paper,fleqn]{cas-dc}

\usepackage[numbers]{natbib}
\usepackage{mathtools, amsfonts, amssymb, amsthm, bm}
\usepackage{graphicx}
\usepackage{adjustbox}
\usepackage{booktabs, tabularx, multirow}
\usepackage{enumitem}
\usepackage{ragged2e}
\usepackage{microtype}
\usepackage{placeins}
\usepackage{aliascnt}
\usepackage{float}
\floatstyle{ruled}
\newfloat{algorithm}{tbp}{loa}
\floatname{algorithm}{Algorithm}
\newlist{algsteps}{enumerate}{1}
\setlist[algsteps]{label=\footnotesize\arabic*:,leftmargin=1.9em,
  itemsep=1pt,topsep=2pt}
\usepackage[nameinlink,noabbrev,capitalize]{cleveref}
\usepackage{xr-hyper}

\usepackage{subcaption}
\usepackage{lastpage}
\allowdisplaybreaks[2]
\renewcommand{\arraystretch}{0.96}
\newcommand{\suppmat}{Supplementary Material}
\AtBeginDocument{}

\setlist{nosep}
\numberwithin{equation}{section}

\newcommand{\E}{\mathbb{E}}
\newcommand{\Q}{\mathbb{Q}}

\newcommand{\R}{\mathbb{R}}

\newcommand{\1}{\mathbf{1}}
\newcommand{\dd}{\,\mathrm{d}}
\newcommand{\ee}{\mathrm{e}}

\newcommand{\scrF}{\mathcal F}
\newcommand{\scrG}{\mathcal G}

\newcommand{\Act}{\operatorname{Act}}

\newcommand{\Cov}{\operatorname{Cov}}

\newcommand{\Var}{\operatorname{Var}}
\newcommand{\diag}{\operatorname{diag}}

\newcommand{\Real}{\operatorname{Re}}
\newcommand{\argmin}{\operatorname*{arg\,min}}
\newcommand{\brac}[2]{\langle #1,#2\rangle}
\newcommand{\norm}[1]{\lVert #1\rVert}

\newcommand{\mcarma}{\operatorname{MCARMA}}
\newcommand{\logit}{\operatorname{logit}}

\newcommand{\eps}{\varepsilon}
\theoremstyle{plain}
\newtheorem{theorem}{Theorem}[section]
\newaliascnt{proposition}{theorem}
\newtheorem{proposition}[proposition]{Proposition}
\aliascntresetthe{proposition}
\newtheorem{lemma}[theorem]{Lemma}

\crefname{proposition}{Proposition}{Propositions}
\Crefname{proposition}{Proposition}{Propositions}
\theoremstyle{definition}

\newtheorem{assumption}[theorem]{Assumption}
\crefname{assumption}{Assumption}{Assumptions}
\Crefname{assumption}{Assumption}{Assumptions}
\theoremstyle{remark}
\newtheorem{remark}[theorem]{Remark}
\crefname{algorithm}{Algorithm}{Algorithms}
\Crefname{algorithm}{Algorithm}{Algorithms}

\graphicspath{{empirical/}}

\ExplSyntaxOn

\cs_set:Npn \__first_footerline:
  {
    \group_begin:
      \small
      \sffamily
      \ifnum \theblind > 0\relax
      \else
        \texttt{contact:}~%
        {\rmfamily\itshape
          \texttt{k.chatziandreou@uva.nl},\,%
          \texttt{sven@karbach.org}%
        }%
      \fi
    \group_end:
  }

\ExplSyntaxOff

\begin{document}

\shorttitle{Renewable PPA Pricing and Semi-static Hedging}
\shortauthors{Chatziandreou and Karbach}
\title[mode=title]{Pricing and Semi-static Hedging of Green Pay-as-produced Power Purchase Agreements}

\author[1,2]{Konstantinos Chatziandreou}
\cormark[1]
\ead{k.chatziandreou@uva.nl}

\author[1,2]{Sven Karbach}
\ead{sven@karbach.org}

\affiliation[1]{
  organization={Korteweg-de Vries Institute for Mathematics, University of Amsterdam},
  addressline={Science Park 105--107},
  city={Amsterdam}, postcode={1098 XG}, country={The Netherlands}}

\affiliation[2]{
  organization={Informatics Institute, University of Amsterdam},
  addressline={LAB42, Science Park 900},
  city={Amsterdam}, postcode={1098 XH}, country={The Netherlands}}

\cortext[1]{Corresponding author.}

\begin{abstract}
Pay-as-produced power purchase agreements (PPAs) expose buyers and sellers to the joint risk of power prices and renewable production. This paper develops a theoretical framework for hedging this exposure using a semi-static strategy: liquid futures hedge traded price risk dynamically, while a fixed portfolio of renewable-linked claims targets residual volume and covariance risk. The pricing and hedging decomposition is model-free, whereas the empirical implementation for German wind and solar generation uses a calibrated stochastic model. Conditional on a valuation measure, the fair strike is a production-weighted expected spot price. We show that it decomposes exactly into the baseload forward level, a deterministic production-profile correction, and a stochastic price--volume covariance correction, where the covariance term measures the pricing effect of renewable cannibalisation. The static hedge is selected through a finite-dimensional variance projection onto claims linked to renewable volume, delivery-period average prices, and price-volume covariance. We estimate a Lévy-driven bivariate MCARMA state-space model with state-dependent price spikes using hourly German data for 2023--2024 and apply it to monthly PPAs over the January--December 2025 delivery horizon. The results distinguish deterministic profile risk from stochastic covariance risk and show how sparse static overlays reduce residual exposures that
fixed-volume futures cannot hedge. The selected portfolios also indicate which claim types are most effective for hedging residual renewable shape risk.
\end{abstract}

\begin{keywords}
power purchase agreements \sep volumetric risk \sep cannibalisation \sep
capture price \sep semi-static hedging \sep MCARMA \sep wind power \sep solar power \sep sparse hedging

\vspace{27mm}

\sep JEL: Q41, Q42, G13, C58
\end{keywords}

\maketitle
\hypersetup{
  pdftitle={Pricing and Semi-static Hedging of Green Pay-as-produced Power Purchase Agreements},
  pdfauthor={Konstantinos Chatziandreou and Sven Karbach},
  pdfkeywords={power purchase agreements, volumetric risk, cannibalisation, capture price, semi-static hedging, MCARMA, wind power, solar power}
}

\section{Introduction}\label{sec:intro}

The expansion of wind and solar generation has changed the risk profile of
electricity contracting. Renewable output is stochastic, strongly seasonal and
only partly forecastable; electricity prices are non-storable, mean-reverting,
seasonal and prone to spikes. A green pay-as-produced power purchase agreement
(PPA) combines both variables in a single contract: the buyer receives the
realized renewable production and pays a fixed strike per MWh, while the seller
receives the strike and gives up the spot or market-indexed revenue. The
financial exposure is therefore neither a pure price exposure nor a pure volume
exposure but a product of both, and the contract cannot be reduced to a strip of
power forwards unless production is deterministic.

The price-volume interaction follows directly from the payoff. Let $S_t$ denote the spot power
price and $q_t$ the renewable production rate. Over a delivery horizon $[0,T]$,
the floating-minus-fixed PPA payoff is
\begin{equation}\label{eq:intro-ppa-payoff}
  H^K=\int_0^T q_u(S_u-K)\dd u ,
\end{equation}
and conditioning gives the basic pricing identity
\begin{equation}\label{eq:intro-covariance-identity}
  \E_t[q_uS_u]=\E_t[q_u]\,\E_t[S_u]+\Cov_t(q_u,S_u).
\end{equation}
A PPA is thus a price--volume covariance claim. When high renewable output tends to
occur during low-price periods, the covariance term is negative and the
production-weighted capture price lies below the baseload price. This covariance
term is the pricing counterpart of renewable cannibalisation, whose economics are
documented by \citet{Hirth2013}: rising renewable penetration depresses the
market value of renewable output relative to the baseload price. Both wind and
solar exhibit the effect, but through different channels. Wind output exhibits persistent
negative price-volume covariance because weather systems can raise generation and
depress prices for several days. Solar combines a pronounced deterministic intraday profile
with stochastic cloud-cover shortfall risk: generation is concentrated around
midday, when high system-wide solar output depresses wholesale prices. We therefore
separate deterministic profile effects from
stochastic covariance effects using the same
definitions for wind and solar.

The paper first formulates the green-PPA hedging
problem by distinguishing risks that
can be hedged dynamically with liquid instruments from risks that require
a fixed portfolio of structured claims. The first part gives
a model-free decomposition of pricing and
hedging. We define the PPA payoff, delivery-period power futures, production-index futures and auxiliary over-the-counter (OTC)
claims, all defined on a common filtered probability space. 

Under a chosen valuation measure, the PPA fair
strike decomposes exactly into a level term, a profile term and a covariance
term. The conditional price, production, and price-production moments also determine the hedge. Because renewable volume is
not spanned by exchange-traded futures, the market is incomplete, and we
use quadratic hedging as the optimality criterion: variance-optimal hedging in the sense of
F\"ollmer--Sondermann local risk minimization under a martingale measure, built
on the Galtchouk--Kunita--Watanabe (GKW) decomposition
\citep{Galchouk1976,KunitaWatanabe1967,FollmerSondermann1985,Schweizer1992,Schweizer1994,Pham2000,CernyKallsen2007}.
Following the semi-static literature of
\citet{SemistaticHedging,semistaticsparse,chatziandreou2026semistaticvarianceoptimalhedgingcovariance}, every target and auxiliary claim is
first projected dynamically onto the gains of the traded futures stack, and the static portfolio
is then chosen by projecting the remaining orthogonal residual onto the residuals of
the auxiliary claims. The delivery-period feature of power futures is handled explicitly:
in the convention of \citet{BenthDetering2015}, trading in a delivery future
stops at the start of its delivery period and the last position is held to
settlement, which creates a continuous--discrete quadratic hedging problem
analogous to restricted hedging of Asian options; an idealized during-delivery
convention extends rebalancing through delivery. 

The second part estimates these components for German wind and
solar PPAs. The spot price is written as an additive
deterministic seasonal level, a continuous mean-reverting stochastic component
and a mean-reverting spike component whose negative-jump intensity increases
with renewable penetration. The bivariate continuous state formed by the
transformed renewable variable and the deseasonalized spot residual follows a
L\'evy-driven multivariate continuous-time autoregressive moving-average
($\mcarma$) process
\citep{Brockwell2001CARMA,MarquardtStelzer2007,SchlemmStelzer2012}. The MCARMA
specification differs from the renewable-covariate models of
\citet{RowinskaVeraartGruet2021} in three respects relevant to PPA
pricing: it preserves exact continuous-time pricing for delivery-period claims, it admits
exact sampled-state quasi-maximum-likelihood estimation
\citep{SchlemmStelzerQMLE2012}, and it allows recovery of the driving L\'evy
increments from the fitted state \citep{BrockwellSchlemm2013}, so that a
heavy-tailed multivariate normal-inverse Gaussian driver can be estimated and
used for simulation.

\paragraph{Contribution.}
This paper makes four main contributions. First, it derives an exact
decomposition of the fair strike of a pay-as-produced PPA into the baseload
forward level, a deterministic renewable-profile correction, and a stochastic
price--volume covariance correction. This decomposition shows that the
covariance correction is the model-free pricing term associated with renewable
cannibalisation. Second, it formulates PPA hedging as a two-stage variance-optimal procedure that
respects delivery-period trading restrictions: the contract is first projected onto dynamically
traded delivery-period futures, and the residual that remains after dynamic futures hedging
is then projected onto a static portfolio of auxiliary claims. Explicit
variance-optimal hedge ratios and residual variances are derived under both the
pre-delivery and the idealised during-delivery trading conventions. Third, it
shows that power--renewable orthant quantos and capture-spread options provide
a finite set of payoffs that can hedge covariance exposure remaining after
dynamic futures hedging. A sufficient-condition result also
shows how the signs of the optimal static weights can inform the design of OTC
hedging products. Fourth, it applies the same valuation, calibration, and
hedging procedure to wind and solar PPAs using hourly German data, thereby
distinguishing technology-specific deterministic profile risk from residual
stochastic price--volume covariance risk using a common set of definitions and
assumptions.
Together, the results show how
covariance-sensitive OTC claims can be selected to hedge the residual risk left after
futures hedging. Their payoff structures indicate
which price--volume states a renewable shape-risk
contract should cover.

The remainder of the paper is organised as follows.
\Cref{sec:literature} reviews the related literature.
\Cref{sec:model-free-pricing} introduces the market objects and the exact
fair-strike decomposition. \Cref{subsec:empirical-analysis} documents
the empirical price-renewable dependence, capture prices
and factors, state-dependent negative spikes, and
the relation between spot, futures, and observed PPA
prices. These findings motivate the model
and hedge design.
\Cref{sec:model-free-hedging} develops the dynamic and semi-static
variance-optimal hedging theory. \Cref{sec:model} builds a joint renewable-price MCARMA model with
state-dependent spikes using the state variables introduced
in \cref{subsec:empirical-analysis}. \Cref{sec:calibration} presents the estimation
methodology. \Cref{sec:empirical-design} describes the simulation, hedging-instrument and
backtesting design used in the empirical study.

\section{Related literature}\label{sec:literature}

\paragraph{Renewable cannibalisation and electricity price formation:}

The market value of variable renewable generation and the associated
cannibalisation effect are documented by \citet{Hirth2013}. Empirical evidence
that wind and solar production affect wholesale price levels, volatility and
spikes is provided by \citet{Ketterer2014} for German wind, \citet{Clo2015} for
Italian wind and solar, and \citet{Rintamaki2017} for Denmark and Germany.
Equilibrium electricity-forward pricing and the role of demand, variance and
skewness in forward premia go back to \citet{BessembinderLemmon2002} and
\citet{LongstaffWang2004}. Our decomposition expresses this mechanism in
the PPA fair-strike formula: cannibalisation enters the fair strike
through the conditional covariance between production and spot prices, while
deterministic profile effects enter through the alignment of expected production
with the forward curve.

\paragraph{Price--volume risk and renewable PPA hedging:}
Quantity risk is a central source of incompleteness in power and energy markets.
Static and functional hedging of joint price--quantity exposure with standard
power options is studied by \citet{Oum2006,Oum2010}; \citet{COULON2013976}
develop a structural model for hedging load and price risk in the Texas market,
and \citet{RONCORONI2017415} analyse size risk in the US gas market. For wind
power, \citet{Pircalabu2016} and \citet{Pircalabu2017} model joint price--volume
risk with copula methods and study the hedging of fixed-price wind agreements;
\citet{Rowinska2020} develops statistical inference for joint electricity--wind
models. Hedging of renewable PPAs with standard instruments is studied by
\citet{PENA2024101513}. Energy quanto options, which pay on products of a price
index and a volumetric or weather index, are analysed by
\citet{BenthLangeQuanto} and \citet{Benth2022}; product options as
market-completing instruments for functions of two underlyings go back to
\citet{CarrCorso2001} and \citet{Madan2021Pricing}. Closely related renewable
risk-transfer designs combine energy and weather derivatives for joint
price--volume exposure \citep{MATSUMOTO2021105101}, introduce QF (Quality Factor)-settled
shape-risk contracts for capture-rate risk \citep{en19133044}, and analyse
optimized fixed-volume wind swaps under shape and basis risk
\citep{LUCY2021105603}. Relative to this literature,
the present paper works with the exact PPA strike decomposition, formulates the
hedge as a two-layer $L^2$ projection that respects delivery-period trading
restrictions, and selects quanto-type claims according to how
well they hedge the residual left after dynamic futures hedging.

\paragraph{Electricity spot, futures and renewable factor models:}
Electricity spot models must accommodate seasonality, mean reversion, spikes,
negative prices, delivery-period settlement and non-storability; see
\citet{Benth2008} for a comprehensive treatment of arithmetic and geometric
models. Pricing and hedging of Asian-style energy claims under delivery-period
trading restrictions are analysed by \citet{BenthDetering2015}, and additive
mean-reverting forward-curve specifications in a Heath--Jarrow--Morton framework
by \citet{Benth2019}. The empirical model is closest in spirit to
\citet{DeschatreVeraart2018}, who model spot prices and wind penetration jointly
with seasonal functions, continuous autoregressive factors, a logistic
wind-penetration transform and a wind-dependent spike intensity, and to
\citet{RowinskaVeraartGruet2021}, who use a multi-factor arithmetic model with
wind-related exogenous variables. The present specification replaces the
short-term factor pair by a bivariate L\'evy-driven MCARMA state
\citep{MarquardtStelzer2007,SchlemmStelzer2012, BenthKarbach2023MCARMACones}, which nests multivariate
Ornstein--Uhlenbeck and higher-order autoregressive dynamics, can
represent oscillatory autocovariance at daily and weekly scales, and supports
exact discrete-time inference. Estimation rests on the
quasi-maximum-likelihood theory of \citet{SchlemmStelzerQMLE2012} and the driver
recovery of \citet{BrockwellSchlemm2013}; the multivariate normal-inverse
Gaussian law used for the recovered driver goes back to
\citet{BarndorffNielsen1997}.

\paragraph{Variance--optimal and semi--static hedging:}
Quadratic hedging in incomplete markets originates with the GKW decomposition
\citep{Galchouk1976,KunitaWatanabe1967} and local risk minimization
\citep{FollmerSondermann1985}; mean-variance hedging for general claims and
semimartingale price processes is developed in
\citet{Schweizer1992,Schweizer1994}, surveyed by \citet{Pham2000}, and treated
in full generality by \citet{CernyKallsen2007}. Explicit variance-optimal
hedges for L\'evy and affine stochastic volatility models are obtained by
\citet{Hubalek.2006}, \citet{KallsenPauwels2010} and \citet{Goutte2014}. The
semi-static layer follows \citet{SemistaticHedging,semistaticsparse}: each
target and auxiliary claim is first dynamically projected onto the traded
futures stack, and the static portfolio is chosen by a finite-dimensional
covariance projection of the orthogonal residuals. Fourier methods transfer
these projections to claims with integral payoff representations
\citep{CarrMadan2001,Eberlein2010}. In a companion paper,
\citet{chatziandreou2026semistaticvarianceoptimalhedgingcovariance} develop the
general semi-static variance-optimal hedging theory for covariance risk in
multi-asset derivatives: a multivariate Galtchouk--Kunita--Watanabe
decomposition splits the global mean-variance problem into an inner dynamic
$L^2$ projection onto the stochastic integrals generated by the traded assets
and an outer finite-dimensional quadratic program over static portfolios of
vanilla, product, and spread options. Multidimensional functional-spanning results
then identify the static option strips that reduce cross-gamma exposures not
hedgeable through dynamic trading alone. In this paper the semi-static construction
is specialized to delivery-period futures with a time-varying active set, which
is the setting in which power futures enter and leave the tradable set at contract-specific dates, and the auxiliary
claims are chosen to target the price--volume covariance channel identified by
the fair-strike decomposition.

\section{Model-free PPA pricing}\label{sec:model-free-pricing}

\subsection{Valuation convention and market objects}\label{subsec:valuation-convention}

Fix a finite horizon $T>0$ and a filtered probability space
$(\Omega,\scrF,(\scrF_t)_{0\le t\le T},\Q)$ satisfying the usual conditions. The
measure $\Q$ is a valuation measure under which the discounted futures value
processes used for hedging are square-integrable martingales. Interest rates are
deterministic and all cash flows are expressed in discounted units; equivalently,
the money-market account is the numeraire. Since renewable production is not
fully spanned by exchange-traded power futures, no-arbitrage does not identify a
unique joint law for price and production. All fair strikes and hedge
projections below are therefore conditional on the chosen valuation measure and
model convention.

Let $\mathcal R=\{\mathrm{W},\mathrm{S}\}$ denote the wind and solar
technologies. For $r\in\mathcal R$, the physical production rate is
\begin{equation}\label{eq:qC-general}
  q_t^r=\bar q^r C_t^r,\qquad C_t^r\in[0,1],\qquad \bar q^r>0 ,
\end{equation}
where $C^r$ is the technology-specific capacity factor and $\bar q^r$ the
installed capacity. We write $\E_t[\cdot]=\E[\cdot\mid\scrF_t]$. For $t\le u$
define the three pointwise conditional objects
\begin{align}
  F_t(u)&=\E_t[S_u],\label{eq:F-general}\\
  Q_t^r(u)&=\E_t[q_u^r],\label{eq:Q-general}\\
  R_t^r(u)&=\E_t[q_u^rS_u],
  \label{eq:R-general}
\end{align}
and the conditional covariance kernel
\begin{align*}\label{eq:Gamma-general}
  \Gamma_t^r(u) &
  =R_t^r(u)-Q_t^r(u)F_t(u)\\
 & =\Cov_t(q_u^r,S_u)\\
 & =\E_t\bigl[(q_u^r-Q_t^r(u))(S_u-F_t(u))\bigr].
\end{align*}
The last identity records that $\Gamma_t^r(u)$ is exactly the conditional
covariance of the centered variables; a negative $\Gamma_t^r(u)$ means that high
production occurs in low-price states. The notation is point-delivery notation;
delivery-period instruments are obtained by deterministic integration over their
delivery windows. Dimensionally, $S_u$ is in EUR/MWh and $q_u^r$ in MWh per unit
time, so $R_t^r(u)$ and $\Gamma_t^r(u)$ are in EUR per unit time.

\begin{assumption}[Model-free integrability]\label{ass:integrability}
For every delivery period and every technology considered in the paper,
$\int |S_u|\dd u$, $\int |q_u^r|\dd u$ and $\int |q_u^rS_u|\dd u$ are integrable,
and all payoff variables entering the hedging problem are square-integrable.
Conditional Fubini holds for the displayed conditional expectations.
\end{assumption}

\subsection{Fair strike and capture-price decomposition}\label{subsec:ppa-decomp}

For a fixed strike $K\in\R$, the floating-minus-fixed pay-as-produced PPA payoff
for technology $r$ is
\begin{equation}\label{eq:ppa-payoff-general}
  H_r^K=\int_0^T q_u^r(S_u-K)\dd u .
\end{equation}
Assume $\int_0^TQ_0^r(u)\dd u>0$. The fair strike is the unique $K_r^\star$
satisfying $\E[H_r^{K_r^\star}]=0$:
\begin{equation}\label{eq:fair-strike-raw}
  K_r^\star=
  \frac{\E\left[\int_0^T q_u^rS_u\dd u\right]}
       {\E\left[\int_0^T q_u^r\dd u\right]}.
\end{equation}
This is the model-implied production-weighted capture price under $\Q$. It
equals the flat baseload forward only when the profile and covariance
corrections sum to zero.

\begin{proposition}[Exact PPA decomposition]\label{prop:ppa-decomp}
Under Assumption \ref{ass:integrability}, the time-$t$ value of the remaining PPA cash
flow is
\begin{equation}
\begin{aligned}
  V_{r,t}^K
  &=\E_t\left[\int_t^T q_u^r(S_u-K)\dd u\right]\\
  &=\int_t^T\{R_t^r(u)-KQ_t^r(u)\}\dd u ,
\end{aligned}
  \label{eq:ppa-value-RQ}
\end{equation}
or equivalently
\begin{equation}\label{eq:ppa-value-decomp}
  V_{r,t}^K
  =\int_t^T Q_t^r(u)\{F_t(u)-K\}\dd u
  +\int_t^T\Gamma_t^r(u)\dd u .
\end{equation}
Consequently,
\begin{equation}\label{eq:fair-strike-gamma}
  K_r^\star=
  \frac{\int_0^T\{Q_0^r(u)F_0(u)+\Gamma_0^r(u)\}\dd u}
       {\int_0^TQ_0^r(u)\dd u}.
\end{equation}
\end{proposition}

\begin{proof}
The first equality follows from the tower property and conditional Fubini under
Assumption \ref{ass:integrability}. Adding and subtracting $Q_t^r(u)F_t(u)$ inside the
integral gives \eqref{eq:ppa-value-decomp}. Setting $t=0$ and imposing
$\E[H_r^K]=0$ gives \eqref{eq:fair-strike-gamma}.
\end{proof}

To compare the pay-as-produced strike with a flat baseload quote, define
\begin{equation}\label{eq:flat-and-volume-weighted}
  \bar F=\frac1T\int_0^T F_0(u)\dd u,\qquad
  \bar F_r^q=
  \frac{\int_0^T Q_0^r(u)F_0(u)\dd u}{\int_0^T Q_0^r(u)\dd u}.
\end{equation}
Then
\begin{align}
  K_r^\star-\bar F
  &=\Delta_r^{\rm prof}+\Delta_r^{\rm cov},\label{eq:flat-profile-cov}\\
  \Delta_r^{\rm prof}
  &=\bar F_r^q-\bar F,\label{eq:profile-correction}\\
  \Delta_r^{\rm cov}
  &=
  \frac{\int_0^T\Gamma_0^r(u)\dd u}
       {\int_0^TQ_0^r(u)\dd u}.
  \label{eq:covariance-correction}
\end{align}
The profile term $\Delta_r^{\rm prof}$ is deterministic conditional on the
forward curve and the expected production profile; the covariance term
$\Delta_r^{\rm cov}$ is the unspanned stochastic contribution. This distinction
is central for comparing wind and solar: a technology-specific PPA strike can lie below baseload
because expected output is concentrated in low-forward hours, because high
output coincides with low spot prices after conditioning, or because both
channels are present. The two components also require different hedging
instruments, which motivates the two-layer hedge of
\cref{sec:model-free-hedging}.

\subsection{Delivery-period futures}\label{subsec:delivery-futures}

Let $I=[I^-,I^+]\subset[0,T]$ be a delivery period and let $w_I(u)\ge0$ be a
deterministic delivery weight with
\begin{equation}\label{eq:delivery-weights}
  D(I)=\int_{I^-}^{I^+}w_I(u)\dd u>0 .
\end{equation}
Define the realized weighted power and production settlements
\begin{align}
  A^S(I)
  &=\int_{I^-}^{I^+}w_I(u)S_u\dd u,\label{eq:settlement-power}\\
  A^{q,r}(I)
  &=\int_{I^-}^{I^+}w_I(u)q_u^r\dd u .
  \label{eq:settlement-production}
\end{align}
The average-price power future and the production-volume future are
\begin{equation}\label{eq:delivery-futures}
\begin{aligned}
  F^S(t,I)&=\frac{\E_t[A^S(I)]}{D(I)},\\
  F^{q,r}(t,I)&=\E_t[A^{q,r}(I)],\qquad t\le I^+ .
\end{aligned}
\end{equation}
Before delivery starts,
\begin{align}
  F^S(t,I)
  &=\frac1{D(I)}\int_{I^-}^{I^+}w_I(u)F_t(u)\dd u,\label{eq:pre-delivery-power}\\
  F^{q,r}(t,I)
  &=\int_{I^-}^{I^+}w_I(u)Q_t^r(u)\dd u ,\qquad t\le I^- .
  \label{eq:pre-delivery-production}
\end{align}
During delivery, the realized part is known path by path and only the remaining
part is conditional:
\begin{align}
  F^S(t,I)
  &=\frac1{D(I)}
  \int_{I^-}^{I^+}\1_{\{u\le t\}}w_I(u)S_u\dd u\notag\\
  &\quad
  +\frac1{D(I)}
  \int_{I^-}^{I^+}\1_{\{u>t\}}w_I(u)F_t(u)\dd u,\label{eq:during-delivery-power}\\
  F^{q,r}(t,I)
  &=\int_{I^-}^{I^+}\1_{\{u\le t\}}w_I(u)q_u^r\dd u\notag\\
  &\quad
  +\int_{I^-}^{I^+}\1_{\{u>t\}}w_I(u)Q_t^r(u)\dd u ,
  \label{eq:during-delivery-production}
\end{align}
for $I^-<t<I^+$. The notional-weighted settlement martingales
\begin{equation}\label{eq:X-martingales}
  X_t^S(I)=D(I)F^S(t,I),\qquad
  X_t^{q,r}(I)=F^{q,r}(t,I)
\end{equation}
are square-integrable $\Q$-martingales by construction and serve as the integrators
in the hedging problem; in particular $X^S_{I^+}(I)=A^S(I)$ and
$X^{q,r}_{I^+}(I)=A^{q,r}(I)$. This martingale property, together with the stated square-integrability
conditions, is sufficient for the projection results below. Hedge ratios below are computed against
the integrator vector $\bm X=(X^S(I),X^{q,r}(I),\ldots)$, not against quoted
futures labels; a position in the quoted average-price future equals the
notional-weighted position multiplied by $D(I)$. The empirical implementation
compares a pre-delivery convention, in which a futures position is held fixed
during its delivery period, with an idealized during-delivery convention in
which the remaining settlement value can be rebalanced throughout delivery,
following the Asian-style energy-contract analysis of \citet{BenthDetering2015};
the precise conventions are formalized in \cref{sec:single-period} and
\cref{sec:empirical-design}.

\subsection{Empirical analysis of the German market}\label{subsec:empirical-analysis}

The decomposition of the previous section identifies why a pay-as-produced PPA deviates from a flat baseload contract:
the expected production profile, the covariance between production and
prices, and the volumetric components that survive a futures hedge. This
section documents these empirical features using German data before any
stochastic model is introduced: the deterministic seasonal structure of
prices and renewable output, the serial dependence and spectral structure of
the deseasonalized states, the negative and load-dependent
price--renewable dependence (cannibalisation), realized capture prices and
capture factors for wind and solar, the state dependence of negative price
spikes, and the relation between spot outcomes, futures quotes and observed
PPA prices. These findings motivate the design choices for
the hedging theory of \cref{sec:model-free-hedging,sec:single-period}, for
the joint renewable--price model of \cref{sec:model}, and for its estimation
in \cref{sec:calibration}.

\subsubsection{Data and market overview}\label{subsec:emp-data}

The empirical panel combines the four series summarized in
\cref{tab:emp-data-sources}: German day-ahead electricity prices
$S_t$ from the ENTSO-E Transparency Platform, the ENWEX Germany wind and
solar utilisation indices $C_t^{\mathrm W},C_t^{\mathrm S}\in[0,1]$, and
DE--LU actual total load $L_t$ from the ENTSO-E Transparency Platform
\citep{ENTSOE2025EnergyPrices,ENWEX2026Data,
ENTSOE2026ActualLoad}. After aligning the four series on
their common hourly timestamps, the resulting panel contains $43{,}819$
observations from January~2021 through December~2025.

\begin{table*}[t]
\centering
\caption{Data sources for the hourly German empirical panel. The table reports
the variables used in the analysis, their definitions, the corresponding
provider datasets, and the official documentation and data-access sites.
The final sample is obtained by aligning all series on their common hourly
timestamps.}
\label{tab:emp-data-sources}
\small
\renewcommand{\arraystretch}{1.18}
\setlength{\tabcolsep}{5pt}
\begin{tabularx}{\textwidth}{
  @{}
  p{0.09\textwidth}
  p{0.25\textwidth}
  p{0.28\textwidth}
  X
  @{}
}
\toprule
Variable
& Series used
& Provider and dataset
& Official reference sites \\
\midrule

$S_t$
&
DE--LU day-ahead electricity price, in EUR/MWh
&
ENTSO-E Transparency Platform,
\emph{Energy Prices [12.1.D]}
\citep{ENTSOE2025EnergyPrices}
&
\href{https://transparency.entsoe.eu/}
     {ENTSO-E data portal};
\href{https://transparencyplatform.zendesk.com/hc/en-us/articles/16647234190100-Energy-Prices-12-1-D}
     {dataset definition} \\

\addlinespace

$C_t^{\mathrm W}$
&
ENWEX Germany wind utilisation index, divided by $100$ and expressed on
$[0,1]$
&
enwex GmbH,
\emph{Enwex Germany Wind}
\citep{ENWEX2026Data}
&
\href{https://enwex.com/product/}
     {index methodology};
\href{https://enwex.com/getdataapitest/}
     {historical-data page} \\

\addlinespace

$C_t^{\mathrm S}$
&
ENWEX Germany solar utilisation index, divided by $100$ and expressed on
$[0,1]$
&
enwex GmbH,
\emph{Enwex Germany Solar}
\citep{ENWEX2026Data}
&
\href{https://enwex.com/product/}
     {index methodology};
\href{https://enwex.com/getdataapitest/}
     {historical-data page} \\

\addlinespace

$L_t$
&
DE--LU actual total electricity load, in MW
&
ENTSO-E Transparency Platform,
\emph{Actual Total Load [6.1.A]}
\citep{ENTSOE2026ActualLoad}
&
\href{https://transparency.entsoe.eu/}
     {ENTSO-E data portal};
\href{https://transparencyplatform.zendesk.com/hc/en-us/articles/16647979768084-Actual-Total-Load-Day-ahead-Per-Bidding-Zone-6-1-A-6-1-B}
     {dataset definition} \\

\bottomrule
\end{tabularx}

\vspace{0.4em}
\begin{minipage}{0.98\textwidth}
\footnotesize
\emph{Notes:}
The ENWEX wind and solar series are standardized, nationwide
weather-based utilisation indices. They should therefore be interpreted as
country-level renewable-capacity-utilisation measures rather than as metered
capacity factors for an individual wind or solar installation. The ENTSO-E
series correspond to the DE--LU bidding zone. All variables are aligned to the
common hourly grid used in the empirical analysis.
\end{minipage}
\end{table*}

Each hour is labelled \emph{crisis} (2021--2022, which contains the
energy-crisis regime), \emph{calibration} (2023--2024, the window on which the
stochastic model of \cref{sec:model} is later estimated), or
\emph{out-of-sample} (January--December~2025, the PPA delivery window of the
empirical hedging study). 

Three sample characteristics are important for the analysis. First, the crisis window is dominated
by the 2022 shock, with a mean price of $166.2$~EUR/MWh and a standard deviation
of $133.1$~EUR/MWh, which is why estimation deliberately begins in 2023.
The calibration and out-of-sample mean prices are nearly identical
($86.8$ versus $89.3$~EUR/MWh), so the 2025 experiment is not driven by a
shift in the average price level. Second, the share of negative-price hours
increases out of sample, from $4.3\%$ in the calibration period to
$6.6\%$. Third, 2025 is a weak wind year, with a mean wind utilisation factor
$C^{\mathrm W}$ of $0.20$ compared with $0.24$ during calibration, while the
mean solar utilisation factor remains close to $0.10$. The out-of-sample period therefore provides
a demanding test of hedges against volume and capture-price
risk.

\begin{figure*}[t]
\centering
\includegraphics[width=.92\textwidth]{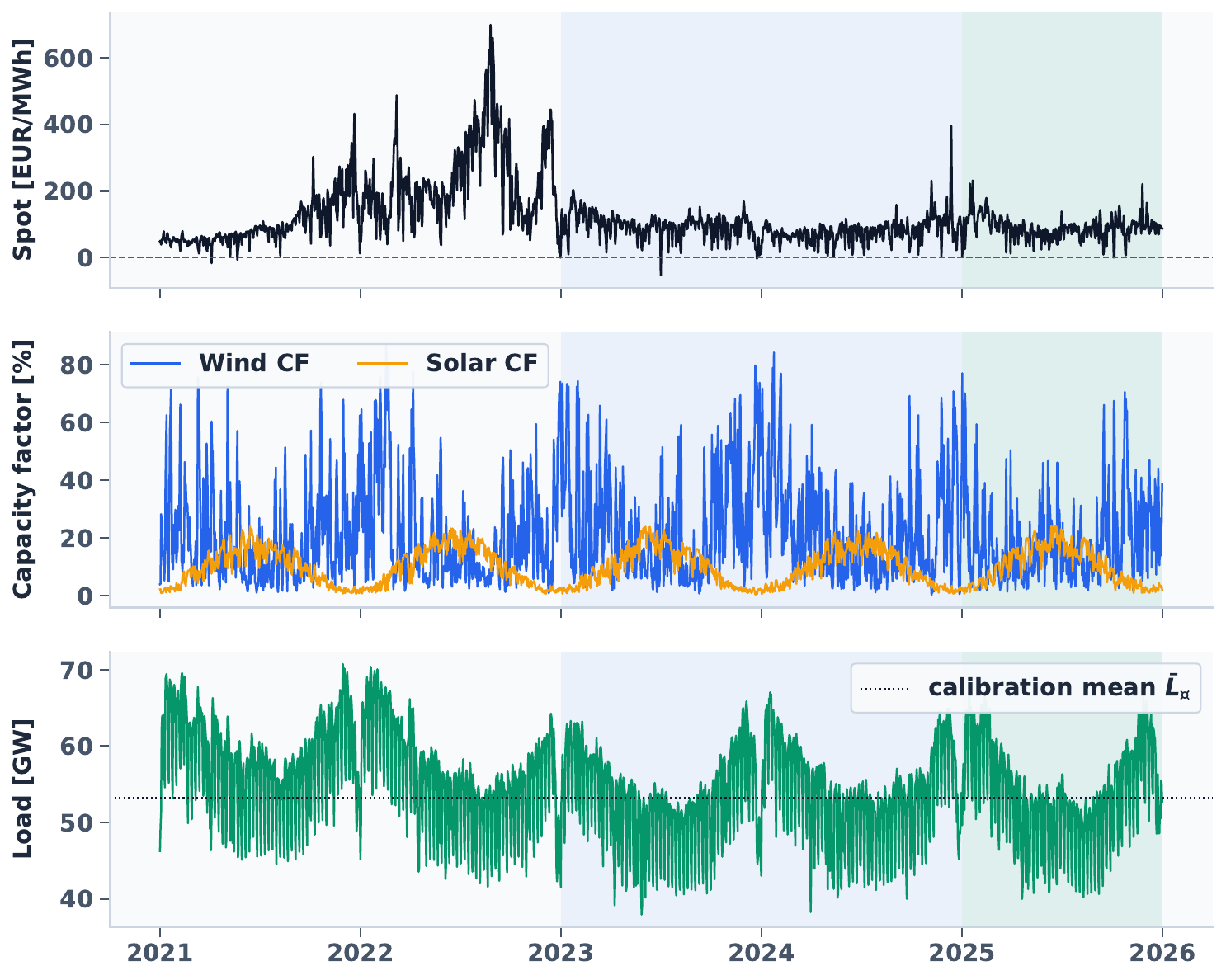}
\caption{Daily means of the hourly German panel, 2021--2025: day-ahead spot
price (top, dashed line at zero), wind and solar capacity factors (middle) and
system load with the calibration-window mean $\bar L_{\mathcal C}$ (bottom).
Blue shading marks the 2023--2024 calibration window, green shading the 2025
out-of-sample delivery window.}
\label{fig:emp-overview}
\end{figure*}

\subsubsection{Seasonal structure and deseasonalized components}
\label{subsec:emp-seasonality}

Prices and renewable output are dominated by deterministic calendar
structure, and their stochastic dependence should therefore be assessed after
removing that structure. We use three
transformations to separate deterministic from stochastic
variation. The resulting variables form the state vector in
\cref{sec:model}.

The spot price is represented additively,
\begin{equation}\label{eq:spot-additive}
  S_t=\Lambda^S(t)+Y_t^S+J_t ,
\end{equation}
where $\Lambda^S$ is a deterministic price level of the product form
$\Lambda^S(t)=\bar S\,F^{2Y,S}_tF^{2D,S}_t$, combining a slowly varying
annual factor with a normalized intraday profile \citep{Paraschiv2015},
$Y^S$ is the continuous deseasonalized residual in EUR/MWh, and $J$ is a
mean-reverting spike component identified by a threshold procedure on the
deseasonalized price. For valuation, if liquid futures quotes are available, the deterministic price
level may alternatively be inferred from a smooth, market-implied hourly price
forward curve (HPFC).

For wind, the relevant state is not the capacity factor alone but wind
penetration relative to demand. With $L_t$ the system load and
$\bar L_{\mathcal C}$ its calibration-window average, define the
load-adjusted wind penetration index
\begin{equation}\label{eq:wind-pressure}
  \Pi_t^{\mathrm W}
  =
  \min\left\{1,\max\left\{0,\;C_t^{\mathrm W}\frac{\bar L_{\mathcal C}}{L_t}\right\}\right\},
\end{equation}
which preserves the ordering of penetration because, for approximately
constant installed capacity, penetration is proportional to
$C^{\mathrm W}_t/L_t$. The latent wind coordinate is the clipped logit
\begin{equation}\label{eq:wind-logit}
\begin{aligned}
  X_t^{\mathrm W}
  &=
  \logit\bigl(\Pi_t^{\mathrm W,\eps}\bigr)
  =
  \Lambda^{\mathrm W}(t)+Y_t^{\mathrm W},\\
  \Pi_t^{\mathrm W,\eps}
  &=\min\{1-\eps,\max\{\eps,\Pi_t^{\mathrm W}\}\},
\end{aligned}
\end{equation}
for a small clipping constant $\eps\in(0,1/2)$, where $\Lambda^{\mathrm W}$
is a deterministic seasonal component and $Y^{\mathrm W}$ the stochastic
residual.

For solar, we separate deterministic daylight availability from
stochastic cloud-cover shortfall and model the realized capacity
factor $C_t^{\mathrm S}\in[0,1]$ directly. Write $d_t\in\{1,\ldots,366\}$
and $h_t\in\{0,\ldots,23\}$ for day of year and hour of day, and let
$\operatorname{dist}_{366}$ denote circular distance on the annual grid. For
each day--hour cell, the clear-sky construction first estimates the locally
pooled upper envelope
\begin{equation}\label{eq:solar-local-envelope}
\begin{aligned}
 \mathcal N_{d,h}
 &:=\{t\in\mathcal C:h_t=h,\ 
       \operatorname{dist}_{366}(d_t,d)\le b\},\\
 \widetilde q_{d,h}
 &:=Q_{\tau}\bigl(C_t^{\mathrm S}:t\in\mathcal N_{d,h}\bigr).
\end{aligned}
\end{equation}
where $\mathcal C$ is the calibration window, $\tau=0.98$ and $b=15$. The sequence
$d\mapsto\widetilde q_{d,h}$ is smoothed circularly, separately for every
hour, with a Gaussian kernel of standard deviation $\sigma_d=3$ days; denote
the resulting envelope by $\widetilde C_{d,h}^{\mathrm{cs}}$. A scale factor estimated from
high-visibility observations raises the envelope so that it covers nearly all observed output,
\begin{equation}\label{eq:solar-envelope-scale}
\begin{aligned}
 \mathcal H
 &:=\{t\in\mathcal C:
 \widetilde C_{d_t,h_t}^{\mathrm{cs}}\ge c_{\mathrm{ref}}\},\\
  \kappa
  &:=\max\left\{1,
  Q_{\eta}\!\left(
  \frac{C_t^{\mathrm S}}{\widetilde C_{d_t,h_t}^{\mathrm{cs}}}
  \,\middle|\,t\in\mathcal H
  \right)\right\},
  \qquad \eta=0.999,\\
 c_{\mathrm{ref}}&:=0.03.
\end{aligned}
\end{equation}
To distinguish nighttime cells with zero solar availability from small positive metering values, let
$\widehat p_{d,h}$ be the local-window frequency of
$C_t^{\mathrm S}>\epsilon_{\mathrm{obs}}$, with
$\epsilon_{\mathrm{obs}}=10^{-4}$. The final climatology is
\begin{equation}\label{eq:solar-structural-night}
\begin{aligned}
 E_{d,h}&:=[\kappa\widetilde C_{d,h}^{\mathrm{cs}}]_{[0,1]},\\
 C_{d,h}^{\mathrm{cs}}
 &:={}
 \begin{cases}
  0,
  & E_{d,h}<c_{\mathrm n}\ \text{and}\ \widehat p_{d,h}<p_{\min},\\
  E_{d,h},&\text{otherwise}.
 \end{cases}
\end{aligned}
\end{equation}
Here $c_{\mathrm n}=10^{-3}$, $p_{\min}=0.10$, and
$[z]_{[a,b]}:=\min\{b,\max\{a,z\}\}$. We set
$C_t^{\mathrm{cs}}=C_{d_t,h_t}^{\mathrm{cs}}$.

At very low irradiance the physical output does not identify cloud cover:
when $C_t^{\mathrm{cs}}=0$, production is zero for every weather state. The
observation map therefore uses the identification scale $\delta=0.03$,
\begin{equation}\label{eq:solar-risk-driver}
  R_t^\odot
  =
  \left[
  1-\frac{C_t^{\mathrm S}}{\max(C_t^{\mathrm{cs}},\delta)}
  \right]_{[0,1]},
  \qquad \delta=0.03,
\end{equation}
so that $R_t^\odot\approx0$ under clear conditions and
$R_t^\odot\approx1$ under heavy cloud cover or curtailment. With the finite
boundary buffer $\eps=0.01$, define
\begin{equation}\label{eq:solar-latent}
\begin{aligned}
  U_t^\odot&=[R_t^\odot]_{[\eps,1-\eps]},\\
  X_t^{\mathrm S}&=\logit(U_t^\odot)
  =\Lambda^{\mathrm S_\odot}(t)+Y_t^{\mathrm S_\odot},\\
  Y_t^{\mathrm S_\odot}
  &=X_t^{\mathrm S}-\Lambda^{\mathrm S_\odot}(t).
\end{aligned}
\end{equation}
The buffer bounds the latent coordinate by
$|X_t^{\mathrm S}|\le\log(99)$. In the non-clipped region its local
sensitivity satisfies
\begin{equation}\label{eq:solar-latent-sensitivity}
\begin{aligned}
 \left|\frac{\partial X_t^{\mathrm S}}{\partial C_t^{\mathrm S}}\right|
 &=\frac{1}{\max(C_t^{\mathrm{cs}},\delta)
 U_t^\odot(1-U_t^\odot)}
 \\[-0.15em]
 &\le\frac{1}{\delta\eps(1-\eps)}.
\end{aligned}
\end{equation}
and it is zero where clipping is active. Thus we map negligible
nighttime metering noise to the boundary state instead of treating
it as variation in cloud cover.

The deterministic latent level is estimated using a robust
calendar--Fourier regression. Put
$\vartheta_t=2\pi d_t/365.25$. With hour~0 and January as reference classes,
we take the following design:
\begin{align}\label{eq:solar-calendar-design}
 \bm x_t^\top\bm b
 ={}&b_0
 +\sum_{h=1}^{23}b_h^{\mathrm H}\1_{\{h_t=h\}}\nonumber\\
 &+\sum_{m=2}^{12}b_m^{\mathrm M}\1_{\{m_t=m\}}\nonumber\\
 &+\sum_{k=1}^{3}
   \{a_k\sin(k\vartheta_t)+c_k\cos(k\vartheta_t)\}\nonumber\\
 &+\sum_{h=1}^{23}\sum_{k=1}^{2}
   a_{h,k}\1_{\{h_t=h\}}\sin(k\vartheta_t)\nonumber\\
 &+\sum_{h=1}^{23}\sum_{k=1}^{2}
   c_{h,k}\1_{\{h_t=h\}}\cos(k\vartheta_t).
\end{align}
The hour--Fourier interactions allow the annual daylight cycle to vary across
the intraday profile without introducing discontinuous day-of-year cells. The
coefficient vector solves
\begin{equation}\label{eq:solar-robust-seasonality}
\begin{aligned}
 \mathcal J(\bm b)
 &:={}
 \sum_{t\in\mathcal C}w_t\rho_{1.345}
   (X_t^{\mathrm S}-\bm x_t^\top\bm b)\\
 &\quad+\frac{\lambda_s}{N_*}\|D\bm X_*\bm b\|_2^2
 +\lambda_r\|\bm b_{-0}\|_2^2,\\
 \widehat{\bm b}&:=\argmin_{\bm b}\mathcal J(\bm b).
\end{aligned}
\end{equation}
where $\bm X_*$ is the chronological design grid, $N_*$ is the number of its
first differences, $D$ is the first-difference operator,
$\lambda_s=50$, $\lambda_r=10^{-5}$, and
\[
 \rho_c(u)=
 \begin{cases}
  u^2/2,&|u|\le c,\\
  c|u|-c^2/2,&|u|>c
 \end{cases}
\]
is the Huber loss. The visibility weight is
\begin{equation}\label{eq:solar-visibility-weights}
\begin{aligned}
  w_t&=w_0+(1-w_0)
  \min\{C_t^{\mathrm{cs}}/\delta,1\},\\
  w_0&=0.35.
\end{aligned}
\end{equation}
The problem is solved by iteratively reweighted least squares and the fitted
intercept is recentered so that the calibration residual has zero arithmetic
mean. The construction keeps low-visibility observations in the
seasonality fit but downweights them when estimating weather-driven
variation. Finally,
$\Lambda^{\mathrm S_\odot}(t)=\bm x_t^\top\widehat{\bm b}$. All quantities
defined in \crefrange{eq:solar-local-envelope}{eq:solar-visibility-weights} are
estimated on 2023--2024 only and projected onto 2025, so the panels below
are fully out-of-sample projections.

\Cref{fig:emp-spot-seasonality} displays the price decomposition: the
seasonal level $\Lambda^S$ tracks the daily price level through both
calibration years and, when projected out of sample, preserves the estimated level
and calendar pattern through December 2025 without re-estimation. The lower panel
separates the continuous residual $Y^S_t=S_t-\Lambda^S_t-J_t$ from the spike
state $J_t$; the spike filter assigns the December 2024
\emph{Dunkelflaute} episode and the 2025 negative-price afternoons mainly to $J$,
leaving a mean-reverting continuous residual.
\Cref{fig:emp-renewable-seasonality} shows the renewable decompositions: the
wind penetration $\Pi^{\mathrm W}_t$ against its seasonal logistic level with
the logit residual $Y^{\mathrm W}_t$, and the realized solar capacity factor
against the clear-sky envelope and the seasonal level in physical units, with
the bounded latent cloud-risk residual $Y^{\mathrm S_\odot}_t$.

\begin{figure*}[p]
\centering
\begin{subfigure}[t]{.96\textwidth}
\centering
\includegraphics[width=\linewidth,height=.36\textheight,keepaspectratio]{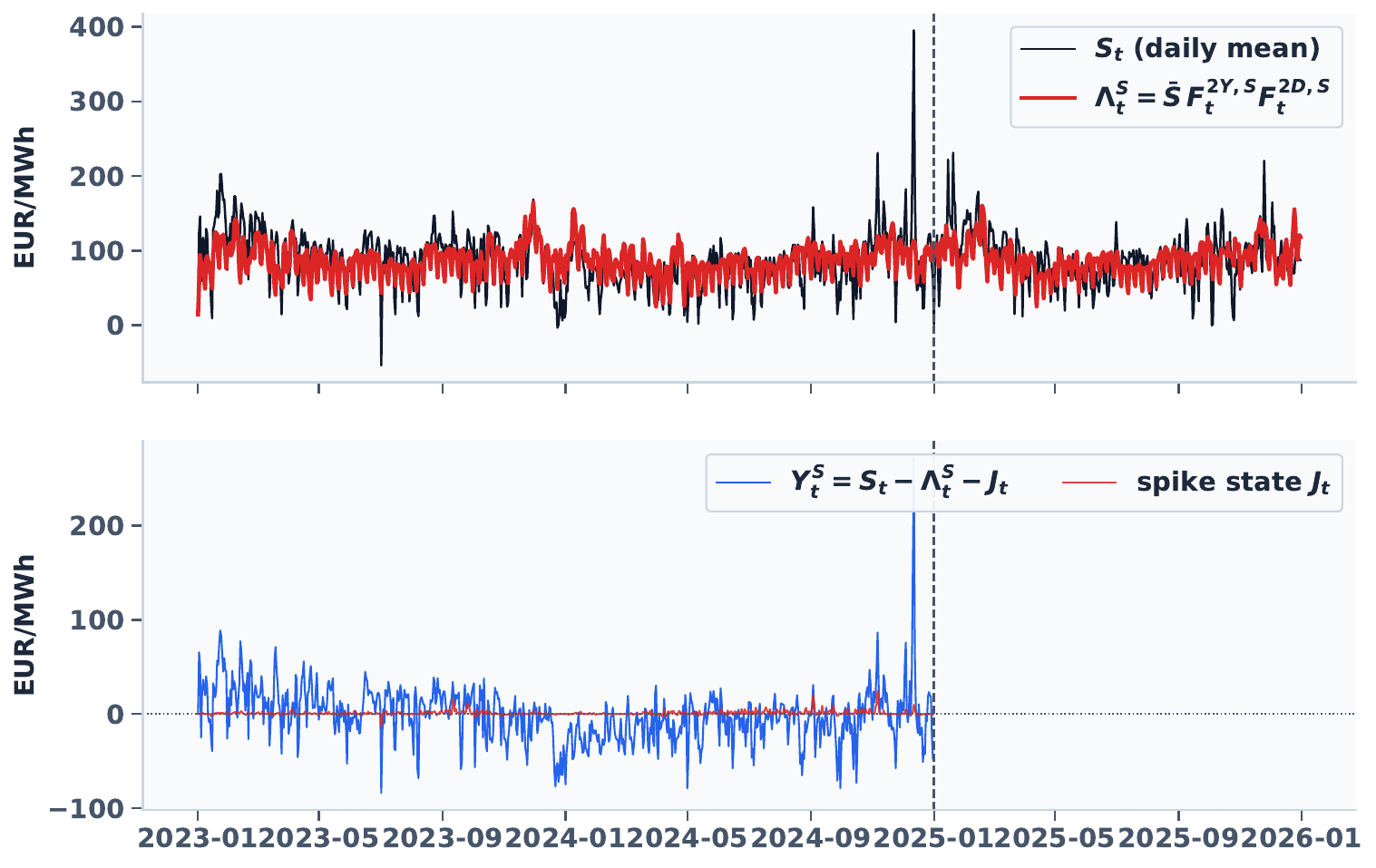}
\caption{Spot-price seasonal level, continuous residual and spike component.}
\end{subfigure}\\[0.4em]
\begin{subfigure}[t]{.96\textwidth}
\centering
\includegraphics[width=\linewidth,height=.36\textheight,keepaspectratio]{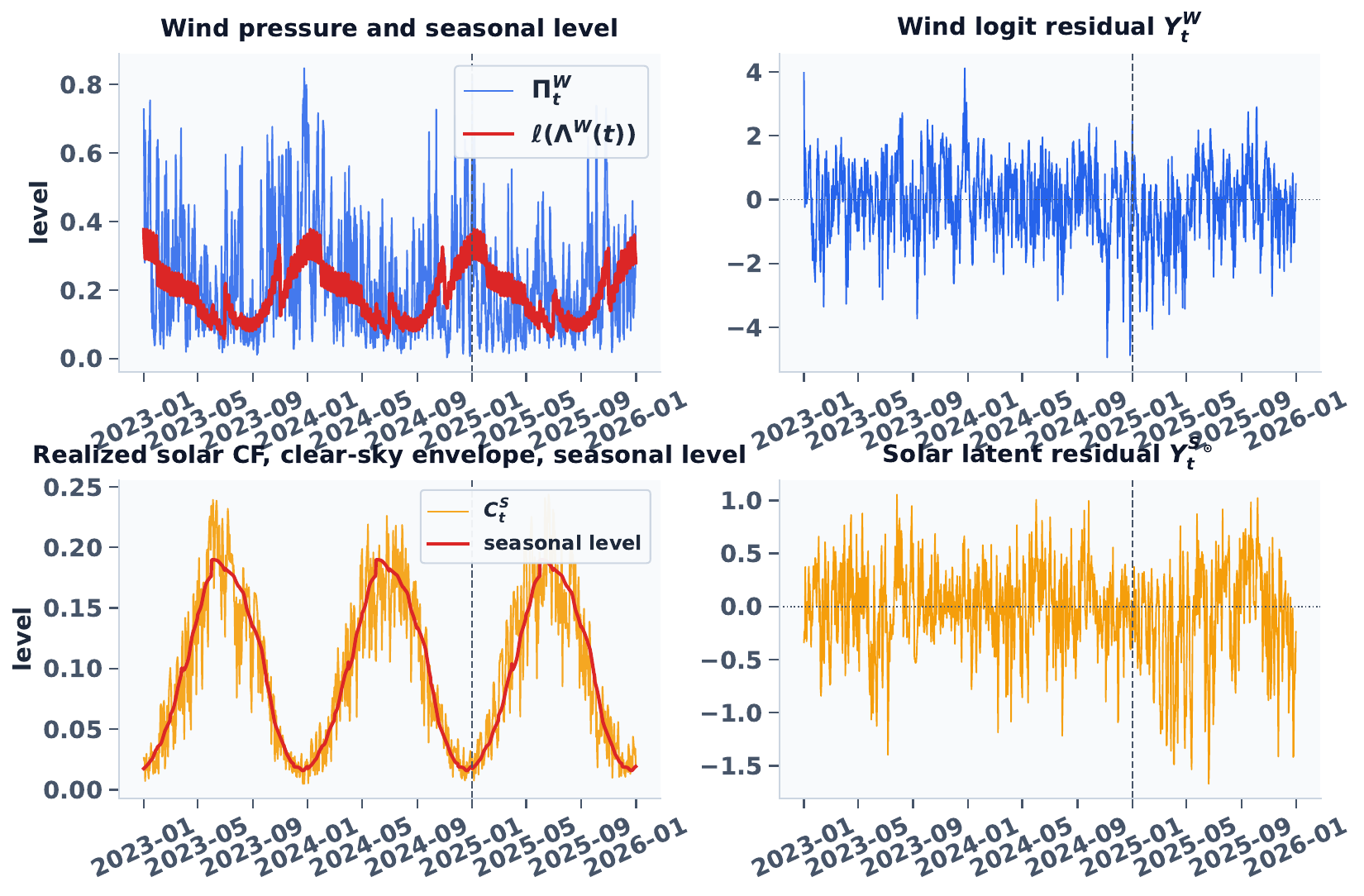}
\caption{Wind-penetration and solar clear-sky state constructions with their
deseasonalized latent residuals.}
\end{subfigure}
\caption{Seasonal decomposition and renewable state construction. The dashed
vertical lines separate the 2023--2024 calibration window from the 2025
out-of-sample projection.}
\label{fig:emp-spot-seasonality}
\label{fig:emp-renewable-seasonality}
\end{figure*}

\subsubsection{Serial dependence and spectral structure of the states}
\label{subsec:emp-acf}

Removing the seasonal layers does not leave white noise: the deseasonalized
states remain serially dependent, and their autocorrelation patterns guide
the order and spectral specification of the continuous-time model. For a
state $x$ with $n$ hourly observations and sample mean $\bar x$, the
empirical autocorrelation is
\begin{equation}\label{eq:emp-acf}
  \widehat\rho(k)
  =\frac{\frac{1}{n-k}\sum_{t=1}^{n-k}(x_t-\bar x)(x_{t+k}-\bar x)}
        {\frac{1}{n}\sum_{t=1}^{n}(x_t-\bar x)^2},
\end{equation}
computed on the 2023--2024 calibration states for lags
$k=1,\dots,336$ hours. \Cref{fig:emp-state-acf} shows $\widehat\rho$ over
these two weeks of lags. The three states display three
distinct dependence shapes. The wind residual $Y^{\mathrm W}$ decays
smoothly but on two clearly separated time scales: a fast component that
decays within a day and a slower synoptic component that remains visible
after a week, with only mild calendar oscillations. The spot
residual $Y^S$ combines a slowly decaying envelope with pronounced
oscillations at multiples of $24$ hours and a secondary peak around the weekly
lag of 168 hours. The solar residual $Y^{\mathrm S_\odot}$ combines a fast-decaying
component with pronounced positive autocorrelation at multiples of 24 hours. The
bounded low-irradiance observation map prevents nighttime zero output from being
treated as an extreme weather shock; the remaining daily autocorrelation reflects
persistence in cloud-cover regimes across successive daylight periods.
This persistence of calendar-frequency dependence after deterministic
deseasonalization and short-memory filtering is consistent with
\citet{DeschatreVeraart2018}. In their spot-price application, significant
residual autocorrelation remains at lags $24$ and $48$ hours after an
AR(24) model is fitted to the deseasonalized and filtered series; in their
wind application, significant residual autocorrelation remains at lag $24$
hours after an autoregressive continuous-time model is fitted to
deseasonalized logit wind penetration. It is also consistent with
\citet{kloster2025ambit}, who remove a 90-day trend and quarterly and annual
seasonal components from an hourly price panel and obtain a fitted dependence
kernel with slow temporal decay and a dependence pattern linking adjacent within-day
hours and corresponding hours across midnight. These studies support the narrower
claim that calendar-scale dependence can remain after seasonal adjustment and
short-memory filtering; they do not identify the four frequencies or the
MCARMA order used below.

We summarize the frequency content by a descriptive spectral diagnostic. Let
$\widehat f(\nu)$ be the Welch estimate of the power spectral density of a
standardized state, computed from $K$ Hann-tapered segments with 50\% overlap.
For a calendar period $T$, define the peak-to-local-background ratio
\begin{equation}\label{eq:emp-peak-ratio}
  R_T
  =\frac{\displaystyle
    \max_{|\nu-1/T|\le\delta}\widehat f(\nu)}
  {\displaystyle
    \operatorname*{median}_{\nu\in\mathcal B_T}\widehat f(\nu)},
\end{equation}
where the peak window has half-width $\delta$ and $\mathcal B_T$ is a
surrounding frequency band from which the peak window is removed. A period enters the
descriptive candidate set when
\begin{equation}\label{eq:emp-peak-screen}
  R_T>\tau_R,\qquad \tau_R=2.3,\qquad
  \frac{2\pi}{T}<\frac{\pi}{h},
\end{equation}
where the last inequality is the Nyquist condition for sampling step $h$ and
$\tau_R$ is a fixed, non-inferential ratio cutoff used uniformly across
states and frequencies.

\Cref{fig:emp-state-psd} displays the Welch spectra and
\cref{tab:emp-harmonics} reports the ratios for $K=16$ segments. All four
periods $T\in\{24,12,168,8\}$ hours exceed the descriptive cutoff for each
state and hence enter the candidate set. The largest ratios occur at the
intraday frequencies ($R_8=8.7$ for wind, $R_{12}=6.8$ for the spot
residual and $R_8=9.9$ for solar), while the weekly ratios range from
$R_{168}=3.8$ to $4.2$. These values quantify local spectral prominence;
they are not significance statistics.

The ACF and PSD diagnostics \cref{fig:emp-state-acf,fig:emp-state-psd} show dependence that cannot be reproduced
simultaneously by the scalar, single-rate Ornstein--Uhlenbeck
autocorrelation $\rho(k)=\exp(-\kappa kh)$. They motivate,
but do not uniquely identify, a higher-order continuous-time specification
combining several relaxation scales with damped oscillator pairs. The
screening-and-selection rule retains the daily and
half-daily periods under a two-pair model-order cap; the weekly and eight-hour
peaks remain reported diagnostics rather than additional fitted oscillator
pairs. The screening, model-order selection and fitting procedure is documented in the \suppmat.

\begin{table}[t]
\caption{Descriptive peak-to-local-background ratios
\eqref{eq:emp-peak-ratio} at the calendar periods
$T\in\{24,12,168,8\}$ hours for the three deseasonalized states, 2023--2024
($K=16$ Welch segments). The threshold $\tau_R=2.3$ is the fixed
non-inferential screen in \eqref{eq:emp-peak-screen}; a checkmark means that
the period enters the candidate set, not that a hypothesis is accepted.}
\label{tab:emp-harmonics}
\resizebox{\columnwidth}{!}{
\begin{tabular}{lrrrl}
\toprule
State & Period $T$ [h] & $R_T$ & Threshold & Included \\
\midrule
$Y^{W}$ & 24 & 2.7 & 2.3 & $\checkmark$ \\
$Y^{W}$ & 12 & 6.6 & 2.3 & $\checkmark$ \\
$Y^{W}$ & 168 & 3.8 & 2.3 & $\checkmark$ \\
$Y^{W}$ & 8 & 8.6 & 2.3 & $\checkmark$ \\
$Y^{S}$ & 24 & 4.3 & 2.3 & $\checkmark$ \\
$Y^{S}$ & 12 & 6.8 & 2.3 & $\checkmark$ \\
$Y^{S}$ & 168 & 4.0 & 2.3 & $\checkmark$ \\
$Y^{S}$ & 8 & 4.9 & 2.3 & $\checkmark$ \\
$Y^{S_\odot}$ & 24 & 9.0 & 2.3 & $\checkmark$ \\
$Y^{S_\odot}$ & 12 & 10.1 & 2.3 & $\checkmark$ \\
$Y^{S_\odot}$ & 168 & 4.1 & 2.3 & $\checkmark$ \\
$Y^{S_\odot}$ & 8 & 40.0 & 2.3 & $\checkmark$ \\
\bottomrule
\end{tabular}
}
\end{table}

\begin{figure*}[p]
\centering
\begin{subfigure}[t]{.96\textwidth}
\centering
\includegraphics[width=\linewidth,height=.34\textheight,keepaspectratio]{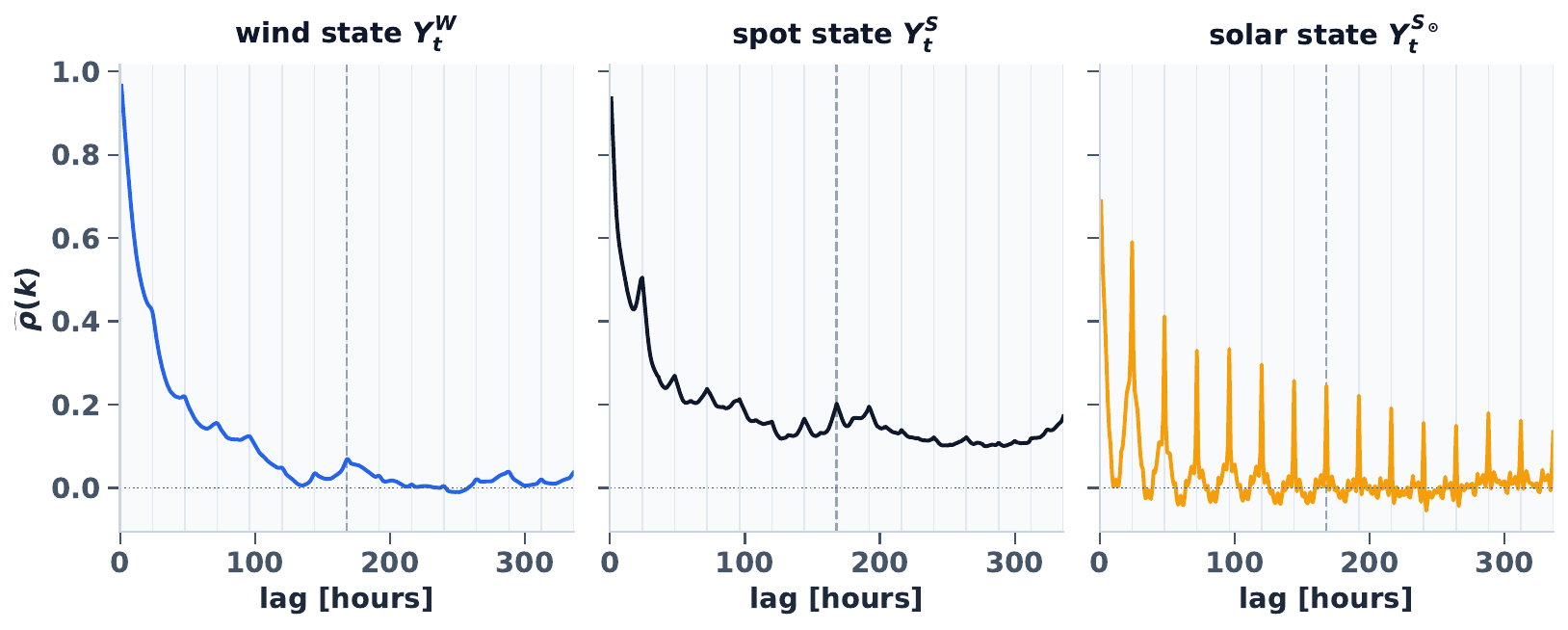}
\caption{Empirical autocorrelations over two weeks of hourly lags.}
\end{subfigure}\\[0.4em]
\begin{subfigure}[t]{.96\textwidth}
\centering
\includegraphics[width=\linewidth,height=.34\textheight,keepaspectratio]{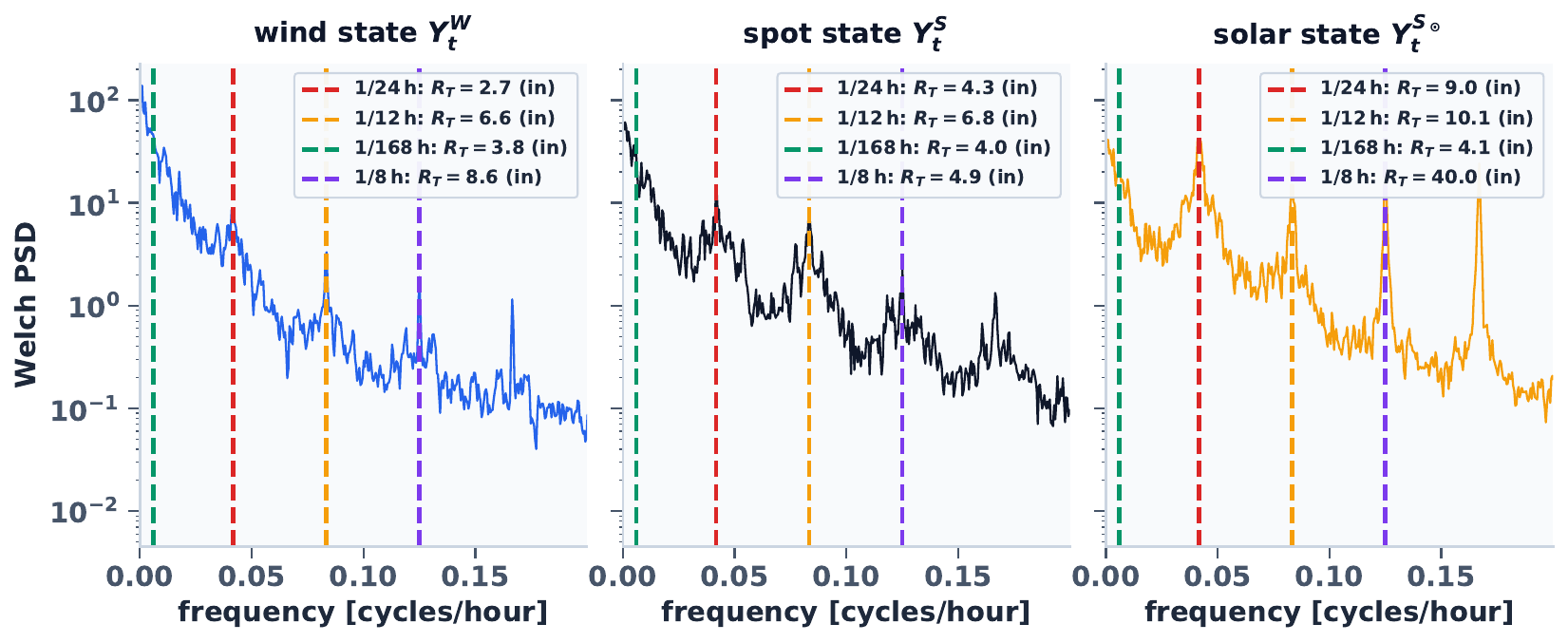}
\caption{Welch power spectral densities and calendar-frequency prominence
ratios.}
\end{subfigure}
\caption{Serial dependence and spectral structure of the deseasonalized wind,
spot and solar states. Thin markers identify the principal daily and weekly
calendar frequencies.}
\label{fig:emp-state-acf}
\label{fig:emp-state-psd}
\end{figure*}

\subsubsection{The isolated renewable price effect}\label{subsec:emp-wind-effect}

The covariance correction $\Delta^{\rm cov}_r$ in
\eqref{eq:covariance-correction} depends on the joint variation of
production and price, so its empirical counterpart must be measured after
removing the deterministic calendar and demand structure that both series
share. For each calendar year we estimate both an univariate regression and a controlled
regression of the daily load-weighted price $P_d=\sum_{t\in d}S_tL_t/\sum_{t\in d}L_t$ on the daily wind
capacity factor. The controlled (isolated) slope is the
Frisch--Waugh--Lovell partial regression coefficient
\begin{equation}\label{eq:iso-slope}
  \beta^{\rm iso}
  =\frac{\langle M_ZC^{\mathrm W},M_ZP\rangle}
        {\langle M_ZC^{\mathrm W},M_ZC^{\mathrm W}\rangle},
  \quad
  M_Z=I-Z(Z^\top Z)^{-1}Z^\top,
\end{equation}
where the control design $Z$ collects an intercept, the solar capacity
factor, load, a holiday indicator, annual sine/cosine terms and weekday and
month dummies; $\rho^{\rm iso}=\operatorname{corr}(M_ZC^{\mathrm W},M_ZP)$ is
the controlled correlation. Swapping the roles of wind and solar gives the
isolated solar effect. \Cref{fig:emp-wind-effect} and
\cref{tab:emp-wind-effect} show that the controls materially change both the wind and solar estimates.
For wind, the isolated slope is uniformly \emph{steeper} than the raw slope
outside the crisis year (for 2023--2025, $\beta^{\rm iso}$ between $-1.52$ and
$-1.71$~EUR/MWh per percentage point of capacity factor against raw slopes
near $-1.0$), and the controlled correlation reaches $-0.90$ in 2023 and
$-0.87$ in 2025: almost the entire non-seasonal variation of the daily price
level is conditionally associated with wind once demand and season are removed. For solar, \cref{fig:emp-solar-effect} uses the
same regression with wind and solar interchanged: the daily
load-weighted price is regressed on the daily \emph{solar} capacity factor
(the same realized series $C^{\mathrm S}$ on which the solar state
\eqref{eq:solar-latent} is built), with the wind capacity factor taking solar's place in the
control design $Z$ of \eqref{eq:iso-slope}, so that
$\beta^{\rm iso}_{\mathrm S}
=\langle M_ZC^{\mathrm S},M_ZP\rangle/
 \langle M_ZC^{\mathrm S},M_ZC^{\mathrm S}\rangle$.
Controlling matters far more here than for wind, because solar output is
dominated by its deterministic annual cycle: the raw slope conflates the
seasonal price level with the solar effect and even changes sign in 2022
($+4.5$, when the crisis price peak coincided with the high-solar season),
while the isolated slope is negative in every year ($-1.6$ to $-5.3$, with
$-1.9$ in 2023--2024 and $-2.4$ in 2025), indicating a negative association between solar output and the daily
price level after controlling for seasonality. The controlled correlations
($-0.13$ to $-0.42$) are much weaker than for wind, consistent with a daily
price level driven primarily by wind and demand while the solar effect
is concentrated within the day, as shown by the hourly profile term
$\Delta^{\rm prof}$ and the price-bucket analysis
in \cref{subsec:emp-buckets}. \Cref{fig:emp-rolling-corr} shows the same dependence at the
hourly scale: the 60-day rolling correlation of $S_t$ with the wind capacity
factor is persistently negative, and replacing $C^{\mathrm W}_t$ by the
load-adjusted pressure $\Pi^{\mathrm W}_t$ strengthens the full-sample hourly
correlation from $-0.32$ to $-0.36$ and stabilizes it across regimes. This stronger and
more stable correlation motivates the use of $\Pi_t^{\mathrm W}$ as the wind coordinate
of the joint model. The hourly spot--solar correlation is weaker ($-0.16$) but becomes markedly
more negative in the 2024--2025 springs and summers.

\begin{table}[t]
\caption{Raw versus isolated (controlled) renewable price effects by calendar
year: OLS slopes of the daily load-weighted price on the daily capacity factor
(EUR/MWh per percentage point) and the corresponding raw and controlled
correlations, per \eqref{eq:iso-slope}.}
\label{tab:emp-wind-effect}
\resizebox{\columnwidth}{!}{
\begin{tabular}{lrrrrrr}
\toprule
Technology & Year & Days & Raw slope & Isolated slope & Raw corr. & Controlled corr. \\
\midrule
Wind & 2021 & 365 & -1.11 & -1.71 & -0.26 & -0.67 \\
Wind & 2022 & 365 & -4.34 & -3.79 & -0.63 & -0.73 \\
Wind & 2023 & 365 & -1.07 & -1.55 & -0.56 & -0.90 \\
Wind & 2024 & 366 & -1.00 & -1.52 & -0.45 & -0.74 \\
Wind & 2025 & 365 & -1.07 & -1.71 & -0.47 & -0.87 \\
Solar & 2021 & 365 & -4.49 & -1.62 & -0.39 & -0.13 \\
Solar & 2022 & 365 & 4.53 & -5.35 & 0.23 & -0.21 \\
Solar & 2023 & 365 & -0.88 & -1.91 & -0.16 & -0.34 \\
Solar & 2024 & 366 & -1.69 & -1.93 & -0.27 & -0.20 \\
Solar & 2025 & 365 & -2.30 & -2.42 & -0.43 & -0.41 \\
\bottomrule
\end{tabular}
}
\end{table}

\begin{figure*}[p]
\centering
\begin{subfigure}[t]{.49\textwidth}
\centering
\includegraphics[width=\linewidth,height=.30\textheight,keepaspectratio]{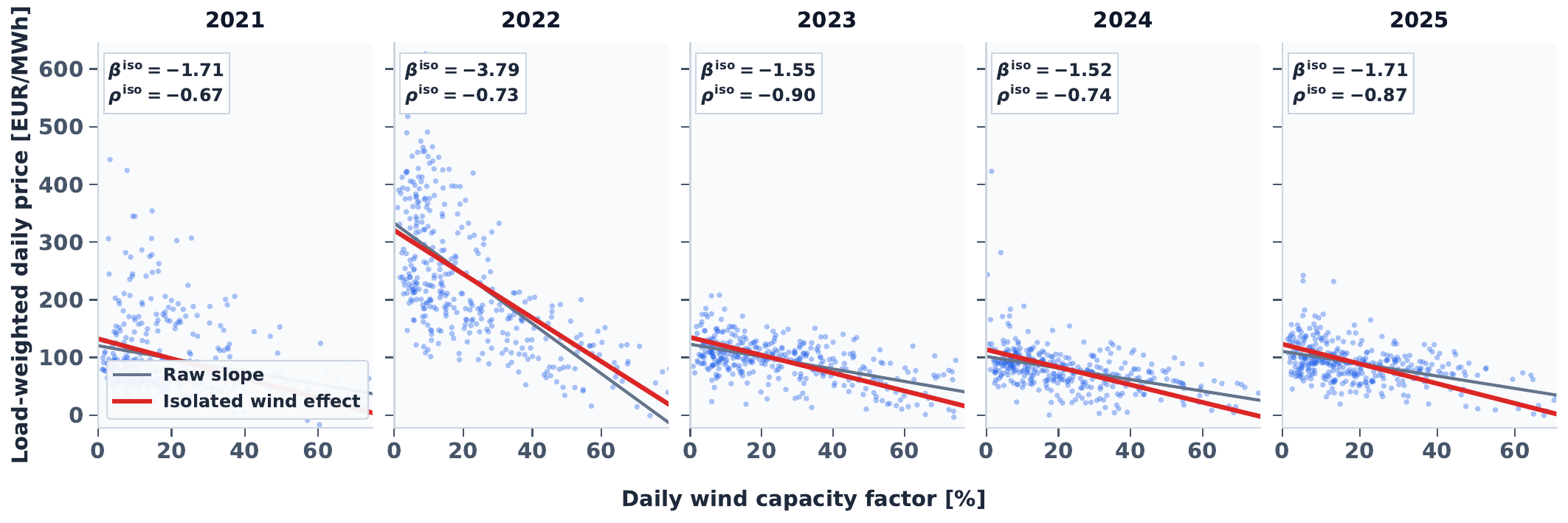}
\caption{Raw and isolated wind effects.}
\end{subfigure}\hfill
\begin{subfigure}[t]{.49\textwidth}
\centering
\includegraphics[width=\linewidth,height=.30\textheight,keepaspectratio]{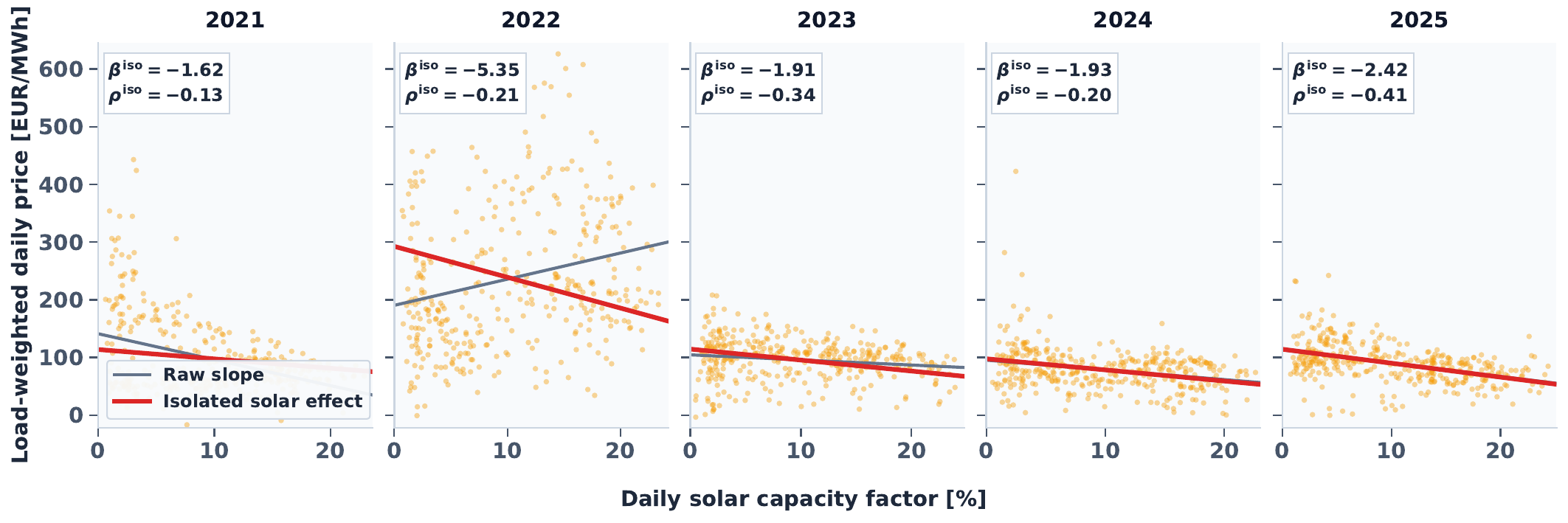}
\caption{Raw and isolated solar effects.}
\end{subfigure}\\[0.4em]
\begin{subfigure}[t]{.92\textwidth}
\centering
\includegraphics[width=\linewidth,height=.31\textheight,keepaspectratio]{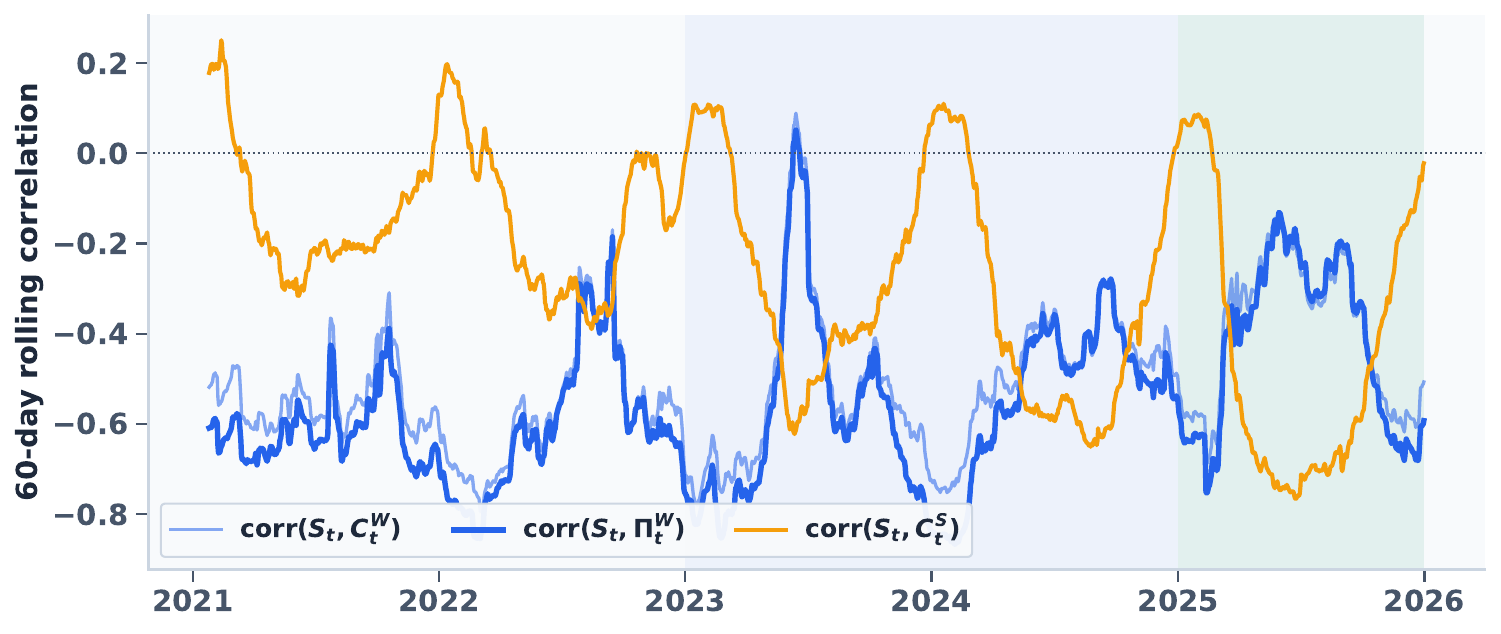}
\caption{Sixty-day rolling hourly spot--renewable correlations.}
\end{subfigure}
\caption{Isolated renewable price effects. The annual regressions control for
the other technology, load, holidays and calendar effects; the rolling panel
compares raw wind, load-adjusted wind pressure and solar output.}
\label{fig:emp-wind-effect}
\label{fig:emp-solar-effect}
\label{fig:emp-rolling-corr}
\end{figure*}

\subsubsection{Capture prices and capture factors}\label{subsec:emp-capture}

For technology $a\in\{\mathrm{wind},\mathrm{solar}\}$ with per-MW output
$q^a_t=C^a_t$ (installed capacity cancels), define for a period $P$ the
generation-weighted capture price, the baseload and demand-weighted market
references, and the value (capture) factor
\begin{equation}\label{eq:emp-capture}
\begin{aligned}
  CP^a_P&=\frac{\sum_{t\in P}q^a_tS_t}{\sum_{t\in P}q^a_t},
  &BL_P&=\frac{1}{|P|}\sum_{t\in P}S_t,\\
  DW_P&=\frac{\sum_{t\in P}L_tS_t}{\sum_{t\in P}L_t},
  &VF^a_P&=\frac{CP^a_P}{DW_P}.
\end{aligned}
\end{equation}
$CP^a_P$ is the realized counterpart of the fair strike
\eqref{eq:fair-strike-raw}, $BL_P$ of the flat reference $\bar F$ in
\eqref{eq:flat-and-volume-weighted}, and the wedge $CP^a_P-BL_P$ is the realized analogue
of the profile-plus-covariance correction
$\Delta^{\rm prof}_a+\Delta^{\rm cov}_a$ of \eqref{eq:flat-profile-cov}.
\Cref{fig:emp-capture-prices} plots monthly capture prices against the
demand-weighted market average and \cref{tab:emp-annual-capture} reports the
annual panel. Wind capture factors are low but comparatively stable
($0.68$ in 2022, $0.77$--$0.78$ in 2021 and 2023--2024, and $0.81$ in 2025); solar capture factors
fall after the crisis, from $0.96$ in 2022 to $0.75$ in
2023, $0.60$ in 2024 and $0.51$ in 2025---implying that the solar capture
price was only about half of the demand-weighted market price in 2025. \Cref{fig:emp-value-factors}
shows the monthly capture factors and their spread: the solar value factor exceeds the wind
value factor in months adjacent to winter and is far below it in summer. Across all 60
complete months the solar--wind value-factor correlation is $-0.25$, and the
spread $VF^{\mathrm S_\odot}-VF^{\mathrm W}$ has mean $0.03$ with standard
deviation $0.24$; solar lies below wind in $36.7\%$ of months. In 2025 alone
the spread averages $-0.16$ and solar lies below wind in seven of twelve
months. This imperfect co-movement creates potential diversification benefits across technologies.
Wind and solar PPAs are exposed to different capture-price
drivers. Combining the two technologies therefore
diversifies price-volume exposure across
meteorological and intraday conditions. The same mechanism underlies the growing
share of multi-technology (portfolio) PPAs in the European market: since wind
and solar generation profiles are imperfectly correlated, blending them
attenuates the intermittency of any single project and gives the offtaker a
generation profile whose timing and volume more closely match consumption, thereby
reducing shape, balancing, and volume risk relative to either technology in
isolation \cite{GABRIELLI2022105980}.

\begin{table}[t]
\caption{Annual realized capture panel for complete calendar years
\eqref{eq:emp-capture}: full-load hours per MW, capture price, baseload and
demand-weighted references (EUR/MWh), value factor and merchant revenue per MW.}
\label{tab:emp-annual-capture}
\resizebox{\columnwidth}{!}{
\begin{tabular}{rlrrrrrr}
\toprule
Year & Technology & FLH [h/MW] & Capture price & Baseload & Demand-wtd & Value factor & Revenue [kEUR/MW] \\
\midrule
2021 & solar & 833.33 & 76.09 & 96.85 & 100.20 & 0.76 & 63.41 \\
2021 & wind & 1733.41 & 78.28 & 96.85 & 100.20 & 0.78 & 135.69 \\
2022 & solar & 905.54 & 226.55 & 235.46 & 237.09 & 0.96 & 205.15 \\
2022 & wind & 1921.95 & 161.18 & 235.46 & 237.09 & 0.68 & 309.78 \\
2023 & solar & 855.62 & 73.65 & 95.19 & 98.11 & 0.75 & 63.01 \\
2023 & wind & 2204.96 & 76.17 & 95.19 & 98.11 & 0.78 & 167.96 \\
2024 & solar & 853.20 & 48.53 & 78.51 & 80.79 & 0.60 & 41.40 \\
2024 & wind & 1997.18 & 62.33 & 78.51 & 80.79 & 0.77 & 124.49 \\
2025 & solar & 892.35 & 47.40 & 89.33 & 92.20 & 0.51 & 42.30 \\
2025 & wind & 1713.24 & 74.64 & 89.33 & 92.20 & 0.81 & 127.87 \\
\bottomrule
\end{tabular}
}
\end{table}

\begin{figure*}[p]
\centering
\begin{subfigure}[t]{.96\textwidth}
\centering
\includegraphics[width=\linewidth,height=.34\textheight,keepaspectratio]{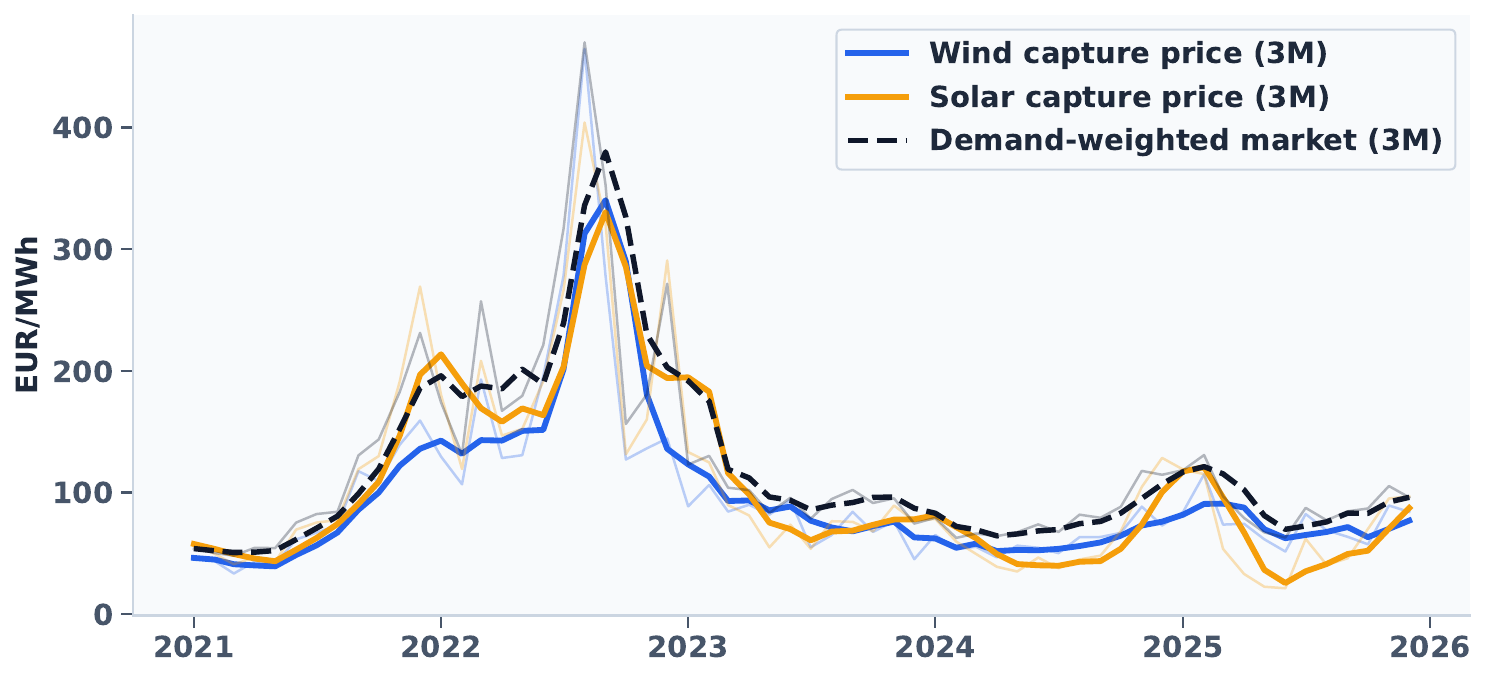}
\caption{Monthly generation-weighted capture prices against the
demand-weighted market average.}
\end{subfigure}\\[0.4em]
\begin{subfigure}[t]{.96\textwidth}
\centering
\includegraphics[width=\linewidth,height=.34\textheight,keepaspectratio]{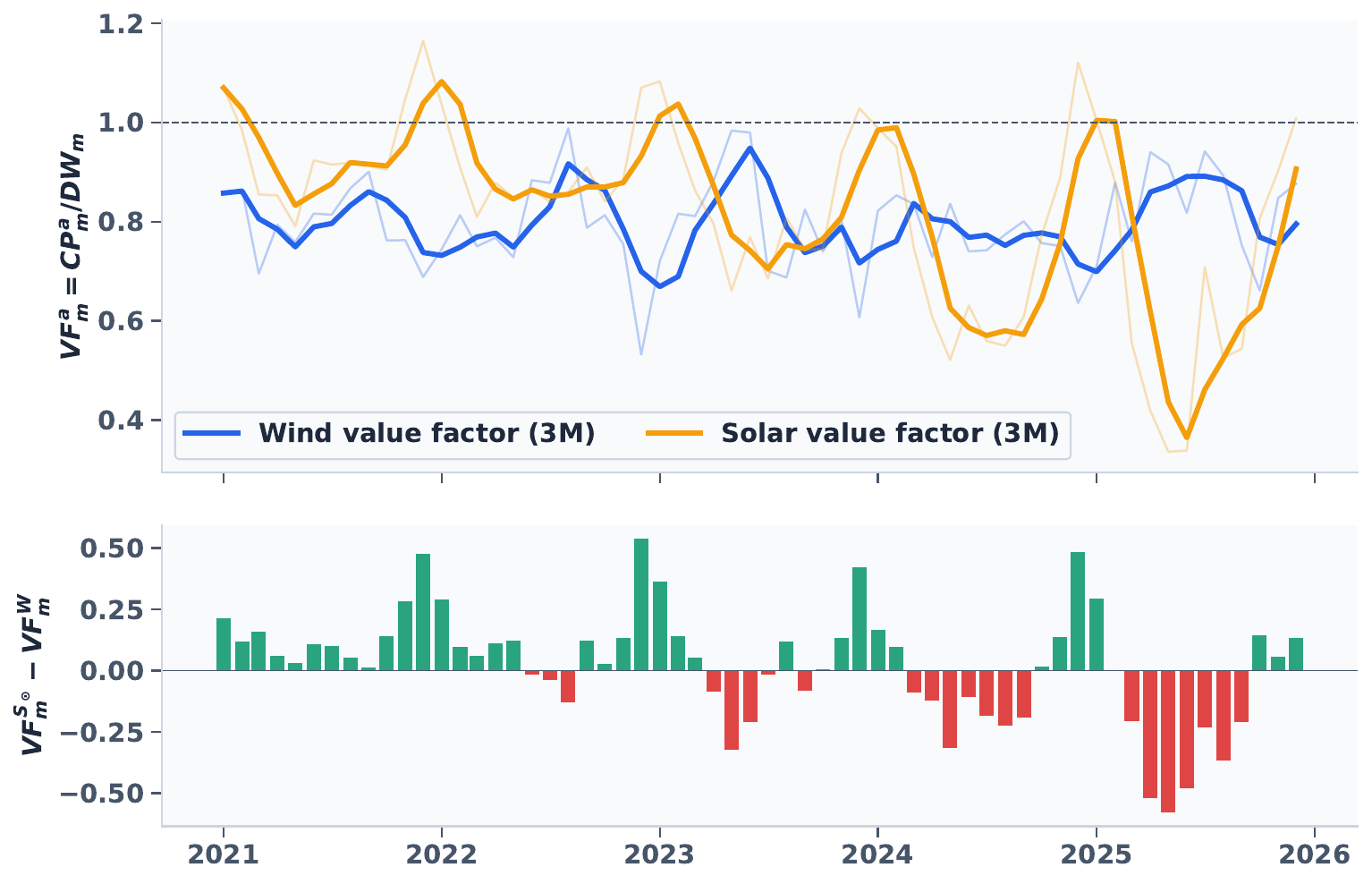}
\caption{Monthly wind and solar value factors and their cross-technology
spread.}
\end{subfigure}
\caption{Capture prices and value factors for German wind and solar
generation. Thick curves report three-month rolling means.}
\label{fig:emp-capture-prices}
\label{fig:emp-value-factors}
\end{figure*}

\subsubsection{Concentration of renewable output across price buckets}\label{subsec:emp-buckets}

Capture factors fall when a larger share of renewable output is produced during low-price hours. With 50~EUR/MWh
spot-price buckets $b$, the technology and demand shares and the lean are
\begin{equation}\label{eq:emp-lean}
  \pi^a_b=\frac{\sum_{t:S_t\in b}q^a_t}{\sum_tq^a_t},\quad
  \pi^L_b=\frac{\sum_{t:S_t\in b}L_t}{\sum_tL_t},\quad
  \ell^a_b=\pi^a_b-\pi^L_b ,
\end{equation}

and the low-price lean aggregates $\ell^a$ over the bottom-quartile price
hours of each year. \Cref{fig:emp-price-buckets} shows that in 2025 a much larger
share of solar generation than of demand occurs in the lowest-price quartile;
the solar low-price lean rises from $+1.8$ percentage points in 2022 to
$+36.2$ in 2025, while the wind lean \emph{falls} from its 2022 peak of
$+25.5$ to $+12.5$ because, in the weak-wind year, a smaller share of wind output occurs
during low-price hours. The share of solar MWh delivered into strictly negative prices
reaches $24.5\%$ in 2025 ($9.3\%$ for wind). In the language of
\cref{prop:ppa-decomp}, $\ell^a$ is an empirical indicator of
$\Gamma^a$: a positive low-price lean indicates negative realized covariance between
production and price, and its growth over time explains why a constant
(profile-only) capture discount would misprice the 2025 window.

\subsubsection{State-dependent negative jump intensity}\label{subsec:emp-jumps}

Negative price spikes occur more frequently when renewable
penetration is high. To quantify this state dependence, we define the
negative-spike indicator
$\mathbf 1\{\Delta S_t<0,\ S_t\le q_{0.20}(S)\}$ identified by the spike
filter of \cref{subsec:emp-seasonality} and estimate, for each candidate
driver $x_t\in\{C_t,\Pi_t\}$, the annualized Epanechnikov kernel intensity
\begin{equation}\label{eq:emp-kernel-intensity}
  \widehat\lambda^-(x)
  =\frac{\sum_tK_h(x_t-x)\,\mathbf 1\{\text{event at }t\}}
        {\sum_tK_h(x_t-x)\,\Delta t},
\end{equation}
summarized by the monotone two-state fit
$\lambda^-(x)=\lambda^-_{\rm low}\mathbf 1\{x\le x^\star\}
+\lambda^-_{\rm high}\mathbf 1\{x>x^\star\}$.
\Cref{fig:emp-jump-intensity} and \cref{tab:emp-jump-intensity} support state dependence
and show that load adjustment improves the separation of the two regimes. Against the raw
wind capacity factor the fitted intensity rises from $35$ to $84$ events per
year across the threshold; against the load-adjusted pressure
$\Pi^{\mathrm W}$ the separation widens to $41$ versus $140$ per year (a
factor $3.4$) with threshold $x^\star=0.60$, and the kernel curve is monotone
over the whole support instead of flattening at high capacity factors. For
solar, we use the realized capacity factor $C^{\mathrm S}_t$ itself as the
driver on which the solar state of \cref{subsec:emp-seasonality} is built:
its intensity rises from $34$ to $122$ events per year across the threshold
$x^\star=0.35$, while the load-adjusted variant is nearly indistinguishable
($37$ versus $128$, threshold $0.38$). Because the load adjustment
changes the solar fit only marginally, the model uses
$C^{\mathrm S}$ directly. For wind, the monotone increase in negative-spike intensity
is consistent with the wind-penetration mechanism documented by
\citet{DeschatreVeraart2018}. The solar relation is a separate empirical
extension supported by the estimates in the present sample, rather than a
result attributed to that study. Together, these findings motivate the
technology-specific state-dependent negative-spike intensities adopted in
\cref{subsec:spikes}.

\begin{table}[t]
\caption{Two-state negative jump intensity fits (low-price filtered events;
79 events on 2023--2024): low and high intensities, driver threshold
$x^\star$, and the constant positive intensity, per
\eqref{eq:emp-kernel-intensity}. A checkmark marks the driver adopted for
the state-dependent spike intensity of the model (load-adjusted $\Pi$ for
wind; the realized capacity factor $C$ for solar).}
\label{tab:emp-jump-intensity}
\resizebox{\columnwidth}{!}{
\begin{tabular}{lllrrrrr}
\toprule
Panel & Driver & Model & Events & $\lambda^-_{\rm low}$ [yr$^{-1}$] & $\lambda^-_{\rm high}$ [yr$^{-1}$] & Threshold $x^\star$ & $\lambda^+$ [yr$^{-1}$] \\
\midrule
Wind & realized $C$ &  & 79 & 35.17 & 84.46 & 0.44 & 200.25 \\
Wind & load-adjusted $\Pi$ & $\checkmark$ & 79 & 41.35 & 140.20 & 0.60 & 200.25 \\
Solar & realized $C$ & $\checkmark$ & 79 & 34.36 & 122.38 & 0.35 & 200.25 \\
Solar & load-adjusted $\Pi$ &  & 79 & 36.66 & 128.27 & 0.38 & 200.25 \\
\bottomrule
\end{tabular}
}
\end{table}

\begin{figure*}[p]
\centering
\begin{subfigure}[t]{.96\textwidth}
\centering
\includegraphics[width=\linewidth,height=.33\textheight,keepaspectratio]{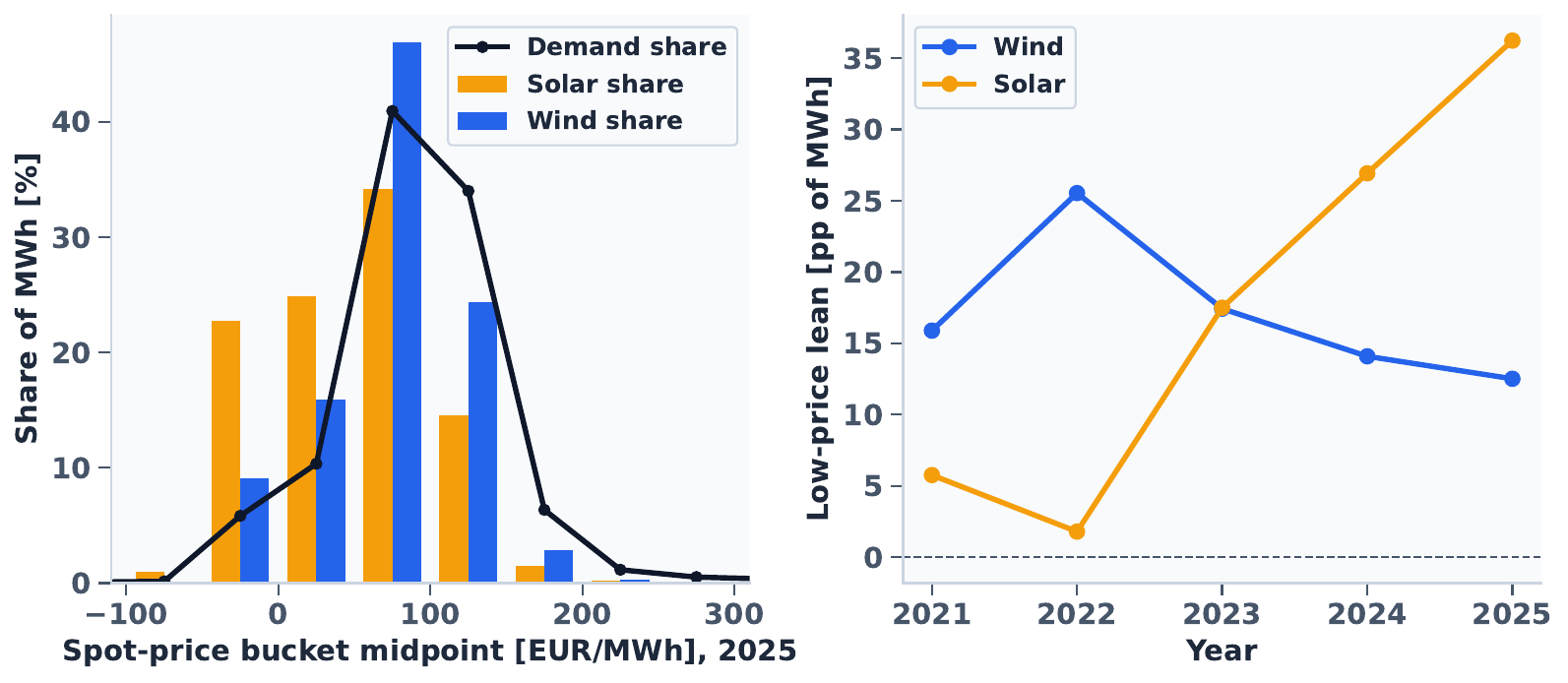}
\caption{Renewable production shares across price buckets and the annual
low-price lean relative to demand.}
\end{subfigure}\\[0.4em]
\begin{subfigure}[t]{.92\textwidth}
\centering
\includegraphics[width=\linewidth,height=.33\textheight,keepaspectratio]{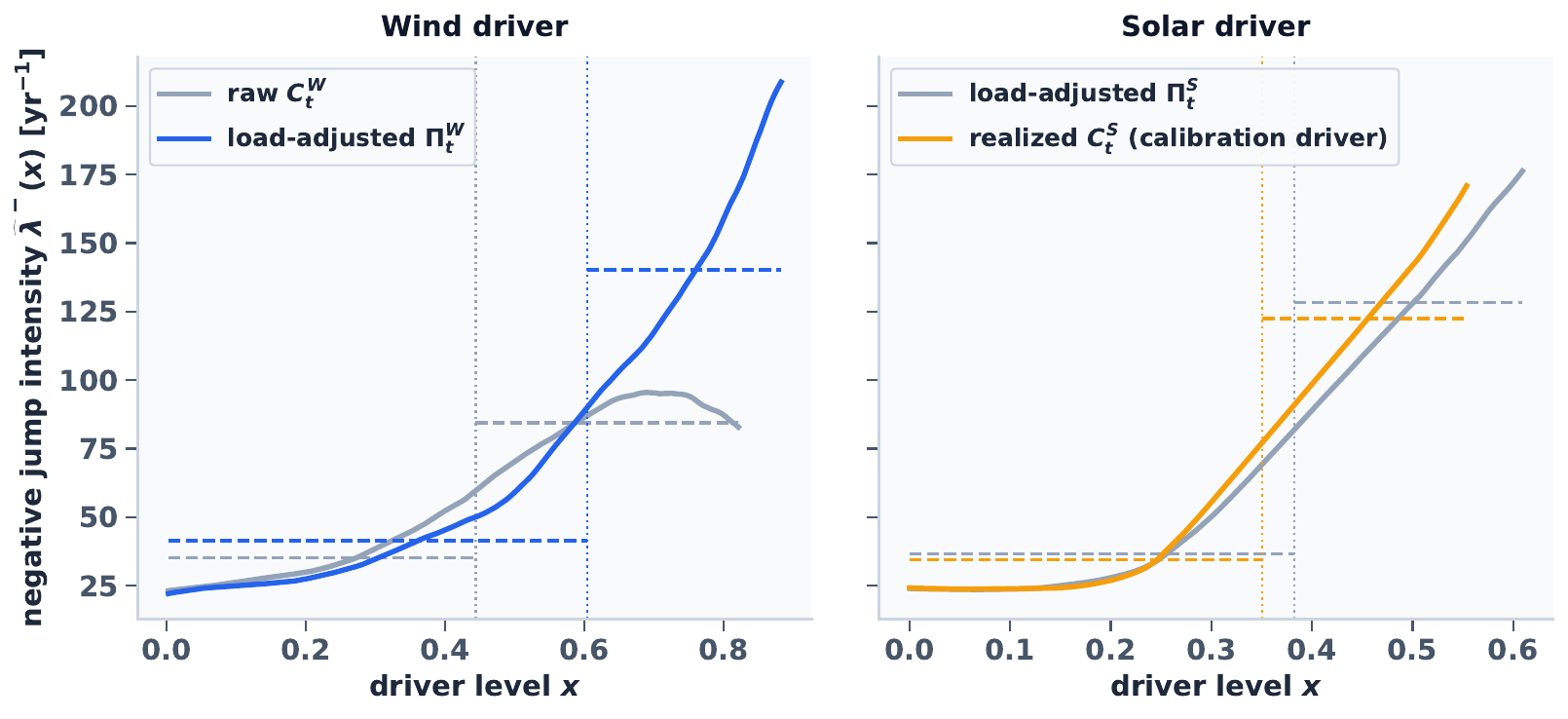}
\caption{State-dependent negative-price jump intensities for wind and solar.}
\end{subfigure}
\caption{Renewable-output concentration in low-price hours and the corresponding
negative-spike intensity. Dashed lines in the lower panel show the fitted two-state jump
intensities and thresholds.}
\label{fig:emp-price-buckets}
\label{fig:emp-jump-intensity}
\end{figure*}

\subsubsection{Spot, futures and observed PPA prices}\label{subsec:emp-forward}

Finally, we compare realized spot prices and capture prices with observed
exchange-traded futures settlements and quoted PPA indices. The comparison is
descriptive because the contracts differ in delivery horizon, volume profile,
credit exposure, and quotation date. For each delivery month \(m\), let
\(T_m\) denote the first delivery hour and define
\[
q_m^- := \max\{q : q < T_m\},
\]
so that \(q_m^-\) is the final available quotation date strictly before
delivery. Let \(F_m^{B}\) and \(F_m^{P}\) denote, respectively, the EEX
baseload- and peakload-month settlement prices observed on \(q_m^-\). The
ex-post baseload futures premium is defined as
\begin{equation}
\label{eq:emp-risk-premium}
  RP_m^{B}
  :=
  F_m^{B} - \bar{S}_m^{\mathrm{obs}},
  \qquad
  \bar{S}_m^{\mathrm{obs}}
  :=
  \frac{1}{\lvert I_m\rvert}
  \sum_{i\in I_m} S_i^{\mathrm{obs}},
\end{equation}
where \(I_m\) is the set of delivery hours in month \(m\). This ex-post
quantity measures the difference between the final pre-delivery baseload-month
settlement and the subsequently realized monthly average spot price. It should
not be interpreted as a risk premium identified under a particular pricing
measure.

\Cref{fig:emp-risk-premia,tab:emp-risk-premia} report the corresponding
descriptive statistics. The mean value of \(RP_m^{B}\) is
\(14.8\)~EUR/MWh over 2021--2022, with a standard deviation of
\(51.1\)~EUR/MWh; \(4.1\)~EUR/MWh over the 2023--2024 calibration window,
with a standard deviation of \(15.3\)~EUR/MWh; and
\(-0.5\)~EUR/MWh over the eight 2025 delivery months for which the final
pre-delivery EEX quote is available, with a standard deviation of
\(4.5\)~EUR/MWh.

\Cref{fig:emp-ppa} provides a separate descriptive comparison between the
observed German PPA index and contemporaneous front-month and front-year
futures prices. Over the 2019--2024 sample, the quoted PPA index lies, on
average, \(52.7\)~EUR/MWh below the front-month futures price and
\(61.0\)~EUR/MWh below the front-year futures price. These spreads are
consistent with differences in delivery tenor, production profile,
price--volume covariance exposure, contract structure, and credit and
liquidity conditions.

\begin{table}[!t]
\caption{Ex-post baseload futures risk premium \eqref{eq:emp-risk-premium} by
sample window (EUR/MWh).}
\label{tab:emp-risk-premia}
\resizebox{\columnwidth}{!}{
\begin{tabular}{lrrrrrr}
\toprule
Sample & Months & Mean & Median & Std & MAE & RMSE \\
\midrule
Context & 24 & 14.82 & 1.82 & 51.05 & 35.18 & 52.13 \\
Out-of-sample 2025 & 8 & -0.47 & -0.37 & 4.53 & 3.14 & 4.27 \\
Calibration window & 24 & 4.07 & 3.67 & 15.31 & 12.24 & 15.53 \\
\bottomrule
\end{tabular}
}
\end{table}

\begin{figure*}[!t]
\centering
\includegraphics[width=.78\textwidth]{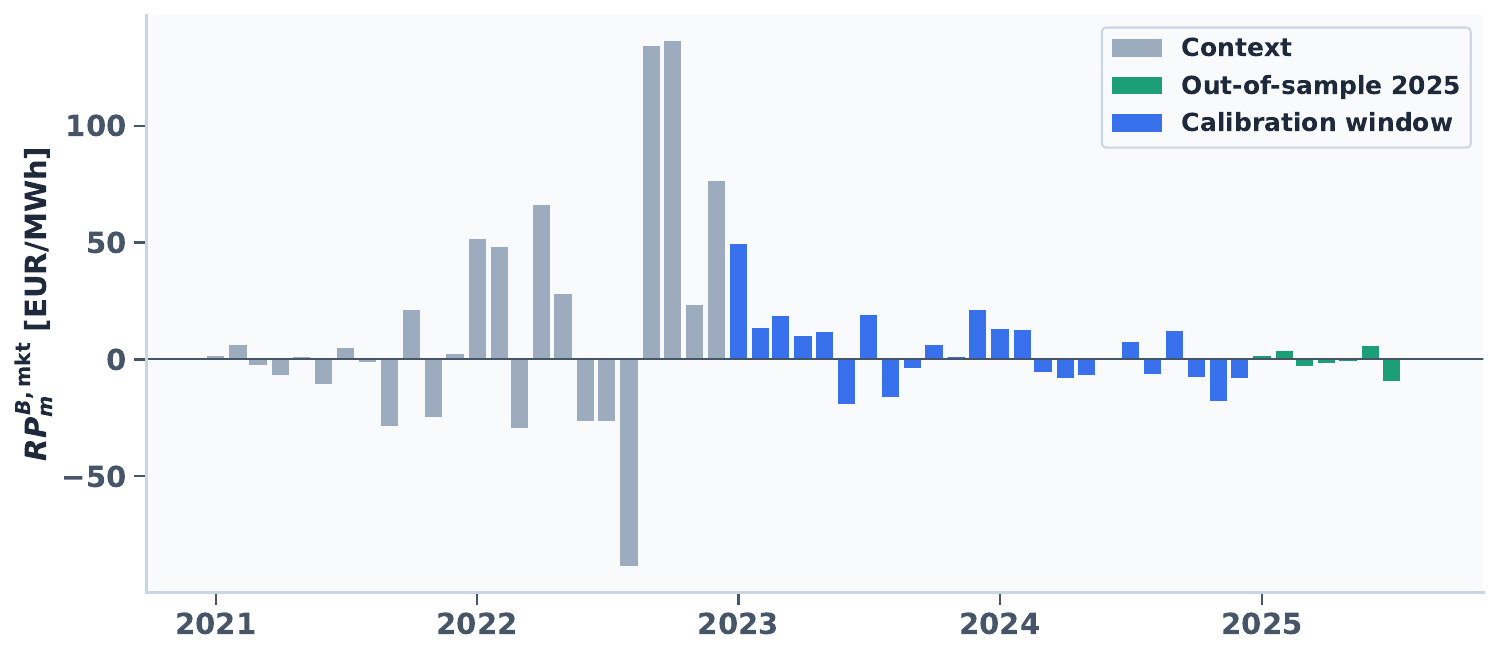}\\[0.6em]
\includegraphics[width=.78\textwidth]{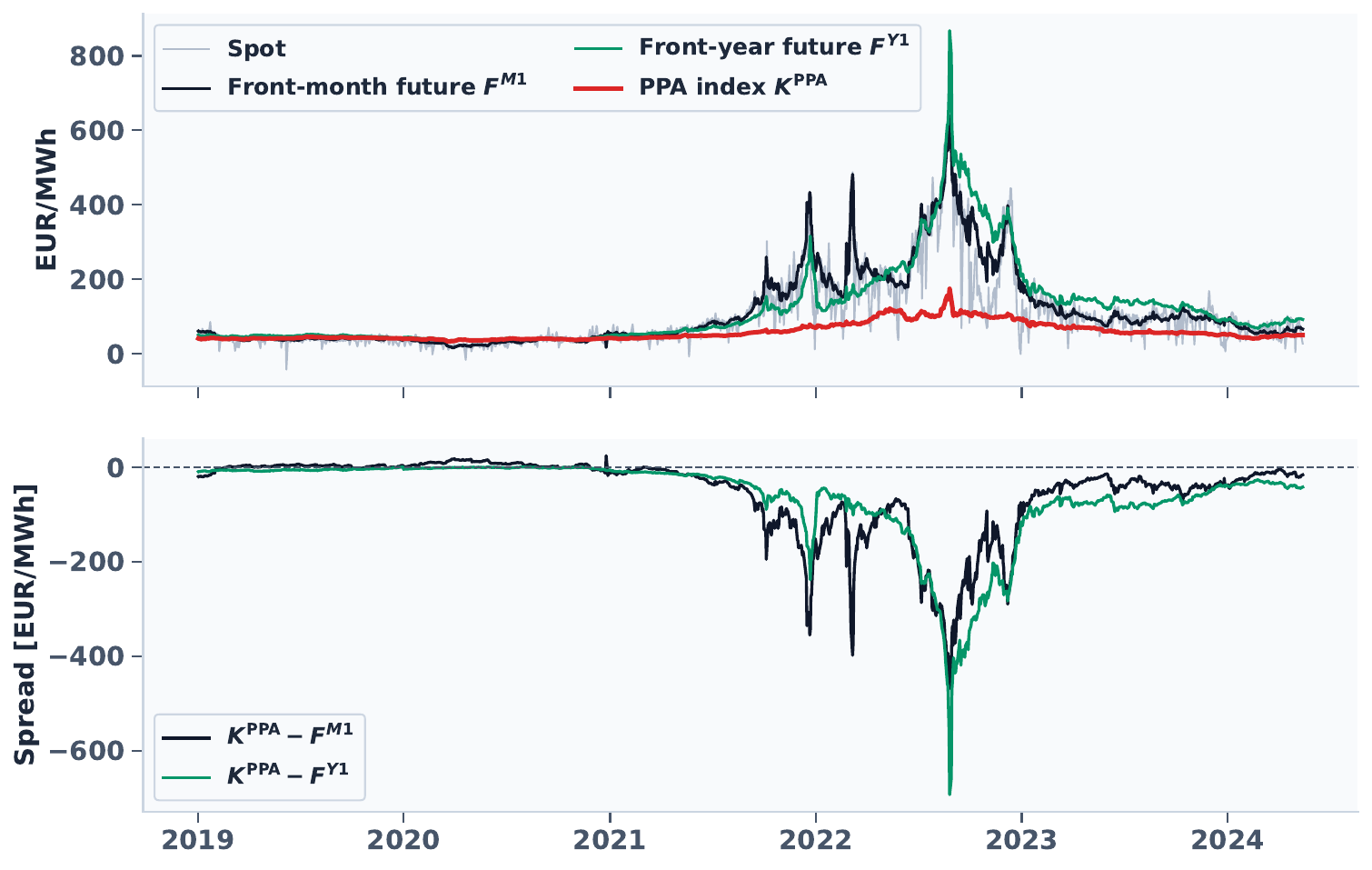}
\caption{Observed futures and PPA quotations compared with realized spot prices. Top panel:
monthly ex-post baseload futures risk premium \eqref{eq:emp-risk-premium},
last pre-delivery EEX base-month settlement minus realized monthly baseload
spot, coloured by sample window. Bottom panels: observed Germany PPA index
against spot and contemporaneous front-month and front-year futures, and the
diagnostics for the spread between the PPA index and futures prices. Data are obtained from
\citet{PENA2024101513} and sourced from Pexapark.}
\label{fig:emp-risk-premia}\label{fig:emp-ppa}
\end{figure*}

We decompose the quoted discount directly. For
technology $a$ with realized monthly capture price $CP^a_m$ and realized
baseload $BL_m=\bar S^{\rm obs}_m$ as in \eqref{eq:emp-capture}, the total
realized wedge between the tradable baseload future and the pay-as-produced
capture price splits exactly as
\begin{equation}\label{eq:emp-forward-wedge}
  F^{B}_m-CP^a_m
  =\underbrace{F^{B}_m-BL_m}_{\text{baseload risk premium}}
  +\underbrace{BL_m-CP^a_m}_{\text{realized capture discount}} ,
\end{equation}
and the realized capture discount itself splits exactly into empirical
counterparts of the profile and covariance corrections
\crefrange{eq:profile-correction}{eq:covariance-correction}. With
$\bar S_{m,h}$ and $\bar q^a_{m,h}$ the hour-of-day means of price and
production within month $m$, and
$CP^{a,\rm prof}_m
=\sum_h\bar q^a_{m,h}\bar S_{m,h}\big/\sum_h\bar q^a_{m,h}$
the capture price of the \emph{average} intraday production profile
evaluated against the \emph{average} intraday price profile,
\begin{equation}\label{eq:emp-profile-cov}
  CP^a_m-BL_m
  =\underbrace{CP^{a,\rm prof}_m-BL_m}_{\widehat\Delta^{a,\rm prof}_m}
  +\underbrace{CP^a_m-CP^{a,\rm prof}_m}_{\widehat\Delta^{a,\rm cov}_m} .
\end{equation}
The profile term measures how the average production profile
weights the average intraday price profile and can be hedged
with fixed-volume shape instruments. The covariance term captures
within-month price-production co-movement around hourly means and
cannot be hedged with fixed-volume instruments.

\begin{table}[!t]
\caption{Realized forward capture decomposition by technology and sample
window (monthly means, EUR/MWh): baseload risk premium $RP^{B}$, realized
capture discount $BL-CP$, its profile and covariance components
\eqref{eq:emp-profile-cov}, and the total wedge $F^{B}-CP$
\eqref{eq:emp-forward-wedge}. Complete delivery months with a pre-delivery
base-month settlement, 2021--2025.}
\label{tab:emp-forward-capture}
\resizebox{\columnwidth}{!}{
\begin{tabular}{llrrrrrr}
\toprule
Technology & Window & Months & $RP^{B}$ & $BL-CP$ & $\widehat\Delta^{\rm prof}$ & $\widehat\Delta^{\rm cov}$ & $F^{B}-CP$ \\
\midrule
Solar & calibration & 24 & 4.06 & 14.52 & -12.32 & -2.21 & 18.59 \\
Solar & context & 24 & 14.78 & 9.58 & -6.43 & -3.15 & 24.36 \\
Solar & out-of-sample & 12 & 1.01 & 24.94 & -21.57 & -3.37 & 25.96 \\
Wind & calibration & 24 & 4.06 & 17.14 & -0.45 & -16.69 & 21.20 \\
Wind & context & 24 & 14.78 & 30.20 & -0.01 & -30.19 & 44.98 \\
Wind & out-of-sample & 12 & 1.01 & 14.05 & 0.29 & -14.34 & 15.06 \\
\bottomrule
\end{tabular}
}
\end{table}

\begin{figure*}[p]
  \centering
  
  \begin{minipage}[c]{0.485\textwidth}
  \centering
  
  \begin{subfigure}[t]{\linewidth}
  \centering
  \includegraphics[
      width=\linewidth,
      height=.225\textheight,
      keepaspectratio
  ]{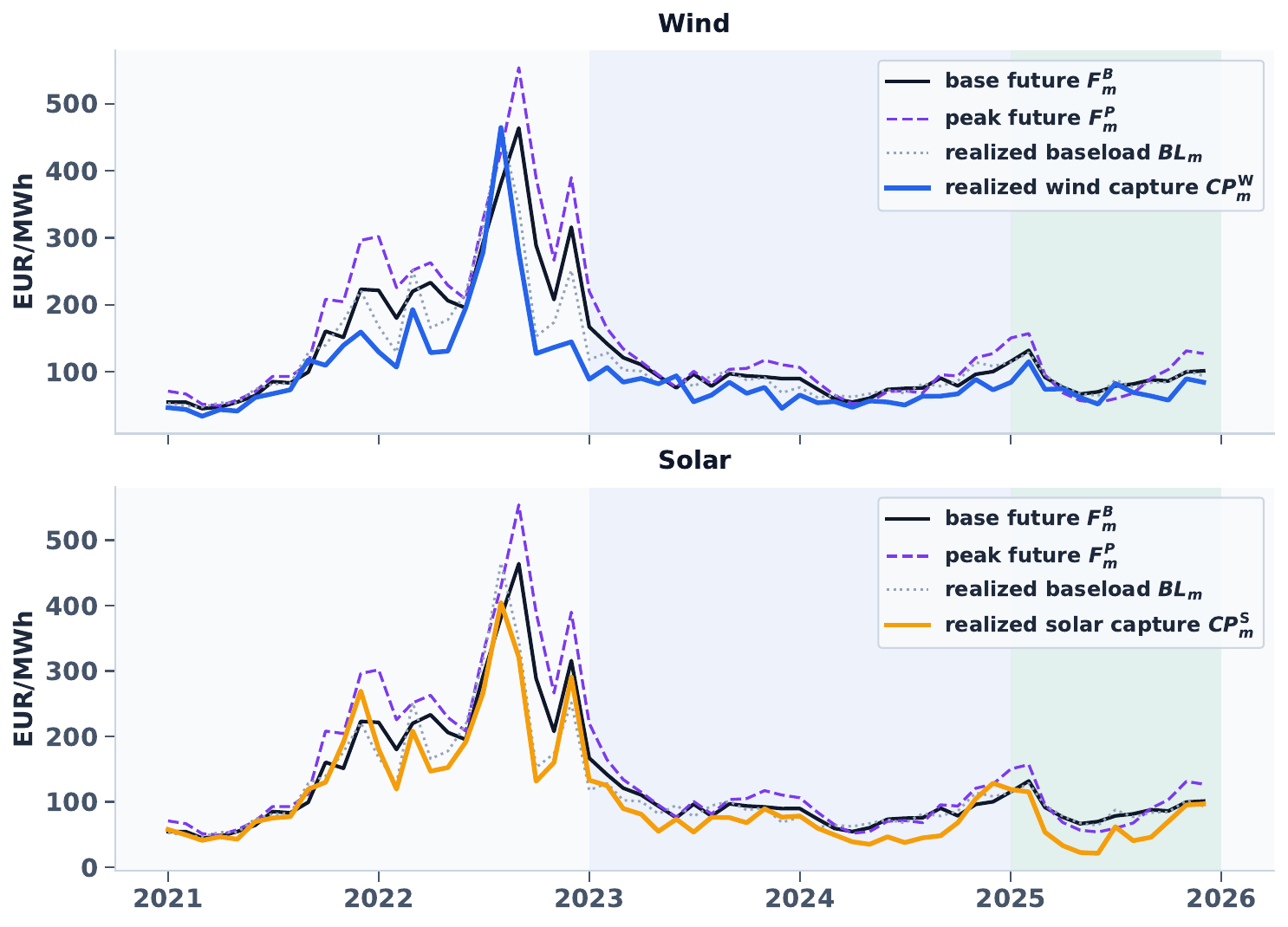}
  \caption{Pre-delivery base and peak futures, realized baseload and realized
  wind/solar capture prices.}
  \label{fig:emp-capture-futures}
  \end{subfigure}
  
  \vspace{0.8em}
  
  \begin{subfigure}[t]{\linewidth}
  \centering
  \includegraphics[
      width=\linewidth,
      height=.225\textheight,
      keepaspectratio
  ]{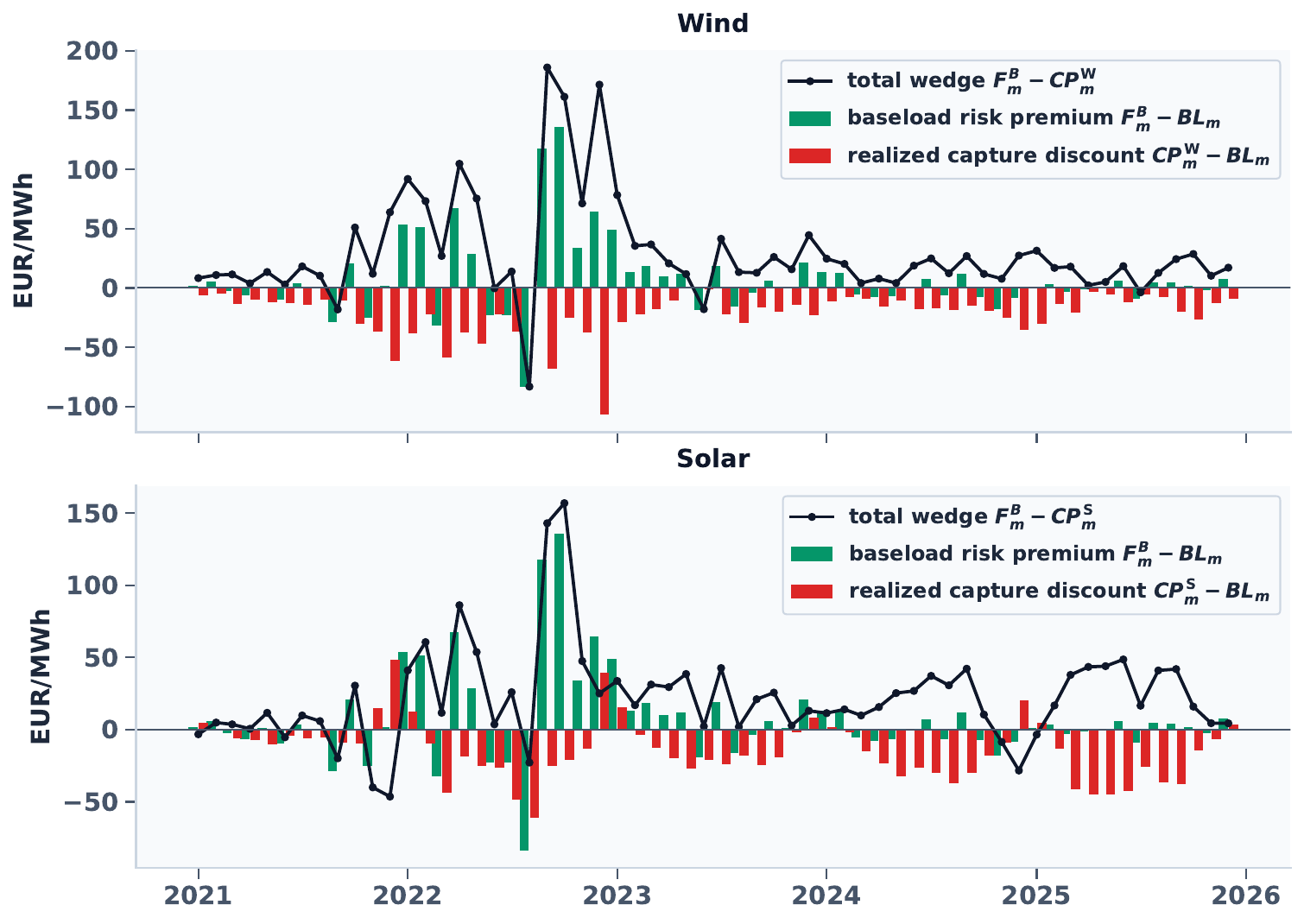}
  \caption{Baseload risk premium, realized capture discount and total
  futures--capture wedge.}
  \label{fig:emp-forward-decomp}
  \end{subfigure}
  
  \end{minipage}
  \hfill
  \begin{minipage}[c]{0.485\textwidth}
  \centering
  
  \begin{subfigure}[c]{\linewidth}
  \centering
  \includegraphics[
      width=\linewidth,
      height=.50\textheight,
      keepaspectratio
  ]{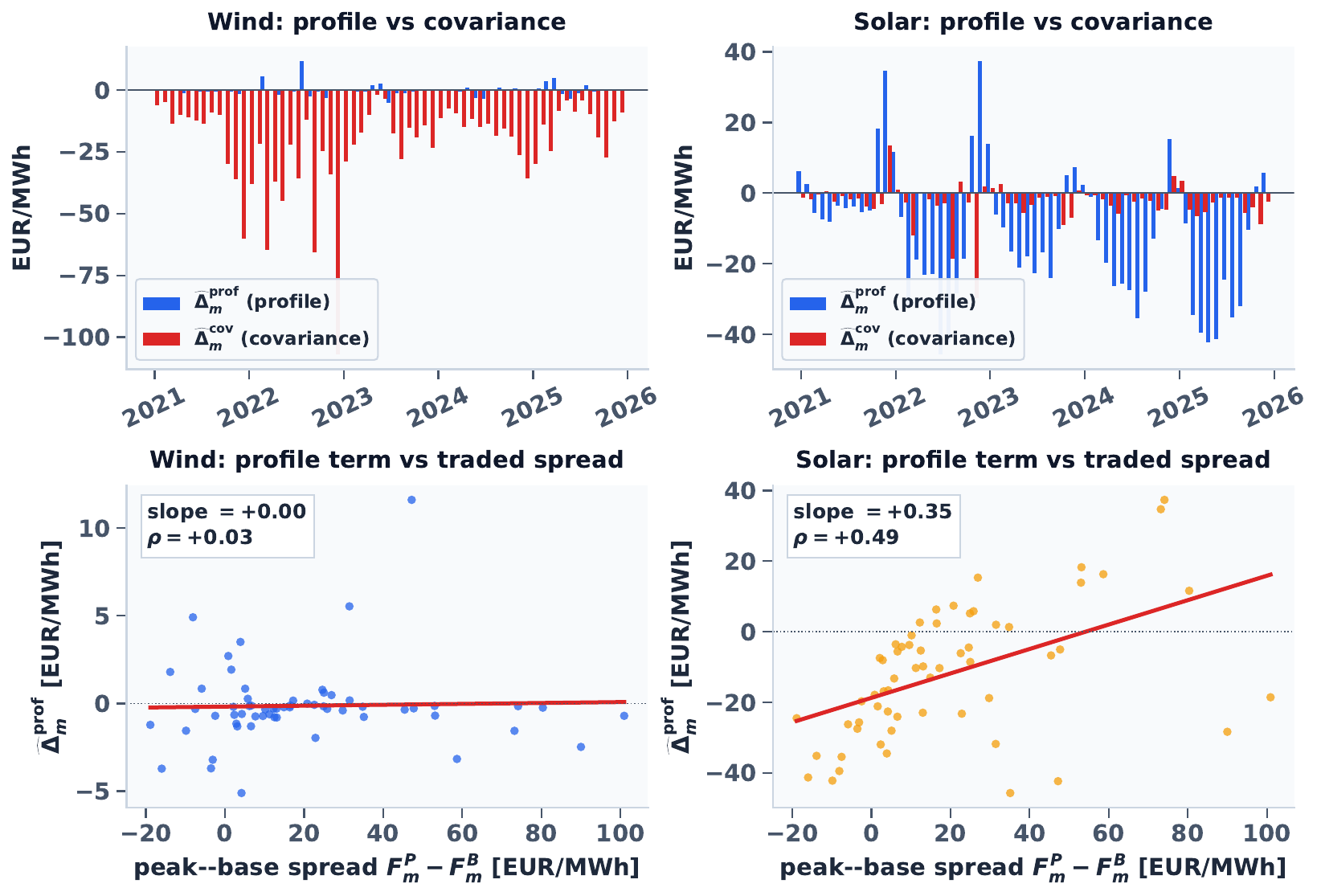}
  \caption{Profile/covariance split and the relation of the profile component to
  the traded peak--base spread.}
  \label{fig:emp-profile-cov}
  \end{subfigure}
  
  \end{minipage}
  
  \caption{Observed futures and the realized forward-capture decomposition for
  Germany, 2021--2025. The panels distinguish futures bias, deterministic
  delivery shape and stochastic price--production covariance.}
  \label{fig:emp-forward-capture-summary}
  \end{figure*}

\Cref{fig:emp-capture-futures} shows the four monthly price series over
2021--2025, \cref{fig:emp-forward-decomp} the decomposition
\eqref{eq:emp-forward-wedge}, and \cref{fig:emp-profile-cov} the split
\eqref{eq:emp-profile-cov}; \cref{tab:emp-forward-capture} summarizes by
sample window. Three findings emerge. First, the total wedge
$F^{B}_m-CP^a_m$ is large and persistent for both technologies ($+21.2$
(wind) and $+18.6$ (solar) EUR/MWh on average over the calibration window,
$+15.1$ and $+26.0$ over the twelve complete out-of-sample months), while the
mean baseload risk premium contributes only $+1.0$~EUR/MWh out of sample: the
wedge is almost entirely realized capture discount.
Second, the split \eqref{eq:emp-profile-cov} separates the two technologies
sharply. The wind discount is essentially pure covariance: the wind profile
term is economically zero in every window ($|\widehat\Delta^{\mathrm
W,\rm prof}|\le0.5$~EUR/MWh) while the covariance term carries the entire
discount ($-30.2$ in the crisis years, $-16.7$ in calibration, $-14.3$ out
of sample). For solar, the covariance term is small and stable ($-2.2$ to $-3.4$), while the profile term
dominates and deteriorates from $-6.4$ (2021--2022) to
$-12.3$ (calibration) and $-21.6$ (2025), quantifying the
growing midday profile discount in EUR/MWh documented in
\cref{subsec:emp-capture}. Third, the two components have different
hedgeability, and one component can already be hedged with an exchange-traded instrument:
the monthly solar profile term co-moves with the traded peak--base spread
$F^{P}_m-F^{B}_m$ with correlation $+0.49$ (regression slope $+0.35$ per
EUR/MWh of spread), whereas the wind profile term shows no such loading
($+0.03$ correlation). A peak--base futures position therefore spans a substantial part
of the solar \emph{profile} risk, while the dominant wind component
and the residual solar covariance depend on production-weighted
price-volume co-movement that fixed-volume futures cannot hedge. These results motivate dynamic
futures hedging for the spanned components and the semi-static production-linked portfolio
of \cref{subsec:semistatic} for the
remaining covariance risk.

The empirical findings motivate the model and hedging choices used in the remainder
of the paper. Pronounced but well-fitted seasonality
(\cref{subsec:emp-seasonality}) justifies deterministic seasonal layers with
stationary residuals. The multiscale autocorrelations and the large
peak-to-background ratios at the candidate calendar frequencies
(\cref{subsec:emp-acf}) motivate higher-order continuous-time dynamics
combining several relaxation scales with damped oscillator pairs. Under the
two-pair order cap imposed below, the daily and half-daily periods are included in the
two-oscillator specification, while the weekly and eight-hour peaks remain
descriptive diagnostics. The strong, load-dependent negative price--wind
dependence (\cref{subsec:emp-wind-effect}) requires a joint stochastic
state for price and renewable output, which \cref{sec:model} builds on the
coordinates introduced above. Diverging and dispersed capture factors with a
growing low-price lean (\cref{subsec:emp-capture,subsec:emp-buckets}) show that
the covariance correction $\Delta^{\rm cov}$ varies over time and is economically
material, so futures alone are insufficient: the theory of
\cref{sec:model-free-hedging,sec:single-period} therefore complements the
dynamic futures projection with a semi-static layer of production-linked
claims. The regime structure of negative spikes (\cref{subsec:emp-jumps})
calls for a state-dependent spike intensity (\cref{subsec:spikes}). Finally,
the realized forward capture decomposition (\cref{subsec:emp-forward}) shows
that the wedge between tradable baseload futures and pay-as-produced capture
prices is dominated by the capture discount rather than by futures bias, and
that its two components differ in hedgeability: the solar profile term loads
on the traded peak--base spread, while the wind discount and the residual
solar covariance arise from production-weighted price-volume co-movement
that fixed-volume futures cannot hedge. The residual risk of a PPA after
futures hedging is therefore concentrated in covariance components
targeted by the semi-static hedge in \cref{subsec:semistatic};
the resulting risk reduction is reported in
\cref{sec:empirical-design}.

\section{Semi-static variance-optimal hedging}\label{sec:model-free-hedging}
This section separates the PPA hedge into a dynamic projection onto the
active delivery-period futures market, followed by a static projection
of the orthogonal residual onto production-linked claims. The construction is
first stated for a general futures stack, then specialized to a single delivery
window and finally implemented with the time-varying active sets relevant for
power markets. Proofs and extended transform calculations are collected in the
\suppmat.
\subsection{Dynamic projection on liquid futures}\label{subsec:dynamic-projection}

Let $\bm X=(X^1,\ldots,X^d)^\top$ be the vector of discounted value processes of
all liquid futures available over the hedging horizon; components may be power
monthly, quarterly or yearly delivery futures and, in an extended market,
renewable-volume futures. Since contracts enter and leave the hedgeable set,
define the diagonal predictable availability matrix
\begin{equation}\label{eq:availability}
  \bm\Pi_t^{\rm act}
  =\diag\left(\1_{\{1\in\Act_t\}},\ldots,\1_{\{d\in\Act_t\}}\right),
\end{equation}
where $\Act_t$ is the active index set at time $t$. A predictable strategy
$\bm\theta_t\in\R^d$ generates terminal gain
\begin{equation}\label{eq:gain-process}
  G_T^\theta
  =\int_0^T\bm\theta_t^\top\bm\Pi_t^{\rm act}\dd\bm X_t ,
\end{equation}
and the admissible integrand space is
\begin{equation}\label{eq:theta-space}
  \Theta_X=
  \left\{\bm\theta:
  \E\left[\int_0^T
  \bm\theta_t^\top\bm\Pi_t^{\rm act}
  \dd\brac{\bm X}{\bm X}_t
  \bm\Pi_t^{\rm act}\bm\theta_t\right]<\infty\right\}.
\end{equation}
For $H\in L^2(\scrF_T,\Q)$, the dynamic variance-optimal problem is
\begin{equation}\label{eq:dynamic-problem}
  \epsilon_{\rm dyn}^2(H)
  =
  \inf_{\bm\theta\in\Theta_X}
  \E\left[(H-\E[H]-G_T^\theta)^2\right].
\end{equation}
Since $\bm X$ is a square-integrable $\Q$-martingale, mean-variance optimality
coincides with local risk minimization and the F\"ollmer--Schweizer
decomposition reduces to the GKW decomposition
\citep{Galchouk1976,KunitaWatanabe1967,FollmerSondermann1985}; the
semimartingale extension under the structure condition is given by
\citet{Schweizer1992,Schweizer1994}.

\begin{lemma}[Closedness of the active-futures gain space]\label{lem:gain-closedness}
Let $\bm X$ be an $\R^d$-valued square-integrable $\Q$-martingale on $[0,T]$ and
$\bm\Pi^{\rm act}$ a bounded predictable diagonal orthogonal projector. Then the
terminal gain space
$\scrG_X=\{\int_0^T\bm\theta_t^\top\bm\Pi_t^{\rm act}\dd\bm X_t:
\bm\theta\in\Theta_X\}$ is a closed subspace of $L^2(\scrF_T,\Q)$, and the
infimum in \eqref{eq:dynamic-problem} is attained.
\end{lemma}

Let $V_t^H=\E_t[H]$. Choose a predictable increasing process $A^X$ dominating
the predictable covariations of $\bm X$ and $V^H$, and write
\begin{equation}\label{eq:bracket-densities}
  \dd\brac{\bm X}{\bm X}_t=\bm c_t^{XX}\dd A_t^X,
  \qquad
  \dd\brac{\bm X}{V^H}_t=\bm c_t^{XH}\dd A_t^X ,
\end{equation}
with active densities
\begin{equation}\label{eq:active-densities}
  \bm c_t^{XX,\rm act}
  =\bm\Pi_t^{\rm act}\bm c_t^{XX}\bm\Pi_t^{\rm act},\qquad
  \bm c_t^{XH,\rm act}=\bm\Pi_t^{\rm act}\bm c_t^{XH}.
\end{equation}
The GKW decomposition on the active futures market is
\begin{equation}\label{eq:gkw}
  H=\E[H]+
  \int_0^T(\bm\theta_t^H)^\top\bm\Pi_t^{\rm act}\dd\bm X_t+L_T^H,
  \quad
  L^H\perp \bm\Pi^{\rm act}\!\cdot\!\bm X,
\end{equation}
where the integrand solves the local normal equations
\begin{equation}\label{eq:dynamic-normal-equations}
  \bm c_t^{XX,\rm act}\bm\theta_t^H=\bm c_t^{XH,\rm act},
\end{equation}
whose minimum-norm predictable solution is
\begin{equation}\label{eq:dynamic-theta}
  \bm\theta_t^H=(\bm c_t^{XX,\rm act})^\dagger\bm c_t^{XH,\rm act},
\end{equation}
with $\dagger$ the Moore--Penrose pseudoinverse, which is essential when futures
are collinear or inactive.

\begin{proposition}[Dynamic variance-optimality and residual covariance]\label{prop:dynamic-residual}
The strategy \eqref{eq:dynamic-theta} solves \eqref{eq:dynamic-problem}. If
$G,H\in L^2(\scrF_T,\Q)$ have value processes $V^G,V^H$ with
$c_t^{G,H}\dd A_t^X=\dd\brac{V^G}{V^H}_t$, the dynamic residuals in
\eqref{eq:gkw} satisfy
\begin{equation}\label{eq:residual-density-measure}
  \dd\brac{L^G}{L^H}_t=\ell_t^{G,H}\dd A_t^X ,
\end{equation}
where
\begin{equation}\label{eq:residual-density}
  \ell_t^{G,H}
  =
  c_t^{G,H}
  -(\bm c_t^{XG,\rm act})^\top
  (\bm c_t^{XX,\rm act})^\dagger
  \bm c_t^{XH,\rm act}.
\end{equation}
Consequently,
\begin{equation}\label{eq:residual-isometry}
  \E[L_T^GL_T^H]
  =
  \E\left[\int_0^T\ell_t^{G,H}\dd A_t^X\right].
\end{equation}
\end{proposition}

\subsection{Variance-optimal hedging on a single delivery window}\label{sec:single-period}

For a PPA on a single delivery period, the projection can be written
explicitly. The hedge uses the power swap and, when available, the
production-index swap on the same window. Two trading conventions are
compared: the \emph{pre-delivery} convention, in which trading stops at the
delivery start and the last position is held to settlement---the
continuous--discrete structure of energy Asian options analysed by
\citet{BenthDetering2015}---and the idealized \emph{during-delivery}
convention, in which the delivery-period value process remains hedgeable until
the end of delivery. Throughout this subsection the technology index $r$ is
fixed and suppressed.

Fix a delivery period $I_\tau=[\tau_1,\tau_2]\subset[0,T]$ with deterministic
endpoints and a deterministic non-negative delivery weight $w(u)$, and put
$D_\tau=\int_{\tau_1}^{\tau_2}w(u)\dd u>0$. The realized notional-weighted
power and production deliveries are
$A_\tau^S=\int_{\tau_1}^{\tau_2}w(u)S_u\dd u$ and
$A_\tau^q=\int_{\tau_1}^{\tau_2}w(u)q_u\dd u$, and the fixed-strike PPA payoff
on the same window is
\begin{equation}\label{eq:single-ppa-payoff}
  H_\tau^K=\int_{\tau_1}^{\tau_2}w(u)q_u(S_u-K)\dd u .
\end{equation}
The associated square-integrable $\Q$-martingales are the PPA value
$V_t^{K,\tau}=\E_t[H_\tau^K]$, the notional-weighted power swap
$X_t^{S,\tau}=\E_t[A_\tau^S]$ and the production-index swap
$X_t^{q,\tau}=\E_t[A_\tau^q]$; before delivery these are the $w$-weighted
integrals of the pointwise objects
\crefrange{eq:F-general}{eq:R-general}, and during delivery each value
splits into the realized part and the conditional value of the undelivered
remainder. The exchange-quoted power swap is the average-price future
$F^S(t,I_\tau)=X_t^{S,\tau}/D_\tau$; all formulas below use the
notional-weighted martingales. Since $\tau_1$ is deterministic and the spike
point measures carry no fixed-time jumps,
$\scrF_{\tau_1-}=\scrF_{\tau_1}$ almost surely. Assume the predictable
brackets of $X^{S,\tau}$, $X^{q,\tau}$ and $V^{K,\tau}$ are absolutely
continuous in time and define the instantaneous covariance densities
\begin{equation}\label{eq:single-cov-def}
  c_t^{a,b}=\frac{\dd\brac{U^a}{U^b}_t}{\dd t},
  \qquad a,b\in\{S,q,K\},
\end{equation}
where $U^S=X^{S,\tau}$, $U^q=X^{q,\tau}$ and $U^K=V^{K,\tau}$; the martingale
part of a delivery-period value process is generated only by the undelivered
remainder of the window, so all densities vanish as $t\to\tau_2^-$. For the
pre-delivery convention we also need the conditional integrated densities
\begin{equation}\label{eq:single-bar-c}
  \overline c_{\tau_1}^{a,b}
  =\E_{\tau_1-}\!\left[\int_{\tau_1}^{\tau_2}c_s^{a,b}\dd s\right],
  \qquad a,b\in\{S,q,K\},
\end{equation}
which, by the martingale isometry, are the conditional second moments of
$(H_\tau^K,A_\tau^S,A_\tau^q)$ given $\scrF_{\tau_1-}$; for example
$\overline c_{\tau_1}^{S,K}=\Cov_{\tau_1-}(H_\tau^K,A_\tau^S)$ and
$\overline c_{\tau_1}^{S,S}=\Var_{\tau_1-}(A_\tau^S)$.

\paragraph{One power swap.}
Under the pre-delivery convention the hedger trades $X^{S,\tau}$ continuously
on $[0,\tau_1)$, then chooses an $\scrF_{\tau_1-}$-measurable position
$\beta_{\tau_1}$ and holds it to settlement, with terminal gain
\begin{equation}\label{eq:single-power-pre-gain}
  G_T^{S,{\rm pre}}(\vartheta,\beta)
  =\int_0^{\tau_1}\vartheta_t\dd X_t^{S,\tau}
  +\beta_{\tau_1}\,(A_\tau^S-X_{\tau_1}^{S,\tau}).
\end{equation}

\begin{proposition}[Single power swap, pre-delivery trading]\label{prop:single-power-pre}
The variance-optimal pre-delivery hedge of $H_\tau^K$ is
\begin{equation}\label{eq:single-power-pre-theta}
  \vartheta_t^{S,{\rm pre}}
  =\frac{c_t^{S,K}}{c_t^{S,S}},
  \qquad 0\le t<\tau_1,
\end{equation}
with the convention that the ratio is zero on $\{c_t^{S,S}=0\}$, and the optimal position
chosen at the start of delivery is the corresponding conditional covariance-to-variance ratio
\begin{equation}\label{eq:single-power-pre-beta}
  \beta_{\tau_1}^{S,{\rm pre}}
  =\frac{\Cov_{\tau_1-}(H_\tau^K,A_\tau^S)}{\Var_{\tau_1-}(A_\tau^S)}
  =\frac{\overline c_{\tau_1}^{S,K}}{\overline c_{\tau_1}^{S,S}} .
\end{equation}
The minimal squared hedging error is
\begin{equation}\label{eq:single-power-pre-error}
  \epsilon_{S,{\rm pre}}^2
  =\Var(H_\tau^K)
  -\E\left[\int_0^{\tau_1}\frac{(c_t^{S,K})^2}{c_t^{S,S}}\dd t\right]
  -\E\left[\frac{(\overline c_{\tau_1}^{S,K})^2}{\overline c_{\tau_1}^{S,S}}\right].
\end{equation}
\end{proposition}

Under the during-delivery convention, $X^{S,\tau}$ remains tradeable until
$\tau_2$ and the optimal strategy is the same instantaneous beta,
$\vartheta_t^{S,{\rm del}}=c_t^{S,K}/c_t^{S,S}$ for $0\le t<\tau_2$, where for
$t\in[\tau_1,\tau_2)$ the densities are generated only by the remaining
interval $[t,\tau_2]$. The minimal squared error is
\begin{equation}\label{eq:single-power-del-error}
  \epsilon_{S,{\rm del}}^2
  =\Var(H_\tau^K)
  -\E\left[\int_0^{\tau_2}\frac{(c_t^{S,K})^2}{c_t^{S,S}}\dd t\right],
\end{equation}
and since the during-delivery admissible space contains the pre-delivery one,
$\epsilon_{S,{\rm del}}^2\le\epsilon_{S,{\rm pre}}^2$. Finiteness of the
integrals is a standing hypothesis (local integrability of
$(c_t^{S,K})^2/c_t^{S,S}$ and, in the multivariate case,
$\bm b_t\in\operatorname{Range}(\bm C_t)$ with
$\bm b_t^\top\bm C_t^\dagger\bm b_t$ locally integrable on $[0,\tau_2)$).

\paragraph{Power and production swaps.}
Suppose the hedger trades both the notional-weighted power swap and the
production-index swap, $\bm X_t^\tau=(X_t^{S,\tau},X_t^{q,\tau})^\top$, and
define
\begin{equation}\label{eq:single-C-b}
  \bm C_t=
  \begin{pmatrix}
  c_t^{S,S} & c_t^{S,q}\\
  c_t^{S,q} & c_t^{q,q}
  \end{pmatrix},
  \qquad
  \bm b_t=
  \begin{pmatrix}
  c_t^{S,K}\\ c_t^{q,K}
  \end{pmatrix}.
\end{equation}

\begin{proposition}[Power and production swaps, pre-delivery trading]\label{prop:single-two-pre}
For $0\le t<\tau_1$, the variance-optimal dynamic hedge is
$\bm\theta_t^{S,q,{\rm pre}}=\bm C_t^\dagger\bm b_t$. If
$\Delta_t=c_t^{S,S}c_t^{q,q}-(c_t^{S,q})^2>0$, the components are
\begin{equation}\label{eq:single-two-pre-theta}
\begin{aligned}
  \theta_t^{S,{\rm pre}}
  &=\frac{c_t^{S,K}c_t^{q,q}-c_t^{q,K}c_t^{S,q}}{\Delta_t},\\
  \theta_t^{q,{\rm pre}}
  &=\frac{c_t^{q,K}c_t^{S,S}-c_t^{S,K}c_t^{S,q}}{\Delta_t}.
\end{aligned}
\end{equation}
At $\tau_1$, with
$\Delta\bm X_{\tau_1}^\tau=(A_\tau^S-X_{\tau_1}^{S,\tau},
A_\tau^q-X_{\tau_1}^{q,\tau})^\top$, the optimal last position is the
finite-dimensional conditional projection
\begin{equation}\label{eq:single-two-pre-beta}
\begin{aligned}
  \bm\beta_{\tau_1}^{S,q,{\rm pre}}
  &=\overline{\bm C}_{\tau_1}^\dagger\overline{\bm b}_{\tau_1},\\
  \overline{\bm C}_{\tau_1}
  &=\E_{\tau_1-}\!\left[\int_{\tau_1}^{\tau_2}\bm C_s\dd s\right],\\
  \overline{\bm b}_{\tau_1}
  &=\E_{\tau_1-}\!\left[\int_{\tau_1}^{\tau_2}\bm b_s\dd s\right],
\end{aligned}
\end{equation}
whose entries are the conditional variances and covariances of
$(H_\tau^K,A_\tau^S,A_\tau^q)$ given $\scrF_{\tau_1-}$. The minimal squared
error is
\begin{equation}\label{eq:single-two-pre-error}
\begin{aligned}
  \epsilon_{S,q,{\rm pre}}^2
  &=\Var(H_\tau^K)
  -\E\left[\int_0^{\tau_1}\bm b_t^\top\bm C_t^\dagger\bm b_t\dd t\right]\\
  &\quad
  -\E\left[\overline{\bm b}_{\tau_1}^\top
  \overline{\bm C}_{\tau_1}^\dagger\overline{\bm b}_{\tau_1}\right].
\end{aligned}
\end{equation}
\end{proposition}

Under the during-delivery convention the same Moore--Penrose projection
$\bm\theta_t^{S,q,{\rm del}}=\bm C_t^\dagger\bm b_t$ extends to
$0\le t<\tau_2$, and the minimal squared error is
\begin{equation}\label{eq:single-two-del-error}
  \epsilon_{S,q,{\rm del}}^2
  =\Var(H_\tau^K)
  -\E\left[\int_0^{\tau_2}\bm b_t^\top\bm C_t^\dagger\bm b_t\dd t\right]
  \le\epsilon_{S,q,{\rm pre}}^2 .
\end{equation}
The conditional terminal moments in \eqref{eq:single-power-pre-beta} and
\eqref{eq:single-two-pre-beta} admit an equivalent two-time-kernel form: with
\begin{equation}\label{eq:two-time-kernels}
\begin{aligned}
  K_t^{SS}(u,v)&=\E_t[S_uS_v],\qquad
  K_t^{qS}(u,v)=\E_t[q_uS_v],\\
  K_t^{qSS}(u,v)&=\E_t[q_uS_uS_v],
\end{aligned}
\end{equation}
the conditional second moments of $(H_\tau^K,A_\tau^S,A_\tau^q)$ are double
$w$-weighted integrals of these kernels over the delivery window, and for
$u\le v$ the Markov property gives $K_t^{SS}(u,v)=\E_t[S_uF_u(v)]$,
$K_t^{qS}(u,v)=\E_t[q_uF_u(v)]$ and $K_t^{qSS}(u,v)=\E_t[q_uS_uF_u(v)]$, so
the kernels can be computed from conditional transforms; path-dependent terms are
evaluated by simulation. The proofs of
\cref{prop:single-power-pre,prop:single-two-pre} are given in the online
supplement.

\paragraph{Interpretation.}
With one power swap, the hedge $\vartheta_t^S=c_t^{S,K}/c_t^{S,S}$ is the
conditional beta of the PPA value against the same-period power swap. If
production were deterministic, this beta would be close to the deterministic
delivered volume. With stochastic production, the numerator depends on the
covariance between price and production because the PPA value depends on $q_uS_u$,
not only on $S_u$; in a cannibalisation regime the beta systematically
deviates from expected volume. With both swaps, the vector hedge
$\bm\theta_t=\bm C_t^\dagger\bm b_t$ separates price and production exposure:
the power component is the marginal contribution of the power swap after
controlling for the production swap and vice versa, while the off-diagonal density
$c_t^{S,q}$ accounts for risk common to the power and production swaps. The
quantities $c_t^{S,K}$, $c_t^{q,K}$ and $c_t^{S,q}$ measure the local
price-production covariance exposure of the PPA; their term structure across
delivery months is implied by the calibrated model.

\subsection{Delivery-period implementation}\label{subsec:delivery-implementation}

For a delivery window $I=[I^-,I^+]$, let $X^{S}(I)$ and $X^{q}(I)$ denote the
power- and production-index futures in \eqref{eq:X-martingales}.  Before
delivery, the active vector is $\bm X=(X^{S}(I),X^{q}(I))^\top$ when both
contracts are available, and only $X^{S}(I)$ otherwise.  The predictable hedge
is therefore always obtained from the normal equations
\begin{equation}\label{eq:delivery-normal-equations-short}
  \bm c_t^{XX,\mathrm{act}}\bm\theta_t^H
  =\bm c_t^{XH,\mathrm{act}},
  \qquad
  \bm\theta_t^H=
  (\bm c_t^{XX,\mathrm{act}})^\dagger
  \bm c_t^{XH,\mathrm{act}},
\end{equation}
with the active set changed only at contract-expiry dates.  This formulation
covers both market conventions relevant for the application.  Under
pre-delivery trading, rebalancing stops at $I^-$ and the terminal position is
held while the settlement averages are formed.  Under the idealised
during-delivery convention, the conditional settlement martingales remain
tradable until $I^+$; the same equations apply, with the bracket integrated over
the longer active interval.  The empirical study adopts the former convention
and hence does not use information revealed after delivery begins.

For a power future alone, \eqref{eq:delivery-normal-equations-short} reduces
to $\theta_t^H=c_t^{SH}/c_t^{SS}$ whenever $c_t^{SS}>0$; with power and
production futures it is the two-dimensional covariance regression of
\cref{prop:single-two-pre}.  The production contract hedges the linear
component of volume risk, while $L_T^H$ retains the nonlinear interaction
between price and production.  The auxiliary static claims of
\cref{subsec:semistatic} are selected to hedge that residual.

\subsection{Semi-static completion by auxiliary claims}\label{subsec:semistatic}

Let $H^0:=H_r^{K_r^\star}$ be the fair-strike PPA payoff for a fixed technology
$r$, and let $H^1,\ldots,H^n\in L^2(\scrF_T,\Q)$ be auxiliary claims available
at inception. Natural candidates are power vanillas on delivery-period averages,
renewable-volume vanillas, price-weighted renewable quantos, orthant quantos and
capture-spread claims; the universe used empirically is specified in
\cref{sec:empirical-design}. Set $\bm H=(H^1,\ldots,H^n)^\top$ and
$\bm\pi=(\pi_1,\ldots,\pi_n)^\top$ with $\pi_i=\E[H^i]$. For $i=0,\ldots,n$,
write the GKW decomposition on the active futures market,
\begin{equation}\label{eq:claim-gkw}
  H^i=\pi_i+
  \int_0^T(\bm\theta_t^i)^\top\bm\Pi_t^{\rm act}\dd\bm X_t
  +L_T^i ,
  \qquad L^i\perp\bm\Pi^{\rm act}\!\cdot\!\bm X .
\end{equation}
A semi-static strategy consists of initial capital $c$, static weights
$\bm\nu\in\R^n$ held fixed from inception to maturity, and a dynamic futures
strategy $\bm\theta\in\Theta_X$. The terminal hedging error is
\begin{equation}\label{eq:semistatic-error}
  \eps_T(c,\bm\nu,\bm\theta)
  =H^0-c-\bm\nu^\top(\bm H-\bm\pi)
  -\int_0^T\bm\theta_t^\top\bm\Pi_t^{\rm act}\dd\bm X_t ,
\end{equation}
and the variance-optimal semi-static problem is
\begin{equation}\label{eq:semistatic-problem}
  \epsilon_{\rm ss}^2
  =
  \inf_{c\in\R,\ \bm\nu\in\R^n,\ \bm\theta\in\Theta_X}
  \E[\eps_T(c,\bm\nu,\bm\theta)^2].
\end{equation}
Define the residual covariance inputs
\begin{equation}\label{eq:outer-abc}
  a_0=\E[(L_T^0)^2],\qquad
  \bm b=\E[\bm L_TL_T^0],\qquad
  \bm C=\E[\bm L_T\bm L_T^\top],
\end{equation}
where $\bm L_T=(L_T^1,\ldots,L_T^n)^\top$. The semi-static layer is a
finite-dimensional regression of the PPA residual left after futures hedging onto the
corresponding residuals of the auxiliary claims. The following theorem is the
semi-static variance-optimal decomposition of
\citet{chatziandreou2026semistaticvarianceoptimalhedgingcovariance}, stated in the present
delivery-period notation with the active projector
$\bm\Pi^{\rm act}$; it is restated here because the time-varying active set is needed
to accommodate contract-specific delivery and expiry dates, and the proof is given in the
online supplement.

\begin{theorem}[Semi-static variance-optimal hedge; cf.\ \citealp{SemistaticHedging,semistaticsparse}]\label{thm:semistatic}
Assume $H^0,\ldots,H^n\in L^2(\Q)$ and that the futures stack $\bm X$ is a
square-integrable $\Q$-martingale. Then:
\begin{enumerate}[label=(\roman*),leftmargin=*]
\item the optimal initial capital is $c^\star=\pi_0$;
\item the cross-covariance vector satisfies $\bm b\in\operatorname{Range}(\bm C)$
and the minimum-norm optimal static weights are
\begin{equation}\label{eq:nu-star}
  \bm\nu^\star=\bm C^\dagger\bm b;
\end{equation}
\item the optimal dynamic futures strategy is
\begin{equation}\label{eq:theta-star}
  \bm\theta_t^\star
  =
  \bm\theta_t^0-[\bm\theta_t^1,\ldots,\bm\theta_t^n]\bm\nu^\star;
\end{equation}
\item the minimal residual variance is
\begin{equation}\label{eq:min-error}
  \epsilon_{\rm ss}^2=a_0-\bm b^\top\bm C^\dagger\bm b .
\end{equation}
\end{enumerate}
\end{theorem}

\subsection{Orthant quantos and capture spreads for hedging covariance risk}
\label{subsec:orthants}

Fix a delivery period $I$ and an observation or settlement time
$\tau\le I^+$. Define the average power and capacity-factor underliers by
\begin{equation}\label{eq:orthant-underliers}
  Z_I^S(\tau)=F^S(\tau,I),
  \qquad
  Z_I^{C,r}(\tau)
  =\frac{F^{q,r}(\tau,I)}{\bar q^rD(I)} .
\end{equation}
At final settlement, $Z_I^S(I^+)$ is the realized delivery-period baseload
price and $Z_I^{C,r}(I^+)$ is the realized average capacity factor. For
thresholds $(a,b)$, introduce the centered underliers
\begin{equation}\label{eq:XY-def}
  X=Z_I^S(\tau)-a,
  \qquad
  Y=Z_I^{C,r}(\tau)-b .
\end{equation}
The four power--renewable orthant quanto payoffs are
\begin{equation}\label{eq:orthant-payoffs}
\begin{array}{ll}
  \mathrm{CC}_{a,b}^{I,\tau}=X^+Y^+,
  &\mathrm{PP}_{a,b}^{I,\tau}=(-X)^+(-Y)^+,\\[0.35em]
  \mathrm{CP}_{a,b}^{I,\tau}=X^+(-Y)^+,
  &\mathrm{PC}_{a,b}^{I,\tau}=(-X)^+Y^+ .
\end{array}
\end{equation}
They satisfy the pathwise identity
\begin{equation}\label{eq:orthant-product}
  XY
  =
  \mathrm{CC}_{a,b}^{I,\tau}
  +\mathrm{PP}_{a,b}^{I,\tau}
  -\mathrm{CP}_{a,b}^{I,\tau}
  -\mathrm{PC}_{a,b}^{I,\tau}.
\end{equation}
Hence the centered product $XY$, which represents price-volume co-movement,
lies in the linear span of the four orthant payoffs: the co-movement
orthants $\mathrm{CC}$ and $\mathrm{PP}$ enter with positive signs, whereas
the divergent orthants $\mathrm{CP}$ and $\mathrm{PC}$ enter with negative
signs. Product options therefore provide exposure to joint nonlinear movements in price and renewable output that
cannot be generated by marginal power and renewable vanillas alone; their use
as a spanning family for covariance-sensitive claims follows the
multidimensional option-spanning arguments of \citet{Madan2021Pricing} and
\citet{chatziandreou2026semistaticvarianceoptimalhedgingcovariance}.

The polarization identity shows how spread options provide additional exposure
to price-capture covariance. In particular, for the monthly baseload and achieved
prices $F_j$ and $A_j$ of \cref{subsec:claim-universe}, and arbitrary anchors
$\alpha,\beta\in\R$, setting $D_j:=F_j-A_j$ gives
\begin{align}\label{eq:capture-spread-polarization}
  &(F_j-\alpha)(A_j-\beta)= \nonumber\\
  & \frac12\left\{
    (F_j-\alpha)^2
    +(A_j-\beta)^2
    -\bigl[D_j-(\alpha-\beta)\bigr]^2
  \right\}.
\end{align}
Thus a mixed baseload--capture-price exposure decomposes into two
marginal quadratic terms and one quadratic term in the spread, which is
the polarization principle underlying covariance hedges constructed from
spread options in \citet{CarrCorso2001} and
\citet{chatziandreou2026semistaticvarianceoptimalhedgingcovariance}. The
capture spread $D_j$ is directly linked to the PPA exposure:
writing $\bar q_j^r:=V_j/D(I_j)$ and
\begin{equation*}
  \Cov_{I_j}^{w}(S,q^r)
  :=
  \frac{1}{D(I_j)}
  \int_{I_j}
  w_{I_j}(u)
  \bigl(S_u-F_j\bigr)
  \bigl(q_u^r-\bar q_j^r\bigr)\,\dd u ,
\end{equation*}
the settlement identities $A_j=R_j/V_j$ and $D_j=F_j-A_j$ imply
\begin{align}\label{eq:capture-spread-covariance}
  & V_jD_j
  =
  V_jF_j-R_j
  =
  -D(I_j)\Cov_{I_j}^{w}(S,q^r),
  \qquad \\
  &D_j
  =
  -\frac{\Cov_{I_j}^{w}(S,q^r)}{\bar q_j^r}.
\end{align}
Consequently, the calls and puts on $D_j$ specified in
\eqref{eq:capture-discount} provide nonlinear exposure to extreme realizations of the within-month capture
spread, which is proportional to the realized price--production
covariance: capture-spread calls pay in severe
cannibalisation states with $A_j<F_j$, whereas capture-spread puts pay when
renewable production earns a capture premium.

For negative raw price--volume covariance, an insurance portfolio that offsets
the associated PPA loss is naturally short the co-movement orthants
$\mathrm{CC}$ and $\mathrm{PP}$ and long the divergent orthants
$\mathrm{CP}$ and $\mathrm{PC}$. Two qualifications are essential. First,
neither power--renewable orthant quantos nor capture-spread options are
standard exchange-traded products in current European power markets; they
should be interpreted as potential OTC hedging products
that would need to be issued bilaterally. Second, the
semi-static optimizer acts on \emph{dynamically residualized} claims: the
pathwise identities motivate the auxiliary basis, but the optimal weights are
determined by the residual covariance problem \eqref{eq:outer-abc}, rather
than being fixed by the raw payoff identity. For the
four-orthant sub-basis, the following proposition gives sufficient conditions
under which the raw insurance sign pattern is inherited by the optimal
weights.

\begin{proposition}[Orthant sign structure]\label{prop:orthant-sign}
Label the four orthant claims $\mathrm{CC}$, $\mathrm{PP}$, $\mathrm{CP}$ and
$\mathrm{PC}$ in this order, and let
$\bm s=(-1,-1,+1,+1)^\top$. With $\bm b$ and $\bm C$ denoting the outer
inputs of \eqref{eq:outer-abc} for this basis, define
\[
  m=\min_i\frac{|b_i|}{C_{ii}},
  \qquad
  M=\max_i\frac{|b_i|}{C_{ii}} .
\]
Suppose that
\begin{enumerate}[label=(\roman*),leftmargin=*]
\item $s_ib_i>0$ for every $i=1,\ldots,4$; and
\item $\bm C$ is strictly diagonally dominant with positive diagonal entries
and dominance ratio
\[
  \delta
  =
  \max_i\frac{\sum_{j\ne i}|C_{ij}|}{C_{ii}}
  <
  \frac{m}{M+m}.
\]
\end{enumerate}
Then $s_i\nu_i^\star>0$ for every $i$, and therefore
\[
  \operatorname{sign}(\bm\nu^\star)=\bm s .
\]
\end{proposition}

\begin{remark}[Strength of the conditions]\label{rem:orthant-sign-caveat}
Condition~(ii) requires the four dynamically residualized orthant claims to
have sufficiently weak mutual dependence relative to the signal-to-noise ratio
$m/M$. Condition~(i) is not implied by negative aggregate price--volume
covariance alone: after dynamic projection on the active futures market, the
PPA residual must covary negatively with the co-movement orthant residuals and
positively with the divergent orthant residuals. Thus
\eqref{eq:orthant-product} motivates the raw sign pattern but does not impose it
on the residualized moments
$b_i=\E[L_T^iL_T^0]$. In the empirical analysis, both conditions are checked
directly using the estimated residual covariance vector and matrix, and the
proposition is interpreted as a sufficient condition for the sign pattern,
rather than as a characterization of every possible optimal portfolio.
\end{remark}

\section{A joint renewable production--price MCARMA model}\label{sec:model}

The stochastic model is built from the economic variables entering the PPA
payoff: the spot price and the technology-specific renewable state. The
empirical analysis motivates three modeling requirements in
\cref{subsec:empirical-analysis}: boundedness of the capacity factor, additivity
of the price in EUR/MWh so that negative prices are admissible, and a dependence structure through
which renewable shocks move prices. The notation uses a generic technology
$r\in\{\mathrm{W},\mathrm{S}\}$; the wind and solar transformations were
introduced in \cref{subsec:emp-seasonality} and are recalled below together
with their inverse production maps.

\subsection{State variables and technology-specific transformations}
\label{subsec:state-variables}

The state variables are the deseasonalized components
constructed in \cref{subsec:emp-seasonality}: (i) the continuous
price residual $Y^S$ and spike component $J$ from \eqref{eq:spot-additive};
(ii) the load-adjusted wind-penetration index $\Pi^{\mathrm W}$ of
\eqref{eq:wind-pressure} with logit residual $Y^{\mathrm W}$ from
\eqref{eq:wind-logit}; and (iii) the regularized solar cloud-cover
shortfall $R^\odot$ of \eqref{eq:solar-risk-driver} with latent residual
$Y^{\mathrm S_\odot}$ from \eqref{eq:solar-latent}. The additive price
specification admits the negative prices documented in
\cref{subsec:emp-data}, the logit transforms respect the boundedness of the
capacity factor, and the load and clear-sky normalizations give the two
renewable coordinates whose dependence with the price residual is
documented in \cref{subsec:emp-wind-effect}. This dependence must
be represented by the joint model.

To evaluate PPA payoffs on simulated paths, we map the latent
states back to physical production. Simulated wind paths are
mapped back through
\begin{equation}\label{eq:wind-reconstruction}
  \widehat C_t^{\mathrm W}
  =
  \min\left\{1,\max\left\{0,\,\widehat\Pi_t^{\mathrm W}
  \frac{\widehat L_t}{\bar L_{\mathcal C}}\right\}\right\},
  \quad
  q_t^{\mathrm W}=\bar q^{\mathrm W}\widehat C_t^{\mathrm W},
\end{equation}
with $\widehat\Pi_t^{\mathrm W}=\ell(\Lambda^{\mathrm W}(t)+Y_t^{\mathrm W})$
and $\ell(x)=(1+\ee^{-x})^{-1}$ the logistic function. The PPA payoff is thus
always evaluated on physical MWh, while the state variable captures wind
penetration relative to load. In the empirical backtest, the realized
historical load path is treated as an exogenous conditioning input and is not
modeled stochastically.

The corresponding solar reconstruction maps the cloud-risk coordinate
\eqref{eq:solar-latent} back to physical output, combining the simulated
cloud-shortfall state with the deterministic daylight envelope:
\begin{equation}\label{eq:solar-inverse}
\begin{aligned}
  \widehat C_t^{\mathrm S}
  &=
  \min\Bigl\{1,\max\Bigl\{0,\;\\
  &\qquad
  C_t^{\mathrm{cs}}
  \bigl[1-
  \ell\bigl(\Lambda^{\mathrm S_\odot}(t)+Y_t^{\mathrm S_\odot}\bigr)\bigr]
  \Bigr\}\Bigr\},\\
  q_t^{\mathrm S}&=\bar q^{\mathrm S}\widehat C_t^{\mathrm S}.
\end{aligned}
\end{equation}
We use the clear-sky envelope as the solar baseline because it removes
the deterministic envelope before the stochastic state is fitted; the
residual state then measures weather-driven shortfall on a comparable scale
across hours with sufficient daylight. The asymmetry between
\eqref{eq:solar-risk-driver} and \eqref{eq:solar-inverse} is deliberate:
$\delta$ regularizes only the low-information observation map, whereas the
physical reconstruction uses the actual envelope $C_t^{\mathrm{cs}}$. Hence
cells classified as nighttime remain zero in every simulated production path.
Production in \eqref{eq:solar-inverse} is a \emph{decreasing} transform of the
logistic latent variable, whereas wind production in
\eqref{eq:wind-reconstruction} is increasing; this difference reverses the relation
between latent and physical covariance for solar.

\begin{remark}[Monotonicity of the production maps and the latent sign convention]
\label{rem:latent-sign-convention}
The two technologies share a common transformation structure: an unconstrained
real-valued latent coordinate is mapped through the logistic function to a
bounded physical capacity factor and then to production
$q_t^r=\bar q^r\widehat C_t^r$.

\emph{Monotonicity of the physical maps.}
Write $\ell(x)=(1+\ee^{-x})^{-1}$, so that
$\ell'(x)=\ell(x)\{1-\ell(x)\}>0$. Conditional on the exogenous load path and deterministic
clear-sky envelope used in the reconstruction, and away from the saturation
boundaries of the clipping operators, the derivatives are
\begin{equation}\label{eq:map-derivatives}
\begin{aligned}
  \frac{\partial\widehat C_t^{\mathrm W}}
        {\partial Y_t^{\mathrm W}}
  &=
  \frac{\widehat L_t}{\bar L_{\mathcal C}}\,
  \ell'\!\bigl(X_t^{\mathrm W}\bigr)>0,\\
  \frac{\partial\widehat C_t^{\mathrm S}}
        {\partial Y_t^{\mathrm S_\odot}}
  &=
  -C_t^{\mathrm{cs}}\,
  \ell'\!\bigl(X_t^{\mathrm S}\bigr)<0,
  \qquad C_t^{\mathrm{cs}}>0 .
\end{aligned}
\end{equation}
Thus the wind latent residual raises production because
$X_t^{\mathrm W}$ is the logit of load-adjusted wind penetration, whereas
the solar latent residual lowers production because
$X_t^{\mathrm S}$ is the logit of cloud-cover shortfall: a larger solar
latent state represents cloudier conditions and lower output. At night,
$C_t^{\mathrm{cs}}=0$ and hence $\widehat C_t^{\mathrm S}=0$ independently
of the stochastic state. At clipping or saturation boundaries, both maps are
Lipschitz and their derivatives are zero almost everywhere outside the
interior region.

\emph{Continuous-state covariance and the Gaussian benchmark.}
For this paragraph, write
\[
  \widetilde Y_t^{\mathrm W}:=Y_t^{\mathrm W},
  \qquad
  \widetilde Y_t^{\mathrm S}:=Y_t^{\mathrm S_\odot},
\]
and, conditional on the deterministic or exogenous reconstruction inputs,
write $\widehat C_u^r=g_u^r(X_u^r)$ with
$X_u^r=\Lambda^r(u)+\widetilde Y_u^r$. Since
$S_u=\Lambda^S(u)+Y_u^S+J_u$, the conditional capture covariance decomposes
exactly as
\begin{equation}\label{eq:capture-covariance-components}
\begin{aligned}
  \Gamma_t^r(u)
  &=
  \Cov_t(q_u^r,S_u)\\
  &=
  \bar q^r\Cov_t\!\left(g_u^r(X_u^r),Y_u^S\right)
  +
  \bar q^r\Cov_t\!\left(g_u^r(X_u^r),J_u\right)\\
  &:=\Gamma_{t,\mathrm c}^r(u)+\Gamma_{t,\mathrm J}^r(u).
\end{aligned}
\end{equation}
The second term cannot generally be omitted in the state-dependent-jump
specification because the intensity of negative price spikes depends on the
renewable state.

In the Gaussian benchmark, $\nu_L=0$ and $J\equiv0$. Conditional on $\scrF_t$, the pair
$(X_u^r,Y_u^S)$ is then jointly Gaussian. Defining
\[
  \sigma_{rS,t}(u)
  :=
  \Cov_t\!\left(\widetilde Y_u^r,Y_u^S\right)
  =
  \Cov_t\!\left(X_u^r,Y_u^S\right),
\]
Stein's identity gives the exact relation
\begin{equation}\label{eq:sign-chain}
  \Gamma_{t,\mathrm c}^r(u)
  =
  \bar q^r\,\sigma_{rS,t}(u)\,
  \E_t\!\left[(g_u^r)'(X_u^r)\right],
  \quad \nu_L=0,\quad J\equiv0 .
\end{equation}
Consequently, whenever the expected derivative is nonzero,
\begin{equation}\label{eq:latent-cov-signs}
\begin{aligned}
  \operatorname{sign}\Gamma_{t,\mathrm c}^{\mathrm W}(u)
  &=
  \operatorname{sign}\sigma_{\mathrm WS,t}(u),\\
  \operatorname{sign}\Gamma_{t,\mathrm c}^{\mathrm S}(u)
  &=
  -\operatorname{sign}\sigma_{\mathrm S_\odot S,t}(u).
\end{aligned}
\end{equation}
Equation \eqref{eq:sign-chain} specializes to
$\Gamma_t(u)=\bar q\,\sigma_{RS}L_1(m_R,\sigma_R)$ for wind and to
\[
  \Gamma_t(u)
  =
  -\bar q\,C^{\mathrm{cs}}(u)
  \sigma_{RS}L_1(m_R,\sigma_R)
\]
for solar.

\emph{What remains valid under the MNIG driver.}
Under the fitted MNIG law, \eqref{eq:sign-chain} is not an exact identity.
A first-order expansion around
$m_{r,t,u}=\E_t[X_u^r]$ gives the diagnostic approximation
\[
  \Gamma_{t,\mathrm c}^r(u)
  \approx
  \bar q^r(g_u^r)'(m_{r,t,u})\,
  \sigma_{rS,t}(u),
\]
but the approximation contains a nonlinear remainder depending on higher
joint moments. Equivalently, conditioning on the inverse-Gaussian mixing
variables collected in $\mathcal M_{t,u}$ gives
\begin{align*}
  &\hspace{-5mm}\Cov_t\!\left(g_u^r(X_u^r),Y_u^S\right)
  =
  \E_t\!\left[
    \Cov_t\!\left(g_u^r(X_u^r),Y_u^S\mid\mathcal M_{t,u}\right)
  \right]\\
  &\qquad+
  \Cov_t\!\left(
    \E_t[g_u^r(X_u^r)\mid\mathcal M_{t,u}],
    \E_t[Y_u^S\mid\mathcal M_{t,u}]
  \right).
\end{align*}
The second term and the mixing-dependent conditional covariance prevent the
physical covariance sign from being inferred from
$\sigma_{rS,t}(u)$ alone without additional dependence assumptions. The
empirical signs
\[
  \operatorname{corr}(Y^{\mathrm W},Y^S)<0,
  \qquad
  \operatorname{corr}(Y^{\mathrm S_\odot},Y^S)>0
\]
are calibration diagnostics that are consistent with the opposite monotonicities
of the wind and solar maps; they are not implications of the MNIG model
alone. They must also be distinguished from the contemporaneous
correlation of the recovered driver increments: the MCARMA transfer function
maps $\Sigma_L$ into the entire state autocovariance sequence
$\Gamma(k)$, so the off-diagonal sign of $\Sigma_L$ need not equal the
lag-zero sign of the filtered state covariance.

\emph{Consistency across estimation, pricing and hedging.}
The calibration procedure matches the state autocovariances and the cross-spectral phase using the exact
second-order law of the sampled MCARMA process; this step requires finite
second moments, not Gaussian state marginals. The MNIG law is subsequently
fitted to the recovered driver increments and rescaled to preserve the
calibrated covariance, thereby adding skewness, heavy tails and tail dependence
without changing the targeted second-order structure. In analytic pricing,
the decreasing physical map produces the solar sign reversal. The Gaussian formulas are used only as a closed-form
benchmark, whereas the full L\'evy covariance densities include the jump
integral. Finally, the empirical design
evaluates every payoff, fair strike and hedge covariance after mapping each
simulated latent path through \eqref{eq:wind-reconstruction} or
\eqref{eq:solar-inverse} and adding the simulated spike component. Hence the
sign of simulated price-production covariance is determined after mapping the
latent states to physical production and adding the spike component, rather than
imposing it through a Gaussian covariance identity.
\end{remark}

\subsection{The MCARMA continuous state}\label{subsec:mcarma-state}

For each technology, define the state vector with the renewable latent
residual first and the spot residual second:
\begin{equation}\label{eq:Y-vector}
  \bm Y_t=(Y_t^{r},Y_t^S)^\top\in\R^2 .
\end{equation}
The ordering matters: joint lower-triangular restrictions on the autoregressive
and moving-average coefficient matrices allow renewable weather shocks to propagate
into spot prices while excluding direct dynamic feedback from the spot driver into
the renewable coordinate, reflecting the assumption that renewable shocks affect prices, but not conversely.

Let $p>q\ge0$ and let $d=2$ denote the observation dimension and $m=2$ the
driver dimension. The monic autoregressive and moving-average matrix
polynomials are
\begin{align}
  P(z)&=I_dz^p+A_1z^{p-1}+\cdots+A_p,\label{eq:mcarma-P}\\
  Q(z)&=B_0z^q+B_1z^{q-1}+\cdots+B_q,\label{eq:mcarma-Q}
\end{align}
with $A_i\in\R^{d\times d}$, $B_j\in\R^{d\times m}$. The formal MCARMA equation
\begin{equation}\label{eq:mcarma-formal}
  P(D)\bm Y_t=Q(D)D\bm L_t,\qquad D=\frac{\dd}{\dd t},
\end{equation}
where $\bm L$ is an $m$-dimensional zero-mean finite-variance L\'evy process, is
interpreted through the continuous-time state-space representation of
\citet{MarquardtStelzer2007,SchlemmStelzer2012}:
\begin{equation}\label{eq:mcarma-state-equation}
  \dd\bm X_t=\mathcal A\bm X_t\dd t+\beta\dd\bm L_t,\qquad
  \bm Y_t=\mathcal C\bm X_t ,
\end{equation}
with companion state $\bm X_t\in\R^{pd}$, observation matrix
$\mathcal C=(I_d,0_d,\ldots,0_d)\in\R^{d\times pd}$, companion matrix
\begin{equation}\label{eq:mcarma-companion}
  \mathcal A=
  \begin{pmatrix}
    0_d&I_d&0_d&\cdots&0_d\\
    0_d&0_d&I_d&\cdots&0_d\\
    \vdots&\vdots&\vdots&\ddots&\vdots\\
    0_d&0_d&0_d&\cdots&I_d\\
    -A_p&-A_{p-1}&-A_{p-2}&\cdots&-A_1
  \end{pmatrix}\in\R^{pd\times pd},
\end{equation}
and loading matrix
$\beta=(\beta_1^\top,\ldots,\beta_p^\top)^\top\in\R^{pd\times m}$ determined by
the recursion
\begin{equation}\label{eq:mcarma-beta-recursion}
\begin{aligned}
  \beta_{p-j}
  &=-\sum_{i=1}^{p-j-1} A_i\beta_{p-j-i}
  +\1_{\{0,\ldots,q\}}(j)B_{q-j},\\
  &\hspace{2.5em} j=0,\ldots,p-1,
\end{aligned}
\end{equation}

with empty sums equal to zero. In particular, for $(p,q)=(1,0)$ the recursion
gives $\beta_1=B_0$, as required by
$\mathcal C(zI_d-\mathcal A)^{-1}\beta=P(z)^{-1}Q(z)$. If all eigenvalues of
$\mathcal A$ have strictly
negative real part, the unique causal stationary solution is
\begin{equation}\label{eq:mcarma-stationary}
  \bm X_t=\int_{-\infty}^t\ee^{\mathcal A(t-s)}\beta\dd\bm L_s,
  \quad
  \bm Y_t=\int_{-\infty}^t\mathcal C\ee^{\mathcal A(t-s)}\beta\dd\bm L_s ,
\end{equation}
and the transfer function
\begin{equation}\label{eq:mcarma-transfer}
  H(z)=\mathcal C(zI_{pd}-\mathcal A)^{-1}\beta=P(z)^{-1}Q(z)
\end{equation}
links the state-space and polynomial representations.

\paragraph{Block lower-triangular design.}
The generic polynomials \crefrange{eq:mcarma-P}{eq:mcarma-Q} are
specialized so that the coefficient matrices encode the causal ordering of
\eqref{eq:Y-vector}. In the basis $(Y^r,Y^S)$ every autoregressive and every
moving-average coefficient is taken \emph{lower triangular}, i.e. its
upper-right entry vanishes,
\begin{equation}\label{eq:lower-triangular-coeffs}
  A_k=\begin{pmatrix} a_k^{r}&0\\[2pt] a_k^{\times}&a_k^{S}\end{pmatrix},
  \quad
  B_j=\begin{pmatrix} b_j^{r}&0\\[2pt] b_j^{\times}&b_j^{S}\end{pmatrix},
  \quad B_0=I_d ,
\end{equation}
for $k=1,\dots,p$ and $j=1,\dots,q$; the normalization $B_0=I_d$ fixes the
driver gauge required for the L\'evy recovery.
Because the inverse of a lower-triangular matrix polynomial is again lower
triangular, the transfer function \eqref{eq:mcarma-transfer} inherits the
structure, $H(z)=P(z)^{-1}Q(z)$ with $H_{rS}(z)\equiv0$. The renewable
coordinate is therefore a functional of its own driving increments alone,
while the spot coordinate is driven by both: renewable weather shocks
propagate dynamically into prices, whereas price shocks never feed back into
the renewable state. This one-directional (Granger) coupling specifies the one-way
effect of renewable shocks on prices documented in
\cref{subsec:emp-wind-effect,subsec:emp-forward}; contemporaneous dependence
is retained through the off-diagonal of the driver covariance $\Sigma_L$
(\cref{subsec:mnig-driver}), so the triangular design removes dynamic feedback,
not instantaneous correlation.

\paragraph{Modal structure.}
Since the determinant of a triangular matrix polynomial is the product of its
diagonal entries, $\det P(z)=p_r(z)\,p_S(z)$ with monic scalar polynomials
$p_c(z)=z^{p}+a_1^{c}z^{p-1}+\cdots+a_p^{c}$ for $c\in\{r,S\}$, and the
companion spectrum splits accordingly; the cross-coefficients $a_k^\times$ do not
affect the eigenvalues and only set the strength with which renewable modes are
transmitted into the spot coordinate. Each diagonal polynomial is specified in
factored form as a product of real relaxation modes and damped oscillator
pairs,
\begin{equation}\label{eq:scalar-poly-factorization}
\begin{aligned}
  p_c(z)&=\prod_{i=1}^{n_c^{\rm re}}\bigl(z+\rho_{c,i}\bigr)
  \prod_{j=1}^{n_c^{\rm os}}\Bigl[\bigl(z+\alpha_{c,j}\bigr)^2+\omega_{c,j}^2\Bigr],\\
  p&=n_c^{\rm re}+2\,n_c^{\rm os},
\end{aligned}
\end{equation}
with $\rho_{c,i},\alpha_{c,j}>0$ and $0<\omega_{c,j}<\pi/h$ for the hourly
sampling step $h$. A real mode contributes an exponentially decaying
autocovariance with relaxation time $1/\rho_{c,i}$; a conjugate pair
contributes a damped cosine of period $2\pi/\omega_{c,j}$ under an
exponentially decaying envelope. Combining several real decay
rates with oscillator pairs at the selected calendar frequencies allows
the model to represent the multiscale and periodic autocorrelation
of the deseasonalized states.

\begin{assumption}[MCARMA standing assumptions]\label{ass:mcarma}
On the pricing horizon:
(i) $\bm L$ is a two-dimensional L\'evy process with $\E[\bm L_1]=0$,
$\E\norm{\bm L_1}^2<\infty$, L\'evy triplet $(\gamma_L,\Sigma_G,\nu_L)$ and
cumulant exponent
\begin{equation}\label{eq:levy-cumulant}
\begin{aligned}
  \kappa_L(\bm u)
  &=\bm u^\top\gamma_L+\frac12\bm u^\top\Sigma_G\bm u\\
  &\quad
  +\int_{\R^2}\left(\ee^{\bm u^\top\bm x}-1-
  \bm u^\top\bm x\1_{\{\norm{\bm x}\le1\}}\right)\nu_L(\dd\bm x)
\end{aligned}
\end{equation}
on its exponential moment domain;
(ii) $\mathcal A$ is stable,
$\max\{\Real\lambda:\lambda\in\sigma(\mathcal A)\}<0$;
(iii) the realization $(\mathcal A,\beta,\mathcal C)$ is minimal (controllable
and observable) on the fitted model class; and
(iv) $P$ and $Q$ are left-coprime in the fitted parametrization.
\end{assumption}

\begin{proposition}[Stationary second-order structure]\label{prop:mcarma-second-order}
Under Assumption \ref{ass:mcarma}, with $\Sigma_L=\Cov(\bm L_1)$, the stationary companion
covariance $V_X=\Cov(\bm X_t)$ is the unique positive semidefinite solution of
the Lyapunov equation
\begin{equation}\label{eq:mcarma-lyapunov}
  \mathcal A V_X+V_X\mathcal A^\top+\beta\Sigma_L\beta^\top=0,
\end{equation}
and for $h\ge0$ the autocovariance of the observed state is
\begin{equation}\label{eq:mcarma-acov}
  \Gamma_Y(h)=\Cov(\bm Y_{t+h},\bm Y_t)
  =\mathcal C\ee^{\mathcal A h}V_X\mathcal C^\top .
\end{equation}
\end{proposition}

These identities underlie the multiscale autocovariance diagnostics, the
spectral initialization and the exact sampled-state likelihood of
\cref{sec:calibration}: the model matches empirical auto- and cross-covariances
at hourly, daily and weekly lags through $\ee^{\mathcal A h}$ acting on $V_X$,
including the oscillatory components induced by complex eigenvalue pairs of
$\mathcal A$ at calendar frequencies.

\subsection{The MNIG driver}\label{subsec:mnig-driver}

The driver $\bm L$ is specified as a bivariate normal-inverse Gaussian (NIG)
L\'evy process with a common inverse Gaussian subordinator
\citep{BarndorffNielsen1997}: over a sampling interval of length $h$,
\begin{equation}\label{eq:mnig-increment}
\begin{aligned}
  \Delta\bm L
  &\overset{d}{=}
  \bm\mu h+\bm b\,W+\sqrt{W}\,\bm Z,\\
  \bm Z&\sim N(0,\bm\Sigma),\qquad
  W\sim\mathrm{IG}(\delta h,\gamma),
\end{aligned}
\end{equation}
with $\bm Z$ independent of $W$ and parameters constrained so that
$\E[\Delta\bm L]=0$ and $\Cov(\Delta\bm L)=\Sigma_L h$ matches the covariance
identified in estimation. The common subordinator generates joint heavy tails
and tail dependence between renewable and price shocks, producing heavier
joint tails and stronger tail dependence than a Gaussian driver
while preserving the exact second-order structure
of \cref{prop:mcarma-second-order}. The MNIG law is estimated from the recovered
driver increments.

\subsection{State-dependent price spikes}\label{subsec:spikes}

The spike component is a price-only mean-reverting jump process,
\begin{equation}\label{eq:jump-ou}
\begin{aligned}
  \dd J_t&=-\kappa_JJ_t\dd t+\dd Z_t^+ +\dd Z_t^- ,\\
  \dd Z_t^\pm&=\int_\R x\,N^\pm(\dd t,\dd x),
\end{aligned}
\end{equation}
where positive marks have law $\nu^+$ on $(0,\infty)$ and negative marks law
$\nu^-$ on $(-\infty,0]$. Positive spikes have deterministic intensity
$\lambda^+(t)$; negative spikes have an intensity that depends on the
left-continuous renewable state,
\begin{equation}\label{eq:state-dependent-intensity}
  \lambda_t^-
  =
  \lambda_L^-+\Delta\lambda^-
  \1_{\{X_{t-}^{r}>k_\star^r\}},
  \qquad k_\star^r=\logit(c_\star^r),
\end{equation}
so that episodes of high renewable penetration raise the arrival rate of
negative price spikes. For wind, $X^r$ is the wind-penetration logit state,
and this mechanism is consistent with \citet{DeschatreVeraart2018}. For
solar, the analogous state dependence is an empirical extension developed in
this paper: the latent shortfall state enters with reversed sign so that the intensity
switches to the high regime when shortfall is low and solar output is high. Whenever
analytic state derivatives of the jump transform are required, the threshold is
replaced by the smooth sigmoid
$\lambda_L^-+\Delta\lambda^-\ell(a(X^r-k_\star^r))$ with finite steepness $a$;
simulations use the hard-threshold specification.

\begin{assumption}[Jump admissibility]\label{ass:jump-admissibility}
The marked point measures $N^\pm$ are defined on an extension of the MCARMA
filtration with predictable intensities $\lambda^+(t)$ and $\lambda_t^-$ and
i.i.d.\ marks with laws $\nu^\pm$; the intensities are bounded and
non-negative, the mark second moments $m_{2,\pm}=\int x^2\nu^\pm(\dd x)$ are
finite, and the mark moment generating functions are finite on every transform
strip used below \citep{Bremaud1981}.
\end{assumption}

\section{Estimation methodology}\label{sec:calibration}

The model is estimated on hourly German data: day-ahead spot prices, wind and
solar capacity factors, system load and temperature. The calibration window
runs from January~2023 through December~2024 and the out-of-sample PPA delivery
window covers January through December~2025, so that all pricing and hedging
results are evaluated on data not used in estimation. Estimation proceeds in four stages: deterministic seasonality,
the continuous MCARMA state, recovered L\'evy increments, and
the spike layer.

\subsection{Seasonal decomposition and state construction}\label{subsec:seasonality}

The deterministic spot level has the product form of \citet{Paraschiv2015},
\begin{equation}\label{eq:spot-seasonality}
  \Lambda^S(t)=\bar S\,F_t^{2Y,S}F_t^{2D,S},
\end{equation}
where the factor-to-year term $F^{2Y,S}$ is estimated on daily average prices by
a regression on day-of-week and month indicators, holiday effects and
heating/cooling temperature terms $(15-\Theta_d)^+$ and $(\Theta_d-15)^+$, and
the factor-to-day term $F^{2D,S}$ is the normalized intraday profile within
twenty calendar classes (weekday months and weekend season groups). The
seasonality is fitted directly to prices rather than to a shifted logarithm, so
the residual $Y^S$ is an additive price residual in EUR/MWh. The spike state
$J_t$ is identified first by an iterative threshold procedure on the
deseasonalized price and removed, so that the continuous residual entering the
MCARMA estimation is $\widehat Y_t^S=S_t-\widehat\Lambda_t^S-\widehat J_t$.

The renewable seasonal components $\Lambda^{\mathrm W}$ and
$\Lambda^{\mathrm S_\odot}$ are calendar--Fourier regressions on the
corresponding logit states \eqref{eq:wind-logit} and \eqref{eq:solar-latent}.
For solar, the locally pooled quantile
\eqref{eq:solar-local-envelope}, circular smoother, visibility scaling
\eqref{eq:solar-envelope-scale} and support rule
\eqref{eq:solar-structural-night} are fitted first. The regularized observation
map \eqref{eq:solar-risk-driver} then separates the nonidentified
low-irradiance region from daylight weather variation, and
\eqref{eq:solar-robust-seasonality} estimates the latent calendar level with
Huber and chronological smoothness regularization. The final calibration series is
the bivariate hourly residual vector
$\widehat{\bm Y}_t=(\widehat Y_t^{r},\widehat Y_t^S)^\top$; no
post-calibration observation enters these estimates.

\section{Calibration Results}\label{sec:calibration-results}

This section reports the fitted wind and solar models that are used in the
empirical hedging experiment of \cref{sec:empirical-design}.

\providecommand{\inputcalibrationtable}[1]{
  \begingroup
  \expandafter\let\expandafter\table\csname table*\endcsname
  \expandafter\let\expandafter\endtable\csname endtable*\endcsname
  \input{#1}
  \endgroup}

The graphical diagnostics are complemented by full-level one-step prediction
metrics. For a physical component $X_t$ and its one-step fitted value
$\widehat X_{t|t-h}$, the reported root mean squared error, mean absolute
error and bias are
\begin{align*}
  \operatorname{RMSE}(X)
  &=
  \left\{N^{-1}\sum_t
  \bigl(X_t-\widehat X_{t|t-h}\bigr)^2\right\}^{1/2},\\
  \operatorname{MAE}(X)
  &=
  N^{-1}\sum_t\bigl|X_t-\widehat X_{t|t-h}\bigr|,\\
  \operatorname{Bias}(X)
  &=
  N^{-1}\sum_t\bigl(\widehat X_{t|t-h}-X_t\bigr).
\end{align*}
The same tables report realized and fitted standard deviations and the
empirical correlation between $X_t$ and $\widehat X_{t|t-h}$. These statistics
are computed after applying the inverse physical maps
\eqref{eq:wind-reconstruction} and \eqref{eq:solar-inverse}, so they measure
fit quality on the variables entering the PPA payoff rather than only on the
latent state scale.

\subsection{Wind calibration}\label{subsec:calibration-results-wind}

For wind, the fitted state is
$\bm Y_t=(Y_t^{\mathrm W},Y_t^S)^\top$, where
$Y^{\mathrm W}$ is the logit residual of the load-adjusted wind pressure
$\Pi^{\mathrm W}$ and $Y^S$ is the additive spot residual. Physical
wind production is reconstructed using \eqref{eq:wind-reconstruction}, so
an increase in $Y_t^{\mathrm W}$ increases $\Pi_t^{\mathrm W}$ and wind
production. The estimated dependence therefore has the same sign on the latent
and physical wind scales: the relevant cannibalisation relation is the
negative association between wind pressure and spot prices.

\Cref{fig:cal-wind-state-acf-psd} assesses the fit of the sampled-state law
before the MNIG marginal layer is assessed. The wind-state autocorrelation
decays from near one at the first lag to values close to zero after roughly one
week, and the realized curve is close to the simulated median for the short lags
that drive one-step prediction. The spectrum shows similar agreement in the
frequency domain: the low-frequency wind persistence and the selected daily
and half-daily oscillator frequencies are reproduced by the exact resolvent
spectrum of the fitted state-space model. The spot state exhibits stronger daily periodicity
associated with market hours. Its autocorrelation has visible daily peaks throughout
the 336-hour diagnostic window, and the PSD displays peaks at the corresponding
frequencies. The fitted median captures these peaks and the realized spectrum
remains inside the simulated envelope over the economically relevant
low-frequency and calendar-frequency bands.

The wind-pressure equation has an
RMSE of $0.0319$ and correlation $0.9872$ against realized
$\Pi_t^{\mathrm W}$, while the spot equation has an RMSE of about
$12.22$ EUR/MWh and correlation $0.9709$ against realized $S_t$. The biases
are small relative to the respective standard deviations, and the fitted
standard deviations are close to the realized ones. The filter therefore matches
the main persistence and spectral features and yields close one-step
conditional-mean predictions for the two payoff variables.

The dynamic cross-dependence is summarized by
\cref{fig:cal-wind-lead-lag}. The plotted statistic is
$k\mapsto\operatorname{Corr}(S_{t+k},\Pi_t^{\mathrm W})$. It has a pronounced
minimum at $k=0$, approximately $-0.45$ to $-0.50$, and remains negative over
lags near zero. This pattern is consistent with wind
cannibalisation: high contemporaneous wind pressure is
associated with low spot prices, and the effect persists over neighbouring
hours. The fitted model matches the contemporaneous minimum and the
asymmetric return of the correlation toward zero; the realized curve is
well inside the simulated band over almost the whole window.

\Cref{fig:cal-wind-mnig-residuals} then compares the marginal distributions of
the deseasonalized state coordinates. The wind residual is close to symmetric
after the load-adjusted logit transform, and the simulated law matches its dispersion and the
concentration of observations near the center. The spot residual is heavier tailed and more skewed; the fitted MNIG density is smoother than the empirical
distribution near the center.

\begin{figure*}[p]
\centering
\begin{subfigure}[t]{.96\textwidth}
\centering
\includegraphics[width=\linewidth,height=.34\textheight,keepaspectratio]{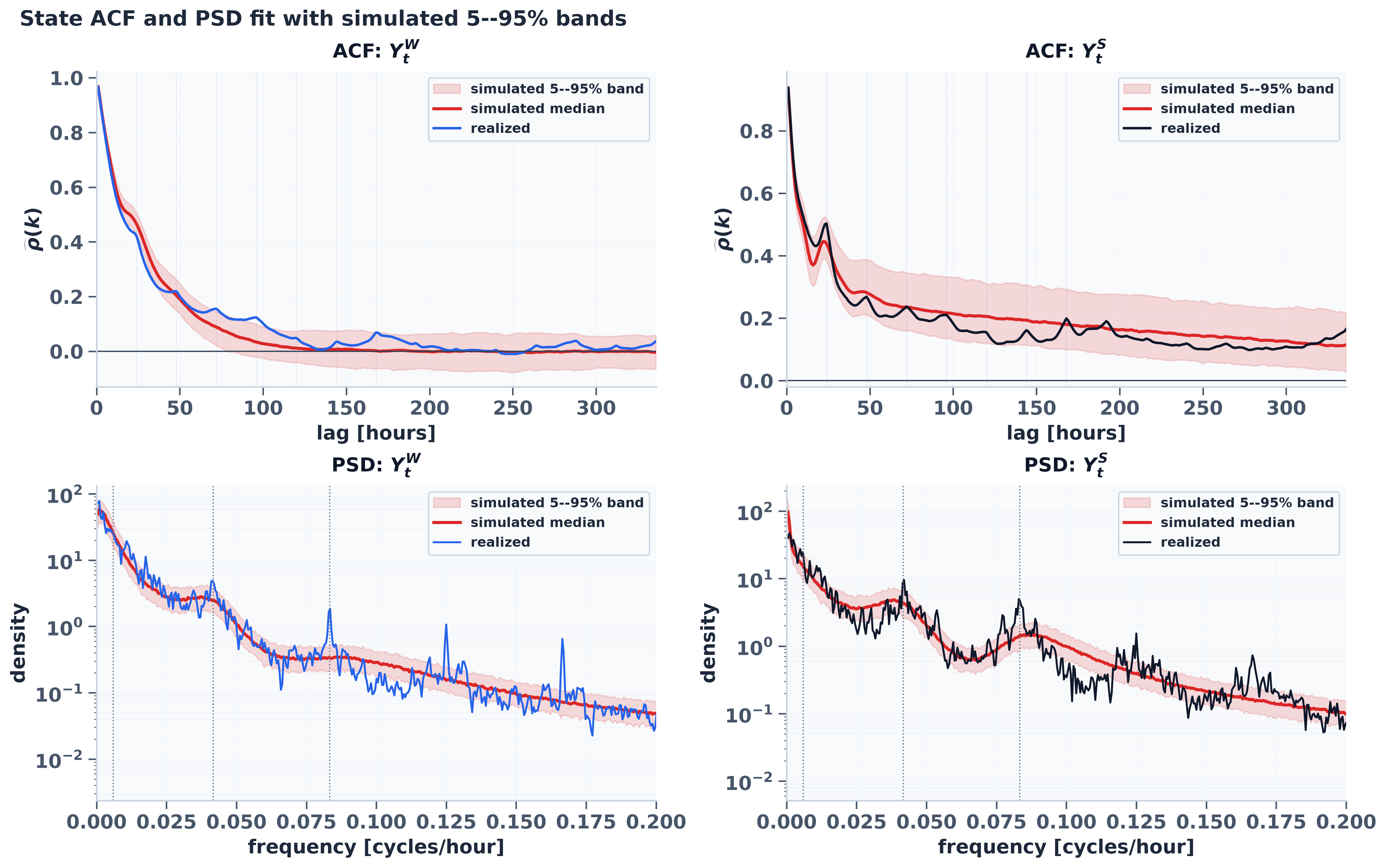}
\caption{State autocorrelations and Welch/resolvent spectra.}
\end{subfigure}\\[0.3em]
\begin{subfigure}[t]{.48\textwidth}
\centering
\includegraphics[width=\linewidth,height=.27\textheight,keepaspectratio]{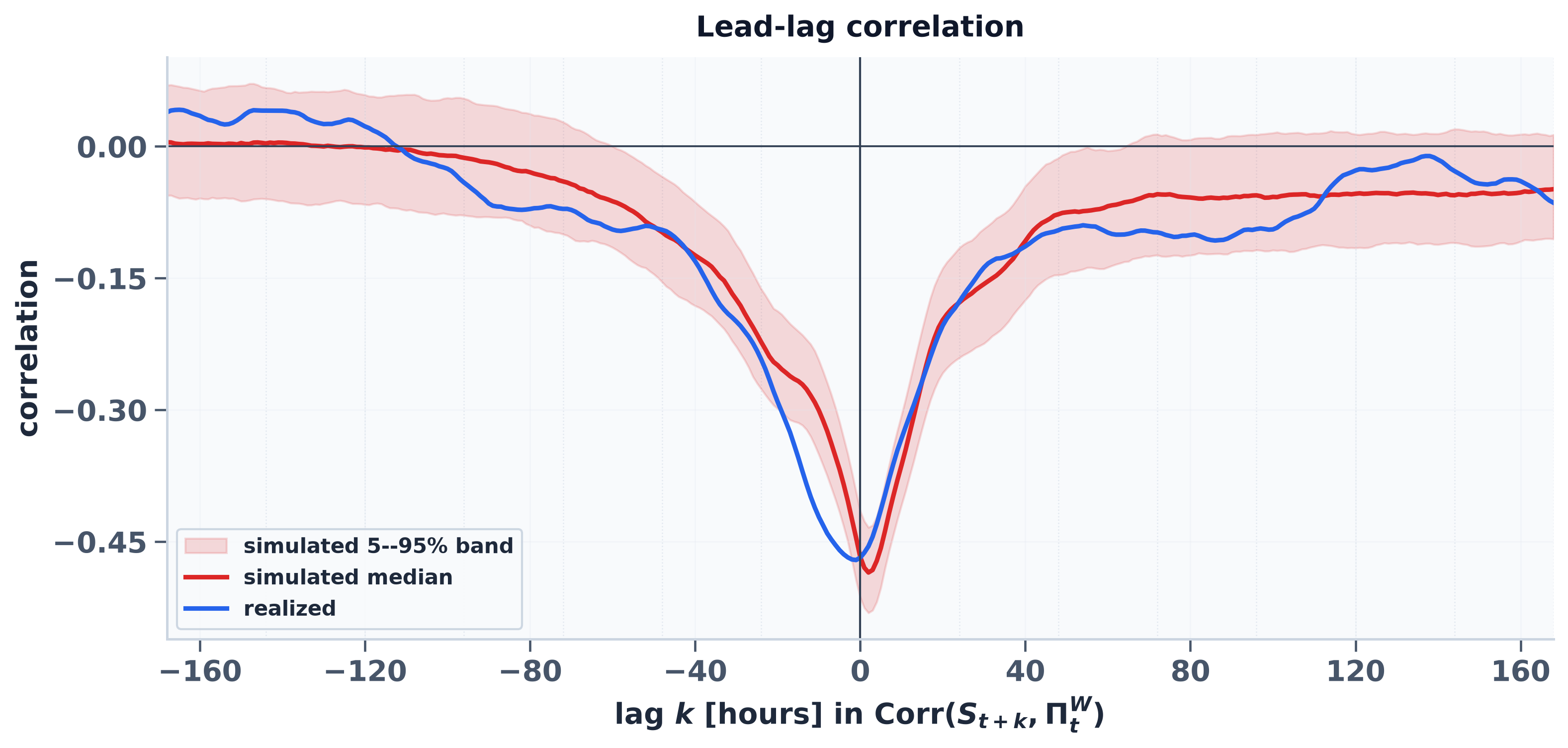}
\caption{Full-level wind--price lead--lag dependence.}
\end{subfigure}\hfill
\begin{subfigure}[t]{.48\textwidth}
\centering
\includegraphics[width=\linewidth,height=.27\textheight,keepaspectratio]{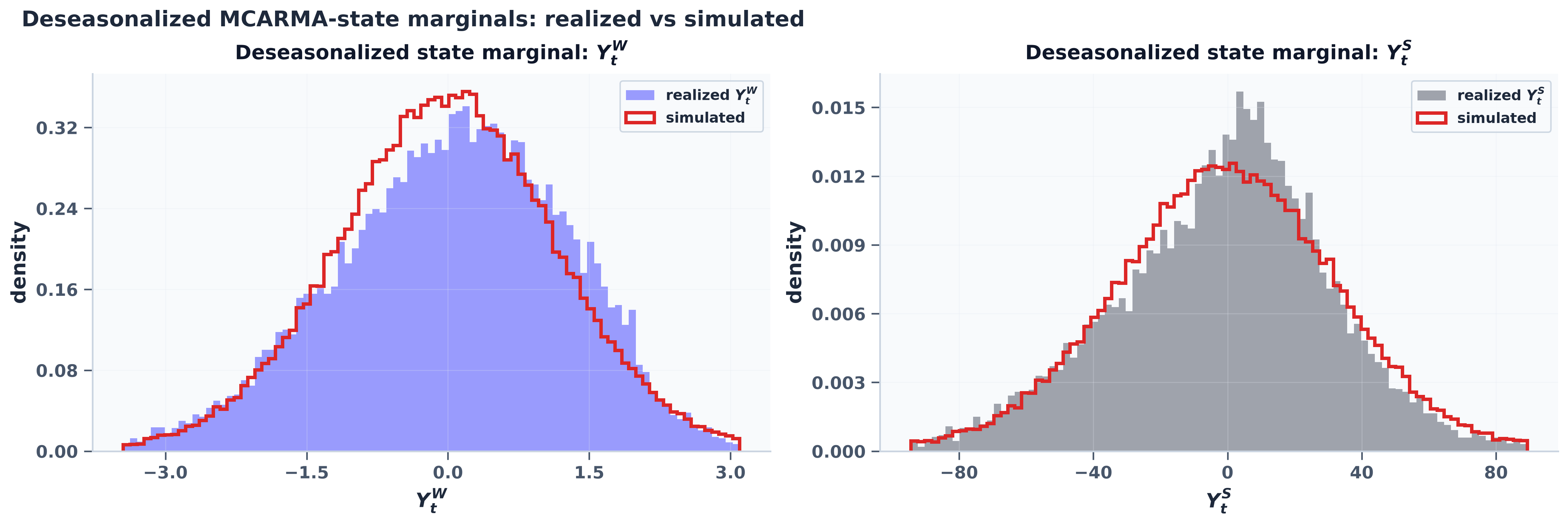}
\caption{MNIG-driven residual marginals.}
\end{subfigure}
\caption{Wind calibration diagnostics for the deseasonalized states and their dependence.
Bands are simulated from the fitted exact-sampled MCARMA--MNIG law.}
\label{fig:cal-wind-state-acf-psd}
\label{fig:cal-wind-lead-lag}
\label{fig:cal-wind-mnig-residuals}
\end{figure*}

At the physical level, \cref{fig:cal-wind-mnig-marginals,fig:cal-wind-mnig-bivar}
show that the fitted model reproduces the main unconditional features of the
joint law of $(\Pi_t^{\mathrm W},S_t)$. The marginal law of
$\Pi_t^{\mathrm W}$ has the correct support, positive skewness, and a long upper tail
corresponding to high-wind hours, while the spot marginal has the correct broad level and
dispersion. The bivariate density is concentrated along a downward-sloping region: high
prices occur mainly when wind pressure is low, whereas high wind pressure is concentrated at
lower prices. The continuous component underrepresents isolated high-price
events at very low wind; these events are instead captured by
the spike component rather than to the continuous MCARMA driver.

\Cref{fig:cal-wind-mnig-acf} checks that the same MNIG-driven physical paths
preserve the full-level serial dependence. The wind-pressure ACF combines
slow persistence with residual daily oscillations, and the simulated median
tracks the realized curve over the full 336-hour horizon. The spot ACF has
the familiar daily and half-daily peaks; the model reproduces their location
and most of their amplitude, with the realized curve remaining close to the
simulated band at the major market-clock lags. This shows that the physical-scale simulation
retains the main serial-dependence pattern, not only the state-scale fit in
\cref{fig:cal-wind-state-acf-psd}.

\subsection{Solar calibration}\label{subsec:calibration-results-solar}

For solar, the estimated state is
$\bm Y_t=(Y_t^{\mathrm S_\odot},Y_t^S)^\top$. The first coordinate is not the
capacity factor itself but the bounded, visibility-regularized and
deseasonalized latent cloud-shortfall state. The physical capacity factor is
obtained through the decreasing map
\eqref{eq:solar-inverse}. Consequently, a positive latent covariance
$\Cov(Y_t^{\mathrm S_\odot},Y_t^S)$ corresponds to a negative physical
covariance between solar production and prices, as explained in
Remark \ref{rem:latent-sign-convention}. All solar results below should therefore be
read with this sign convention in mind: model fitting occurs on
$Y^{\mathrm S_\odot}$, while PPA cash flows and full-level diagnostics are
reported for $C_t^{\mathrm S}$.

The state-scale second-order diagnostics in
\cref{fig:cal-solar-state-acf-psd} show that the fitted MCARMA law captures
the dominant solar and spot time scales. The solar latent ACF has a strong
daily structure induced by the clear-sky/weather decomposition and the simulated
median reproduces the main decay and the daily peaks, while the realized curve
is sharper at some calendar lags. The PSD makes this more explicit: the
selected daily and half-daily frequencies are present in both realized and
simulated spectra, and the low-frequency envelope is matched, although the
realized solar spectrum contains additional narrow high-frequency peaks that are
not included in the two-oscillator specification. The spot residual behaves as in
the wind calibration: the model reproduces the low-frequency decay and the
calendar peaks of $Y^S$, with realized values largely inside the simulated
band.

The full-level one-step metrics show similar agreement. After applying
\eqref{eq:solar-inverse}, the solar capacity factor has an RMSE of $0.0168$,
correlation $0.9931$ and bias $-0.0003$; the spot equation has an RMSE of
$11.80$~EUR/MWh and correlation $0.9730$. The close agreement between realized
and fitted standard deviations shows that the nonlinear reconstruction
preserves the calibrated state information on the physical scale used in
valuation.

The lead--lag diagnostic in \cref{fig:cal-solar-lead-lag} is plotted on the
physical scale as
$k\mapsto\operatorname{Corr}(S_{t+k},C_t^{\mathrm S})$. Its magnitude is much
smaller than the wind lead--lag curve, with values concentrated around zero
and a local contemporaneous peak of roughly $0.07$. This reflects the fact
that full-level solar output contains a deterministic daylight envelope:
intraday price and production profiles partly offset the negative latent
weather--price relation. The fitted model reproduces the sign change around
the contemporaneous lag and the simulation band contains the realized curve. The physical-scale correlation is relevant
for PPA cash flows because \eqref{eq:solar-inverse} is applied before prices are multiplied by solar
production.

\begin{figure*}[p]
\centering
\begin{subfigure}[t]{.96\textwidth}
\centering
\includegraphics[width=\linewidth,height=.38\textheight,keepaspectratio]{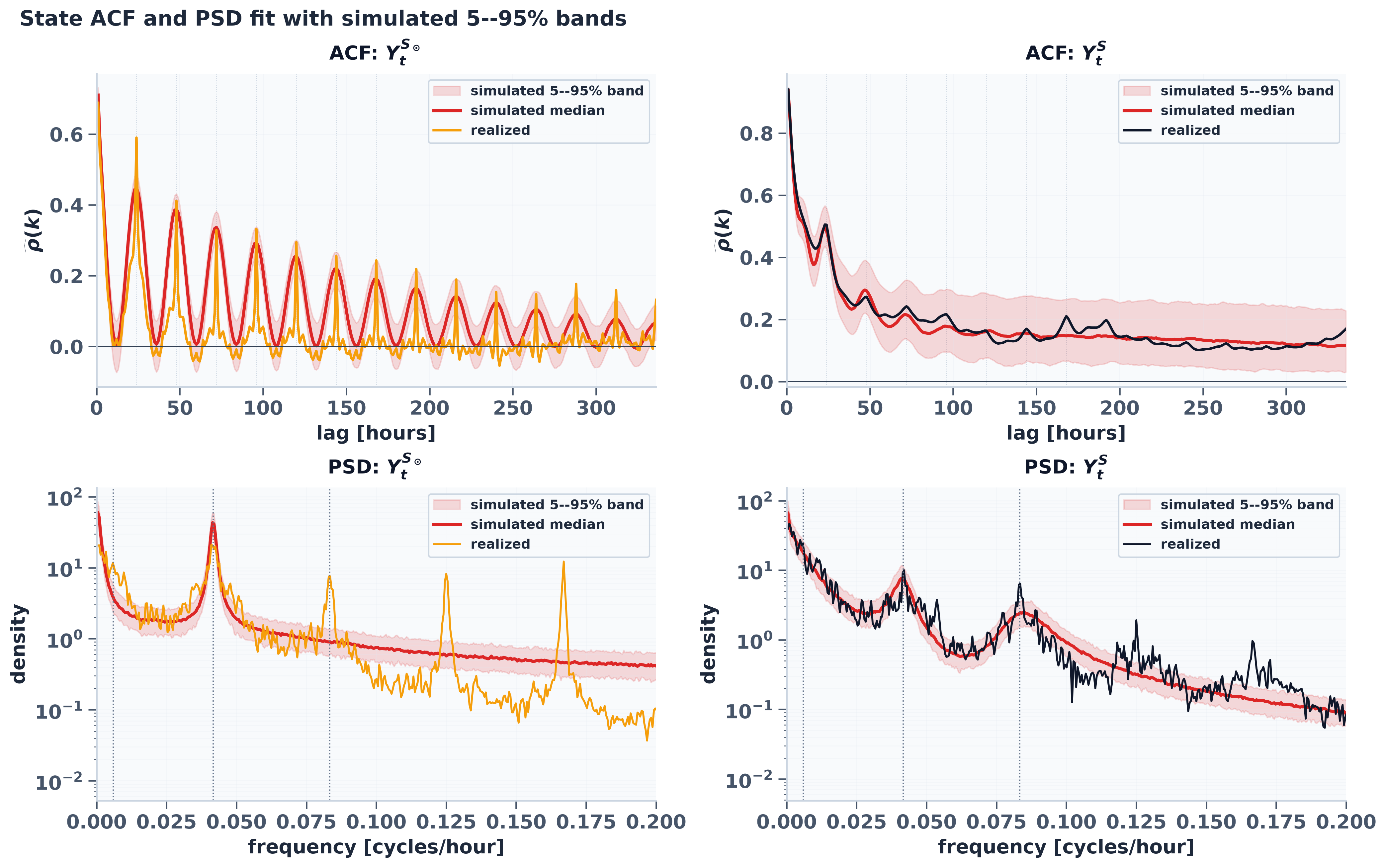}
\caption{Solar-shortfall and spot autocorrelations and spectra.}
\end{subfigure}\\[0.35em]
\begin{subfigure}[t]{.92\textwidth}
\centering
\includegraphics[width=\linewidth,height=.30\textheight,keepaspectratio]{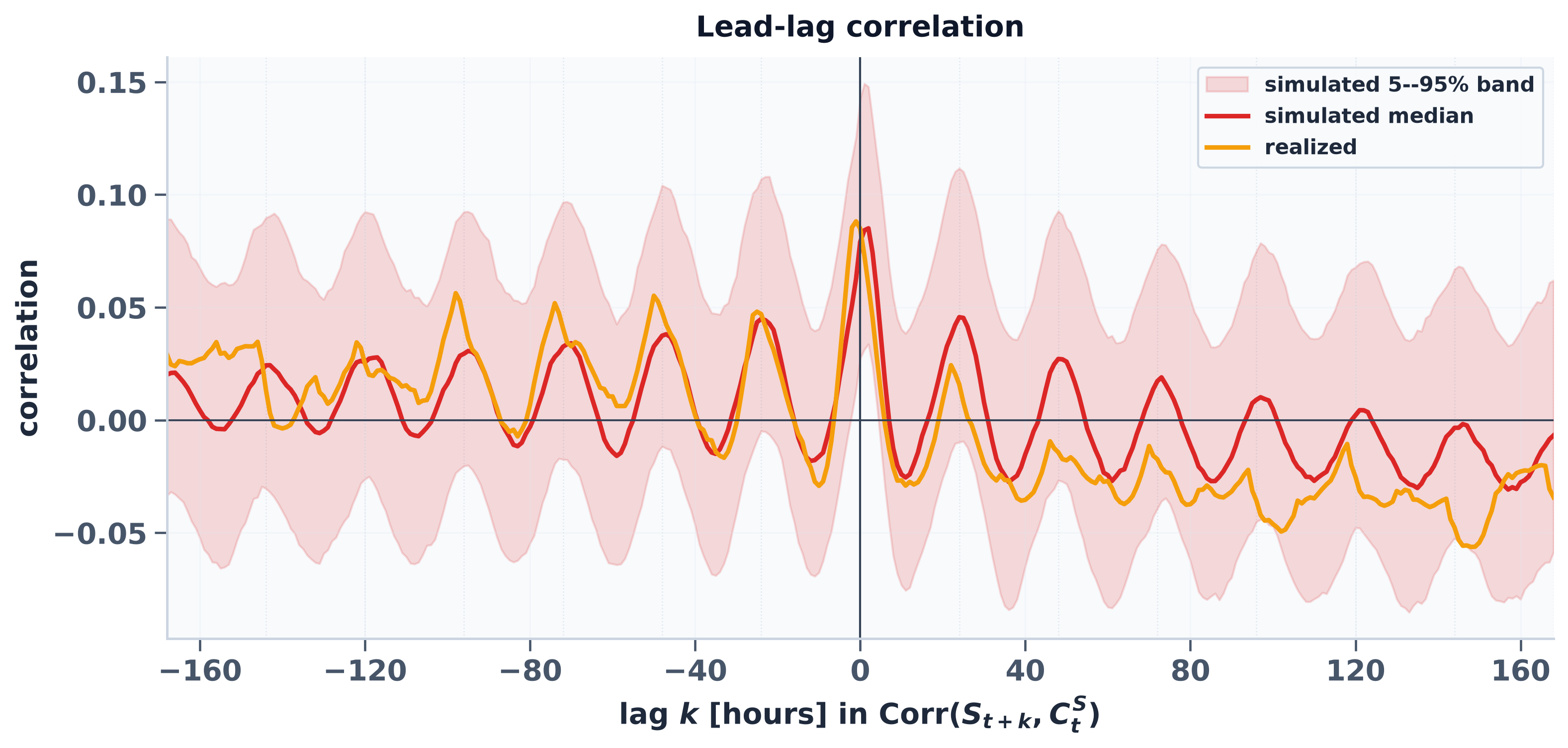}
\caption{Full-level solar--price lead--lag dependence after the decreasing
physical reconstruction.}
\end{subfigure}
\caption{Solar calibration on the deseasonalized state and dependence scales.}
\label{fig:cal-solar-state-acf-psd}
\label{fig:cal-solar-lead-lag}
\end{figure*}

\Cref{fig:cal-solar-mnig-marginals} shows that the model
reproduces the large mass of $C_t^{\mathrm S}$ near zero and the long daytime
right tail of solar capacity factors. The spot marginal again captures the
central price dispersion and heavy tails while smoothing the empirical
concentration near zero. The bivariate law in \cref{fig:cal-solar-mnig-bivar}
has the correct triangular support: high solar output is concentrated in lower
spot-price regions, while low output admits the full range of prices. After the decreasing reconstruction
map reverses the latent sign, the physical variables display the expected negative price-solar relation
in \eqref{eq:solar-inverse}.

The fitted full-level ACF in \cref{fig:cal-solar-mnig-acf} closely matches the empirical ACF of the solar
capacity factor. This is expected: once the clear-sky envelope and the
bounded inverse map are restored, the deterministic diurnal component dominates
the full-level autocorrelation. The fitted simulation therefore matches the
realized daily oscillations over the full 336-hour window. For $S_t$, the
model captures the daily and half-daily price peaks and the longer decay,
with residual discrepancies concentrated at the most pronounced daily price peaks.

\begin{figure*}[p]
\centering
\begin{subfigure}[t]{.49\textwidth}
\centering
\includegraphics[width=\linewidth,height=.22\textheight,keepaspectratio]{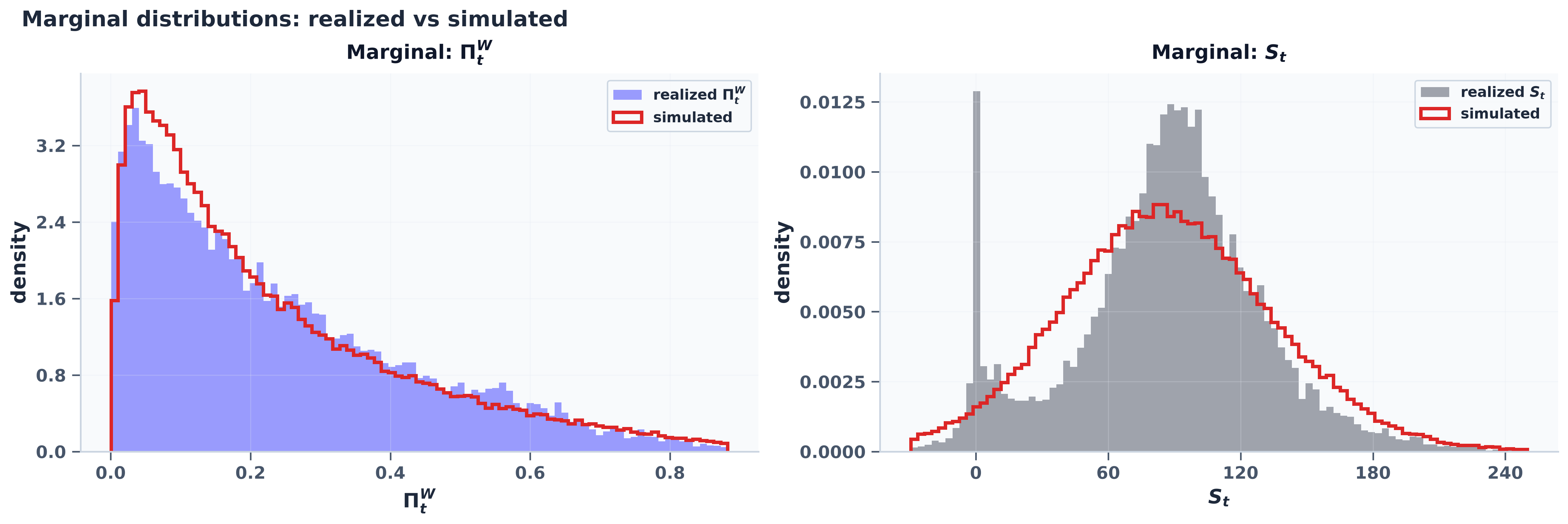}
\caption{Wind-pressure and spot marginals.}
\end{subfigure}\hfill
\begin{subfigure}[t]{.49\textwidth}
\centering
\includegraphics[width=\linewidth,height=.22\textheight,keepaspectratio]{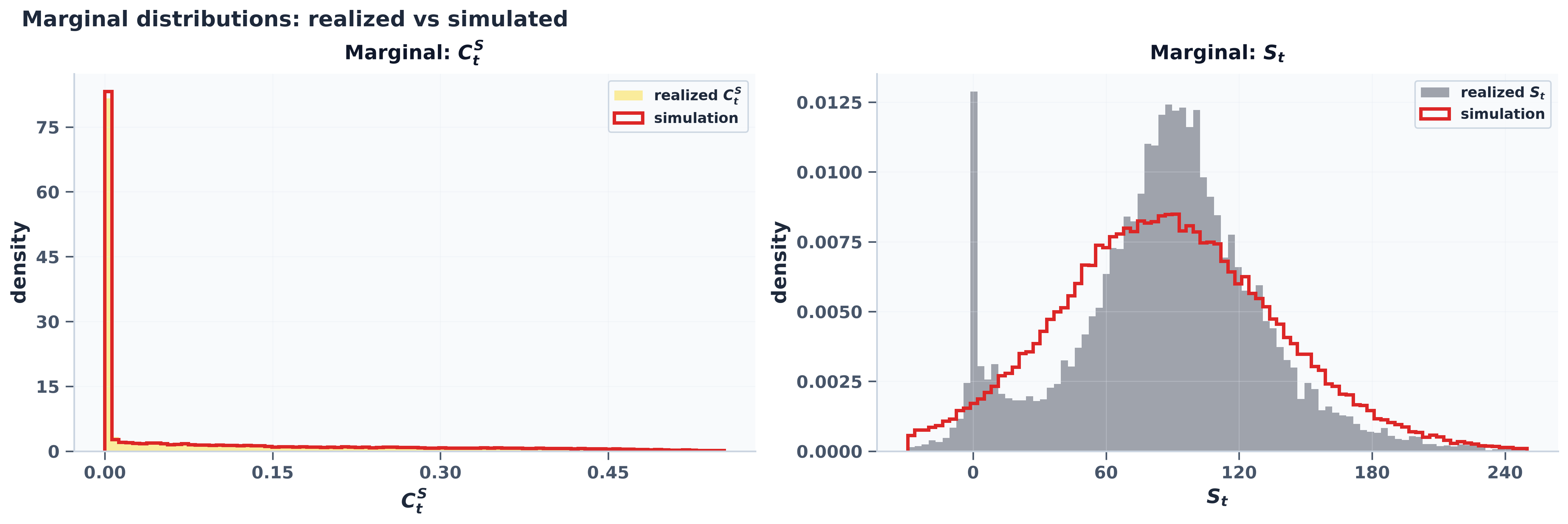}
\caption{Solar-capacity-factor and spot marginals.}
\end{subfigure}\\[0.2em]
\begin{subfigure}[t]{.49\textwidth}
\centering
\includegraphics[width=\linewidth,height=.22\textheight,keepaspectratio]{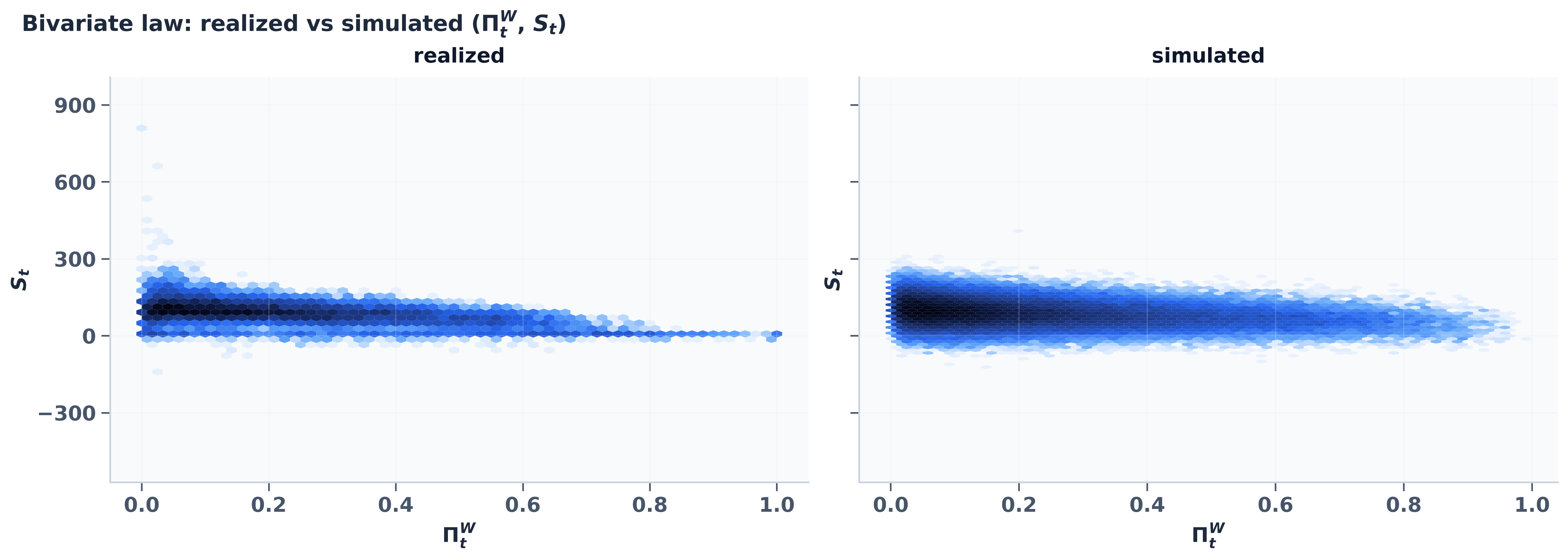}
\caption{Wind full-level bivariate law.}
\end{subfigure}\hfill
\begin{subfigure}[t]{.49\textwidth}
\centering
\includegraphics[width=\linewidth,height=.22\textheight,keepaspectratio]{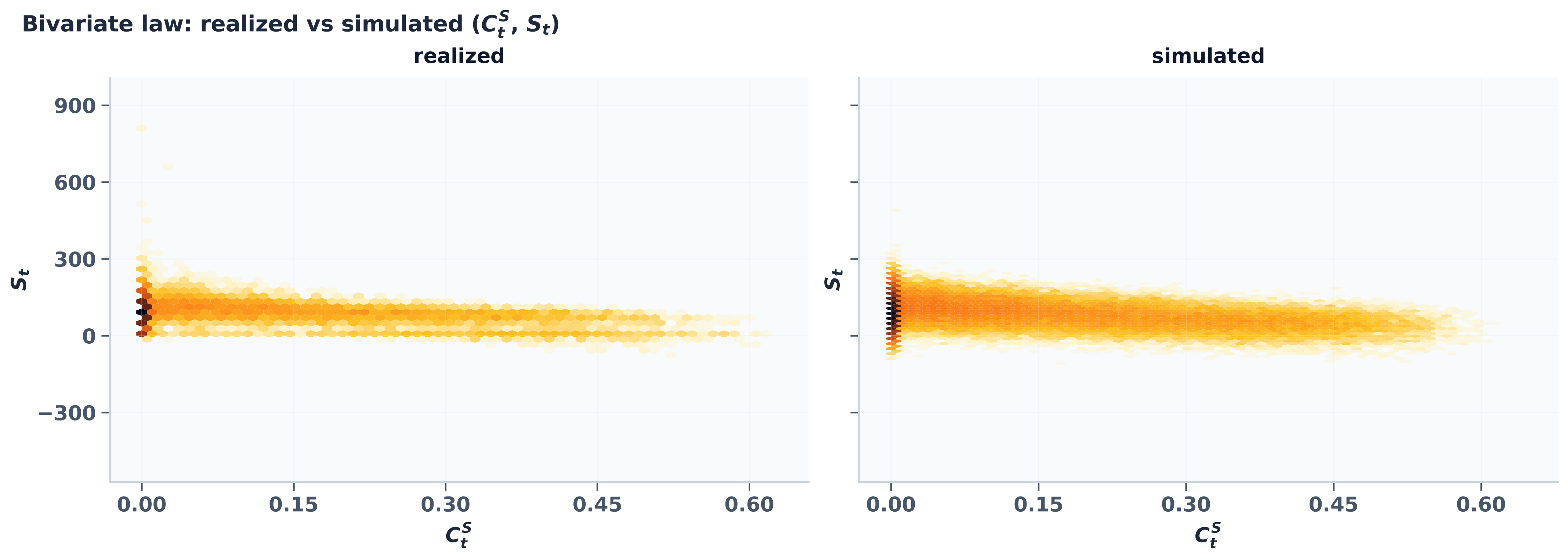}
\caption{Solar full-level bivariate law.}
\end{subfigure}\\[0.2em]
\begin{subfigure}[t]{.49\textwidth}
\centering
\includegraphics[width=\linewidth,height=.22\textheight,keepaspectratio]{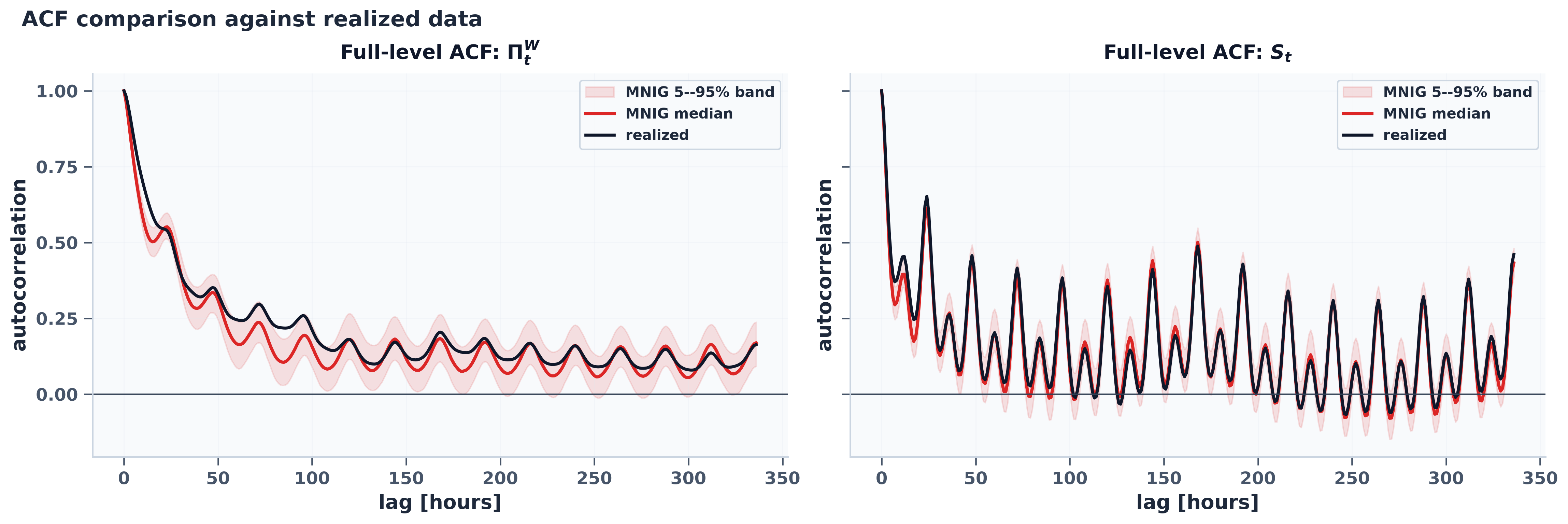}
\caption{Wind full-level autocorrelations.}
\end{subfigure}\hfill
\begin{subfigure}[t]{.49\textwidth}
\centering
\includegraphics[width=\linewidth,height=.22\textheight,keepaspectratio]{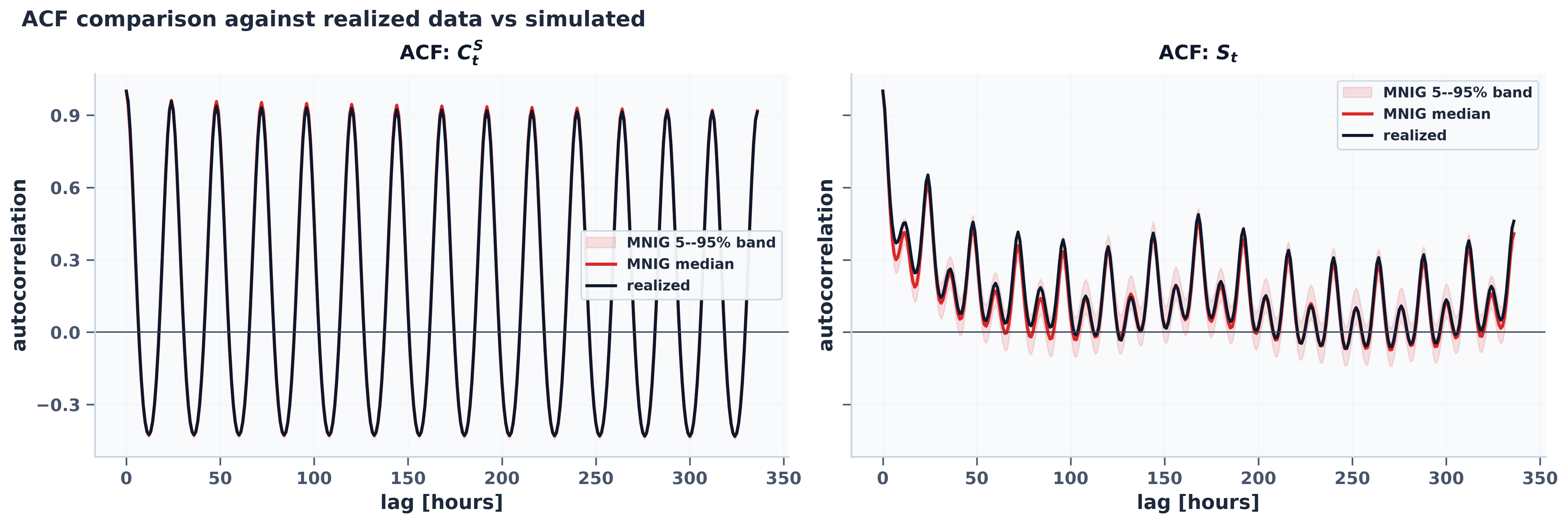}
\caption{Solar full-level autocorrelations.}
\end{subfigure}
\caption{Physical-scale calibration diagnostics for wind, solar,
and spot prices reconstructed from the latent MCARMA-MNIG
states.}
\label{fig:cal-wind-mnig-marginals}
\label{fig:cal-wind-mnig-bivar}
\label{fig:cal-wind-mnig-acf}
\label{fig:cal-solar-mnig-marginals}
\label{fig:cal-solar-mnig-bivar}
\label{fig:cal-solar-mnig-acf}
\end{figure*}

Overall, the calibration matches the main persistence and periodicity patterns of
the deseasonalized renewable and price states, while the MNIG and spike
components accommodate non-Gaussian variation and isolated price extremes. The
inverse physical maps also ensure that simulated renewable production remains
within its economically admissible bounds. The diagnostics nevertheless reveal
residual serial dependence in the squared recovered driver increments and
the simulated spot-price distribution understates the observed
skewness and excess kurtosis. The calibrated model should therefore be
interpreted as a parsimonious benchmark for joint price--production
dynamics rather than as a complete description of their conditional distribution. The hedging analysis in
\cref{sec:empirical-design} uses the model in this role. Future work could combine
richer short-run dynamics with stochastic volatility and stochastic covariance to account for the serial dependence of the squared residuals. The remaining discrepancies are concentrated in
economically important stress periods. The year 2025 was a low-wind year, so realized
$\Pi_t^{\mathrm W}$ often lies near the lower simulated band. Solar production
is well fitted in levels, but the realized 2025 solar value factor is an
unusually severe cannibalisation outcome. The backtest evaluates
the hedge under these out-of-sample volume and capture-price
conditions.

\section{Empirical hedging design}\label{sec:empirical-design}
\begingroup
\linespread{0.97}\selectfont

This section specifies the discrete implementation used in the empirical study:
the discretized payoff and fair strike, the backward-projection estimator that
implements the semi-static decomposition of \cref{thm:semistatic} on
simulated paths, the two trading conventions, the auxiliary claim universe, and
the backtesting protocol with its benchmark strategies. The
same protocol is used for wind and
solar; the corresponding results appear in
\cref{subsec:wind-backtest-results,subsec:solar-backtest-results}.

\subsection{Discretized payoff and fair strike}\label{subsec:discrete-payoff}

Let $\mathcal T=\{t_1<\cdots<t_M\}$ be the hourly delivery grid with discount
factors $d_i$ and weights $\Delta t_i$. The discretized PPA payoff and Monte
Carlo fair strike over $N$ simulated paths are
\begin{equation}\label{eq:discrete-ppa}
\begin{aligned}
  H^K&=\sum_{i=1}^Md_i\,q_{t_i}(S_{t_i}-K)\Delta t_i,\\
  \widehat K^\star&=
  \frac{\sum_{p=1}^N\sum_id_i\,q_{t_i}^{(p)}S_{t_i}^{(p)}\Delta t_i}
       {\sum_{p=1}^N\sum_id_i\,q_{t_i}^{(p)}\Delta t_i}.
\end{aligned}
\end{equation}
As a ratio of sample means, $\widehat K^\star$ is consistent under
\cref{ass:integrability} though biased in finite samples; the delivery horizon
is partitioned into monthly delivery cells $I_j$, and the same calculation is
applied separately to each monthly delivery period to obtain the wind and solar
fair strikes and their level, profile, and covariance components.

\subsection{Backward projection}\label{subsec:backward-projection}

Let $\rho_0<\rho_1<\cdots<\rho_{K'}$ be the rebalancing grid. At $\rho_k$, let
$\Delta\bm X_k$ be the vector of block gains of the active contracts over
$(\rho_k,\rho_{k+1}]$ and $\bm Z_{\rho_k}$ the hedge state (current state
variables and remaining-delivery summaries). For a centered payoff
$Y^H=H-\widehat\pi$, set $R_{K'}=Y^H$ and compute backward conditional
least-squares projections
\begin{equation}\label{eq:backward-argmin}
\begin{aligned}
  \bm\theta_k(\bm z)&=\argmin_{\bm\theta}
  \widehat{\E}\left[(R_{k+1}-\bm\theta^\top\Delta\bm X_k)^2
  \mid\bm Z_{\rho_k}=\bm z\right],\\
  R_k&=R_{k+1}-\bm\theta_k(\bm Z_{\rho_k})^\top\Delta\bm X_k ,
\end{aligned}
\end{equation}
whose first-order condition is the empirical analogue of the normal equations
\eqref{eq:dynamic-normal-equations},
\begin{equation}\label{eq:backward-normal}
\begin{aligned}
  &\widehat{\E}[\Delta\bm X_k\Delta\bm X_k^\top\mid\bm Z_{\rho_k}=\bm z]\,
  \bm\theta_k(\bm z)\\
  &\qquad
  =\widehat{\E}[\Delta\bm X_kR_{k+1}\mid\bm Z_{\rho_k}=\bm z].
\end{aligned}
\end{equation}
Running the recursion for $H^0,H^1,\ldots,H^n$ yields discrete dynamic
residuals $\widehat L^i$, and the static regression is then estimated from the following empirical residual moments
\begin{equation}\label{eq:discrete-abc}
\begin{aligned}
  \widehat a_0&=\frac1N\sum_{p}(\widehat L^{0,(p)})^2,
  \qquad
  \widehat{\bm b}=\frac1N\sum_{p}\widehat{\bm L}^{(p)}\widehat L^{0,(p)},\\
  \widehat{\bm C}&=\frac1N\sum_{p}\widehat{\bm L}^{(p)}\widehat{\bm L}^{(p)\top},
  \qquad
  \widehat{\bm\nu}^\star=\widehat{\bm C}^\dagger\widehat{\bm b},
\end{aligned}
\end{equation}
which estimate \eqref{eq:outer-abc}.

We implement the conditional projections in
\eqref{eq:backward-normal} using three design choices. First, the hedge state $\bm Z_{\rho_k}$
collects the simulated spot level, the capacity factor, the deseasonalized
price and production states and the spike-intensity state at $\rho_k$; all
coordinates are standardized on the training paths and winsorized at six
standard deviations. Second, the pre-delivery value $X_{\rho_k}^S(I_j)$ of each
active monthly future is itself a conditional expectation and is estimated by a
Gaussian-kernel nearest-neighbour regression of the discounted settlement
$A^S(I_j)$ on the standardized hedge state, with predictions clipped at extreme
training quantiles, so that block gains use model-implied conditional
marks rather than realized settlements.

Because the spot--renewable model is not jointly calibrated to the observed
EEX futures curve, these state-dependent marks are conditional projections of
simulated settlements under the fitted historical law rather than observed
tradable prices. The resulting dynamic strategy is therefore a model-implied benchmark rather
than a directly implementable market hedge. A market-implementable backtest
would instead require observed EEX mark-to-market histories, or
a separately estimated valuation model whose marks are validated
for every hedged maturity.

Third, the conditional
moments in \eqref{eq:backward-normal} are estimated by state-group local least
squares: the standardized hedge state is projected onto its first principal
component, the training paths are partitioned into five quantile groups along
this score, and within each group, the local moment matrices are shrunk
toward the global estimates using a fixed weight, and a proportional ridge
penalty regularizes near-collinear futures increments; for sparsely populated groups,
the global fit is used instead. The estimated hedge ratios
$\bm\theta_k(\cdot)$ are therefore piecewise constant in the state score.
At every block the empirical variance identity
\begin{equation}\label{eq:block-variance-identity}
\begin{aligned}
  \widehat{\Var}(R_{k+1})-\widehat{\Var}(R_k)
  &=2\,\widehat{\Cov}\bigl(R_{k+1},\bm\theta_k^\top\Delta\bm X_k\bigr)\\
  &\quad-\widehat{\Var}\bigl(\bm\theta_k^\top\Delta\bm X_k\bigr)
\end{aligned}
\end{equation}
is verified numerically, which separates the covariance term from the
variance contributed by the hedge and provides a numerical check of the projection
calculation. All prices, hedge ratios and static weights are estimated on a
training subset of paths and every reported risk statistic is computed on
held-out test paths, so that the reported hedge effectiveness is evaluated on paths that were
not used to estimate the projection; the split is specified in
\cref{subsec:backtest-protocol}.

\subsection{Trading conventions}\label{subsec:trading-conventions}

The two conventions of \cref{sec:single-period} are implemented through different active sets
and block-gain definitions. Under the pre-delivery convention, a contract with delivery period
$I$ is rebalanced only strictly before $I^-$, so
$\Act^{\rm pre}(\rho)=\{I:\rho<I^-\}$ and the final admissible block
includes the gain from the last pre-delivery mark to settlement:
\begin{equation}\label{eq:pre-gain}
  \Delta X^{S,{\rm pre}}(\rho,\rho';I)=
  \begin{cases}
  X_{\rho'}^S(I)-X_{\rho}^S(I),
  & \rho'<I^-,\\[0.35em]
  A^S(I)-X_{\rho}^S(I),
  & \rho<I^-\le\rho',
  \end{cases}
\end{equation}
and analogously for volume contracts. This is the discrete analogue of the
restricted Asian-option hedging space
$G^{\vartheta,I}=\int_0^{I^-}\vartheta_s\dd X_s(I)
+\vartheta_{I^-}(X_{I^+}(I)-X_{I^-}(I))$ of \citet{BenthDetering2015}. Under
the idealized during-delivery convention, contracts remain tradable through delivery. Without modelling the actual
cascade from monthly into daily contracts, we define the active set as
$\Act^{\rm EEX}(\rho)=\{I:\rho<I^+\}$ and block gains are ordinary increments
of the during-delivery value processes
\crefrange{eq:during-delivery-power}{eq:during-delivery-production},
truncated at $I^+$. The backward normal equations \eqref{eq:backward-normal}
are unchanged; only the active set and block-gain definitions differ, so the comparison
measures the model-implied effect of allowing rebalancing during delivery under
the stated idealization. The first hedge-state observation is taken one day before delivery begins
so that the first monthly future is tradable before its delivery begins, and
positions are rolled month by month across the delivery horizon. At each roll
date the stack holds the front monthly contract together with up to five
deferred monthly contracts: the front-contract gain is measured from its current mark
to settlement in the second line of \eqref{eq:pre-gain}, deferred contracts contribute
mark-to-market increments between consecutive roll dates, and contracts whose
block gains are numerically degenerate on the training paths are removed from
the active set. Economically, the front-month
position covers the current delivery month through settlement,
while deferred positions hedge changes in the estimated values
of later delivery months.

\subsection{Auxiliary claim universe}\label{subsec:claim-universe}

The static portfolio contains European-style claims on four discounted quantities
defined for each monthly delivery period $I_j$: the baseload average price
$F_j=A^S(I_j)/D(I_j)$, the delivered volume $V_j=A^{q,r}(I_j)$, the average
capacity factor $\bar C_j$, and the achieved (production-weighted) price $A_j=R_j/V_j$,
where $R_j$ is the discounted revenue earned from production during $I_j$. Let
$\bar F_j=\widehat{\E}[F_j]$, $\bar V_j=\widehat{\E}[V_j]$ and
$\bar c_j=\widehat{\E}[\bar C_j]$ denote training-sample means and
$\widehat Q^{\,\rm tr}_u(\cdot)$ the training-sample $u$-quantile of a
settlement variable. Four families are used, with notionals chosen so that their payoff
magnitudes are comparable with the monthly PPA cash flow.
\emph{Power vanillas}: calls $\bar V_j(F_j-K)^+$ and puts $\bar V_j(K-F_j)^+$
on the monthly baseload average, with call strikes at
$\widehat Q^{\,\rm tr}_u(F_j)$ for $u\in\{0.55,0.65,0.75\}$ and put strikes for
$u\in\{0.25,0.35,0.45\}$.
\emph{Renewable-index vanillas}: calls
$\bar q^rD(I_j)\bar F_j(\bar C_j-L)^+$ and puts
$\bar q^rD(I_j)\bar F_j(L-\bar C_j)^+$ on the monthly capacity factor, with
strikes at the same quantile levels of $\bar C_j$, constrained to $[0,1]$.
\emph{Price--renewable orthant quantos}: the four orthant payoffs
\eqref{eq:orthant-payoffs} on the centered pair $(F_j-a_j,\bar C_j-b_j)$ with
notional $\bar q^rD(I_j)$, at the three threshold pairs
\begin{equation}\label{eq:orthant-threshold-pairs}
\begin{aligned}
  (a_j,b_j)\in
  \bigl\{&(\bar F_j,\bar c_j),\\
  &\bigl(\widehat Q^{\,\rm tr}_{0.45}(F_j),
  \widehat Q^{\,\rm tr}_{0.55}(\bar C_j)\bigr),\\
  &\bigl(\widehat Q^{\,\rm tr}_{0.55}(F_j),
  \widehat Q^{\,\rm tr}_{0.45}(\bar C_j)\bigr)\bigr\},
\end{aligned}
\end{equation}
so that the threshold pairs cover the joint training mean and two nearby combinations with one marginal
threshold above and the other below its median. \emph{Capture-spread options}: options on the monetary
capture discount
\begin{equation}\label{eq:capture-discount}
\begin{aligned}
  D_j&=F_j-A_j,\\
  H_j^{{\rm capSpr},C}(K)&=\bar V_j(D_j-K)^+,\\
  H_j^{{\rm capSpr},P}(K)&=\bar V_j(K-D_j)^+,
\end{aligned}
\end{equation}
with strikes at the training quantiles of $D_j$ at the same call and put
levels. Capture-spread options directly target the profile and covariance effects represented
by the term $V_jD_j$ in $R_j=V_jF_j-V_jD_j$: they respond to
within-month alignment between production and hourly prices
that is not represented by options on monthly price or volume
alone.

All strikes are rounded to market quoting steps, one EUR/MWh for price-type
and capture-spread strikes and one percentage point for capacity-factor
strikes; if rounding produces duplicate strikes, one strike is shifted by a single quoting increment so that
strike grids remain distinct. Claims with negligible payoff variance on the
training paths are screened out. Each monthly cell contributes at most six
power vanillas, six renewable-index vanillas, twelve orthant quantos (four
orthants at three threshold pairs) and six capture-spread options: a nominal total of $12\times30=360$ claims for the twelve-month delivery horizon. The wind
catalog retains all 360 claims; four degenerate solar claims are screened out,
leaving 356. All static claims
are priced by training-sample Monte Carlo expectations under the model,
because these claims are treated as synthetic OTC instruments. Their residualized
payoffs are then used to assess how much covariance risk each family can hedge in the
sense of \cref{subsec:orthants}, and \cref{prop:orthant-sign} is checked on
the estimated residual moments \eqref{eq:discrete-abc}.

\subsection{Backtesting protocol and benchmark strategies}
\label{subsec:backtest-protocol}
The MCARMA--MNIG model of \cref{sec:model}, calibrated on the January 2023--December 2024 window as in
\cref{sec:calibration-results}, is used to simulate joint hourly
price--production paths over the delivery horizon January--December 2025 for a
pay-as-produced contract with installed capacity $\bar q^r=50$ MW. We use
$10{,}000$ paths for each of the wind and solar contracts. The hedge-state grid begins twenty-four hours before the first delivery hour.
The paths are split at random into $5{,}600$ training, $1{,}400$ validation
and $3{,}000$ test paths. The fair strike, kernel futures marks, dynamic hedge
ratios, static claim prices and coefficient fits use the training paths;
regularization and sparse-cardinality choices use the validation paths; and
every risk statistic reported below is computed once on the held-out test
paths. Interest rates are set to zero over the twelve-month
horizon, so discounted and nominal cash flows coincide. We backtest two calibrated variants: the continuous
MCARMA--MNIG model and the state-dependent-jump extension of \cref{subsec:spikes}; the state-dependent-jump
variant is the primary specification, and the continuous variant is used as
a robustness check. Under the restricted pre-delivery convention the
rebalancing grid consists of twelve monthly roll dates, each placed one day
before the start of the corresponding delivery month, and at each roll the
active stack holds the front month and up to five deferred monthly futures as
in \cref{subsec:trading-conventions}; under the idealized convention
the same projection is rerun on a daily rebalancing grid through delivery.

Six benchmark strategies are compared, each centered by its training-sample
price so that all profit-and-loss variables are mean-zero by construction on
the training set.
\emph{Unhedged fair-strike PPA}: the centered payoff
$H^{\widehat K^\star}-\widehat\pi_0$ with $\widehat\pi_0$ the training mean;
this is the unhedged benchmark.
\emph{Static baseload forward}: a single buy-and-hold fixed-for-floating
baseload swap over the whole horizon, with gain
$h\sum_id_i\{F_0(t_i)-S_{t_i}\}\Delta t_i$ and the volume $h$ chosen by
one-dimensional least squares on the training paths; this represents the common
stylized practice of hedging expected production with a flat baseload forward, and
it can remove level risk but neither profile nor covariance risk.
\emph{Static baseload plus peak forward}: adds a peak-block forward
(hours 08--20, weekdays) and refits the two volumes jointly; this provides
a simple exchange-traded approximation to the production-weighted delivery
profile.
\emph{Rolling front-month dynamic hedge}: the backward projection of
\cref{subsec:backward-projection} with the active set restricted to the front
monthly contract at each roll, i.e.\ a stack-and-roll
hedge using only the next delivery month.
\emph{Multi-month dynamic hedge}: the same backward projection on the full
active stack of up to six monthly maturities, the empirical version of
\eqref{eq:dynamic-theta}; relative to the front-month hedge it also hedges
changes in the model-implied marks of later monthly contracts.
\emph{Multi-month dynamic plus static claims}: the multi-month dynamic hedge
combined with the static overlay $\widehat{\bm\nu}^\star$ of
\eqref{eq:discrete-abc} estimated from the training residuals remaining after the dynamic futures hedge; this
is the empirical semi-static hedge of \cref{thm:semistatic}.

We evaluate hedge quality using the test-sample standard deviation of
terminal error and two lower-tail loss measures: with $q_{0.05}$ the
$5\%$-quantile of the centered strategy profit and loss $\eps$, we report the
loss statistics $\mathrm{VaR}_{95}=-q_{0.05}$ and
$\mathrm{CVaR}_{95}=-\E[\eps\mid\eps\le q_{0.05}]$, quoted as positive losses.
Risk reductions are stated as $1-{\rm metric}/{\rm metric}_{\rm ref}$,
referenced either to the unhedged fair-strike PPA or to the multi-month
dynamic hedge; using the dynamic hedge as the reference measures the additional risk
reduction obtained from the static portfolio.

We construct sparse static portfolios using two methods. The
first is greedy forward selection: at each cardinality $k$,
the method adds one claim and refits the restricted least-squares
problem in \eqref{eq:outer-abc}. The second is a LASSO frontier
solved by cyclic coordinate descent on an $\ell^1$-penalized version of the
same quadratic objective, followed by an unpenalized refit on the recovered
active set. Both methods are run without sign restrictions and under a
long-only constraint $\bm\nu\ge0$. 

\subsection{Wind PPA backtest}
\label{subsec:wind-backtest-results}

The wind delivery horizon is January--December 2025 and the auxiliary catalog
contains the 360 claims of \cref{subsec:claim-universe}. The model-implied
fair strike is $\widehat K^\star=80.01$ EUR/MWh against an expected discounted
baseload price of $88.90$ EUR/MWh, so the model-implied wind PPA strike is
$8.89$ EUR/MWh below baseload, with the difference decomposed into profile
and covariance components as in \eqref{eq:flat-profile-cov}. The ratio defining
$\widehat K^\star$ differs from the mean pathwise capture price
($80.14$ EUR/MWh) by the $-0.13$ EUR/MWh volume--capture covariance term. The conditional
monthly price--wind correlation implied by the calibrated model is between
$-0.71$ and $-0.68$ in every delivery month, and expected monthly delivered
volumes range from $14.9$ GWh in January to $5.2$ GWh in June before rising
to $13.3$ GWh in December, so both price-volume
covariance and seasonal variation in expected production materially affect
the 2025 backtest.

The wind heatmap in \cref{fig:wind-hour-month-ppa} reports contributions for
every hour of the day: its strongest positive contributions occur in the January--February
morning and evening shoulder hours, whereas negative midday and
afternoon cells indicate hours in which the wind-weighted spot
price is below $\widehat K_{\mathrm W}^\star$, producing negative
PPA contributions.

\subsubsection{Hedging effectiveness of the benchmark strategies}
\label{subsubsec:wind-benchmarks}

\begin{table*}[t]
\centering
\caption{Wind PPA test-set risk levels for the main hedging benchmarks.
Reductions are relative to the unhedged fair-strike PPA.}
\label{tab:wind-main-risk}
\footnotesize
\begin{tabular}{lrrrr}
\toprule
Strategy & Std. & 95\% CVaR loss & Std. red. & CVaR red. \\
 & (kEUR) & (kEUR) & (\%) & (\%) \\
\midrule
Unhedged fair-strike PPA & 359.5 & 800.3 & 0.0 & 0.0 \\
Static baseload forward & 236.8 & 509.9 & 34.1 & 36.3 \\
Static baseload plus peak forward & 236.8 & 510.1 & 34.1 & 36.3 \\
Front-month dynamic futures hedge & 230.2 & 504.3 & 36.0 & 37.0 \\
Multi-month dynamic futures hedge & 228.8 & 495.4 & 36.4 & 38.1 \\
Multi-month dynamic plus static claims & 45.6 & 101.4 & 87.3 & 87.3 \\
\bottomrule
\end{tabular}
\end{table*}

\Cref{tab:wind-main-risk} distinguishes the risk reduction from liquid futures from the
additional reduction provided by the static claims. A static baseload hedge already removes
about one third of the PPA dispersion, and allowing a rolling variance-optimal
stack of monthly futures improves the standard-deviation reduction only
modestly, from $34.1\%$ to $36.4\%$. Most of the reduction comes from the
auxiliary claims: adding the static portfolio to the same multi-month dynamic hedge
reduces the test standard deviation from $228.8$ kEUR to $45.6$ kEUR and the
95\% CVaR loss from $495.4$ kEUR to $101.4$ kEUR. This is consistent with the
decomposition in \cref{thm:semistatic}: most of the risk remaining after futures
hedging comes from production-shape and price-volume covariance effects rather than
average-price exposure. \Cref{fig:wind-pnl-residuals} shows the corresponding
out-of-sample residual distributions: the forward-only and dynamic-only residual
distributions are narrower than the unhedged distribution but remain heavy-tailed,
whereas the semi-static residual is much more concentrated around zero.

Two robustness checks show that the conclusion is stable within the model. First,
the same backtest under the continuous MCARMA--MNIG variant gives nearly identical
risk reductions: the multi-month dynamic hedge reduces the test standard
deviation by $36.4\%$ and the semi-static hedge by $87.2\%$, against $36.4\%$
and $87.3\%$ for the state-dependent-jump specification. Second, rerunning the
projection with daily rebalancing
through delivery, i.e.\ 378 rebalancing times and 377 delivery blocks instead
of twelve monthly rolls, performs worse than the
monthly-roll strategy in this implementation. One possible
explanation is that estimation noise in the intramonth
positions offsets the benefit of more frequent updating.

\subsubsection{Naive volume deltas versus variance-optimal positions}
\label{subsubsec:wind-delta}

The front-month positions show why the variance-optimal hedge differs from the expected-volume hedge.
For a single futures gain $\Delta F_j$ the variance-optimal position is
$\theta_j^\star=\Cov(H_j,\Delta F_j)/\Var(\Delta F_j)$, the one-dimensional
specialization of \eqref{eq:backward-normal}, while the naive delta sells the
model-implied expected delivered volume $\bar q^r\,\widehat{\E}[\bar C_j]$ of each month
forward. In the wind backtest we define the difference between the variance-optimal
position and the naive volume delta as the covariance adjustment
(\cref{tab:wind-delta-diagnostics}). The adjustment is negative in
January--March and August--December and positive in April--July; for example
the February position is $7.99$ MW against a naive model delta of $16.10$ MW.
The hedge is smaller than the volume delta
when high prices and wind output are negatively associated in the
conditional distribution: in such states the PPA has less high-price exposure
because high prices tend to coincide with low production; selling the full expected
volume therefore produces a position larger than the conditional least-squares
coefficient.

\begin{table*}[t]
\centering
\caption{Wind test-set hedging effectiveness of front-month delta rules versus
the variance-optimal projection. Reductions are relative to the unhedged
fair-strike PPA, and $\rho(H,G)$ is the test correlation between the centered
PPA payoff and the hedge gain.}
\label{tab:wind-delta-effectiveness}
\footnotesize
\begin{tabular}{lrrrrr}
\toprule
Hedge & Std. & 95\% CVaR loss & Std. red. & CVaR red. & $\rho(H,G)$ \\
 & (kEUR) & (kEUR) & (\%) & (\%) &  \\
\midrule
Unhedged fair-strike PPA & 359.5 & 800.3 & 0.0 & 0.0 & -- \\
Front-month naive model delta & 243.9 & 533.3 & 32.2 & 33.4 & 0.75 \\
Front-month naive realized delta & 250.2 & 554.0 & 30.4 & 30.8 & 0.72 \\
Front-month variance-optimal $\theta$ & 228.6 & 496.3 & 36.4 & 38.0 & 0.77 \\
Multi-month dynamic futures hedge & 228.8 & 495.4 & 36.4 & 38.1 & 0.77 \\
\bottomrule
\end{tabular}
\end{table*}
Table \ref{tab:wind-delta-effectiveness} quantifies the
additional residual risk from omitting the covariance
adjustment. The variance identity separates the covariance term
$2\Cov(H,G)$ from the hedge variance $\Var(G)$. The naive model
delta attains a $32.2\%$ standard-deviation reduction against $36.4\%$ for the
variance-optimal front-month position, and its tail improvement is
correspondingly weaker ($33.4\%$ against $38.0\%$ CVaR reduction). The large
negative adjustment in February illustrates the mechanism: the
variance-optimal rule avoids the hedge-variance cost of selling the full
model-implied volume when price and output are strongly negatively coupled. A
hindsight delta based on realized volumes performs worse still ($30.4\%$),
which is consistent with the covariance adjustment, rather than volume
information alone, being important for the hedge. Finally, extending the
variance-optimal front-month position to the full multi-month stack leaves the
standard-deviation reduction unchanged to one decimal place; almost all of the estimated risk
reduction from dynamic futures is obtained with the variance-optimal front-month position.

\begin{table*}[t]
\centering
\caption{Wind front-month variance-optimal futures positions versus naive
volume deltas. PPA values are realized January--December 2025 values.}
\label{tab:wind-delta-diagnostics}
\scriptsize
\renewcommand{\arraystretch}{0.88}
\begin{tabular}{lrrrrr}
\toprule
Month & PPA & Naive delta & VO delta & Adjustment & Ratio \\
 & (kEUR) & (MW) & (MW) & (MW) &  \\
\midrule
2025-01 & 46.3 & 19.98 & 14.37 & -5.60 & 0.72 \\
2025-02 & 191.0 & 16.10 & 7.99 & -8.11 & 0.50 \\
2025-03 & -30.7 & 13.90 & 13.55 & -0.35 & 0.97 \\
2025-04 & -24.0 & 11.67 & 15.28 & 3.61 & 1.31 \\
2025-05 & -115.6 & 8.51 & 11.00 & 2.50 & 1.29 \\
2025-06 & -211.2 & 7.22 & 7.32 & 0.10 & 1.01 \\
2025-07 & 11.1 & 9.03 & 12.08 & 3.05 & 1.34 \\
2025-08 & -46.9 & 7.12 & 7.97 & 0.85 & 1.12 \\
2025-09 & -131.4 & 8.40 & 7.73 & -0.67 & 0.92 \\
2025-10 & -263.8 & 13.63 & 11.58 & -2.05 & 0.85 \\
2025-11 & 75.5 & 14.34 & 8.68 & -5.65 & 0.61 \\
2025-12 & 39.1 & 17.85 & 12.83 & -5.02 & 0.72 \\
\bottomrule
\end{tabular}
\end{table*}

\begin{figure*}[!t]
\centering
\begin{subfigure}[t]{.49\textwidth}
\centering
\includegraphics[width=\linewidth,height=.31\textheight,keepaspectratio]{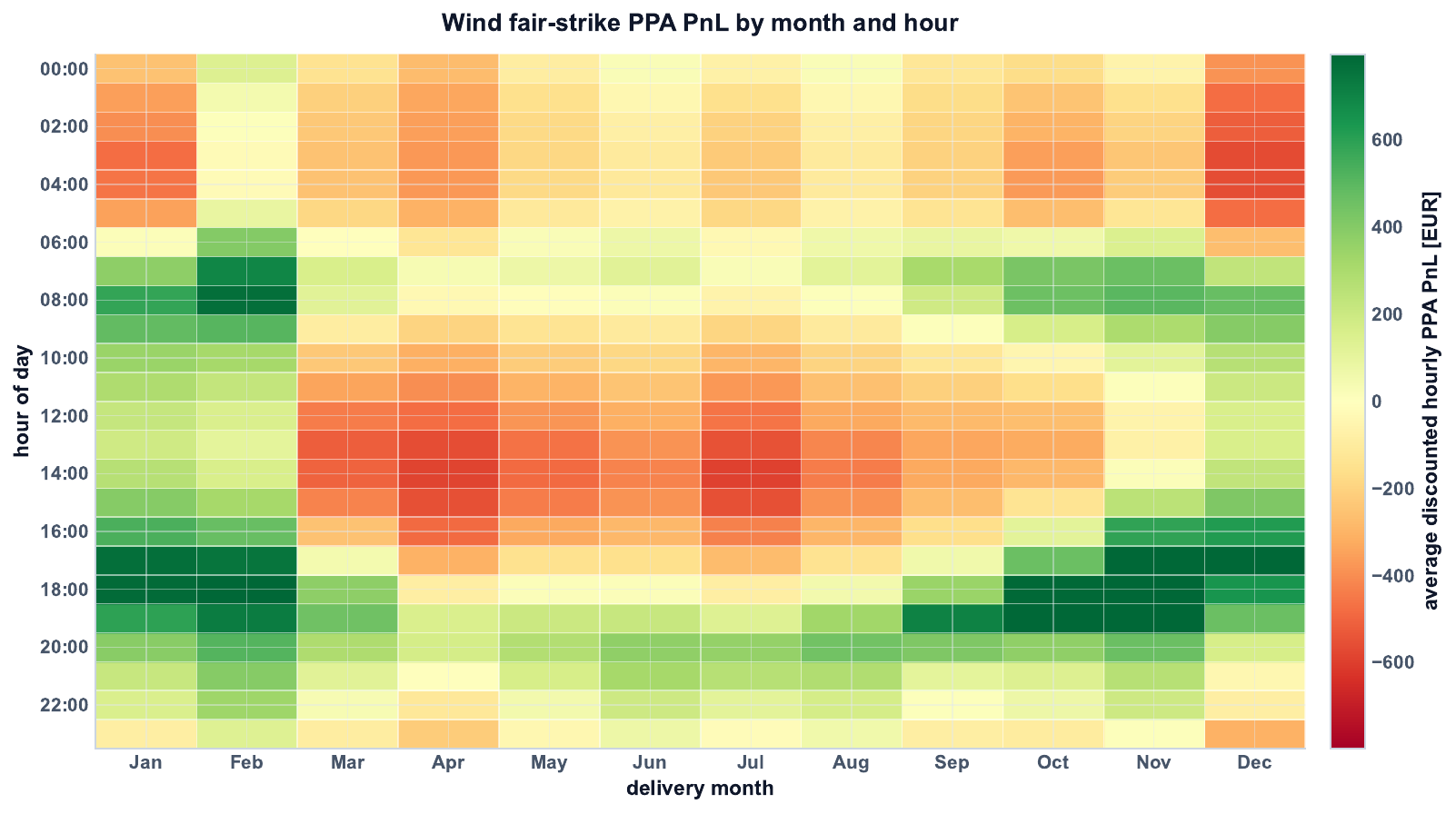}
\caption{Mean fair-strike PPA payoff by delivery month and hour.}
\end{subfigure}\hfill
\begin{subfigure}[t]{.49\textwidth}
\centering
\includegraphics[width=\linewidth,height=.31\textheight,keepaspectratio]{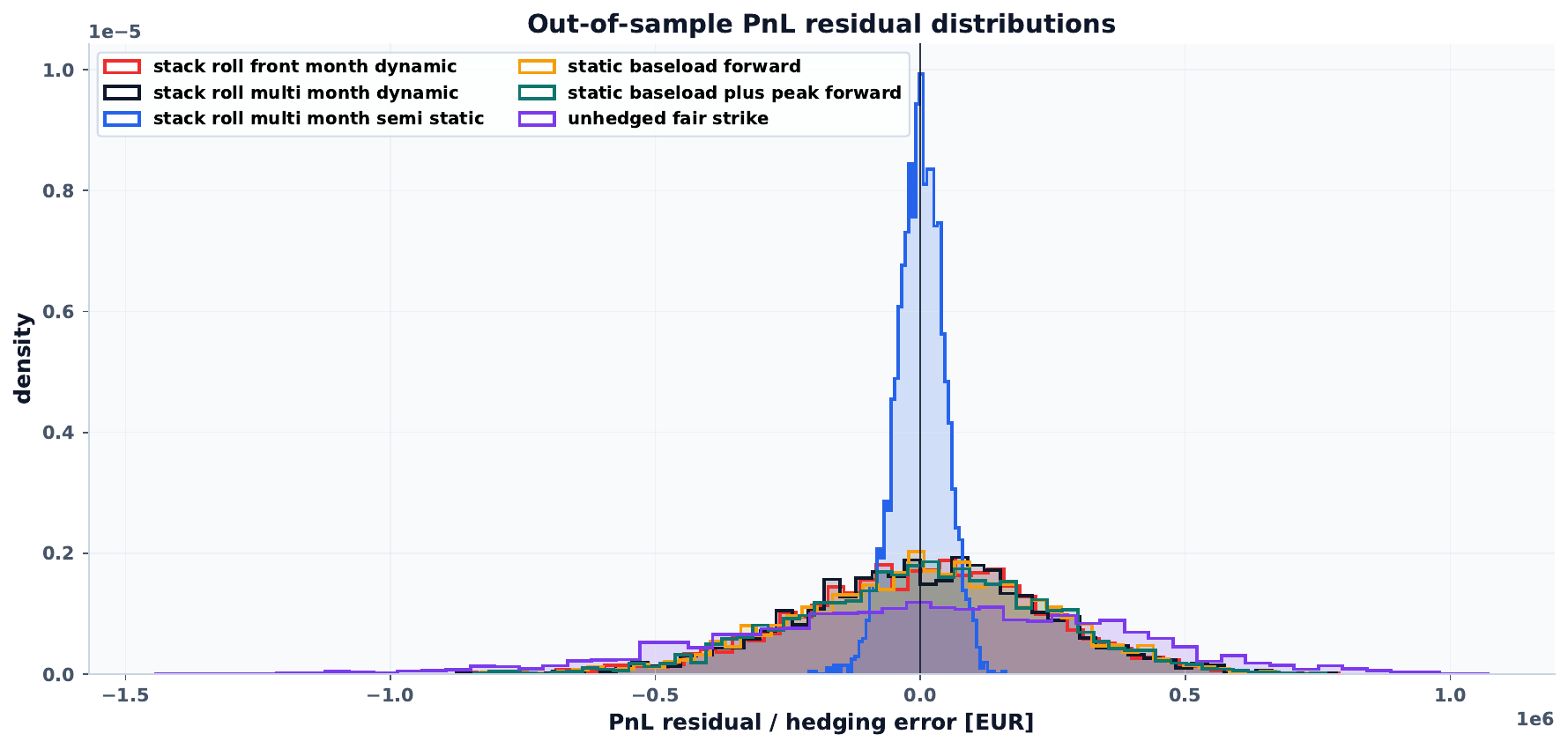}
\caption{Out-of-sample residual distributions across hedge designs.}
\end{subfigure}
\caption{Wind PPA contributions by delivery month and hour, together with hedging residuals.
The heatmap shows which month-hour combinations contribute to the capture discount,
while the distributional panel compares the unhedged, futures-only and
semi-static outcomes.}
\label{fig:wind-hour-month-ppa}
\label{fig:wind-pnl-residuals}
\end{figure*}

\subsubsection{The efficient frontier in the number of static claims}
\label{subsubsec:wind-frontier}

\Cref{fig:wind-efficient-frontier-static-claims,fig:wind-claim-count-risk-reduction}
show risk as a function of the number $m$ of static claims combined
with the multi-month dynamic hedge. Claims are added in the order
selected by the greedy algorithm. Risk falls rapidly for the first selected claims
and then declines more slowly: a small portfolio already removes a substantial share
of the residual dispersion, and the first twenty to thirty claims capture most of the
attainable reduction, after which the validation-selected unconstrained
frontier reaches $76.5\%$ reduction of the post-dynamic residual standard
deviation with 75 claims. Standard deviation and
$95\%$ CVaR loss decrease almost proportionally along the frontier, so the selected
portfolios do not improve one of these metrics at the expense of
the other. For product design, the practical implication is that
a relatively small subset of claims attains most of the risk reduction
achieved by the full catalog.

\begin{figure*}[!t]
\centering
\begin{subfigure}[t]{.78\textwidth}
\centering
\includegraphics[width=\linewidth,height=.31\textheight,keepaspectratio]{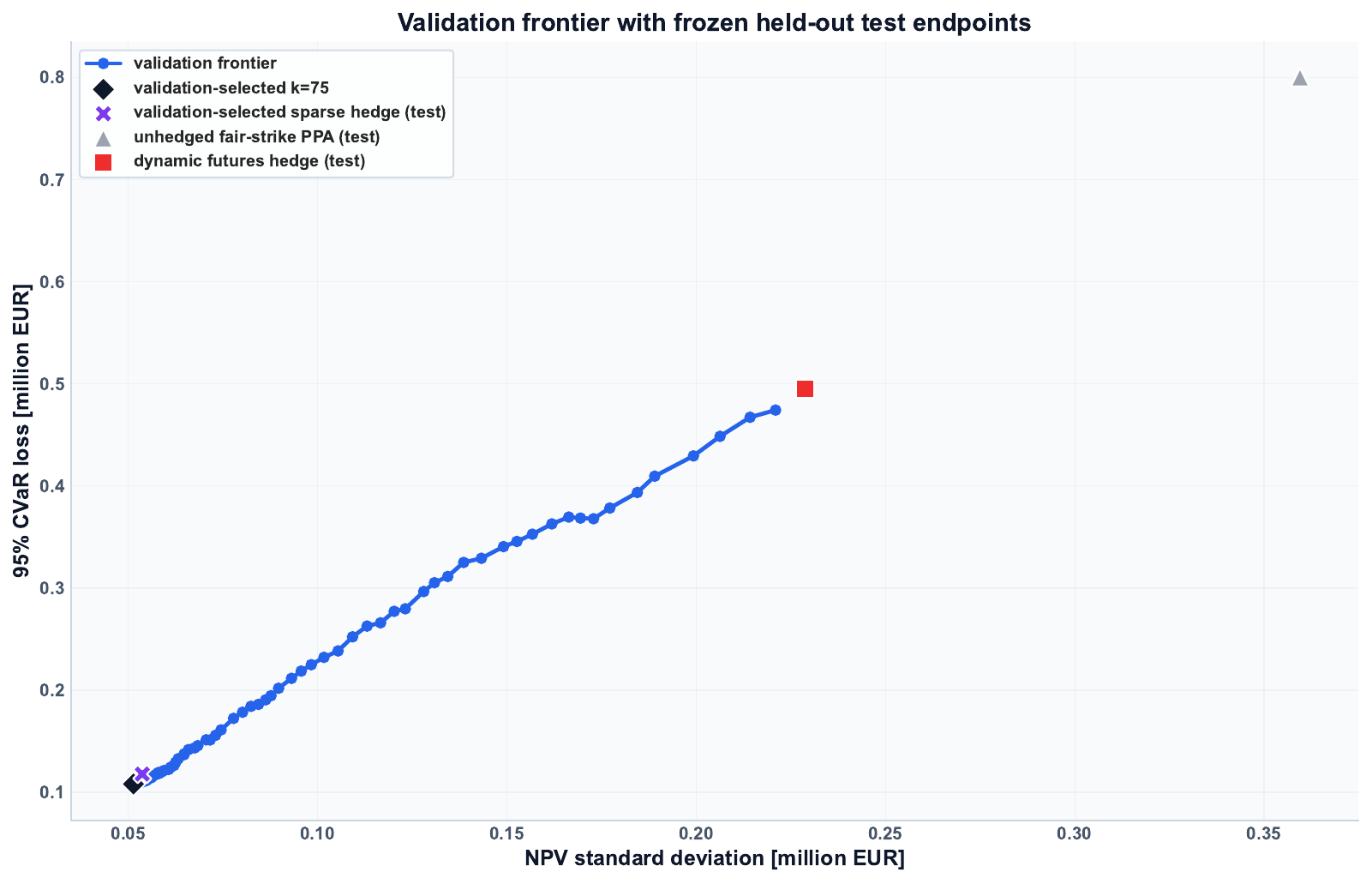}
\caption{Validation frontier in standard deviation and CVaR, with the selected portfolio evaluated once on the test set.}
\end{subfigure}
\caption{Wind efficient frontier in the number of static claims.}
\label{fig:wind-efficient-frontier-static-claims}
\label{fig:wind-claim-count-risk-reduction}
\end{figure*}

\subsubsection{Sparse static portfolios: selection, composition and leverage}
\label{subsubsec:wind-sparse}

\begin{table*}[t]
\centering
\caption{Wind sparse static portfolios selected from the all-product catalog.
Risk reductions are relative to the multi-month dynamic futures hedge and to
the unhedged fair-strike PPA, respectively.}
\label{tab:wind-sparse-risk}
\footnotesize
\begin{tabular}{llrrrrrr}
\toprule
Method & Constraint & $k$ & Std. & 95\% CVaR & Red. vs dyn. & Red. vs unhedged & Gross $|\nu|$ \\
 & & & (kEUR) & (kEUR) & (\%) & (\%) &  \\
\midrule
Greedy forward selection & Long-only & 31 & 119.5 & 224.1 & 47.78 & 66.77 & 27.83 \\
Greedy forward selection & Unconstrained & 75 & 53.7 & 117.9 & 76.53 & 85.06 & 72.69 \\
LASSO coordinate descent & Long-only & 57 & 118.8 & 223.5 & 48.08 & 66.96 & 27.37 \\
LASSO coordinate descent & Unconstrained & 79 & 67.0 & 150.0 & 70.72 & 81.37 & 46.66 \\
\bottomrule
\end{tabular}
\end{table*}

\begin{table*}[t]
\centering
\caption{Composition of the wind sparse static portfolios. Entries are counts
of nonzero static claims in each family.}
\label{tab:wind-sparse-composition}
\footnotesize
\begin{tabular}{llrrrrrr}
\toprule
Method & Constraint & Power & Wind & Quantos & Capture spread & Net $\nu$ & Gross $|\nu|$ \\
\midrule
Greedy forward selection & Long-only & 4 & 7 & 8 & 12 & 27.83 & 27.83 \\
Greedy forward selection & Unconstrained & 12 & 12 & 25 & 26 & 9.76 & 72.69 \\
LASSO coordinate descent & Long-only & 15 & 12 & 18 & 12 & 27.37 & 27.37 \\
LASSO coordinate descent & Unconstrained & 13 & 15 & 27 & 24 & 0.26 & 46.66 \\
\bottomrule
\end{tabular}
\end{table*}

\Cref{tab:wind-sparse-risk,tab:wind-sparse-composition} give the two main
sparse-portfolio comparisons. On the simulated test set, the unconstrained greedy
portfolio has the lowest residual variance: with $k=75$ selected products it reaches
a $53.7$ kEUR test standard deviation and a $117.9$ kEUR 95\% CVaR loss,
i.e.\ an $85.1\%$ standard-deviation reduction relative to the unhedged PPA.
The unconstrained LASSO portfolio selects 79 products and has higher simulated
risk but a much smaller gross absolute weight,
$\norm{\bm\nu}_1=46.66$ instead of $72.69$, and a nearly flat net position.
The $\ell^1$ penalty thus reduces leverage: it attains a similar, though
weaker, risk reduction with roughly two-thirds of the gross notional. The
long-only portfolios are much closer to insurance portfolios. They have similar
risk under greedy selection and LASSO, about $119$ kEUR standard deviation
and $224$ kEUR 95\% CVaR loss, and they reduce risk by slightly less than half
relative to the dynamic hedge because the long-only constraint prevents short option positions that would
offset nonlinear components of the residual payoff. The unconstrained portfolios, by contrast, use long and short
option positions jointly to approximate the residual regression coefficient
\eqref{eq:nu-star}, whose entries are signed by construction.

\begin{table*}[t]
\centering
\caption{Wind single-family sparse frontiers at their selected endpoints.
Reductions in test standard deviation are relative to the multi-month dynamic
futures hedge ($228.8$ kEUR standard deviation, $495.4$ kEUR 95\% CVaR loss).}
\label{tab:wind-family-endpoints}
\footnotesize
\begin{tabular}{lrrrrrrrr}
\toprule
 & \multicolumn{4}{c}{Long-only} & \multicolumn{4}{c}{Unconstrained} \\
\cmidrule(lr){2-5}\cmidrule(lr){6-9}
Candidate family & $k$ & Std. & 95\% CVaR & Red. & $k$ & Std. & 95\% CVaR & Red. \\
 & & (kEUR) & (kEUR) & (\%) & & (kEUR) & (kEUR) & (\%) \\
\midrule
All families & 31 & 119.5 & 224.1 & 47.8 & 75 & 53.7 & 117.9 & 76.5 \\
Capture spread & 11 & 157.8 & 326.5 & 31.0 & 23 & 124.4 & 283.2 & 45.6 \\
Quanto orthant & 11 & 208.7 & 449.4 & 8.8 & 27 & 200.5 & 435.8 & 12.3 \\
Wind vanillas & 9 & 207.6 & 454.9 & 9.2 & 12 & 205.3 & 450.0 & 10.2 \\
Power vanillas & 7 & 221.2 & 468.6 & 3.3 & 8 & 220.6 & 468.7 & 3.6 \\
\bottomrule
\end{tabular}
\end{table*}

\begin{figure*}[p]
\centering
\begin{subfigure}[t]{.49\textwidth}
\centering
\includegraphics[width=\linewidth,height=.29\textheight,keepaspectratio]{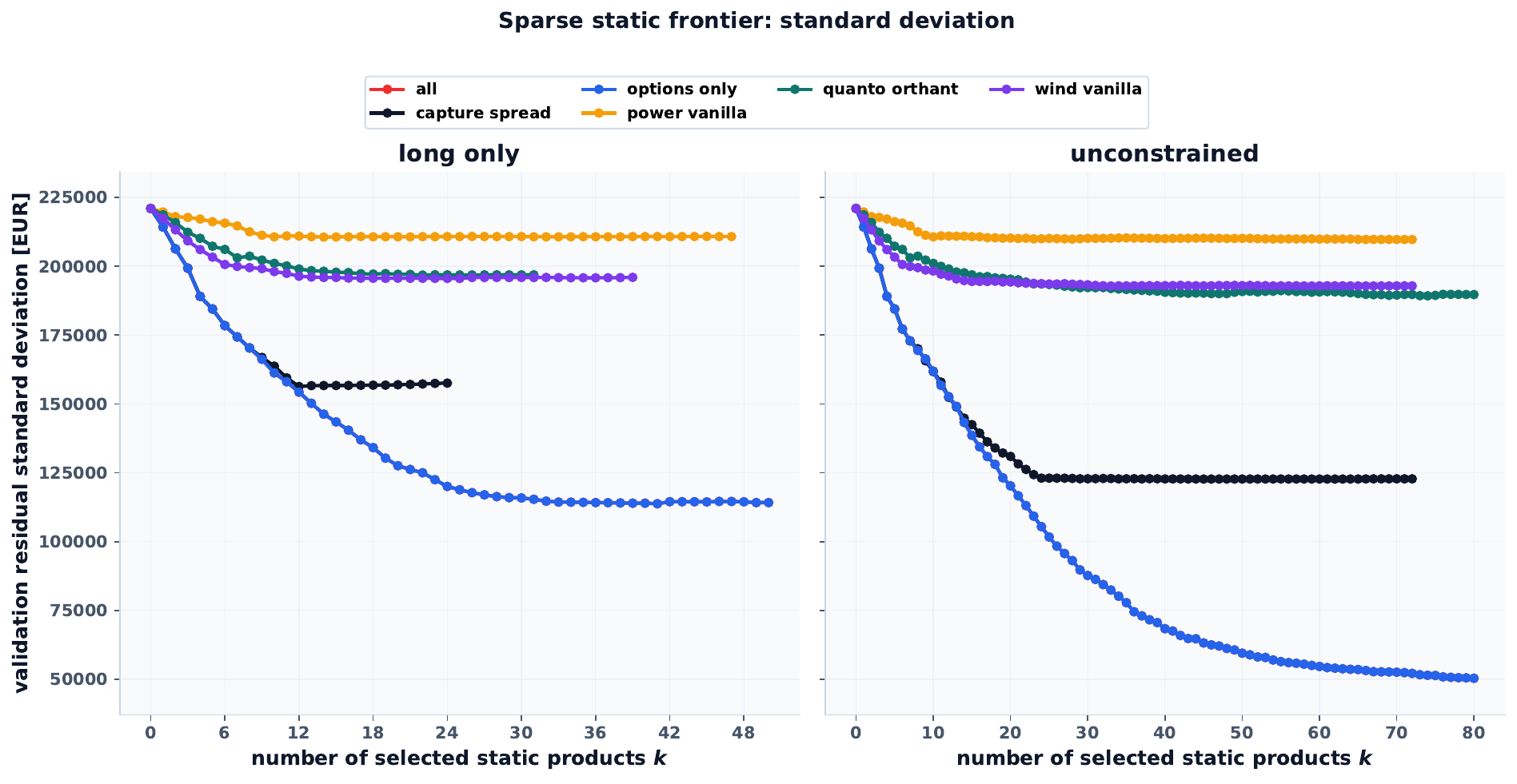}
\caption{Validation standard deviation.}
\end{subfigure}\hfill
\begin{subfigure}[t]{.49\textwidth}
\centering
\includegraphics[width=\linewidth,height=.29\textheight,keepaspectratio]{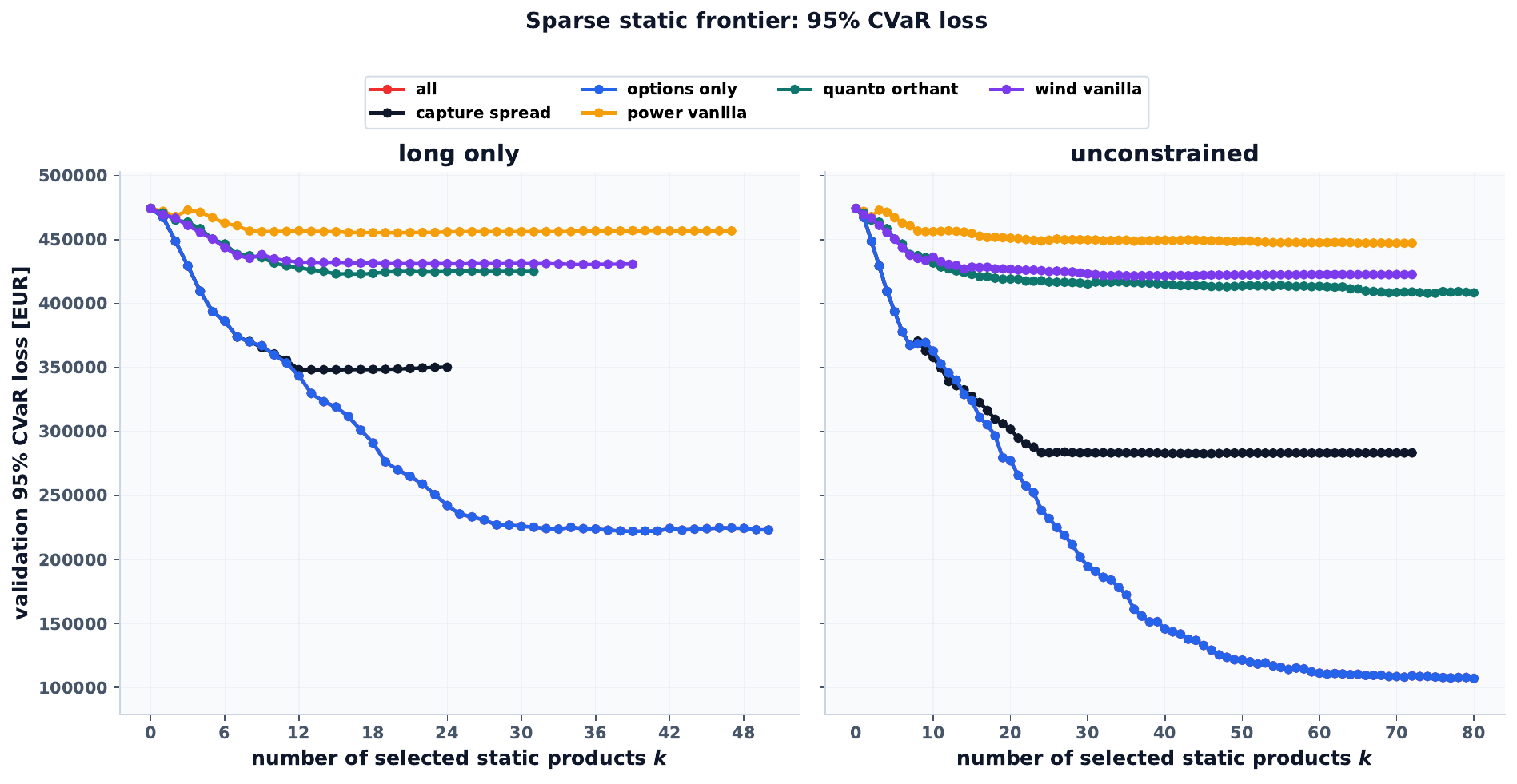}
\caption{Validation 95\% CVaR.}
\end{subfigure}\\[0.35em]
\begin{subfigure}[t]{.78\textwidth}
\centering
\includegraphics[width=\linewidth,height=.29\textheight,keepaspectratio]{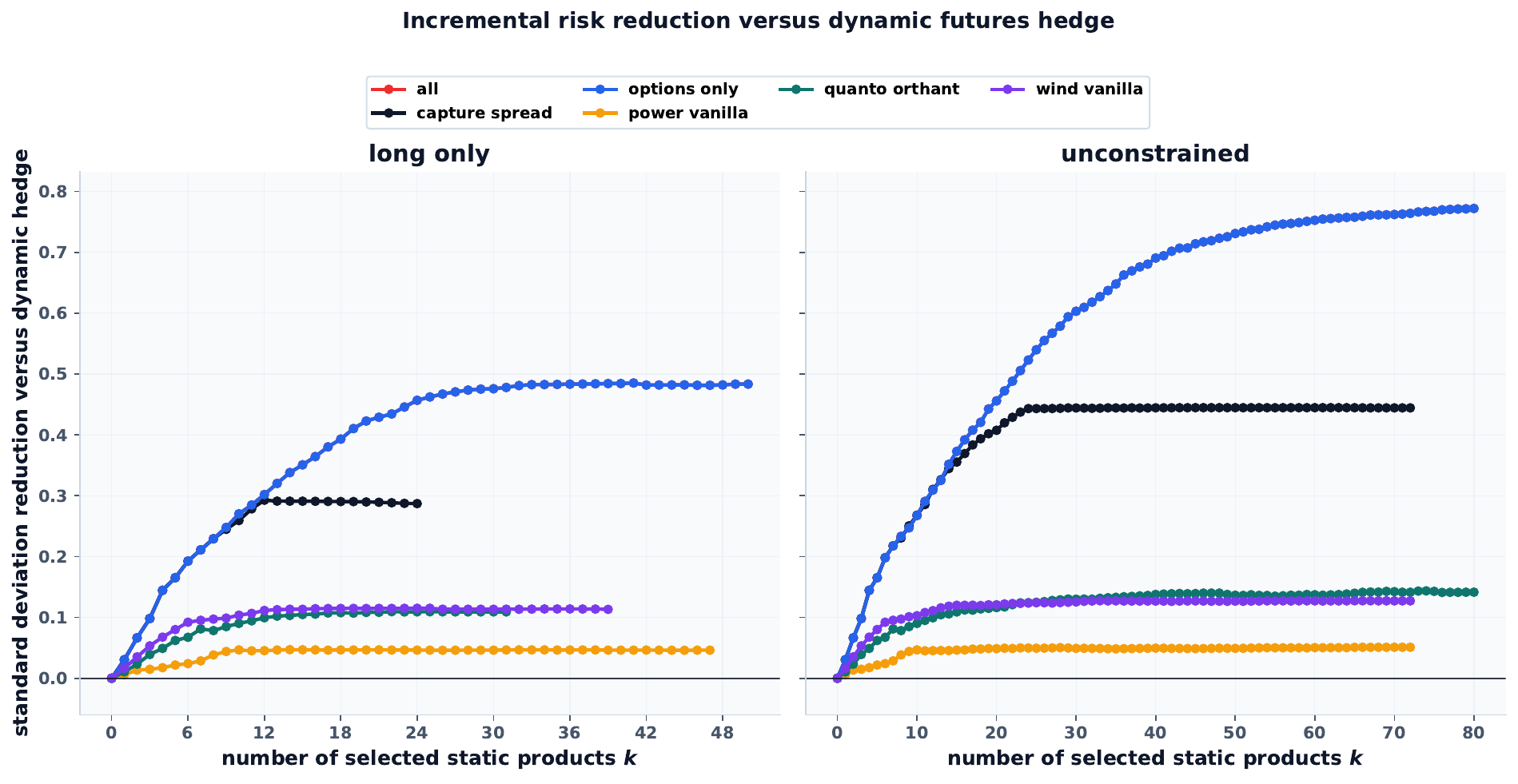}
\caption{Incremental reduction relative to the multi-month dynamic hedge.}
\end{subfigure}
\caption{Wind sparse semi-static frontiers for greedy and LASSO selection,
with and without long-only restrictions.}
\label{fig:wind-sparse-frontiers}
\label{fig:wind-sparse-incremental}
\end{figure*}

\Cref{fig:wind-sparse-frontiers,fig:wind-sparse-incremental} and
\cref{tab:wind-family-endpoints} report separate frontiers for each candidate family. Every
single-family portfolio leaves substantially more residual risk than the combined catalog. Capture-spread options are the most effective stand-alone
family in this experiment. Their underlying variable $D_j$
measures the difference between baseload and
achieved prices, an exposure not represented by fixed-volume
futures. Orthant quantos and wind vanillas each remove
$9$--$12\%$, and power vanillas remove only $3$--$4\%$: after the dynamic
futures hedge, little exposure that is monotone in the monthly baseload price
remains to be hedged. Two further observations qualify the single-family
frontiers. First, the unconstrained single-family quanto book attains its
modest reduction only with an extremely large gross notional built from
offsetting long--short positions, which indicates an unstable approximation
when this single localized basis is used to fit the entire residual; the
combined catalog reaches a far lower risk level at a small fraction of that
leverage. Second, the ordering of families by stand-alone reduction differs
from their marginal value inside the combined portfolio, where the families contribute
different payoff features: capture spreads and orthant quantos respond to joint
price-volume outcomes, wind puts protect against low production, and power vanillas
adjust the remaining price dependence.

\subsubsection{Convexity geometry of the static overlay}
\label{subsubsec:wind-convexity}

The following regression illustrates which nonlinear payoff components are offset by the static portfolio. Let
$F_\Sigma$ and $\bar C_\Sigma$ denote the discount-weighted delivery-horizon
averages of the hourly price and of the capacity factor, the horizon
aggregates of the monthly settlement summaries of
\cref{subsec:claim-universe}. For a centered terminal payoff $Y$ define its
\emph{convexity residual} as the part orthogonal to first-order exposure in
the pair $(F_\Sigma,\bar C_\Sigma)$,
\begin{equation}\label{eq:convexity-residual}
  \kappa^Y
  =Y-\widehat\alpha-\widehat\beta_C\,\bar C_\Sigma-\widehat\beta_S\,F_\Sigma,
\end{equation}
with $(\widehat\alpha,\widehat\beta_C,\widehat\beta_S)$ estimated by least
squares on the training paths. Any static position in linear forwards, and to
first order any futures delta, can only span the fitted plane; $\kappa^Y$ is the component of $Y$ not explained by linear
exposure to $F_\Sigma$ and $\bar C_\Sigma$, and therefore contains interaction and curvature effects. For the PPA, this component is dominated by the product $\bar C_\Sigma F_\Sigma$ inherent in
\eqref{eq:discrete-ppa}. On the test paths the unhedged wind PPA has a payoff
standard deviation of $359.5$ kEUR of which the convexity residual
has a standard deviation of $236.8$ kEUR, corresponding to $43.4\%$
of the unhedged variance. This comparison helps explain why
the dynamic futures hedge leaves a residual of similar magnitude: both remove
mainly linear price and volume exposure, while the dynamic hedge achieves
a modest further improvement by conditioning on interim
information.

For each candidate family, the convexity decomposition
\eqref{eq:convexity-residual} is applied to the unhedged PPA payoff, to the
terminal unconstrained static overlay
$\widehat{\bm\nu}^\top(\bm H-\widehat{\bm\pi})$ of
\cref{subsubsec:wind-sparse} evaluated without the dynamic futures leg, and to
their sum. We then regress the static-overlay convexity residual on the PPA convexity residual to measure
their offsetting relationship,
\begin{equation}\label{eq:convexity-cancellation}
  \widehat\beta_\kappa
  =\frac{\widehat{\Cov}(\kappa^{\rm stat},\kappa^{\rm PPA})}
        {\widehat{\Var}(\kappa^{\rm PPA})},
\end{equation}
whose ideal value is $-1$: a slope of $-1$, together with correlation close to $-1$, indicates
that the static residual offsets the PPA residual with approximately equal scale.

\begin{table*}[t]
\centering
\caption{Wind convexity decomposition of the terminal unconstrained static
overlays on test paths. Payoff and convexity standard-deviation reductions
compare the PPA plus static overlay (without dynamic futures leg) with the
unhedged PPA; $\rho_\kappa$ is the correlation between the static and PPA
convexity residuals and $\widehat\beta_\kappa$ the cancellation slope
\eqref{eq:convexity-cancellation}.}
\label{tab:wind-convexity}
\footnotesize
\begin{tabular}{lrrrr}
\toprule
Candidate family & Payoff std.\ red. & Convexity std.\ red. & $\rho_\kappa$ & $\widehat\beta_\kappa$ \\
 & (\%) & (\%) & & \\
\midrule
All families & 45.44 & 61.36 & $-0.93$ & $-0.98$ \\
Capture spread & 7.94 & 40.57 & $-0.81$ & $-0.60$ \\
Power vanillas & 15.30 & 4.40 & $-0.29$ & $-0.09$ \\
Wind vanillas & 2.33 & 6.03 & $-0.37$ & $-0.18$ \\
Quanto orthant & 1.30 & 1.80 & $-0.35$ & $-0.23$ \\
\bottomrule
\end{tabular}
\end{table*}

\begin{figure*}[p]
\centering
\begin{subfigure}[t]{.96\textwidth}
\centering
\includegraphics[width=\linewidth,height=.33\textheight,keepaspectratio]{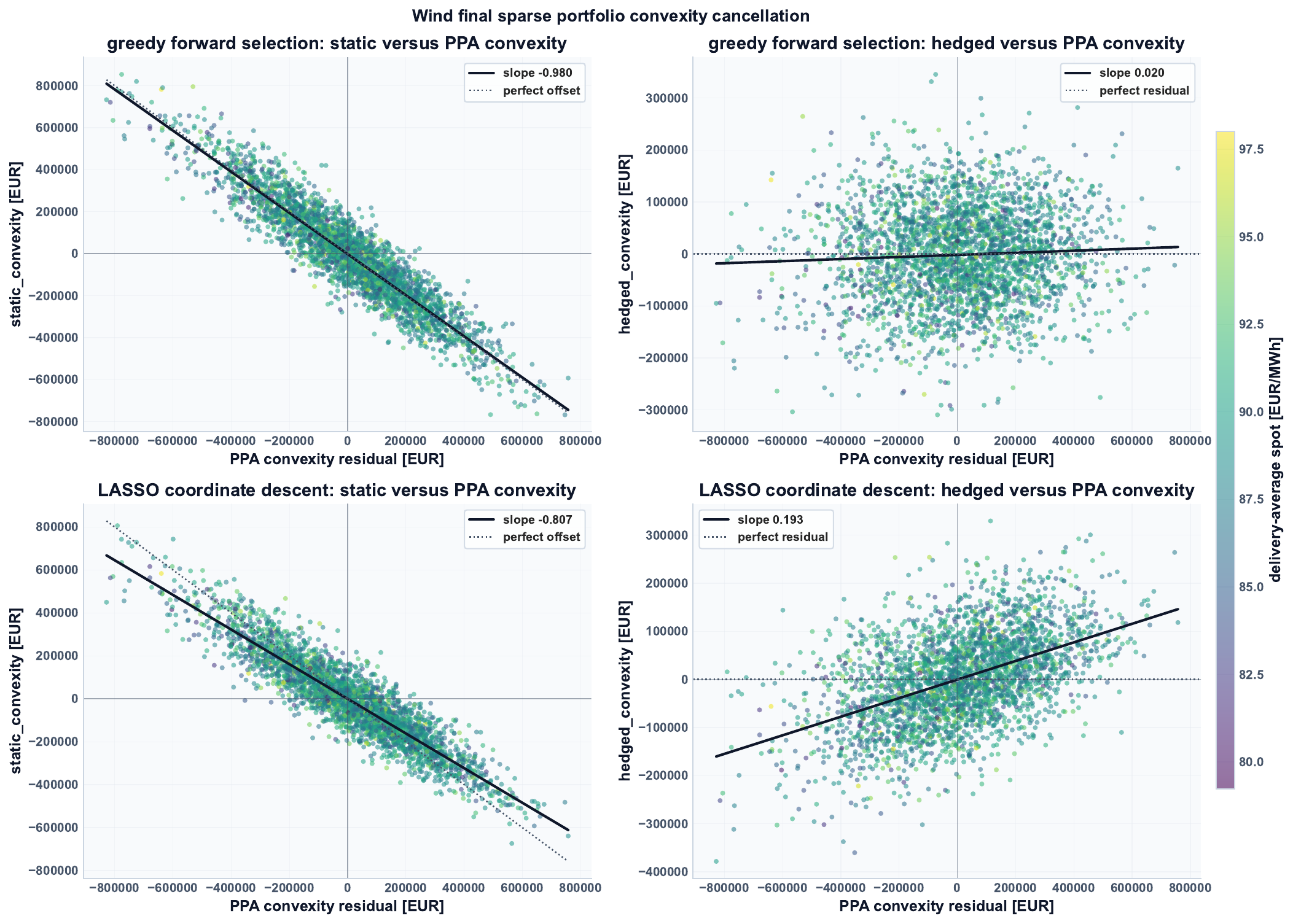}
\caption{Regression-based convexity cancellation.}
\end{subfigure}\\[0.35em]
\begin{subfigure}[t]{.96\textwidth}
\centering
\includegraphics[width=\linewidth,height=.33\textheight,keepaspectratio]{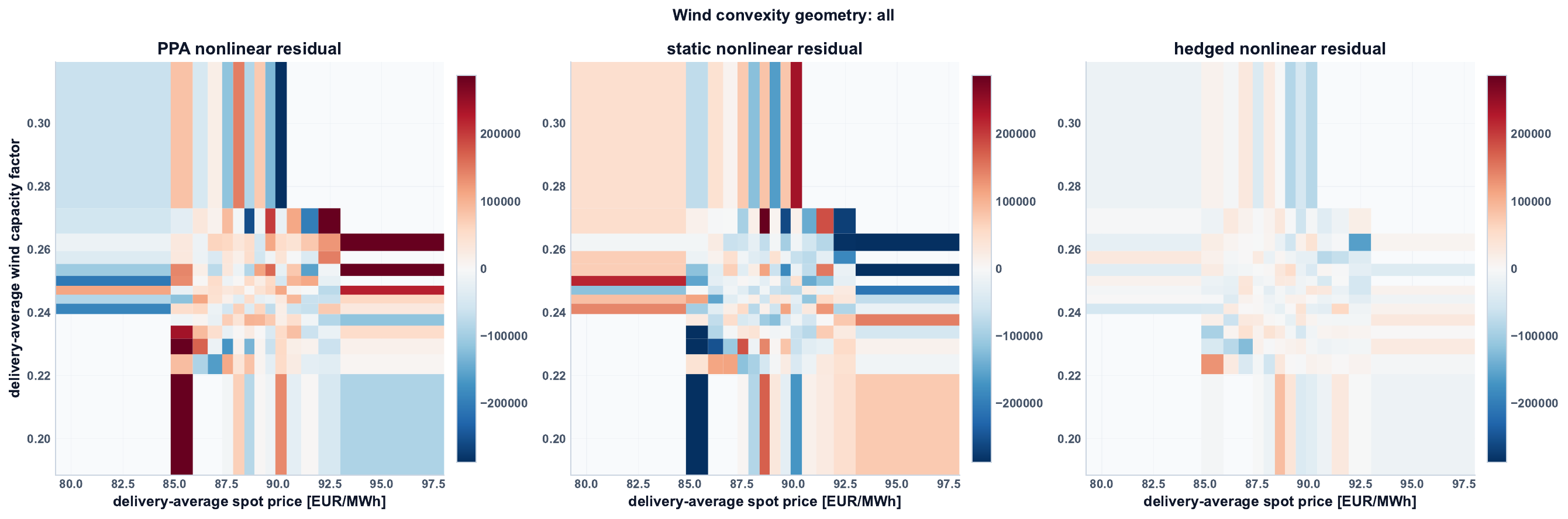}
\caption{Two-dimensional convexity heatmaps.}
\end{subfigure}
\caption{Wind convexity geometry of the sparse static overlay.}
\label{fig:wind-convexity-regression}
\label{fig:wind-convexity-triptych}
\end{figure*}

\Cref{tab:wind-convexity} and
\cref{fig:wind-convexity-regression,fig:wind-convexity-triptych} show that the combined
catalog substantially offsets the nonlinear residual: the selected all-family overlay
attains $\widehat\beta_\kappa=-0.98$ with correlation $-0.93$ and reduces the
convexity dispersion by $61.4\%$, exceeding its $45.4\%$ payoff reduction.
The family decomposition distinguishes linear-payoff reduction from nonlinear-residual reduction more clearly than the frontier alone. Power
vanillas achieve a payoff reduction of $15.3\%$ but almost no convexity
reduction ($4.4\%$, $\widehat\beta_\kappa=-0.09$): they replicate the linear
price exposure that futures also span, which is why they add so little on top
of the dynamic hedge in \cref{tab:wind-family-endpoints}. Capture-spread
options show the opposite pattern: modest total-payoff reduction
but the largest single-family reduction in the nonlinear
residual, because their payoff is convex in the capture
discount. The quanto-orthant book taken alone produces only a
$1.8\%$ convexity reduction because individual orthant payoffs are local.
Combining several orthants and thresholds allows the
portfolio to approximate the signed price--volume interaction described in
\cref{subsec:orthants}. At the monthly level the same decomposition applied to
the monthly summaries $(F_j,\bar C_j)$ shows that the convexity reduction of
the all-family book ranges from $69.2\%$ in August to $79.7\%$ in December,
while its monthly payoff reduction ranges from $11.4\%$ in October to
$86.0\%$ in December;
the selected notionals are larger in high-variance winter months, while the estimated
nonlinear-residual reduction remains relatively stable across months.

\subsubsection{Covariance replication by orthant quantos}
\label{subsubsec:wind-quanto-covariance}

\begin{figure*}[t]
\centering
\includegraphics[width=.92\textwidth]{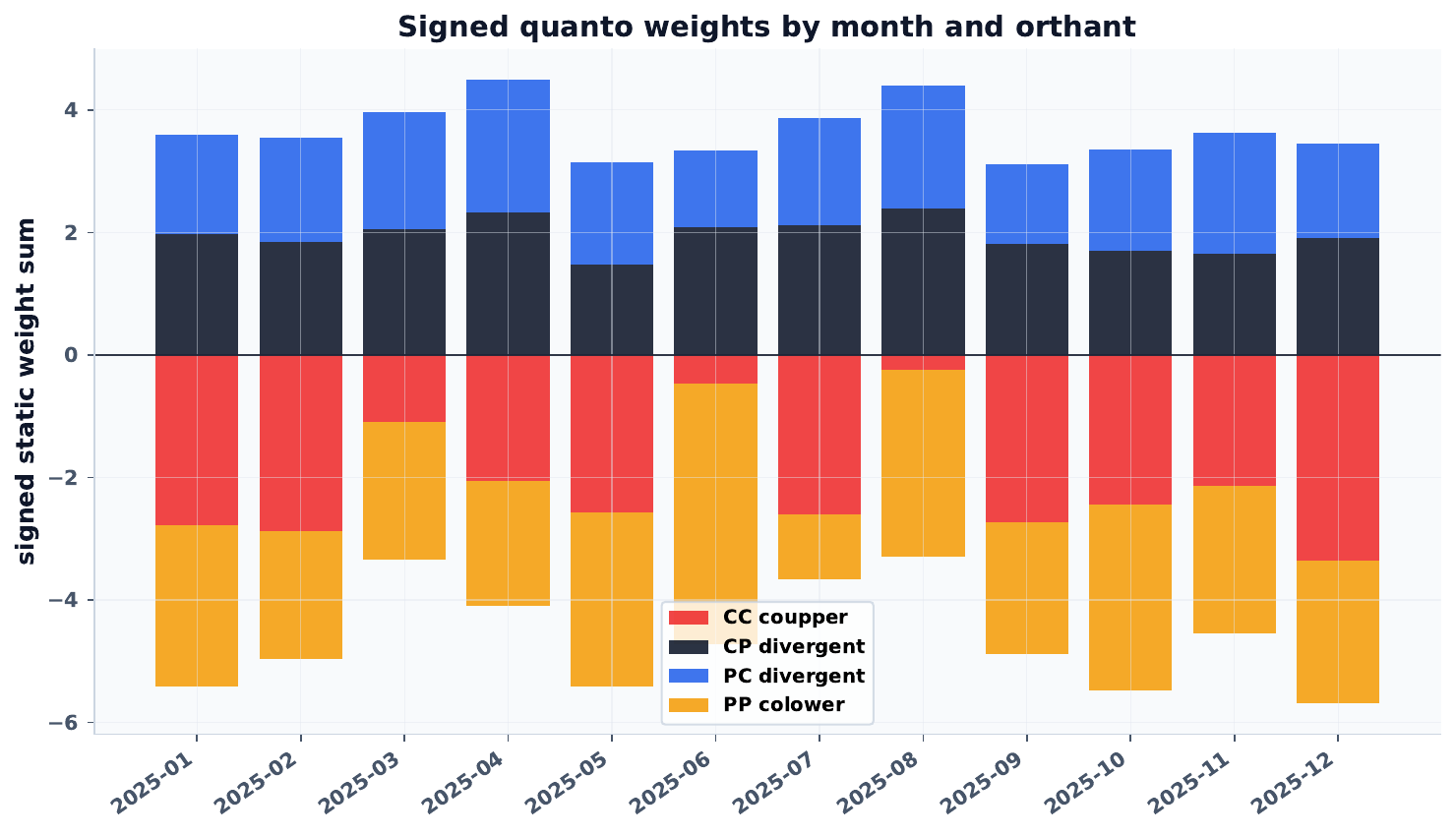}
\caption{Signed wind quanto weights by month and orthant. Positive
$\mathrm{CP}$/$\mathrm{PC}$ (divergent) weights and negative
$\mathrm{CC}$/$\mathrm{PP}$ (concordant) contributions reproduce the sign
structure of the price--wind covariance projection of
\cref{prop:orthant-sign}. The month-by-month pattern is consistent with using
the orthant claims as a sparse approximation to localized price-volume
covariance exposure.}
\label{fig:wind-quanto-weights}
\end{figure*}

The estimated weights match the predicted orthant signs in every
month. In every delivery month and at every threshold pair, the estimated
static weights are positive on the divergent orthants $\mathrm{CP}$ and
$\mathrm{PC}$ and negative on the concordant orthants $\mathrm{CC}$ and
$\mathrm{PP}$, the pattern $\operatorname{sign}(\bm\nu^\star)=\bm s$ of
\cref{prop:orthant-sign}, with an absolute-weighted sign-match
share of $100\%$ (\cref{fig:wind-quanto-weights}). The
coefficient from projecting the selected monthly
quanto payoff onto $(F_j-a_j)(\bar C_j-b_j)$
is negative in every month, ranging from $-0.78$ to
$-1.22$. Gross weights and model costs vary non-monotonically over
the annual production cycle rather than following a single summer--winter
gradient. The negative centered-product loading means that the orthant
portfolio offsets monthly price--wind co-movement, the
type of residual exposure represented by $\Gamma_0^r$ in the
fair-strike decomposition \eqref{eq:fair-strike-gamma}.

\subsubsection{Realized January--December 2025 performance}
\label{subsubsec:wind-realized}

The realized 2025 path illustrates how these claims respond to observed price and
production outcomes. With the fair strike $\widehat K^\star=80.01$ EUR/MWh, the realized
monthly PPA payoff is $H_j=R_j-\widehat K^\star V_j$ evaluated at the observed
monthly settlements (\cref{tab:wind-realized-state}). The period is not
uniformly favorable for the PPA: January and February have high baseload
prices and positive PPA values, while May and June have low capture prices and
large negative PPA values. Losses continue in autumn: September and
October contribute $-131.4$ and $-263.8$ kEUR, respectively, and total
realized PPA value over the year is $-460.6$ kEUR. Realized wind production is
below its model expectation in several high-price months (\textit{Dunkelflaute}).

\begin{table}[t]
\centering
\caption{Realized wind monthly state variables and PPA values,
January--December 2025. The capture discount is $D_j=F_j-A_j$.}
\label{tab:wind-realized-state}
\scriptsize
\begin{adjustbox}{max width=\columnwidth,max totalheight=.155\textheight,center}
\begin{tabular}{lrrrr}
\toprule
Month & $F_j$ & $\bar C_j$ & $D_j$ & PPA value \\
 & (EUR/MWh) &  & (EUR/MWh) & (kEUR) \\
\midrule
2025-01 & 114.14 & 0.3029 & 30.02 & 46.3 \\
2025-02 & 128.52 & 0.1625 & 13.51 & 191.0 \\
2025-03 & 94.61 & 0.1275 & 21.07 & -30.7 \\
2025-04 & 77.94 & 0.1183 & 3.57 & -24.0 \\
2025-05 & 67.34 & 0.1687 & 5.75 & -115.6 \\
2025-06 & 63.99 & 0.2061 & 12.45 & -211.2 \\
2025-07 & 87.80 & 0.1278 & 5.46 & 11.1 \\
2025-08 & 76.99 & 0.1143 & 8.01 & -46.9 \\
2025-09 & 83.51 & 0.2229 & 19.88 & -131.4 \\
2025-10 & 84.51 & 0.3148 & 27.03 & -263.8 \\
2025-11 & 101.88 & 0.2244 & 12.53 & 75.5 \\
2025-12 & 93.47 & 0.2534 & 9.32 & 39.1 \\
\bottomrule
\end{tabular}
\end{adjustbox}
\end{table}

\begin{table*}[t]
\centering
\caption{Realized net payoff of the wind static portfolio by product family. Amounts are
net of model prices, $\sum_i\nu_i(g_i^{\rm realized}-\pi_i)$, in kEUR.}
\label{tab:wind-realized-family-economics}
\footnotesize
\begin{tabular}{llrrrrr}
\toprule
Method & Constraint & Capture spread & Power & Quantos & Wind & Total static \\
\midrule
Greedy forward selection & Long-only & 620.3 & -52.6 & 46.7 & 186.3 & 800.7 \\
Greedy forward selection & Unconstrained & 582.1 & -350.4 & 185.6 & -3.9 & 413.4 \\
LASSO coordinate descent & Long-only & 593.2 & -26.5 & 53.1 & 158.9 & 778.8 \\
LASSO coordinate descent & Unconstrained & 544.6 & -27.2 & 63.6 & 98.9 & 680.0 \\
\bottomrule
\end{tabular}
\end{table*}

The family-level payoffs in
\cref{tab:wind-realized-family-economics} show how each claim type responded to the realized path.
Capture-spread calls insure against months in which the baseload price is high but the
realized production-weighted capture price is materially
lower. They generate large positive payoffs in January and
March, when $D_j=30.02$ and $21.07$ EUR/MWh. Wind puts provide protection in low-production
months, especially February, where the realized capacity factor is $0.1625$. Orthant
quantos add protection whose payoff depends jointly on the price and wind
thresholds: the CP divergent contracts pay when price is above its threshold and
wind is below its threshold, which matches the high-price,
low-wind outcomes observed early in 2025. Power vanillas have the opposite
realized contribution in the unconstrained portfolios because the
variance-optimal regression sells part of the high-price convexity. This reduces
variance in the simulated test set but produces losses when January--March spot prices
finish above the sold call strikes.

The long-only portfolios have larger
realized payoffs on this path because the
observed capture discounts and wind volumes
put capture-spread calls and wind puts in the
money. The unconstrained overlays contribute $413.4$ and $680.0$ kEUR.
The LASSO unconstrained book performs better on the realized path because it
has a smaller short position in power calls while keeping positive exposure to
capture-spread and $\mathrm{CP}$-quanto payoffs. This does not contradict the simulated
frontier: unconstrained greedy is the better test-set variance hedge, whereas
LASSO is the lower-leverage approximation of the same projection. A
single realized path cannot rank the estimators; it only illustrates
that the unconstrained variance-optimal portfolio may include short
option positions, whereas the long-only portfolio cannot.

\subsection{Solar PPA backtest}
\label{subsec:solar-backtest-results}

The solar backtest uses the protocol in \cref{subsec:backtest-protocol}. The solar backtest replaces the wind state and calibration with
the clear-sky solar specification in \cref{subsec:state-variables,subsec:calibration-results-solar}; the
renewable-index vanillas and orthant quantos are written on the solar
capacity factor.

The valuation results already show that wind and solar discounts arise from different
components. The model-implied fair strike is
$\widehat K^\star=63.89$ EUR/MWh, compared with an expected discounted baseload
price of $88.46$ EUR/MWh, a $24.57$ EUR/MWh discount to baseload, almost three times
the wind discount. The mean pathwise capture price is $63.92$ EUR/MWh and the
volume--capture covariance correction is only $0.001$ EUR/MWh, so the Monte
Carlo decomposition attributes essentially all of the discount to the profile
component. Consistent with the empirical decomposition of \cref{subsec:emp-capture},
the model-implied solar PPA strike is far below baseload mainly
because production is concentrated around midday, when prices are typically lower
and aggregate solar output is high; by contrast, the wind discount contains a
larger stochastic covariance component. At the monthly settlement scale the
cross-path correlation between the baseload average $F_j$ and the capacity
factor $\bar C_j$ is mildly positive, between $0.04$ and $0.08$ in every
delivery month, and yet the model capture price lies below baseload in every
month; the discount arises from the within-month alignment of hourly output and prices
rather than from correlation between monthly averages. Expected monthly delivered volume
rises from $0.9$ GWh in January to $6.5$ GWh in June and then falls to
$0.7$ GWh in December, broadly opposite to the wind seasonal pattern.

\Cref{fig:solar-hour-month-ppa} shows that the PPA payoff is nonzero only during daylight hours:
cells classified as night have zero production by construction under the envelope rule
\eqref{eq:solar-structural-night} and the physical inverse map \eqref{eq:solar-inverse},
while the nonzero window widens toward summer.
Positive winter and morning--evening
shoulder-hour contributions coexist with
a broad negative midday band from April
to August, consistent with high summer
production occurring in hours with low solar
capture prices.

\subsubsection{Hedging effectiveness of the benchmark strategies}
\label{subsubsec:solar-benchmarks}

\begin{table*}[t]
\centering
\caption{Solar PPA test-set risk levels for the main hedging benchmarks.
Reductions are relative to the unhedged fair-strike PPA.}
\label{tab:solar-main-risk}
\footnotesize
\begin{tabular}{lrrrr}
\toprule
Strategy & Std. & 95\% CVaR loss & Std. red. & CVaR red. \\
 & (kEUR) & (kEUR) & (\%) & (\%) \\
\midrule
Unhedged fair-strike PPA & 246.1 & 505.1 & 0.0 & 0.0 \\
Static baseload forward & 144.5 & 297.8 & 41.3 & 41.0 \\
Static baseload plus peak forward & 139.2 & 287.4 & 43.5 & 43.1 \\
Front-month dynamic futures hedge & 91.1 & 184.7 & 63.0 & 63.4 \\
Multi-month dynamic futures hedge & 86.7 & 176.8 & 64.8 & 65.0 \\
Multi-month dynamic plus static claims & 9.1 & 18.4 & 96.3 & 96.3 \\
\bottomrule
\end{tabular}
\end{table*}

\Cref{tab:solar-main-risk} reports the same benchmark comparison as
\cref{tab:wind-main-risk}. Linear forward instruments are considerably more
effective for solar than for wind. The static baseload forward already removes
$41.3\%$ of the PPA dispersion, and adding the peak-block forward raises this
to $43.5\%$. The peak block produces no material improvement for wind at the reported
precision. For solar, its delivery hours overlap much of the production window,
so the two-forward hedge better matches the intraday production pattern. The
variance-optimal futures stack improves this further to $64.8\%$, against
$36.4\%$ for wind. The static claims provide the largest
additional reduction: they lower the post-dynamic
residual standard deviation from $86.7$ kEUR to $9.1$
kEUR, bringing the total standard-deviation and CVaR
reductions to $96.3\%$.
\Cref{fig:solar-pnl-residuals} shows the corresponding out-of-sample residual
distributions.

The two robustness checks give similar conclusions for solar. The continuous MCARMA--MNIG
variant gives a $65.6\%$ dynamic and $96.4\%$ semi-static standard-deviation
reduction, against $64.8\%$ and $96.3\%$ for the state-dependent-jump
specification. Rerunning the projection with 378 daily rebalancing times and
377 delivery blocks produces a $93.8\%$ reduction for the semi-static hedge and a
$56.2\%$ reduction for the dynamic futures hedge, below the monthly-roll figure
under the more heavily regularized daily estimator. Under this daily estimator,
within-delivery rebalancing does not improve the futures hedge. The static claims
nevertheless reduce most of the residual risk left by both rebalancing schemes.

\begin{figure*}[!t]
\centering
\begin{subfigure}[t]{.49\textwidth}
\centering
\includegraphics[width=\linewidth,height=.31\textheight,keepaspectratio]{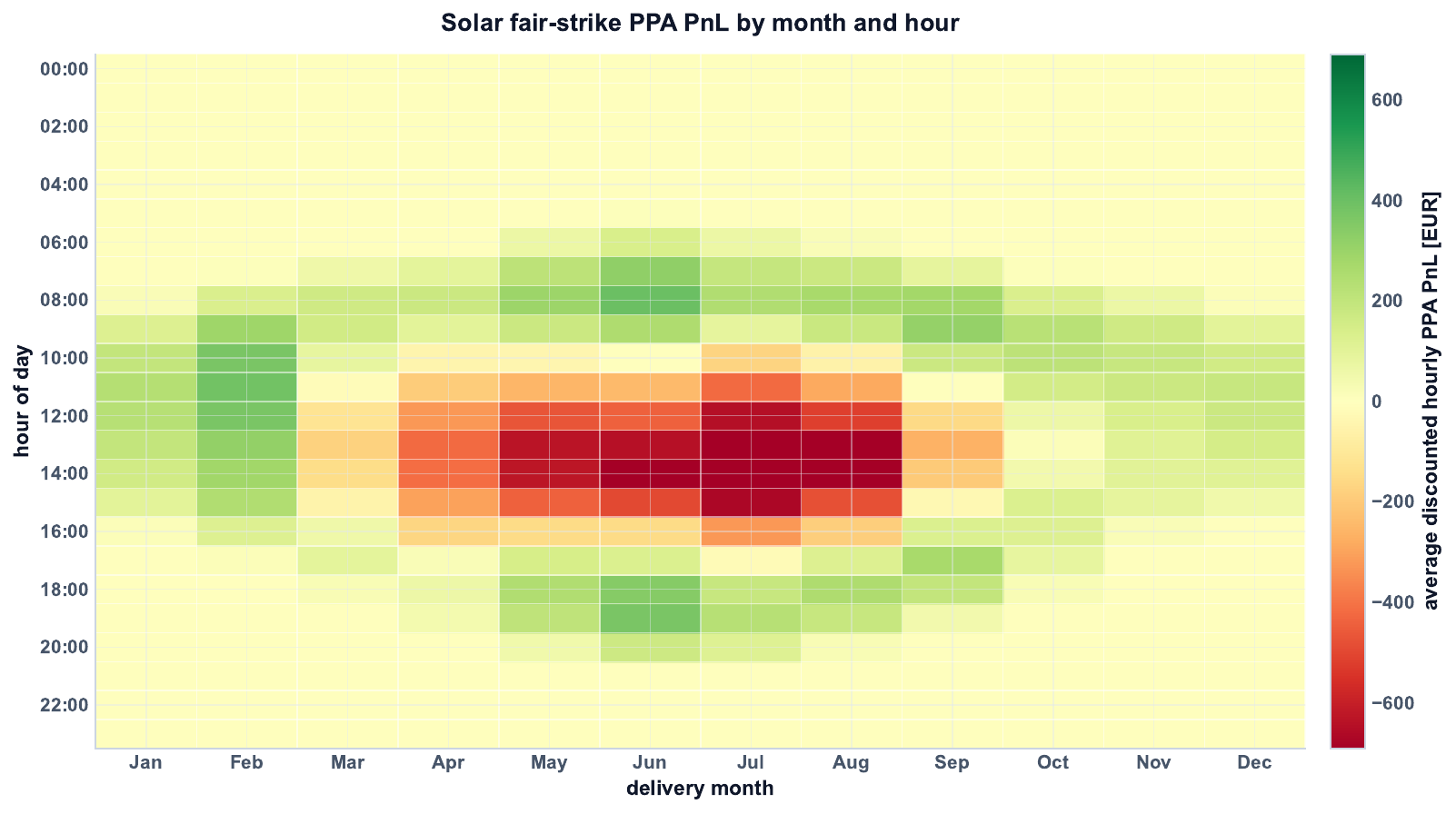}
\caption{Mean fair-strike PPA payoff by delivery month and hour.}
\end{subfigure}\hfill
\begin{subfigure}[t]{.49\textwidth}
\centering
\includegraphics[width=\linewidth,height=.31\textheight,keepaspectratio]{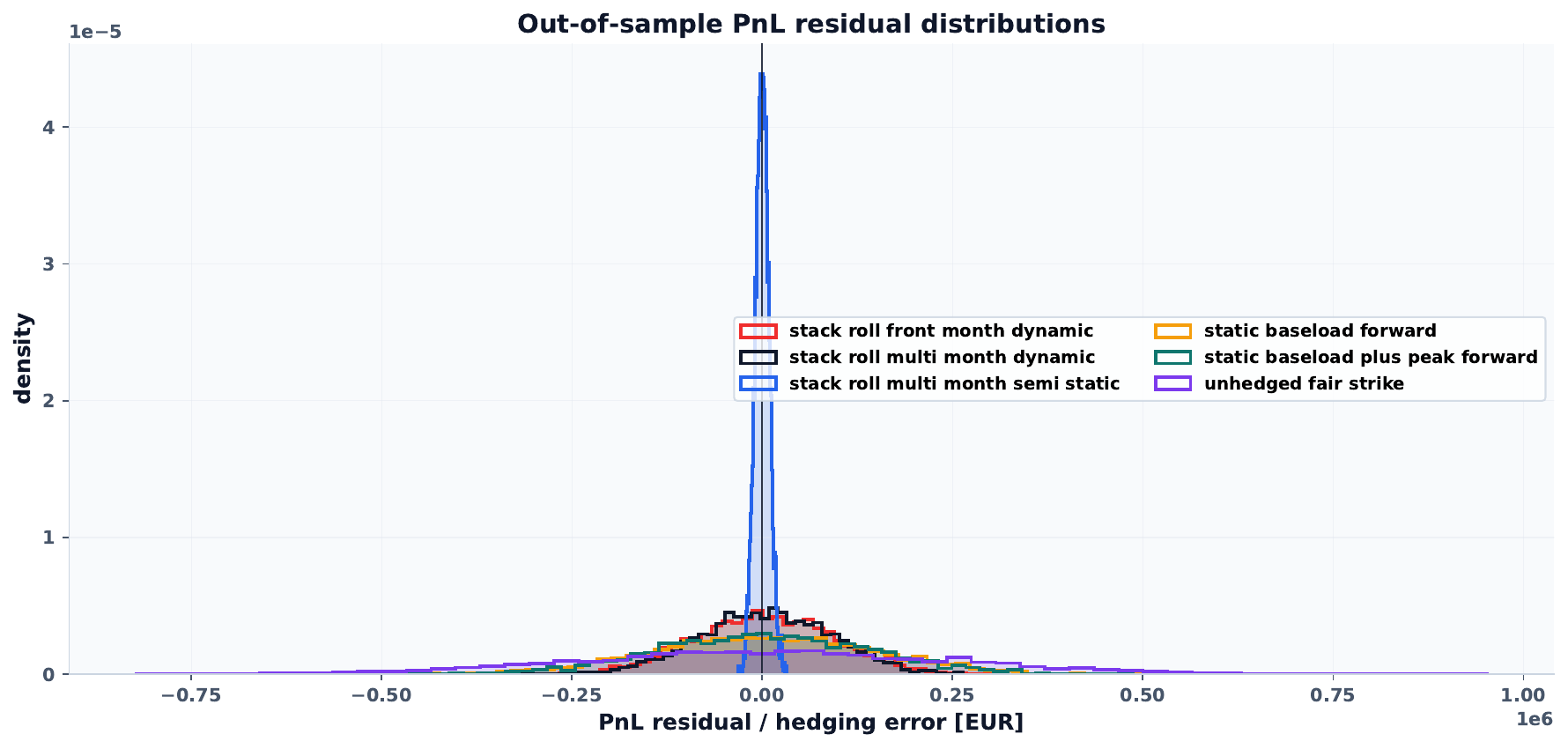}
\caption{Out-of-sample residual distributions across hedge designs.}
\end{subfigure}
\caption{Solar PPA contributions by delivery month and hour, together with hedging residuals.
The heatmap shows the negative midday contribution, while the distributional panel
shows the incremental effect of dynamic futures and the static completion.}
\label{fig:solar-hour-month-ppa}
\label{fig:solar-pnl-residuals}
\end{figure*}

\subsubsection{Naive volume deltas versus variance-optimal positions}
\label{subsubsec:solar-delta}

\begin{table*}[t]
\centering
\caption{Solar test-set hedging effectiveness of front-month delta rules
versus the variance-optimal projection. Reductions are relative to the
unhedged fair-strike PPA, and $\rho(H,G)$ is the test correlation between the
centered PPA payoff and the hedge gain.}
\label{tab:solar-delta-effectiveness}
\footnotesize
\begin{tabular}{lrrrrr}
\toprule
Hedge & Std. & 95\% CVaR loss & Std. red. & CVaR red. & $\rho(H,G)$ \\
 & (kEUR) & (kEUR) & (\%) & (\%) &  \\
\midrule
Unhedged fair-strike PPA & 246.1 & 505.1 & 0.0 & 0.0 & -- \\
Front-month naive model delta & 87.2 & 178.2 & 64.6 & 64.7 & 0.94 \\
Front-month naive realized delta & 90.8 & 186.3 & 63.1 & 63.1 & 0.93 \\
Front-month variance-optimal $\theta$ & 87.4 & 178.8 & 64.5 & 64.6 & 0.94 \\
Multi-month dynamic futures hedge & 86.7 & 176.8 & 64.8 & 65.0 & 0.94 \\
\bottomrule
\end{tabular}
\end{table*}

\begin{table*}[t]
\centering
\caption{Solar front-month variance-optimal futures positions versus naive
volume deltas. PPA values are realized January--December 2025 values.}
\label{tab:solar-delta-diagnostics}
\scriptsize
\renewcommand{\arraystretch}{0.88}
\begin{tabular}{lrrrrr}
\toprule
Month & PPA & Naive delta & VO delta & Adjustment & Ratio \\
 & (kEUR) & (MW) & (MW) & (MW) &  \\
\midrule
2025-01 & 57.8 & 1.26 & 1.29 & 0.03 & 1.03 \\
2025-02 & 92.2 & 2.28 & 2.35 & 0.07 & 1.03 \\
2025-03 & -38.0 & 3.98 & 3.82 & -0.15 & 0.96 \\
2025-04 & -160.3 & 5.81 & 5.81 & 0.00 & 1.00 \\
2025-05 & -260.0 & 7.84 & 7.73 & -0.12 & 0.99 \\
2025-06 & -297.4 & 9.02 & 9.14 & 0.13 & 1.01 \\
2025-07 & -11.7 & 8.49 & 8.49 & 0.00 & 1.00 \\
2025-08 & -135.0 & 7.37 & 6.97 & -0.40 & 0.95 \\
2025-09 & -63.5 & 5.74 & 5.80 & 0.06 & 1.01 \\
2025-10 & 12.6 & 3.16 & 3.05 & -0.10 & 0.97 \\
2025-11 & 37.8 & 1.53 & 1.60 & 0.07 & 1.04 \\
2025-12 & 29.9 & 0.91 & 1.23 & 0.32 & 1.36 \\
\bottomrule
\end{tabular}
\end{table*}

Unlike for wind, the solar variance-optimal and expected-volume positions
are very similar
(\cref{tab:solar-delta-effectiveness,tab:solar-delta-diagnostics}). For solar,
the covariance adjustment to the naive expected-volume delta
$\bar q^r\,\widehat{\E}[\bar C_j]$ is small in every month: the ratio of the
variance-optimal position $\theta_j^\star$ to the naive delta lies between
$0.95$ and $1.36$, against $0.50$ to $1.34$ for wind. The adjustment is small
through most high-volume months and largest in December, when expected solar
volume is lowest. The front-month rules have nearly identical test
performance: the naive model delta attains a $64.6\%$ test
standard-deviation reduction, the variance-optimal front-month position
$64.5\%$, and the full multi-month stack $64.8\%$. For
a solar PPA in this model-based backtest, the expected-volume forward sale performs
about as well as the variance-optimal solar hedge; for wind, it achieves
roughly four percentage points less standard-deviation reduction. For solar, the
joint model contributes mainly to the valuation and selection of
the static claims rather than to the estimated futures position.

\subsubsection{The efficient frontier in the number of static claims}
\label{subsubsec:solar-frontier}

\Cref{fig:solar-efficient-frontier-static-claims,fig:solar-claim-count-risk-reduction}
trace the semi-static frontier in the number of static claims added on top of
the multi-month dynamic hedge. As for wind, risk
falls quickly for the first selected claims and
more slowly after approximately twenty to thirty
claims. The selected solar portfolio achieves a larger incremental reduction than the selected wind portfolio: the validation-selected greedy
overlay removes $82.4\%$ of the post-dynamic residual standard deviation,
against $76.5\%$ for wind, so the solar residual left by the futures stack, although
smaller in absolute terms, is more fully hedged by the option catalog.

\begin{figure*}[!t]
\centering
\begin{subfigure}[t]{.78\textwidth}
\centering
\includegraphics[width=\linewidth,height=.31\textheight,keepaspectratio]{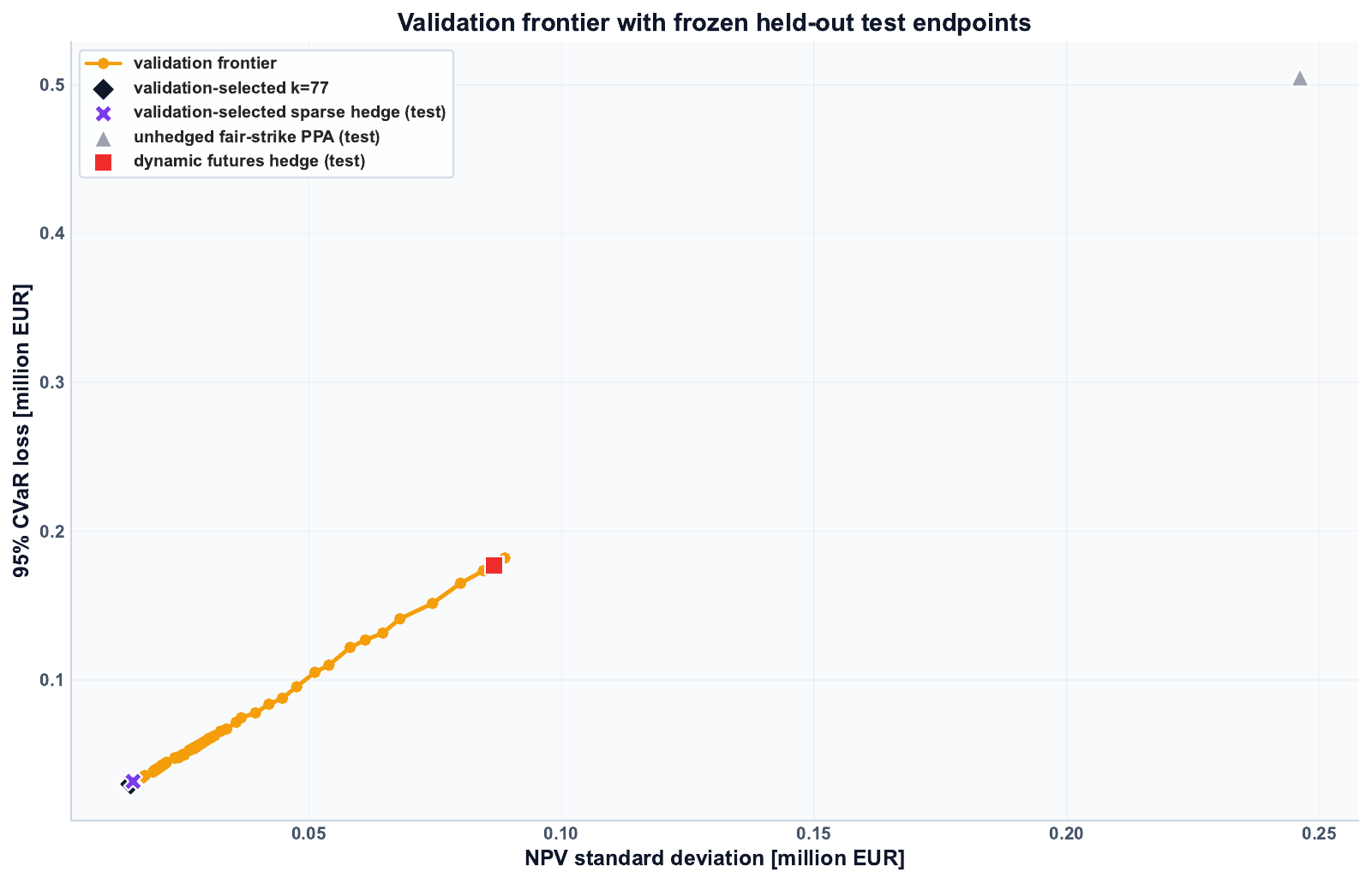}
\caption{Validation frontier in standard deviation and CVaR, with the selected portfolio evaluated once on the test set.}
\end{subfigure}
\caption{Solar efficient frontier in the number of static claims.}
\label{fig:solar-efficient-frontier-static-claims}
\label{fig:solar-claim-count-risk-reduction}
\end{figure*}

\subsubsection{Sparse static portfolios: selection, composition and leverage}
\label{subsubsec:solar-sparse}

\begin{table*}[t]
\centering
\caption{Solar sparse static portfolios selected from the all-product catalog.
Risk reductions are relative to the multi-month dynamic futures hedge and to
the unhedged fair-strike PPA, respectively.}
\label{tab:solar-sparse-risk}
\footnotesize
\begin{tabular}{llrrrrrr}
\toprule
Method & Constraint & $k$ & Std. & 95\% CVaR & Red. vs dyn. & Red. vs unhedged & Gross $|\nu|$ \\
 & & & (kEUR) & (kEUR) & (\%) & (\%) &  \\
\midrule
Greedy forward selection & Long-only & 16 & 51.9 & 95.0 & 40.18 & 78.93 & 27.10 \\
Greedy forward selection & Unconstrained & 77 & 15.2 & 31.8 & 82.44 & 93.81 & 74.38 \\
LASSO coordinate descent & Long-only & 33 & 51.8 & 93.6 & 40.30 & 78.97 & 31.02 \\
LASSO coordinate descent & Unconstrained & 79 & 22.8 & 46.5 & 73.69 & 90.73 & 50.82 \\
\bottomrule
\end{tabular}
\end{table*}

\begin{table*}[t]
\centering
\caption{Composition of the solar sparse static portfolios. Entries are counts
of nonzero static claims in each family.}
\label{tab:solar-sparse-composition}
\footnotesize
\begin{tabular}{llrrrrrr}
\toprule
Method & Constraint & Power & Solar & Quantos & Capture spread & Net $\nu$ & Gross $|\nu|$ \\
\midrule
Greedy forward selection & Long-only & 0 & 1 & 4 & 11 & 27.10 & 27.10 \\
Greedy forward selection & Unconstrained & 25 & 8 & 18 & 26 & -1.21 & 74.38 \\
LASSO coordinate descent & Long-only & 1 & 5 & 15 & 12 & 31.02 & 31.02 \\
LASSO coordinate descent & Unconstrained & 12 & 10 & 33 & 24 & 4.29 & 50.82 \\
\bottomrule
\end{tabular}
\end{table*}

\Cref{tab:solar-sparse-risk,tab:solar-sparse-composition} report the
corresponding solar comparison. The unconstrained greedy portfolio has the lowest
test-set variance among the sparse portfolios: with $k=77$ selected products it reaches
a $15.2$ kEUR test standard deviation and a $31.8$ kEUR 95\% CVaR loss, a
$93.8\%$ standard-deviation reduction relative to the unhedged PPA. The
unconstrained LASSO portfolio selects 79 claims with higher risk but a
$32\%$ lower gross absolute weight, $50.82$ against $74.38$, consistent with the
$\ell^1$ penalty reducing gross leverage in this application. The two long-only
methods produce nearly identical risk levels, at about $51.8$ kEUR
standard deviation and $94$ kEUR CVaR loss. The gap between the long-only
and unconstrained frontiers is wide, $40\%$ against $82\%$ of
the post-dynamic residual: the unconstrained solar portfolio uses short option
positions more extensively. In the power-vanilla family, short calls combined
with long puts create a payoff similar to a short forward, with additional
nonlinear exposure. Such a payoff cannot be formed under the long-only constraint; indeed the
greedy long-only selection includes no power vanillas and LASSO selects only one.

\begin{table*}[t]
\centering
\caption{Solar single-family sparse frontiers at their selected endpoints.
Reductions in test standard deviation are relative to the multi-month dynamic
futures hedge ($86.7$ kEUR standard deviation, $176.8$ kEUR 95\% CVaR loss).}
\label{tab:solar-family-endpoints}
\footnotesize
\begin{tabular}{lrrrrrrrr}
\toprule
 & \multicolumn{4}{c}{Long-only} & \multicolumn{4}{c}{Unconstrained} \\
\cmidrule(lr){2-5}\cmidrule(lr){6-9}
Candidate family & $k$ & Std. & 95\% CVaR & Red. & $k$ & Std. & 95\% CVaR & Red. \\
 & & (kEUR) & (kEUR) & (\%) & & (kEUR) & (kEUR) & (\%) \\
\midrule
All families & 16 & 51.9 & 95.0 & 40.2 & 77 & 15.2 & 31.8 & 82.4 \\
Capture spread & 10 & 53.5 & 98.7 & 38.2 & 25 & 28.7 & 60.3 & 66.9 \\
Quanto orthant & 2 & 86.5 & 176.5 & 0.2 & 8 & 85.7 & 173.6 & 1.1 \\
Solar vanillas & 0 & 86.7 & 176.8 & 0.0 & 2 & 86.1 & 175.5 & 0.6 \\
Power vanillas & 0 & 86.7 & 176.8 & 0.0 & 2 & 86.7 & 176.4 & 0.1 \\
\bottomrule
\end{tabular}
\end{table*}

\begin{figure*}[p]
\centering
\begin{subfigure}[t]{.49\textwidth}
\centering
\includegraphics[width=\linewidth,height=.29\textheight,keepaspectratio]{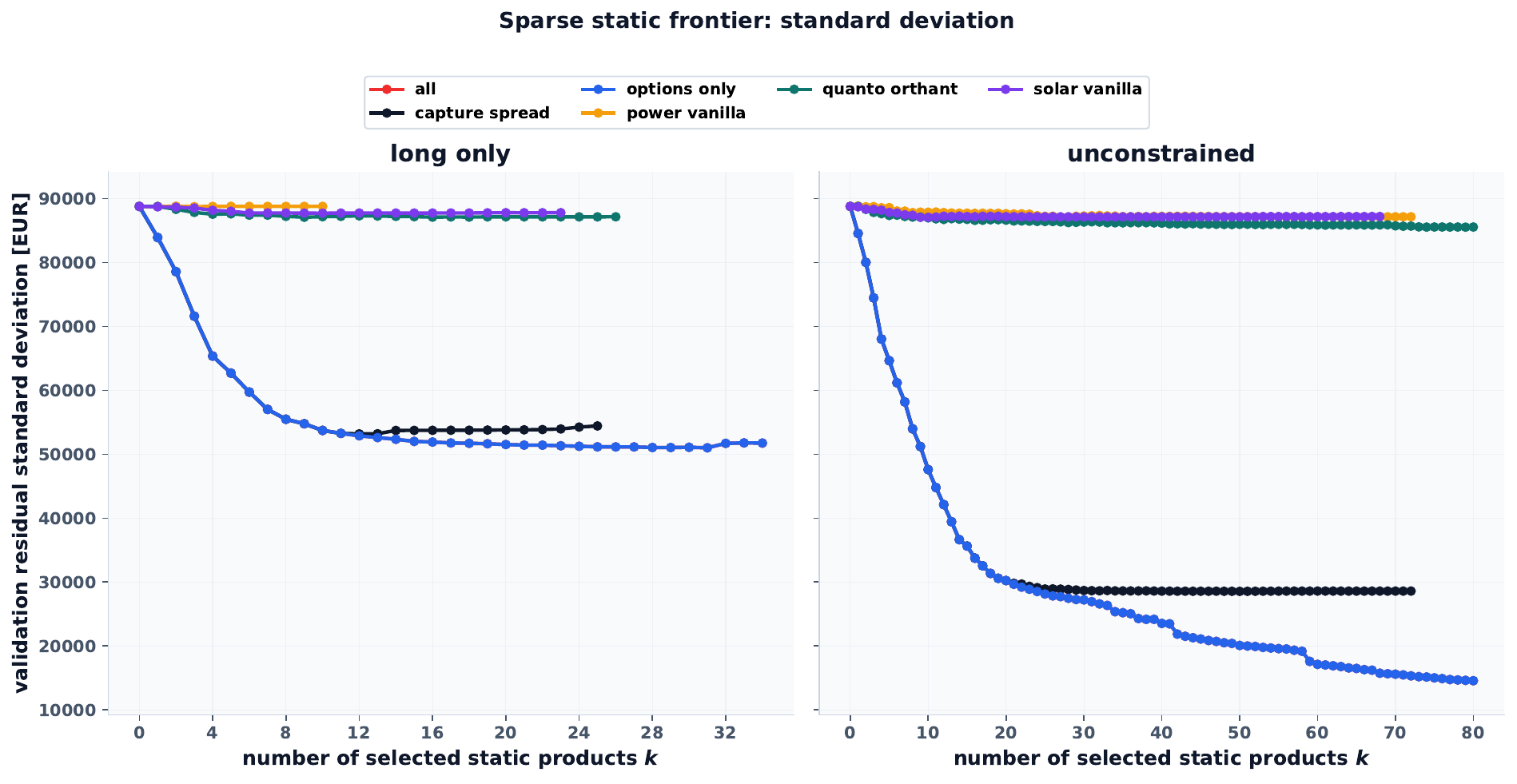}
\caption{Validation standard deviation.}
\end{subfigure}\hfill
\begin{subfigure}[t]{.49\textwidth}
\centering
\includegraphics[width=\linewidth,height=.29\textheight,keepaspectratio]{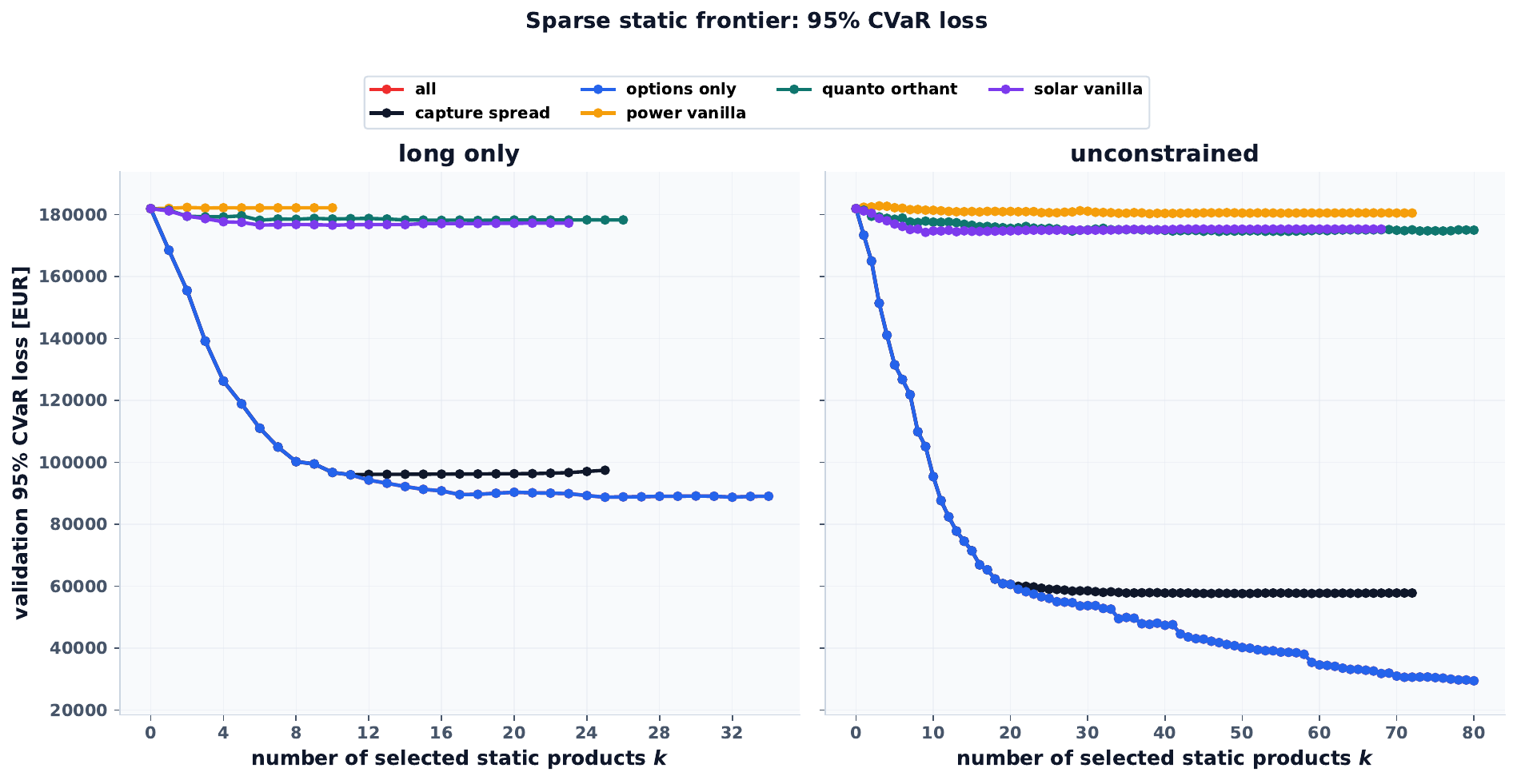}
\caption{Validation 95\% CVaR.}
\end{subfigure}\\[0.35em]
\begin{subfigure}[t]{.78\textwidth}
\centering
\includegraphics[width=\linewidth,height=.29\textheight,keepaspectratio]{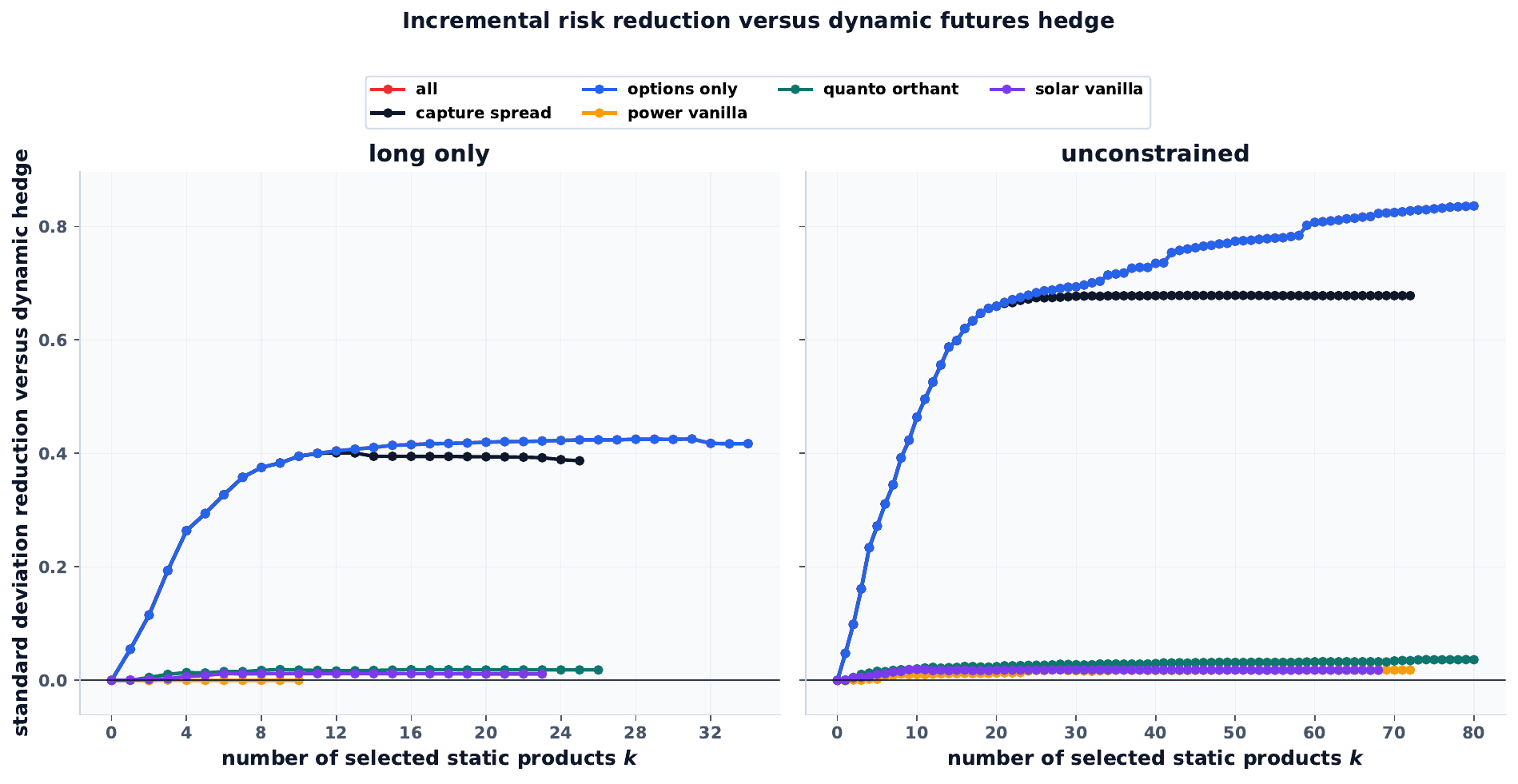}
\caption{Incremental reduction relative to the multi-month dynamic hedge.}
\end{subfigure}
\caption{Solar sparse semi-static frontiers for greedy and LASSO selection,
with and without long-only restrictions.}
\label{fig:solar-sparse-frontiers}
\label{fig:solar-sparse-incremental}
\end{figure*}

\Cref{fig:solar-sparse-frontiers,fig:solar-sparse-incremental} and
\cref{tab:solar-family-endpoints} disaggregate the frontier by candidate
family, and the ranking differs from wind. Capture-spread options
provide the largest stand-alone reduction for solar: alone they remove
$66.9\%$ of the post-dynamic residual dispersion (unconstrained), against
$45.6\%$ for wind, because the solar residual is concentrated in the monetary capture
discount $D_j$, which is the underlying variable of these options and equals the
gap between baseload and achieved price. Orthant quantos, which were the second-most effective family
for wind, contribute only about $1\%$ in isolation: with a small
stochastic covariance component, individual orthant payoffs explain little of the
post-dynamic residual variance when used alone, and the unconstrained quanto-only fit
uses extremely large offsetting long and short positions, indicating
that this localized payoff family is a poor stand-alone approximation
to the residual. Power vanillas alone add less than $0.1\%$ beyond the
futures stack, and the long-only selection procedure chooses no power
vanillas because every candidate long position increases training variance
after the futures hedge. As for wind, combining families
improves the reduction to $82.4\%$, compared with $66.9\%$ for the best
single-family portfolio.

\subsubsection{Convexity geometry of the static overlay}
\label{subsubsec:solar-convexity}

The convexity decomposition \eqref{eq:convexity-residual}, applied to the
delivery-horizon aggregates $(F_\Sigma,\bar C_\Sigma)$ of the solar backtest,
attributes a smaller share of the unhedged variance to curvature than for
wind: the convexity residual has a standard deviation of $144.2$ kEUR, compared with $246.1$
kEUR for the unhedged payoff, $34.3\%$ of the unhedged variance against $43.4\%$ for
wind. The dynamic futures stack attains an $86.7$ kEUR residual, noticeably
below the $144.2$ kEUR residual from the terminal linear regression. The two residuals were similar for wind. For solar, the dynamic
monthly hedge removes more risk because it provides a
separate linear hedge for each delivery month and is
re-estimated at every roll. Month-specific positions are
especially important because expected solar volume changes
by roughly a factor of seven from January to June.

\begin{table*}[t]
\centering
\caption{Solar convexity decomposition of the terminal unconstrained static
overlays on test paths. Payoff and convexity standard-deviation reductions
compare the PPA plus static overlay (without dynamic futures leg) with the
unhedged PPA; $\rho_\kappa$ is the correlation between the static and PPA
convexity residuals and $\widehat\beta_\kappa$ the cancellation slope
\eqref{eq:convexity-cancellation}.}
\label{tab:solar-convexity}
\footnotesize
\begin{tabular}{lrrrr}
\toprule
Candidate family & Payoff std.\ red. & Convexity std.\ red. & $\rho_\kappa$ & $\widehat\beta_\kappa$ \\
 & (\%) & (\%) & & \\
\midrule
All families & 67.70 & 76.73 & $-0.97$ & $-0.97$ \\
Capture spread & 5.89 & 18.46 & $-0.58$ & $-0.32$ \\
Power vanillas & 3.62 & 4.05 & $-0.34$ & $-0.05$ \\
Solar vanillas & $-0.00$ & $-0.06$ & $-0.02$ & $-0.00$ \\
Quanto orthant & 1.06 & 0.87 & $-0.13$ & $-0.02$ \\
\bottomrule
\end{tabular}
\end{table*}

\begin{figure*}[p]
\centering
\begin{subfigure}[t]{.96\textwidth}
\centering
\includegraphics[width=\linewidth,height=.33\textheight,keepaspectratio]{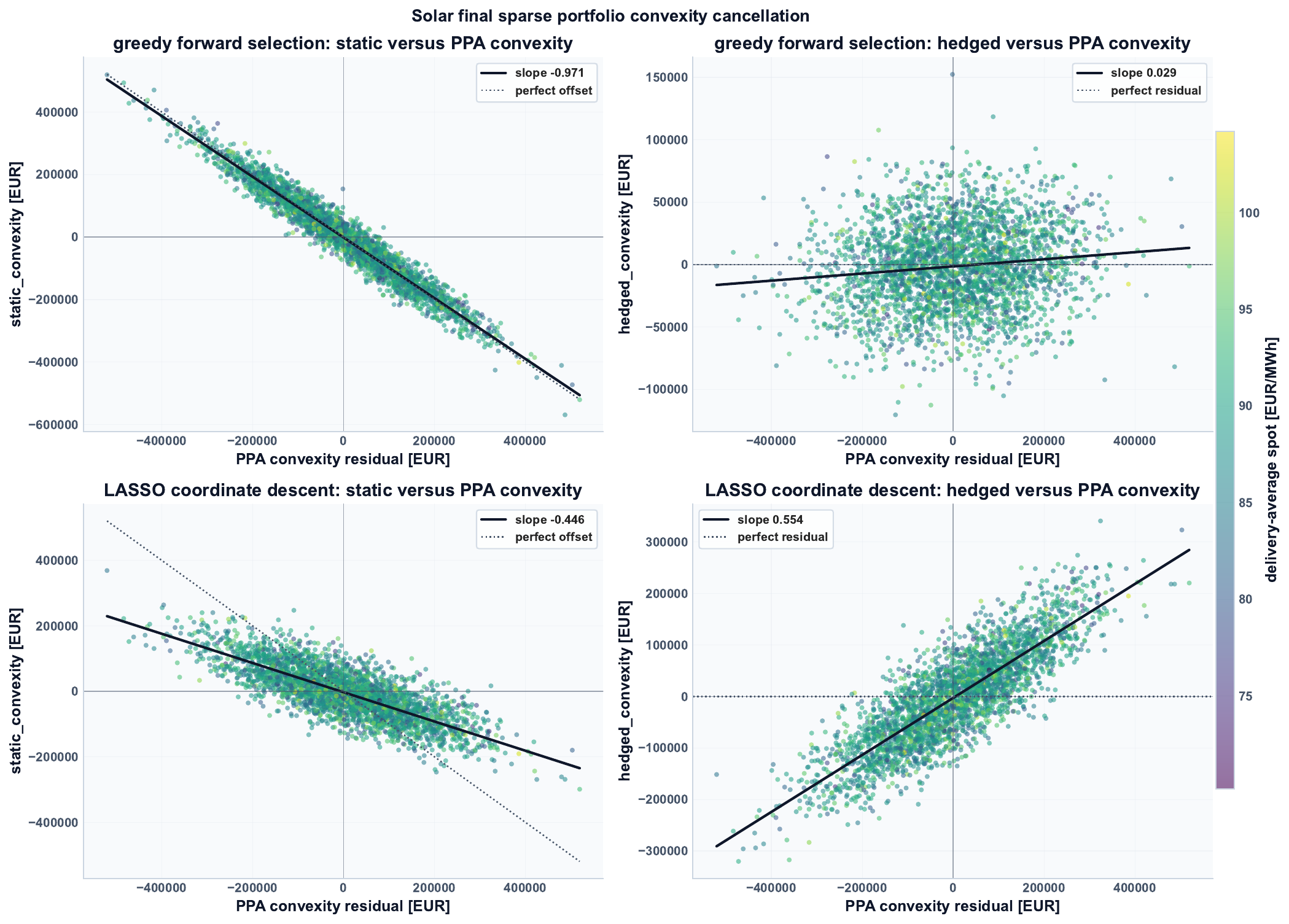}
\caption{Regression-based convexity cancellation.}
\end{subfigure}\\[0.35em]
\begin{subfigure}[t]{.96\textwidth}
\centering
\includegraphics[width=\linewidth,height=.33\textheight,keepaspectratio]{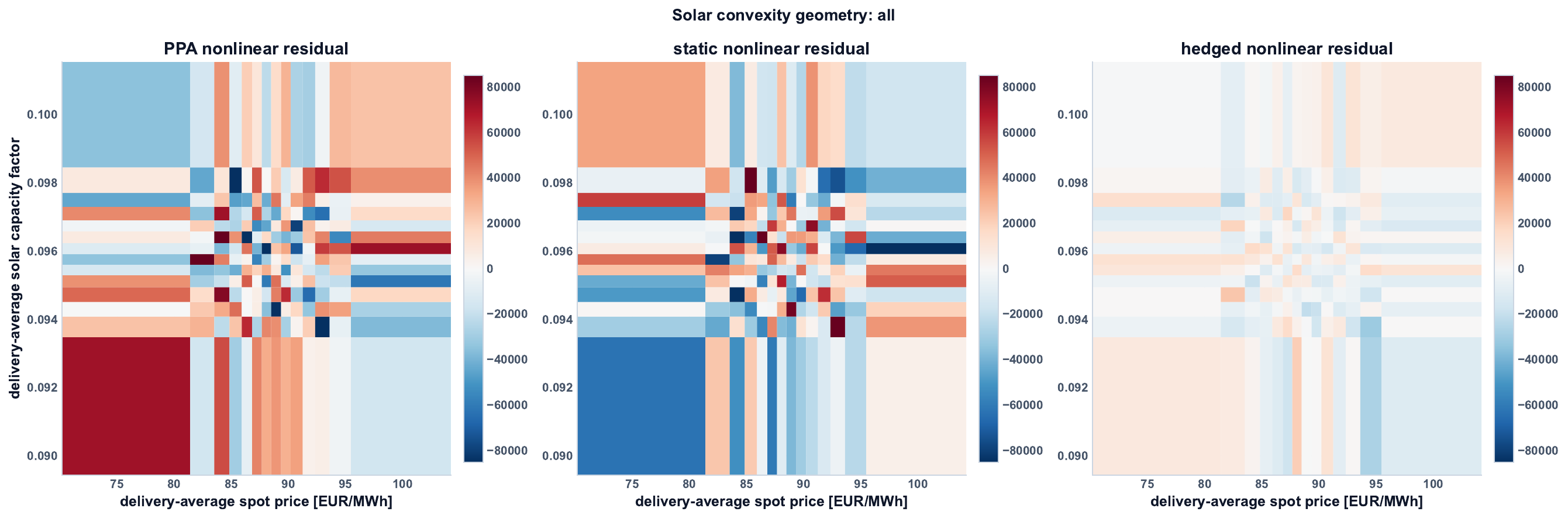}
\caption{Two-dimensional convexity heatmaps.}
\end{subfigure}
\caption{Solar convexity geometry of the sparse static overlay.}
\label{fig:solar-convexity-regression}
\label{fig:solar-convexity-triptych}
\end{figure*}

\Cref{tab:solar-convexity} and
\cref{fig:solar-convexity-regression,fig:solar-convexity-triptych} show that
the combined catalog strongly offsets the nonlinear residual: the selected
all-family overlay attains $\widehat\beta_\kappa=-0.97$ at correlation
$-0.97$, reducing payoff and convexity dispersion by $67.7\%$ and $76.7\%$.
The family-level
results are qualitatively similar to wind, with differences in magnitude.
Capture-spread options produce the largest single-family reduction in the convexity residual
($18.5\%$, $\widehat\beta_\kappa=-0.32$) at a small payoff reduction; power
vanillas remove only $3.6\%$ of payoff and $4.1\%$ of convexity dispersion;
the quanto book alone removes less than $1\%$ of convexity dispersion; and solar
vanillas have essentially zero stand-alone effect at the reported precision, although they
can still contribute conditionally when combined with the joint-state claims.
At the monthly level the all-family payoff reduction ranges from $11.7\%$ in
November to $85.8\%$ in April, while convexity reduction ranges from $66.0\%$
in December to $86.8\%$ in September. Convexity cancellation is therefore
more stable than total-payoff cancellation, but neither is uniform across the
delivery year; the month-specific strike grids reflect seasonality, but the remaining linear
payoff component still varies across months.

\subsubsection{Covariance replication by orthant quantos}
\label{subsubsec:solar-quanto-covariance}
\begin{figure*}[t]
\centering
\includegraphics[width=.92\textwidth]{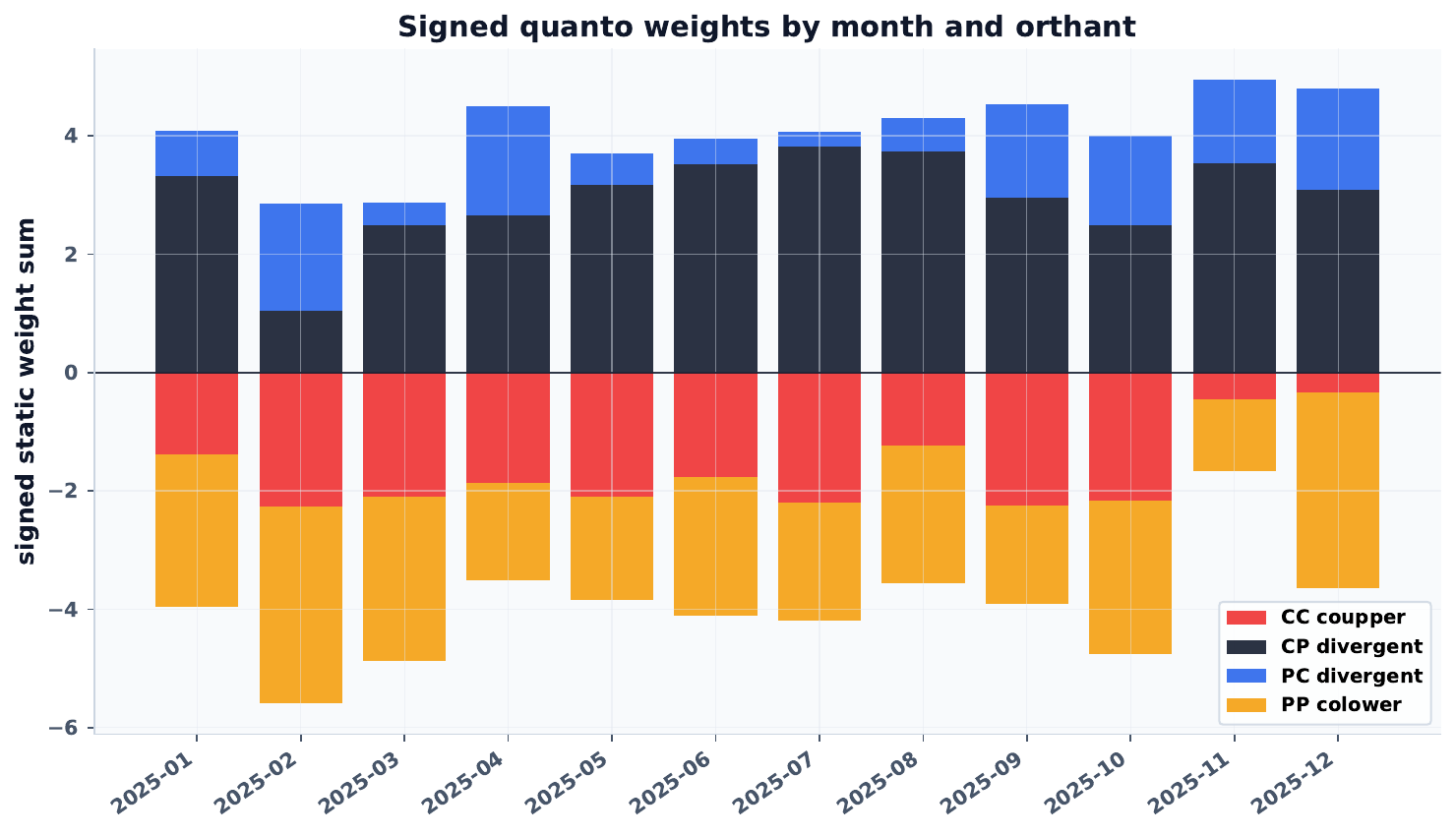}
\caption{Signed solar quanto weights by month and orthant. The aggregate
centered-product loading is negative in every month; exact orthant sign
matching holds in eleven months, with one March exception associated with
a near-zero principal-component coefficient.}
\label{fig:solar-quanto-weights}
\end{figure*}

The aggregate centered-product loading is negative in every solar
delivery month. Exact orthant sign matching holds in eleven of twelve months.
In March, $59.2\%$ of the absolute weight has the predicted
sign. The mismatch is driven by one near-zero coefficient whose
sign is unstable. Projecting the monthly quanto book onto the centered-product direction of
\eqref{eq:orthant-product} gives a negative loading in every month, ranging
from $-0.89$ to $-1.12$ without a monotone summer--winter pattern. This example
illustrates the caveat in Remark \ref{rem:orthant-sign-caveat}. The raw cross-path correlation
between the monthly settlement pair $(F_j,\bar C_j)$ is
\emph{positive} for solar, yet the residualized
weights have the negative sign pattern. The
relevant cross-moments are computed after futures projection
and therefore reflect within-month price-production dependence not
captured by monthly averages.

\subsubsection{Realized January--December 2025 performance}
\label{subsubsec:solar-realized}

\begin{table}[t]
\centering
\caption{Realized solar monthly state variables and PPA values,
January--December 2025. The capture discount is $D_j=F_j-A_j$.}
\label{tab:solar-realized-state}
\scriptsize
\begin{adjustbox}{max width=\columnwidth,max totalheight=.155\textheight,center}
\begin{tabular}{lrrrr}
\toprule
Month & $F_j$ & $\bar C_j$ & $D_j$ & PPA value \\
 & (EUR/MWh) &  & (EUR/MWh) & (kEUR) \\
\midrule
2025-01 & 114.14 & 0.0282 & -4.82 & 57.8 \\
2025-02 & 128.52 & 0.0535 & 13.36 & 92.2 \\
2025-03 & 94.61 & 0.1002 & 40.93 & -38.0 \\
2025-04 & 77.94 & 0.1446 & 44.83 & -160.3 \\
2025-05 & 67.34 & 0.1693 & 44.74 & -260.0 \\
2025-06 & 63.99 & 0.1940 & 42.67 & -297.4 \\
2025-07 & 87.80 & 0.1631 & 25.83 & -11.7 \\
2025-08 & 76.99 & 0.1559 & 36.38 & -135.0 \\
2025-09 & 83.51 & 0.0984 & 37.55 & -63.5 \\
2025-10 & 84.51 & 0.0546 & 14.40 & 12.6 \\
2025-11 & 101.88 & 0.0336 & 6.72 & 37.8 \\
2025-12 & 93.47 & 0.0244 & -3.39 & 29.9 \\
\bottomrule
\end{tabular}
\end{adjustbox}
\end{table}

The realized 2025 PPA payoff is strongly negative. Most of the loss is associated with realized capture prices falling far below baseload
(\cref{tab:solar-realized-state}). With
$H_j=R_j-\widehat K^\star V_j$ evaluated at the observed monthly settlements,
the total realized PPA value is $-735.7$ kEUR. The loss is principally a
capture-price event: realized capture prices fall substantially relative to baseload
from March through September, with capture
discounts $D_j$ of $40.9$ to $44.8$ EUR/MWh in March--June against
model-implied expectations of approximately $16$ to $23$ EUR/MWh for the same months. The
realized May capture price of $22.60$ EUR/MWh against a baseload of $67.34$
EUR/MWh, and June's $21.31$ against $63.99$, are severe
capture-discount outcomes of the type targeted by the auxiliary claims. January
shows the opposite winter configuration, a capture \emph{premium}
($D_j=-4.82$ EUR/MWh), because low winter solar output is concentrated in
the midday peak; December repeats this configuration with a
$3.39$ EUR/MWh capture premium.

Comparing the realized outcomes with the selected portfolios shows that several long option
positions finish in the money in the months for which they were selected. Capture-spread
calls, struck at training quantiles of $D_j$, finish deep in the money in
March--September; the $\mathrm{PC}$ orthant payoffs
are positive in high-volume, low-price late-spring
months, which also account for the large May and
June PPA losses; and long power puts pay in May--June
as baseload settles below the winter-anchored strikes. Unlike
in the wind backtest, the large solar losses occur in outcomes
covered by long capture-spread, orthant, and power-put positions.
Both long-only and unconstrained portfolios therefore earn positive realized
payoffs: the 2025 outcomes include large capture discounts
and high-production/low-price combinations to which both portfolios
have positive exposure. Accordingly, the greedy long-only and
unconstrained overlays contribute $809.2$ and $726.6$ kEUR, while their LASSO
counterparts contribute $748.2$ and $639.8$ kEUR, respectively.

\subsubsection{Wind versus solar: covariance risk against profile risk}
\label{subsubsec:wind-solar-comparison}

The wind and solar results show how the profile and covariance
components are reflected in hedge performance and portfolio
composition. The wind discount
to baseload is moderate ($8.9$ EUR/MWh) but substantially stochastic, driven
by the negative conditional price--wind covariance; the solar discount is
almost three times larger ($24.6$ EUR/MWh) but almost entirely a deterministic
profile effect. This difference is reflected in the futures positions, residual risks, and selected static claims.
The futures hedge reduces solar PPA standard
deviation by $64.8\%$, compared with $36.4\%$
for wind, and the expected-volume delta performs similarly to the variance-optimal
solar position but achieves about four percentage points less standard-deviation reduction for wind. Static
claims provide substantial additional risk reduction
for both technologies. For wind, orthant quantos and
capture-spread options contribute jointly because the
remaining payoff depends on both price and production.
For solar, capture-spread options dominate the
single-family results, while orthant quantos add
little when used alone. In both cases the centered-product
loading is negative throughout the delivery year, but its monthly
magnitude is non-monotone rather than a simple mirror of production
seasonality. These negative loadings are consistent with the price-production
covariance term $\Gamma_t^r$ in the fair-strike decomposition.

\endgroup

\newpage
\section{Conclusion}\label{sec:conclusion}

Pay-as-produced PPAs are claims on the interaction between power prices and
renewable output. Under this pricing convention, the fair strike is the
expected production-weighted spot price, and its difference from baseload
separates exactly into a deterministic profile correction and a stochastic
covariance correction. German data show these terms differ by technology. The
wind discount is driven mainly by adverse price--volume covariance, whereas
the larger solar discount arises primarily from the intraday generation
profile, with a smaller residual covariance contribution.

This distinction determines hedge design. Fixed-volume power futures hedge
much of the linear price and profile exposure, but cannot fully replicate the
nonlinear interaction between output and price. The GKW decomposition leads
to a two-stage strategy: dynamic positions in active futures first hedge the
traded component, after which static production-linked claims target the
remaining orthogonal residual. Capture-spread options respond to the gap
between baseload and realized capture prices, while orthant quantos isolate
adverse joint price--production states.

On held-out paths drawn independently from the fitted MCARMA model, dynamic
futures reduce the standard deviation of wind and solar PPA payoffs by
\(36.4\%\) and \(64.8\%\), respectively. Adding synthetic, model-priced static
claims raises the total reductions to \(87.3\%\) and \(96.3\%\), with similar
CVaR gains. These are within-model results, not an independent market
validation. They measure only model-implied hedgeability, not actual hedge
performance in traded markets. They provide an idealized frictionless
benchmark and likely overstate feasible risk reduction once executable
premiums, bid--ask spreads, margin requirements and transaction costs are
included. Sparse selection retains most of the reduction, and LASSO refitting
lowers gross absolute portfolio weight, but neither step includes these
trading frictions. Within the model, futures mainly hedge linear price risk,
while nonlinear auxiliary claims reduce residual covariance risk from joint
price and production outcomes.

Several extensions follow. Market applications should estimate risk premia,
specify the valuation measure, calibrate spot expectations to an hourly price
forward curve, and model load jointly. The framework could cover multi-site
portfolios of wind, solar and hybrid PPAs across maturities. A spatio-temporal
portfolio hedge would model weather, output, zonal prices and basis risk
jointly, while exploiting diversification across technologies and regions. Extending the framework to co-located renewable--BESS PPAs would also be an interesting next avenue, allowing battery dispatch and hedging decisions to be jointly optimized.

\section*{Acknowledgements}
The authors gratefully acknowledge funding from the University of Amsterdam’s interfaculty Research Priority Area \emph{Energy Transition through the Lens of Sustainable Development Goals} (ENLENS).
\FloatBarrier
\bibliographystyle{elsarticle-num-names}
\bibliography{references/refs,references/refs_extra}

@article{Hirth2013,
  author = {Hirth, Lion},
  title = {The market value of variable renewables},
  journal = {Energy Economics},
  volume = {38}, pages = {218--236}, year = {2013}}

@article{COULON2013976,
  author = {Coulon, Michael and Powell, Warren B. and Sircar, Ronnie},
  title = {A model for hedging load and price risk in the {Texas} electricity market},
  journal = {Energy Economics}, volume = {40}, pages = {976--988}, year = {2013}}

@article{RONCORONI2017415,
  author = {Roncoroni, Andrea and Id Brik, Rachid},
  title = {Hedging size risk: theory and application to the {US} gas market},
  journal = {Energy Economics}, volume = {64}, pages = {415--437}, year = {2017}}

@article{Pircalabu2016,
  author = {Pircalabu, Anca and Hvolby, Thomas and Jung, Jesper and H{\o}g, Esben},
  title = {Joint price and volumetric risk in wind power trading: a copula approach},
  journal = {Energy Economics}, volume = {62}, pages = {139--154}, year = {2017}}

@article{Pircalabu2017,
  author = {Pircalabu, Anca and Jung, Jesper},
  title = {A mixed {C}-vine copula model for hedging price and volumetric risk in wind power trading},
  journal = {Quantitative Finance}, volume = {17}, pages = {1583--1600}, year = {2017}}

@article{RowinskaVeraartGruet2021,
  author = {Rowi{\'n}ska, Paulina A. and Veraart, Almut E. D. and Gruet, Pierre},
  title = {A multi-factor approach to modelling the impact of wind energy on electricity spot prices},
  journal = {Energy Economics}, volume = {104}, pages = {105640}, year = {2021}}

@incollection{DeschatreVeraart2018,
  author = {Deschatre, Thomas and Veraart, Almut E. D.},
  title = {A joint model for electricity spot prices and wind penetration with dependence in the extremes},
  booktitle = {Renewable Energy: Forecasting and Risk Management},
  series = {Springer Proceedings in Mathematics and Statistics}, volume = {254},
  pages = {185--214}, publisher = {Springer}, year = {2018}}

@article{PENA2024101513,
  author = {Pe{\~n}a, Juan Ignacio and Rodr{\'i}guez, Rosa and Mayoral, Silvia},
  title = {Hedging renewable power purchase agreements},
  journal = {Energy Strategy Reviews}, volume = {55}, pages = {101513}, year = {2024}}

@article{Oum2006,
  author = {Oum, Yumi and Oren, Shmuel and Deng, Shijie},
  title = {Hedging quantity risks with standard power options in a competitive wholesale electricity market},
  journal = {Naval Research Logistics}, volume = {53}, pages = {697--712}, year = {2006}}

@article{Oum2010,
  author = {Oum, Yumi and Oren, Shmuel S.},
  title = {Optimal static hedging of volumetric risk in a competitive wholesale electricity market},
  journal = {Decision Analysis}, volume = {7}, pages = {107--122}, year = {2010}}

@book{Benth2008,
  author = {Benth, Fred Espen and Benth, Jurate Saltyte and Koekebakker, Steen},
  title = {Stochastic Modeling of Electricity and Related Markets},
  publisher = {World Scientific}, year = {2008}}

@article{BenthDetering2015,
  author = {Benth, Fred Espen and Detering, Nils},
  title = {Pricing and hedging {Asian}-style options on energy},
  journal = {Finance and Stochastics}, volume = {19}, pages = {849--889}, year = {2015}}

@article{Benth2019,
  author = {Benth, Fred Espen and Piccirilli, Marco and Vargiolu, Tiziano},
  title = {Mean-reverting additive energy forward curves in a {Heath-Jarrow-Morton} framework},
  journal = {Mathematics and Financial Economics}, volume = {13}, pages = {543--577}, year = {2019}}

@article{Benth2022,
  author = {Benth, Fred Espen and Deelstra, Griselda and K{\"o}zpinar, Sinem},
  title = {Pricing energy quanto options in the framework of {Markov}-modulated additive processes},
  journal = {IMA Journal of Management Mathematics}, volume = {34}, pages = {187--220}, year = {2022}}

@techreport{BenthLangeQuanto,
  author = {Benth, Fred Espen and Lange, Nina and Myklebust, Tor {\AA}ge},
  title = {Pricing and hedging quanto options in energy markets},
  institution = {SSRN id 2133935}, year = {2012}}

@article{SemistaticHedging,
  author = {Di Tella, Paolo and Haubold, Martin and Keller-Ressel, Martin},
  title = {Semi-static variance-optimal hedging in stochastic volatility models with {Fourier} representation},
  journal = {Journal of Applied Probability}, volume = {56}, pages = {787--809}, year = {2019}}

@article{semistaticsparse,
  author = {Di Tella, Paolo and Haubold, Martin and Keller-Ressel, Martin},
  title = {Semistatic and sparse variance-optimal hedging},
  journal = {Mathematical Finance}, volume = {30}, pages = {403--425}, year = {2020}}

@article{Schweizer1992,
  author = {Schweizer, Martin},
  title = {Mean-variance hedging for general claims},
  journal = {The Annals of Applied Probability}, volume = {2}, pages = {171--179}, year = {1992}}

@article{Hubalek.2006,
  author = {Hubalek, Friedrich and Kallsen, Jan and Krawczyk, Leszek},
  title = {Variance-optimal hedging for processes with stationary independent increments},
  journal = {The Annals of Applied Probability}, year = {2006}}

@article{Goutte2014,
  author = {Goutte, St{\'e}phane and Oudjane, Nadia and Russo, Francesco},
  title = {Variance optimal hedging for continuous time additive processes and applications},
  journal = {Stochastics}, volume = {86}, pages = {147--185}, year = {2014}}

@article{MarquardtStelzer2007,
  author = {Marquardt, Tina and Stelzer, Robert},
  title = {Multivariate {CARMA} processes},
  journal = {Stochastic Processes and their Applications}, volume = {117}, pages = {96--120}, year = {2007}}

@article{SchlemmStelzerQMLE2012,
  author = {Schlemm, Eckhard and Stelzer, Robert},
  title = {Quasi maximum likelihood estimation for strongly mixing state space models and multivariate {L}{\'e}vy-driven {CARMA} processes},
  journal = {Electronic Journal of Statistics}, volume = {6}, pages = {2185--2234}, year = {2012}}

@article{BrockwellSchlemm2013,
  author = {Brockwell, Peter J. and Schlemm, Eckhard},
  title = {Parametric estimation of the driving {L}{\'e}vy process of multivariate {CARMA} processes from discrete observations},
  journal = {Journal of Multivariate Analysis}, volume = {115}, pages = {217--251}, year = {2013}}

@article{BenthKarbach2023MCARMACones,
  author = {Benth, Fred Espen and Karbach, Sven},
  title = {Multivariate continuous-time autoregressive moving-average processes on cones},
  journal = {Stochastic Processes and their Applications}, volume = {162}, pages = {299--337}, year = {2023}}

@article{BessembinderLemmon2002,
  author = {Bessembinder, Hendrik and Lemmon, Michael L.},
  title = {Equilibrium pricing and optimal hedging in electricity forward markets},
  journal = {The Journal of Finance}, volume = {57}, pages = {1347--1382}, year = {2002}}

@article{LongstaffWang2004,
  author = {Longstaff, Francis A. and Wang, Ashley W.},
  title = {Electricity forward prices: a high-frequency empirical analysis},
  journal = {The Journal of Finance}, volume = {59}, pages = {1877--1900}, year = {2004}}

@article{Ketterer2014,
  author = {Ketterer, Janina C.},
  title = {The impact of wind power generation on the electricity price in {Germany}},
  journal = {Energy Economics}, volume = {44}, pages = {270--280}, year = {2014}}

@article{Clo2015,
  author = {Cl{\`o}, Stefano and Cataldi, Alessandra and Zoppoli, Pietro},
  title = {The merit-order effect in the {Italian} power market: the impact of solar and wind generation on national wholesale electricity prices},
  journal = {Energy Policy}, volume = {77}, pages = {79--88}, year = {2015}}

@article{Rintamaki2017,
  author = {Rintam{\"a}ki, Tuomas and Siddiqui, Afzal S. and Salo, Ahti},
  title = {Does renewable energy generation decrease the volatility of electricity prices? An analysis of {Denmark} and {Germany}},
  journal = {Energy Economics}, volume = {62}, pages = {270--282}, year = {2017}}

@article{kloster2025ambit,
  title={An Ambit Field Framework for the Full Panel of Day-ahead Electricity Prices},
  author={Kloster, Thomas K},
  journal={arXiv preprint arXiv:2509.17236},
  year={2025}
}

@article{GABRIELLI2022105980,
title = {Mitigating financial risk of corporate power purchase agreements via portfolio optimization},
journal = {Energy Economics},
volume = {109},
pages = {105980},
year = {2022},
issn = {0140-9883},
doi = {https://doi.org/10.1016/j.eneco.2022.105980},
url = {https://www.sciencedirect.com/science/article/pii/S0140988322001542},
author = {Paolo Gabrielli and Reyhaneh Aboutalebi and Giovanni Sansavini}
}

@misc{ENTSOE2025EnergyPrices,
  author       = {{European Network of Transmission System Operators
                   for Electricity (ENTSO-E)}},
  title        = {Energy Prices [12.1.D]},
  year         = {2025},
  howpublished = {\url{https://transparencyplatform.zendesk.com/hc/en-us/articles/16647234190100-Energy-Prices-12-1-D}},
}

@misc{ENTSOE2026ActualLoad,
  author       = {{European Network of Transmission System Operators
                   for Electricity (ENTSO-E)}},
  title        = {Actual Total Load and Day-Ahead Load per Bidding Zone
                   [6.1.A and 6.1.B]},
  year         = {2026},
  howpublished = {\href{https://transparencyplatform.zendesk.com/hc/en-us/articles/16647979768084-Actual-Total-Load-Day-ahead-Per-Bidding-Zone-6-1-A-6-1-B}
                   {\texttt{ENTSO-E load documentation}}},
}

@misc{ENWEX2026Data,
  author       = {{enwex GmbH}},
  title        = {Enwex Germany Historical Wind and Solar Index Data},
  year         = {2026},
  howpublished = {\url{https://enwex.com/getdataapitest/}},
}

@article{Schweizer1994,
  author = {Schweizer, Martin},
  title = {Approximating random variables by stochastic integrals},
  journal = {The Annals of Probability}, volume = {22}, number = {3}, pages = {1536--1575}, year = {1994}}

@incollection{Galchouk1976,
  author = {Galchouk, Leonid I.},
  title = {The structure of certain martingales},
  booktitle = {Proceedings of the School and Seminar on the Theory of Random Processes (Druskininkai, 1974), Part I},
  pages = {7--32},
  publisher = {Institute of Mathematics and Cybernetics, Academy of Sciences of the Lithuanian SSR},
  address = {Vilnius}, year = {1975}}

@article{KunitaWatanabe1967,
  author = {Kunita, Hiroshi and Watanabe, Shinzo},
  title = {On square integrable martingales},
  journal = {Nagoya Mathematical Journal}, volume = {30}, pages = {209--245}, year = {1967}}

@incollection{FollmerSondermann1985,
  author = {F{\"o}llmer, Hans and Sondermann, Dieter},
  title = {Hedging of non-redundant contingent claims},
  booktitle = {Contributions to Mathematical Economics},
  pages = {205--223}, publisher = {North-Holland}, address = {Amsterdam}, year = {1986}}

@article{Pham2000,
  author = {Pham, Huy{\^e}n},
  title = {On quadratic hedging in continuous time},
  journal = {Mathematical Methods of Operations Research}, volume = {51}, pages = {315--339}, year = {2000}}

@article{CernyKallsen2007,
  author = {{\v C}ern{\'y}, Ale{\v s} and Kallsen, Jan},
  title = {On the structure of general mean-variance hedging strategies},
  journal = {The Annals of Probability}, volume = {35}, number = {4}, pages = {1479--1531}, year = {2007}}

@article{KallsenPauwels2010,
  author = {Kallsen, Jan and Pauwels, Arnd},
  title = {Variance-optimal hedging in general affine stochastic volatility models},
  journal = {Advances in Applied Probability}, volume = {42}, number = {1}, pages = {83--105}, year = {2010}}

@incollection{CarrMadan2001,
  author = {Carr, Peter and Madan, Dilip B.},
  title = {Towards a theory of volatility trading},
  booktitle = {Option Pricing, Interest Rates and Risk Management},
  publisher = {Cambridge University Press}, pages = {458--476}, year = {2001}}

@article{CarrCorso2001,
  author = {Carr, Peter and Corso, Anthony},
  title = {Commodity covariance contracting},
  journal = {Energy and Power Risk Management}, volume = {4}, year = {2001}}

@article{Madan2021Pricing,
  author = {Madan, Dilip B. and Wang, King},
  title = {Pricing product options and using them to complete markets for functions of two underlying asset prices},
  journal = {Journal of Risk and Financial Management}, volume = {14}, number = {8}, pages = {1--19}, year = {2021}}

@article{chatziandreou2026semistaticvarianceoptimalhedgingcovariance,
      title={Semi-Static Variance-Optimal Hedging of Covariance Risk in Multi-Asset Derivatives}, 
      author={Konstantinos Chatziandreou and Sven Karbach},
      year={2026},
      eprint={2603.25320},
      archivePrefix={arXiv},
      primaryClass={q-fin.MF},
      url={https://arxiv.org/abs/2603.25320}, 
}

@article{Eberlein2010,
  author = {Eberlein, Ernst and Glau, Kathrin and Papapantoleon, Antonis},
  title = {Analysis of {Fourier} transform valuation formulas and applications},
  journal = {Applied Mathematical Finance}, volume = {17}, number = {3}, pages = {211--240}, year = {2010}}

@article{SchlemmStelzer2012,
  author = {Schlemm, Eckhard and Stelzer, Robert},
  title = {Multivariate {CARMA} processes, continuous-time state space models and complete regularity of the innovations of the sampled processes},
  journal = {Bernoulli}, volume = {18}, number = {1}, pages = {46--63}, year = {2012}}

@article{Brockwell2001CARMA,
  author = {Brockwell, Peter J.},
  title = {L{\'e}vy-driven {CARMA} processes},
  journal = {Annals of the Institute of Statistical Mathematics}, volume = {53}, pages = {113--124}, year = {2001}}

@phdthesis{Rowinska2020,
  author = {Rowi{\'n}ska, Paulina A.},
  title = {Stochastic modelling and statistical inference for electricity prices, wind energy production and wind speed},
  school = {Imperial College London}, year = {2020}}

@book{Bremaud1981,
  author = {Br{\'e}maud, Pierre},
  title = {Point Processes and Queues: Martingale Dynamics},
  publisher = {Springer}, address = {New York}, year = {1981}}

@article{BarndorffNielsen1997,
  author = {Barndorff-Nielsen, Ole E.},
  title = {Normal inverse {Gaussian} distributions and stochastic volatility modelling},
  journal = {Scandinavian Journal of Statistics}, volume = {24}, number = {1}, pages = {1--13}, year = {1997}}

@article{Paraschiv2015,
  author = {Paraschiv, Florentina and Fleten, Stein-Erik and Schürle, Michael},
  title = {A spot-forward model for electricity prices with regime shifts},
  journal = {Energy Economics}, volume = {47}, pages = {142--153}, year = {2015}}

@article{MATSUMOTO2021105101,
title = {Simultaneous hedging strategy for price and volume risks in electricity businesses using energy and weather derivatives1},
journal = {Energy Economics},
volume = {95},
pages = {105101},
year = {2021},
issn = {0140-9883},
doi = {https://doi.org/10.1016/j.eneco.2021.105101},
url = {https://www.sciencedirect.com/science/article/pii/S0140988321000062},
author = {Takuji Matsumoto and Yuji Yamada}
}

@Article{en19133044,
AUTHOR = {Matsumoto, Takuji and Yamada, Yuji},
TITLE = {Hedging Shape Risk in Renewable Energy Markets: Empirical Evidence on Hedging Effectiveness Using the Quality Factor (QF) Index},
JOURNAL = {Energies},
VOLUME = {19},
YEAR = {2026},
NUMBER = {13},
ARTICLE-NUMBER = {3044},
URL = {https://www.mdpi.com/1996-1073/19/13/3044},
ISSN = {1996-1073},
DOI = {10.3390/en19133044}
}

@article{LUCY2021105603,
title = {Analysis of fixed volume swaps for hedging financial risk at large-scale wind projects},
journal = {Energy Economics},
volume = {103},
pages = {105603},
year = {2021},
issn = {0140-9883},
doi = {https://doi.org/10.1016/j.eneco.2021.105603},
url = {https://www.sciencedirect.com/science/article/pii/S0140988321004710},
author = {Zachary Lucy and Jordan Kern}
}

\appendix

\clearpage
\onecolumn

\begingroup

\setlength{\parindent}{0pt}
\setlength{\parskip}{1.2pt}
\setlength{\tabcolsep}{3pt}
\renewcommand{\arraystretch}{0.92}

\captionsetup{
  font=scriptsize,
  labelfont=bf,
  skip=2pt,
  justification=justified,
  singlelinecheck=false
}

\section{Additional calibration and out-of-sample diagnostics}
\label{supp:additional-calibration-oos}

\vspace{-0.65em}

\begin{minipage}[t]{0.492\textwidth}
\scriptsize

\textbf{Wind tail dependence.}
The empirical and simulated conditional tail-dependence curves increase
monotonically in the quantile level $q$ in all four orthants. The simulated curves
reproduce the empirical ordering of the four quadrants. The low-wind/high-price and
high-wind/low-price quadrants are especially relevant because they correspond to large joint
price-volume deviations, whereas the high-wind/high-price and low-wind/low-price curves
describe the two concordant tails. The realized high-wind/low-price
curve is somewhat above the simulated median at high $q$. At high
quantiles, the simulated median understates the realized high-wind/low-price
dependence, although the realized curve remains within the simulated
uncertainty band.

\textbf{Wind out-of-sample dependence.}
The out-of-sample diagnostics cover January--December 2025, which is not part
of the calibration window. Out-of-sample simulations using the recovered
MNIG driver reproduce the main autocorrelation pattern. For both the latent wind
state and the full-level wind pressure, the realized 2025 autocorrelation
decays more slowly than the simulated median but remains largely within the simulated band. This
indicates greater persistence in 2025; separately, the year also had below-average wind levels. The realized and
simulated spot autocorrelations both show damped daily oscillations, although some realized
calendar peaks are sharper than the fitted median.

\end{minipage}
\hfill
\begin{minipage}[t]{0.492\textwidth}
\scriptsize

\textbf{Wind out-of-sample levels.}
The simulated spot-price bands contain the realized monthly and weekly means
for most of the delivery window, including the decline from the high winter
level to the lower spring and early-summer levels. The realized wind pressure
$\Pi_t^{\mathrm W}$ lies below the simulated median for much of 2025 and
frequently remains close to the lower edge of the band. The 2025 delivery
period therefore represents an observed low-wind stress year rather than
a large shift in mean spot prices. This matters because wind production
affects both delivered volume and production-weighted revenue.

\textbf{Wind capture-price fit.}
The simulated monthly wind capture-price median is close to the realized level, with median
values in the same range as the realized capture prices and pointwise bands
that contain the January--December variation. The realized wind value factor is
comparatively stable, mostly between $0.70$ and $0.95$, while the simulated
median remains near $0.88$--$0.90$. The fitted model reproduces the
approximate average wind capture discount, and its simulated bands cover much of the
monthly variation in the capture
wedge $CP_P^{\mathrm W}-BL_P$.

\end{minipage}

\vspace{0.35em}

\begin{minipage}{0.96\textwidth}
\raggedright
\scriptsize\textbf{(a) Out-of-sample autocorrelation fit}\\[-0.15em]
\centering
\includegraphics[
  width=\linewidth,
  height=0.205\textheight,
  keepaspectratio
]{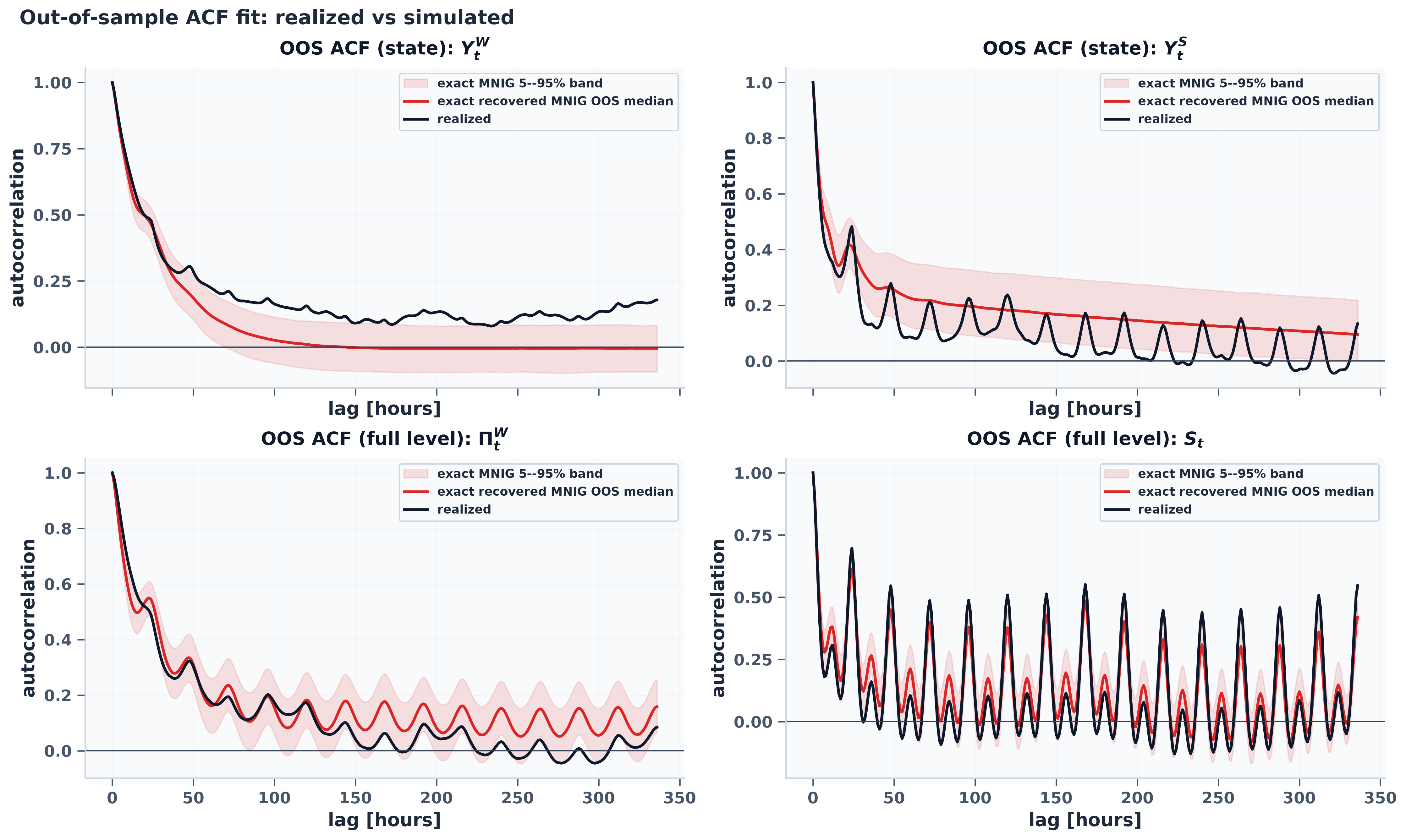}
\end{minipage}

\vspace{0.25em}

\begin{minipage}[t]{0.322\textwidth}
\raggedright
\scriptsize\textbf{(b) Monthly mean levels}\\[-0.15em]
\centering
\includegraphics[
  width=\linewidth,
  height=0.125\textheight,
  keepaspectratio
]{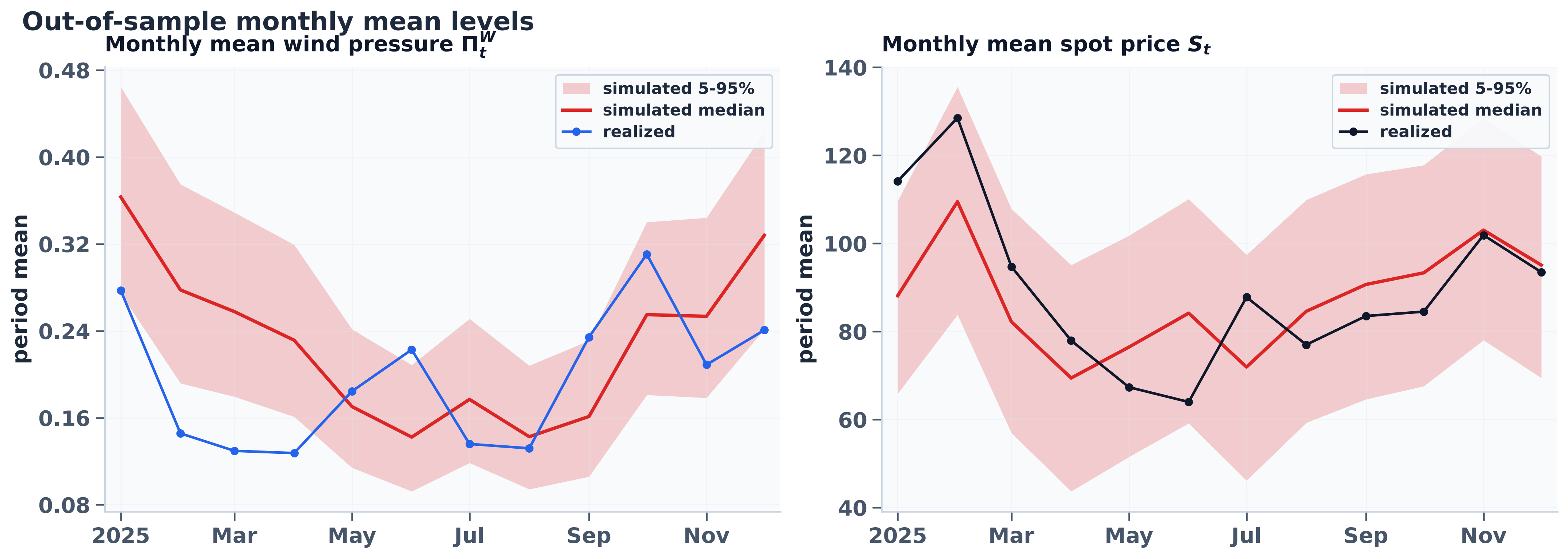}
\end{minipage}
\hfill
\begin{minipage}[t]{0.322\textwidth}
\raggedright
\scriptsize\textbf{(c) Weekly mean levels}\\[-0.15em]
\centering
\includegraphics[
  width=\linewidth,
  height=0.125\textheight,
  keepaspectratio
]{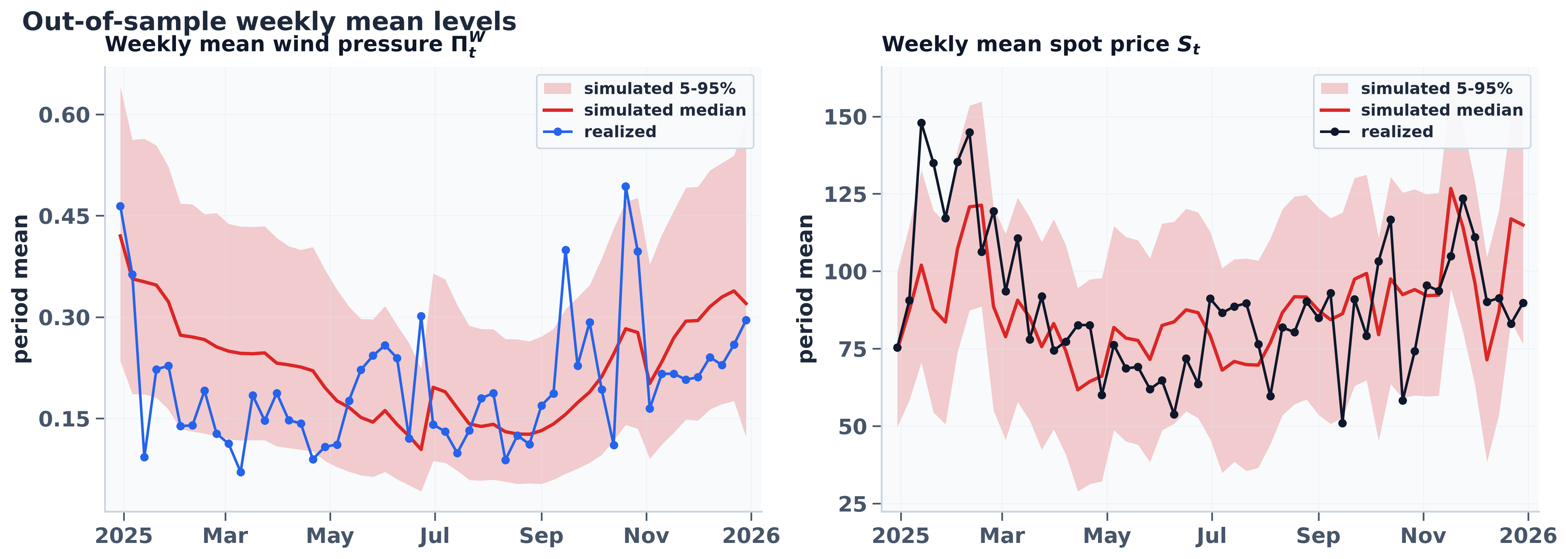}
\end{minipage}
\hfill
\begin{minipage}[t]{0.322\textwidth}
\raggedright
\scriptsize\textbf{(d) Capture price and value factor}\\[-0.15em]
\centering
\includegraphics[
  width=\linewidth,
  height=0.125\textheight,
  keepaspectratio
]{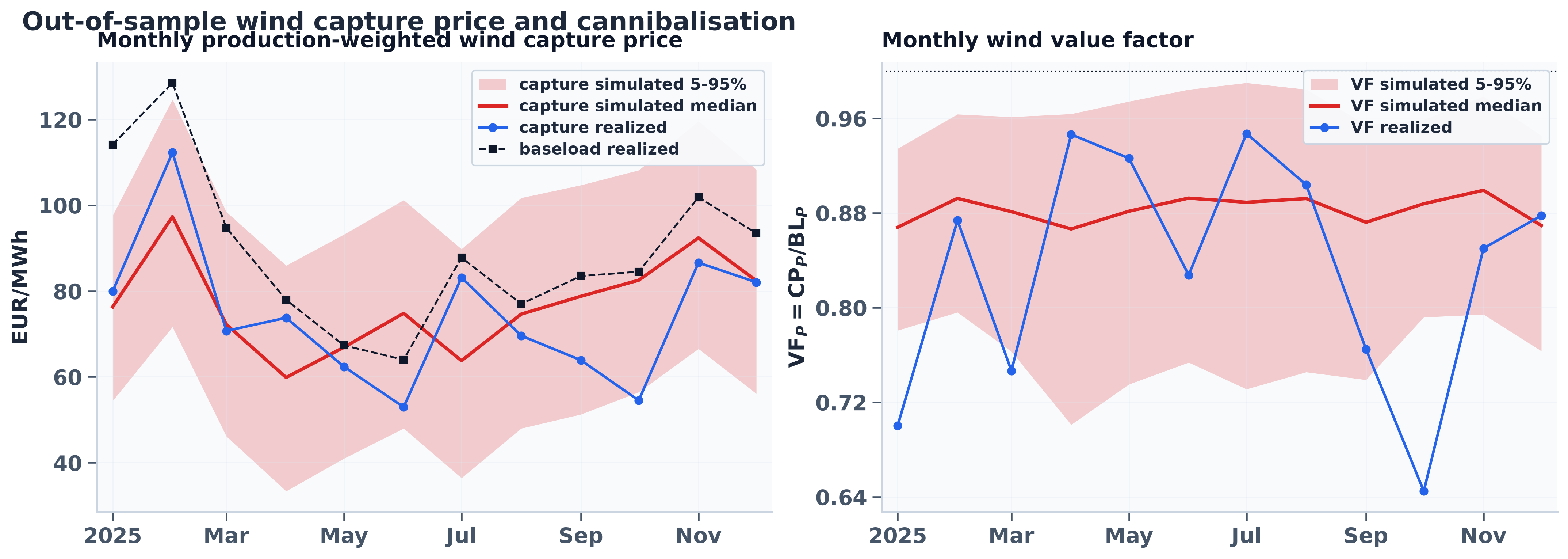}
\end{minipage}

\vspace{0.25em}

\begin{minipage}{0.96\textwidth}
\raggedright
\scriptsize\textbf{(e) Full-level quantile dependence}\\[-0.15em]
\centering
\includegraphics[
  width=\linewidth,
  height=0.145\textheight,
  keepaspectratio
]{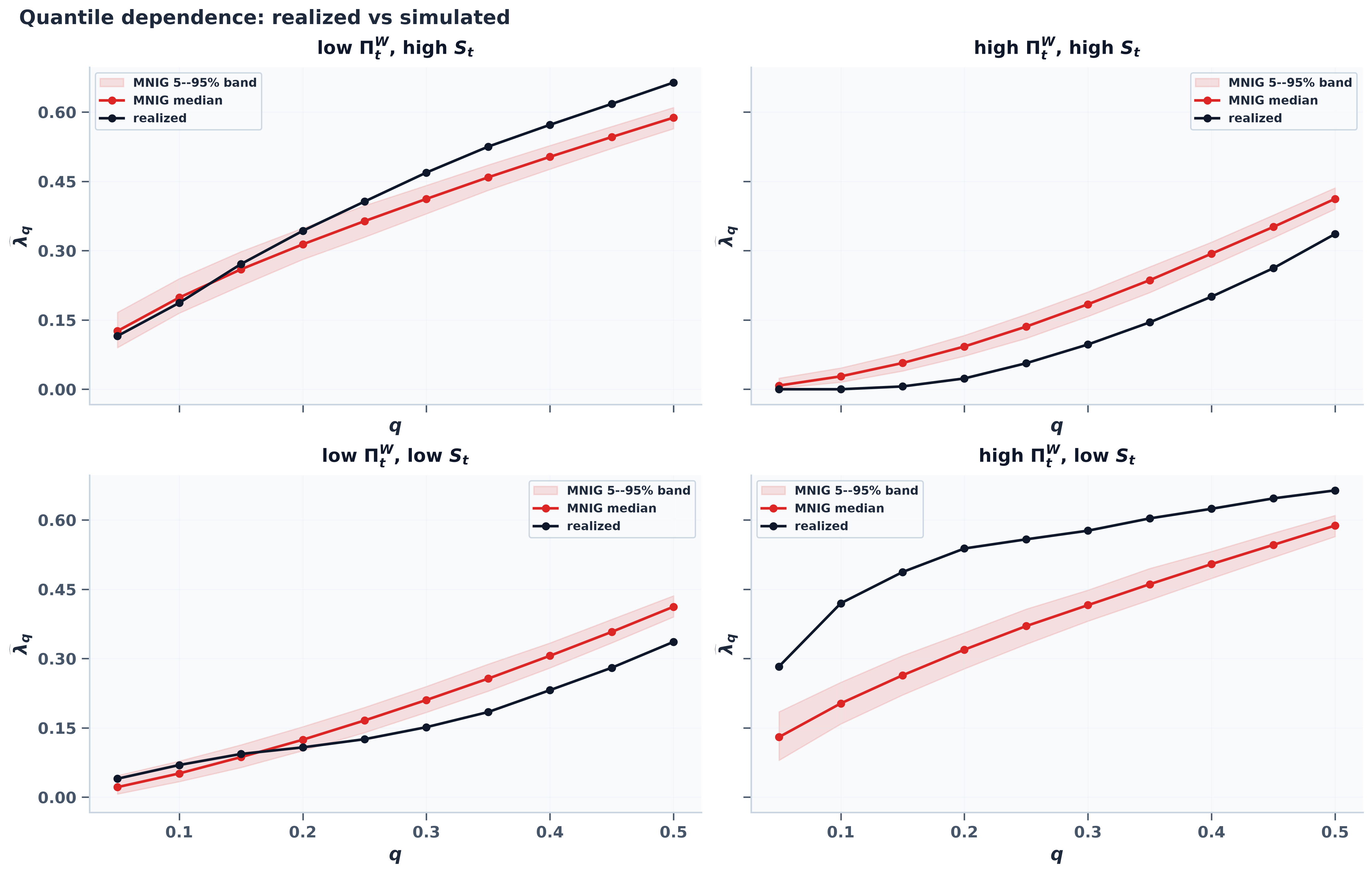}
\end{minipage}

\captionof{figure}{
Wind calibration: additional dependence and out-of-sample diagnostics.
Panel~(a) reports the out-of-sample autocorrelation comparison using simulations
from the recovered MNIG driver; the upper row is the deseasonalized state
$(Y_t^{\mathrm W},Y_t^S)$ and the lower row is the full-level pair
$(\Pi_t^{\mathrm W},S_t)$. Panels~(b) and~(c) compare realized and
simulated monthly and weekly mean levels over January--December 2025.
Panel~(d) reports the monthly production-weighted capture price and value
factor, with the parity level indicated in the value-factor panel.
Panel~(e) reports the full-level quantile dependence of
$(\Pi_t^{\mathrm W},S_t)$ across the four wind--price orthants, together
with the simulated MNIG $5$--$95\%$ bands.
}
\label{fig:cal-wind-oos-composite}
\label{fig:cal-wind-oos-acf}
\label{fig:cal-wind-oos-monthly}
\label{fig:cal-wind-oos-weekly}
\label{fig:cal-wind-oos-capture}
\label{fig:cal-wind-mnig-quantile}

\newpage

\subsection*{Solar calibration and out-of-sample diagnostics}

\vspace{-0.55em}

\begin{minipage}[t]{0.492\textwidth}
\scriptsize

\textbf{One-step physical-scale fit.}
After applying the decreasing inverse production map, the solar capacity
factor has an RMSE of $0.0168$, a correlation of $0.9931$, and a bias of
$-0.0003$. The spot-price equation has an RMSE of approximately
$11.80$~EUR/MWh and a correlation of $0.9730$. The similar realized and
fitted standard deviations show that the nonlinear reconstruction
preserves the marginal scale reasonably well on the physical
scale used for PPA valuation.

\textbf{Solar tail dependence.}
The fitted common-subordinator MNIG driver reproduces several features
of the empirical tail-dependence curves. In the low-solar/high-price and
low-solar/low-price orthants, the simulated median is close
to the realized curve beyond $q\approx0.25$, indicating a good
fit for these two lower-solar tail-dependence measures. In
the high-solar/high-price and high-solar/low-price orthants, the model
reproduces the increasing tail dependence and the realized curve lies close to the simulated band. These curves indicate that adverse solar PPA outcomes can
involve joint production and price movements rather than only extreme marginal prices.

\end{minipage}
\hfill
\begin{minipage}[t]{0.492\textwidth}
\scriptsize

\textbf{Solar out-of-sample dependence and levels.}
On the latent scale, the realized solar-shortfall state has sharper daily
autocorrelation peaks than the simulated median in 2025, indicating
stronger residual daily persistence. Once the deterministic clear-sky envelope is
restored, however, the fitted model closely matches the main full-level
solar autocorrelation pattern. The spot panels show a similar overall decay
and daily periodicity to the wind fit, with stronger realized peaks at
selected daily lags. Monthly and weekly realized means of
$C_t^{\mathrm S}$ closely follow the simulated median from winter through the
summer maximum and remain inside the simulated bands. The realized spot-price mean also
lies within the simulated band for most periods.

\textbf{Solar capture-price fit.}
The simulated median value factor declines from values close to one in winter
to approximately two thirds of baseload during summer. The realized 2025
value factor falls more sharply, reaching approximately one third in late
spring and early summer, partially rebounds in July, remains depressed through
September, and returns close to or above parity in December. The model reproduces the seasonal direction of
the solar capture discount but understates the severity of the realized spring--summer decline. The
realized summer discount lies in the adverse tail of the distribution
fitted to 2023--2024. The hedging experiment tests performance
under this unusually severe out-of-sample outcome.

\end{minipage}

\vspace{0.25em}

\begin{adjustbox}{
  max width=0.82\textwidth,
  max totalheight=0.105\textheight,
  center
}
\begin{minipage}{\textwidth}
\begingroup

\renewenvironment{table}[1][]{%
  \captionsetup{type=table}%
  \centering
  \scriptsize
}{}

\renewenvironment{table*}[1][]{%
  \captionsetup{type=table}%
  \centering
  \scriptsize
}{}

\setlength{\columnwidth}{\linewidth}

\begin{table}[!htbp]
\centering
\caption{Solar in-sample one-step fit on the full-level variables $(C_t^{\mathrm S},S_t)$.}
\label{tab:cal_solar-insample-one-step-fit}
\small
\begin{tabular}{lrrrrrrr}
\toprule
Component & $N$ & RMSE & MAE & Bias & Realized std. & Simulated std. & Correlation \\
\midrule
$C_t^{\mathrm S}$ & 17541.000 & 0.016751 & 0.008503 & -0.000289 & 0.143264 & 0.142669 & 0.993146 \\
$S_t$ & 17541.000 & 11.8046 & 7.25276 & -0.005173 & 50.9051 & 50.6319 & 0.972980 \\
\bottomrule
\end{tabular}
\end{table}

\endgroup
\end{minipage}
\end{adjustbox}

\vspace{0.25em}

\begin{minipage}{0.96\textwidth}
\raggedright
\scriptsize\textbf{(a) Out-of-sample autocorrelation fit}\\[-0.15em]
\centering
\includegraphics[
  width=\linewidth,
  height=0.185\textheight,
  keepaspectratio
]{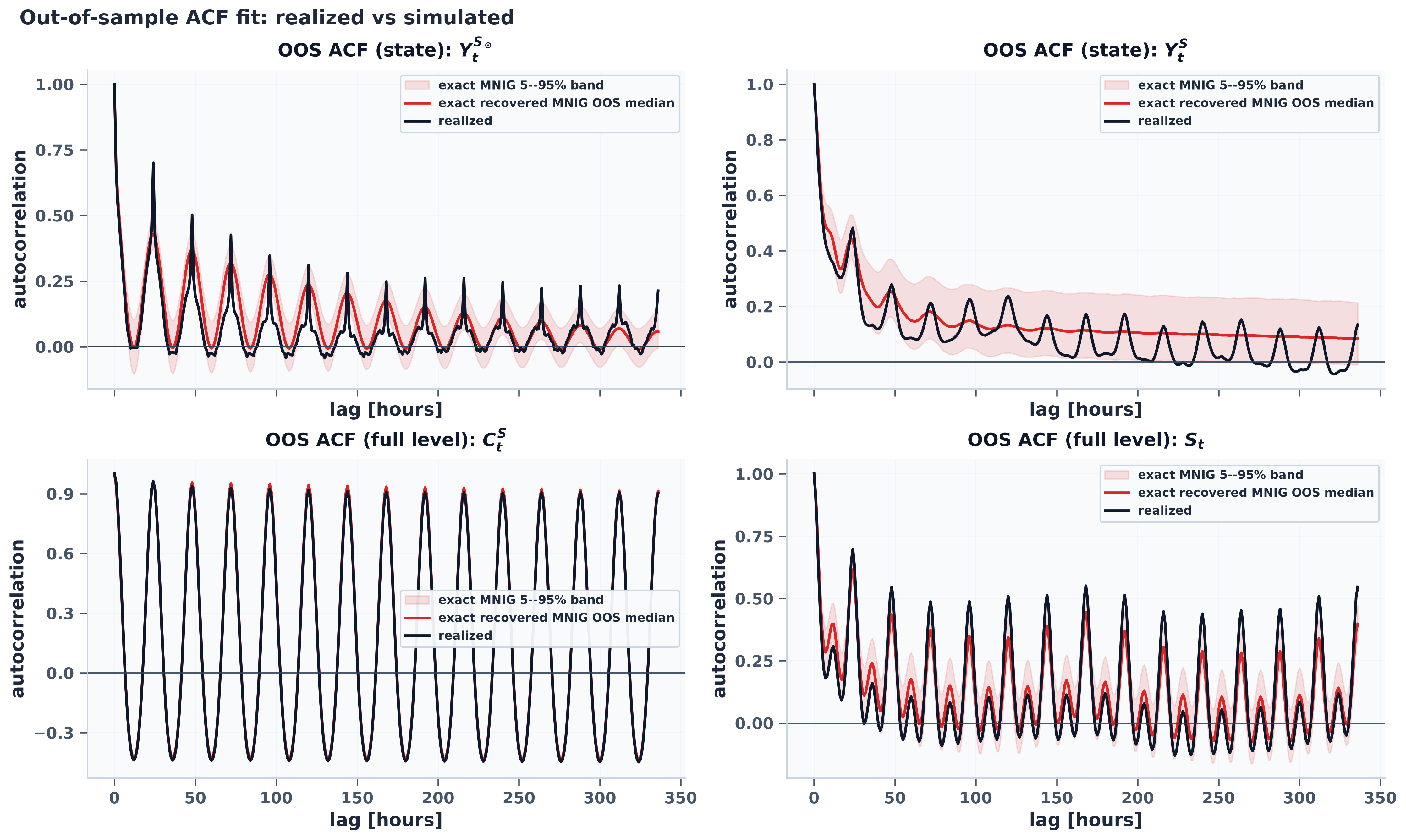}
\end{minipage}

\vspace{0.22em}

\begin{minipage}[t]{0.322\textwidth}
\raggedright
\scriptsize\textbf{(b) Monthly mean levels}\\[-0.15em]
\centering
\includegraphics[
  width=\linewidth,
  height=0.105\textheight,
  keepaspectratio
]{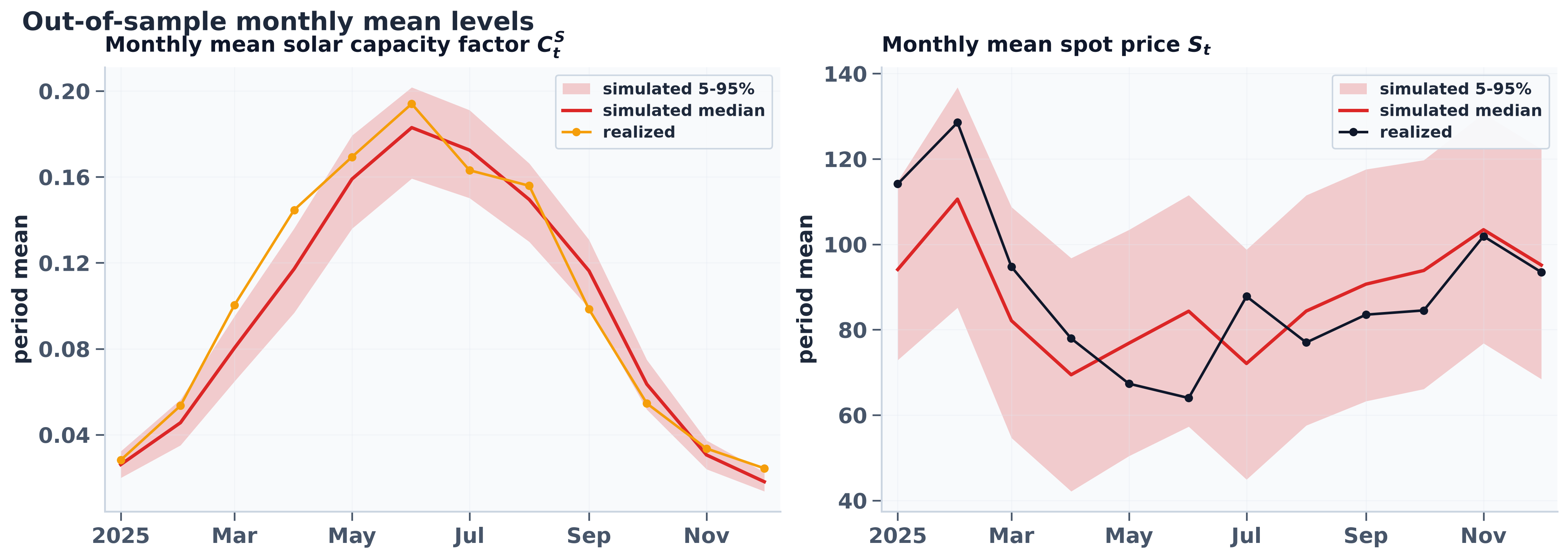}
\end{minipage}
\hfill
\begin{minipage}[t]{0.322\textwidth}
\raggedright
\scriptsize\textbf{(c) Weekly mean levels}\\[-0.15em]
\centering
\includegraphics[
  width=\linewidth,
  height=0.105\textheight,
  keepaspectratio
]{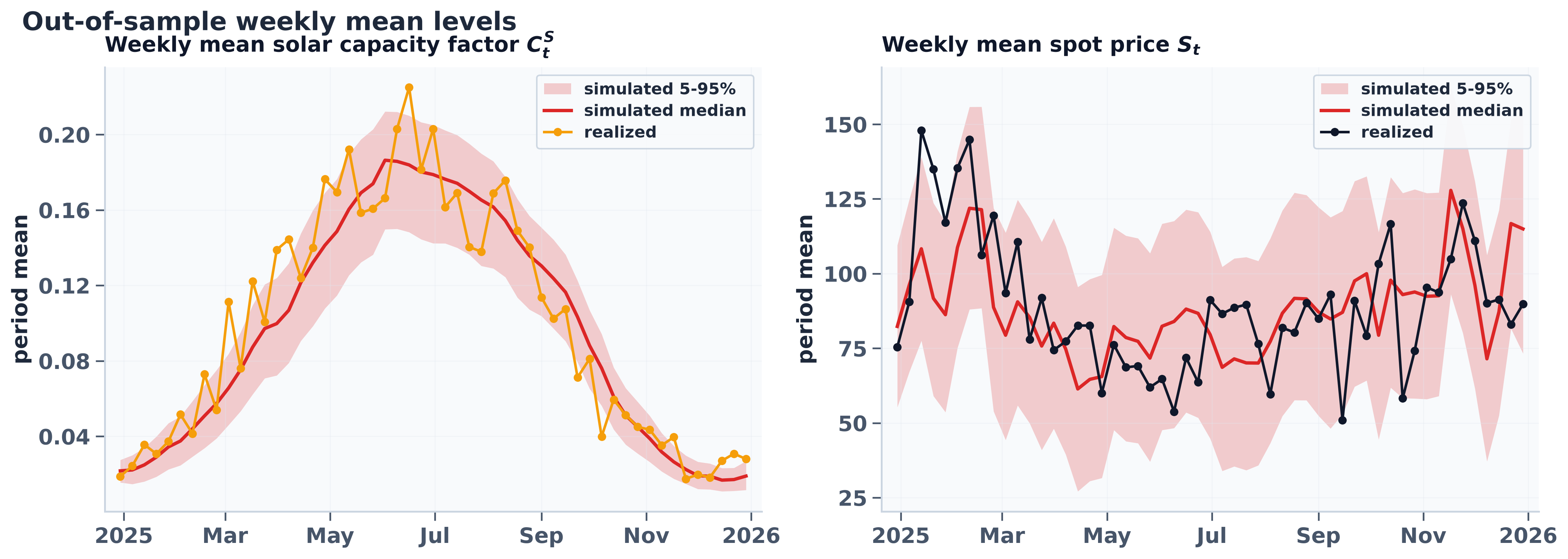}
\end{minipage}
\hfill
\begin{minipage}[t]{0.322\textwidth}
\raggedright
\scriptsize\textbf{(d) Capture price and value factor}\\[-0.15em]
\centering
\includegraphics[
  width=\linewidth,
  height=0.105\textheight,
  keepaspectratio
]{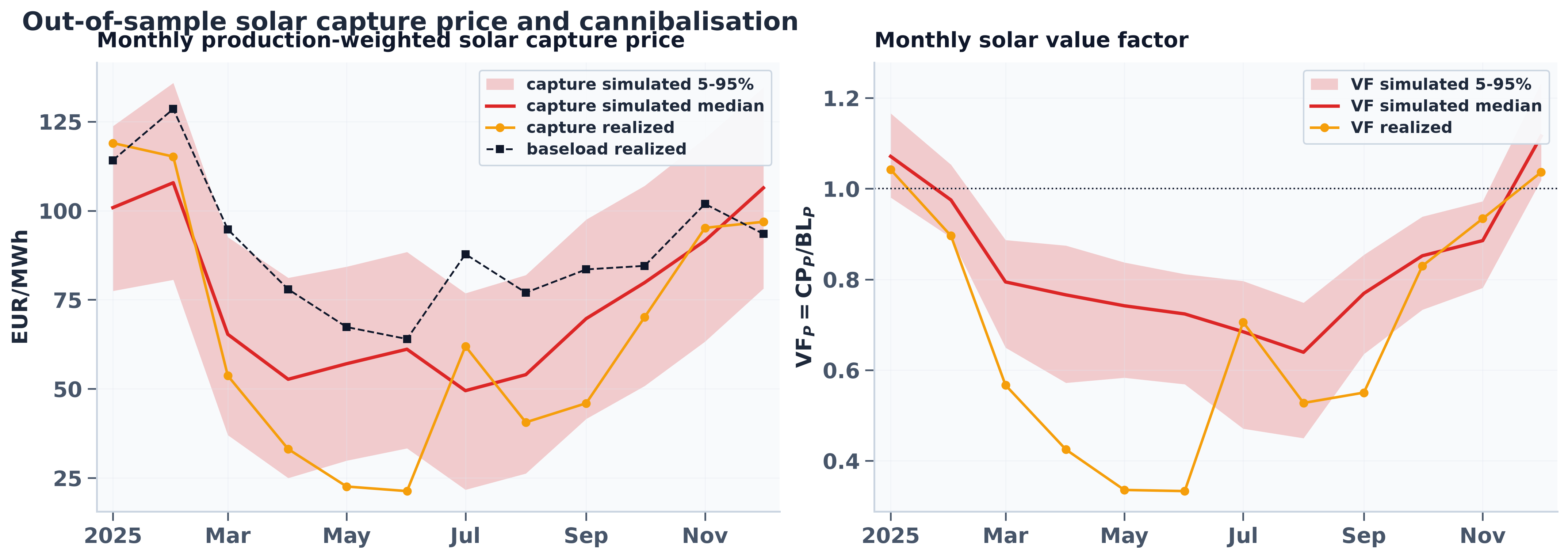}
\end{minipage}

\vspace{0.22em}

\begin{minipage}{0.96\textwidth}
\raggedright
\scriptsize\textbf{(e) Full-level quantile dependence}\\[-0.15em]
\centering
\includegraphics[
  width=\linewidth,
  height=0.125\textheight,
  keepaspectratio
]{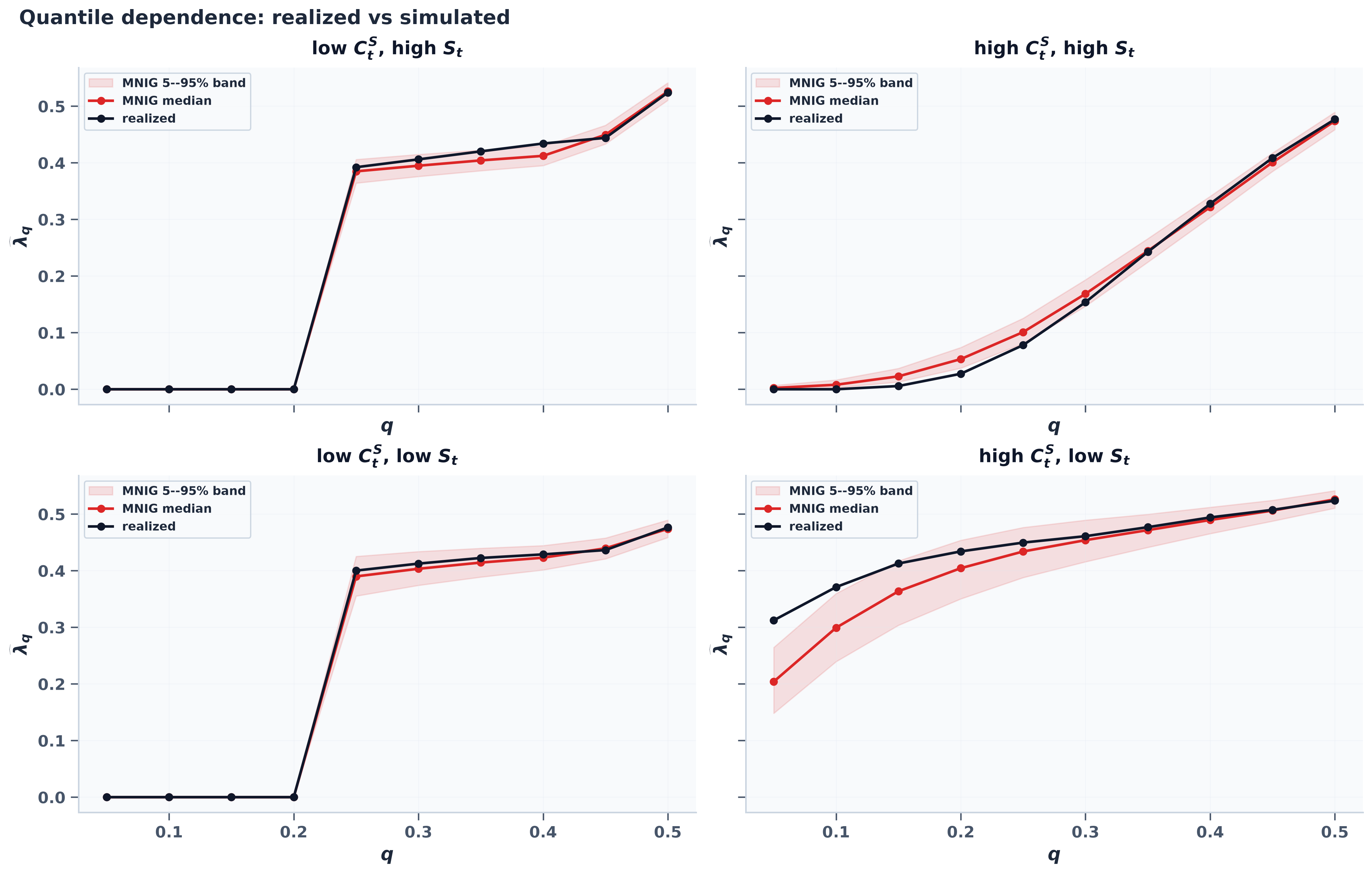}
\end{minipage}

\captionof{figure}{
Solar calibration: additional dependence and out-of-sample diagnostics.
Panel~(a) reports the out-of-sample autocorrelation fit under the exact
recovered MNIG driver; the upper row is the deseasonalized state
$(Y_t^{\mathrm S_\odot},Y_t^S)$ and the lower row is the full-level pair
$(C_t^{\mathrm S},S_t)$. Panels~(b) and~(c) compare realized and simulated
monthly and weekly mean levels over January--December 2025. Panel~(d) reports
the monthly production-weighted capture price and value factor; the realized
spring-summer value factor lies in the lower tail of the fitted
distribution and reflects unusually severe cannibalisation.
Panel~(e) reports the full-level quantile dependence of
$(C_t^{\mathrm S},S_t)$ across the four solar--price orthants.
}
\label{fig:cal-solar-oos-composite}
\label{fig:cal-solar-oos-acf}
\label{fig:cal-solar-oos-monthly}
\label{fig:cal-solar-oos-weekly}
\label{fig:cal-solar-oos-capture}
\label{fig:cal-solar-mnig-quantile}

\endgroup

\clearpage
\twocolumn

\end{document}